%% file: main.tex
\documentclass[12pt]{article}

\usepackage{amsmath,amssymb,amsfonts,amsthm}	
\usepackage{hyperref}							
\usepackage{orcidlink}                          

\usepackage[ngerman,american]{babel}
\usepackage[T1]{fontenc}			
\usepackage[utf8]{inputenc}			
\usepackage{fourier}				
\usepackage{mathtools}              
\usepackage{etoc}                   

\usepackage[onehalfspacing]{setspace}	

\usepackage{tikz}		
\usepackage{graphicx}	

\usepackage{geometry}
\usepackage{fancyhdr}
\usepackage{caption}
\usepackage{booktabs}
\usepackage{minipage-marginpar}

\theoremstyle{plain}				            
\newtheorem{theorem}{Theorem}[section]          
\newtheorem{corollary}{Corollary}[section]      
\theoremstyle{definition}			            
\newtheorem{definition}{Definition}[section]	
\newtheorem{assumption}{Assumption}[section]	
\newtheorem{lemma}{Lemma}[section]	            
\newtheorem{algorithm}{Algorithm}               

\usepackage{enumitem}

\usepackage[figure]{hypcap}
\usepackage{pdfpages}
\usepackage{colortbl} 
\usepackage{xcolor}
\usepackage{placeins}       
\usepackage{subcaption}
\usepackage{tabularx}
\usepackage{appendix}
\usepackage{array}

\DeclareFontFamily{U}{mathx}{}
\DeclareFontShape{U}{mathx}{m}{n}{<-> mathx10}{}
\DeclareSymbolFont{mathx}{U}{mathx}{m}{n}
\DeclareMathAccent{\widecheck}{0}{mathx}{"71}

\makeatletter
\newcommand*\rel@kern[1]{\kern#1\dimexpr\macc@kerna}
\newcommand*\widebar[1]{%
	\begingroup
	\def\mathaccent##1##2{%
		\rel@kern{0.8}%
		\overline{\rel@kern{-0.8}\macc@nucleus\rel@kern{0.2}}%
		\rel@kern{-0.2}%
	}%
	\macc@depth\@ne
	\let\math@bgroup\@empty \let\math@egroup\macc@set@skewchar
	\mathsurround\z@ \frozen@everymath{\mathgroup\macc@group\relax}%
	\macc@set@skewchar\relax
	\let\mathaccentV\macc@nested@a
	\macc@nested@a\relax111{#1}%
	\endgroup
}
\makeatother

\newcommand{\E}{\mathbb{E}}

\renewcommand{\vec}{\text{vec}}
\renewcommand{\u}{\boldsymbol{u}}

\usepackage[authoryear]{natbib}	
\usepackage{chapterbib}         
\usepackage{authblk}  
\title{
    \huge{Inference in Panel SVARs \\
    with Two-Way Dependence}
}
\author{%
    \begin{minipage}[t]{0.46\textwidth}
    \centering
    Lennart Empting*~\orcidlink{0009-0004-5068-4639}\\ \vspace{-2.75ex}
    {\footnotesize University of Duisburg-Essen}
    \end{minipage}
    \hspace{0.01\textwidth}
    \begin{minipage}[t]{0.46\textwidth}
    \centering
    Saskia Öztürk~\orcidlink{0009-0008-1347-3313}\\ \vspace{-2.75ex}
    {\footnotesize Georg August University of Göttingen}
    \end{minipage}

    \vspace{1.5ex}

    \begin{minipage}[t]{0.46\textwidth}
    \centering
    Simone Maxand~\orcidlink{0000-0002-3153-7922}\\ \vspace{-2.75ex}
    {\footnotesize Leipzig University}
    \end{minipage}
    \hspace{0.01\textwidth}
    \begin{minipage}[t]{0.46\textwidth}
    \centering
    Konstantin Wagner~\orcidlink{0009-0006-0216-0755}\\ \vspace{-2.75ex}
    {\footnotesize Humboldt University of Berlin}
    \end{minipage}
    
    \vspace{1.5ex}
    
    {\footnotesize * Corresponding author. Universitätsstraße 12, D-45141 Essen. Email: \textit{lennart.empting@vwl.uni-due.de}}
}
\date{\today}

\begin{document}

\maketitle				
\input{0_Abstract}
\pagenumbering{arabic}	
\thispagestyle{empty}	


\include{MainBody}

\clearpage			
\appendix

\setcounter{page}{1}

\include{Appendix}

\end{document}

%% file: 0_Abstract.tex
\hrule
\begin{abstract}
    We develop inference for heterogeneous panel vector autoregressive (VAR) models and their structural impulse response functions, where the error terms are dependent in the cross-sectional and time dimensions (two-way dependence). For proxy-identified structural VARs, we first adapt mean-group estimation and construct a closed-form pooled identification. Considering the reduced-form VAR dynamics, residual covariances, and structural parameters, we derive a joint central limit theorem under joint limits in the cross-sectional and time dimensions. We then propose a recursive-design panel moving-block bootstrap that resamples the estimated error terms in (i) the temporal, (ii) the cross-sectional, or (iii) both dimensions jointly, and prove consistency of the joint panel-block scheme under two-way dependence. Simulations show coverage close to the nominal level for the joint scheme but severe undercoverage for cross-sectional resampling.
\end{abstract}

\hrule

%% file: MainBody.tex
\input{1_Introduction}

\input{2_PVAR}

\input{3_Bootstrap}

\input{4_Simulation}

\input{5_Conclusion}

    \addcontentsline{toc}{section}{References}
    \bibliographystyle{ecta}  
    \bibliography{Lit_Boot.bib,Lit_Eco.bib,Lit_SoftRefs.bib,Lit_Asym.bib}

%% file: 1_Introduction.tex
\section{Introduction} \label{sec:Introduction}
Bootstrapping has become an established method for the estimation of sampling distributions of multivariate time series processes $ y_{tk}, t=1,\ldots,T,k=1,\ldots,K$. 
Initially, the so-called \textit{iid~bootstrap} \citep{Efron1979,Efron1982} has been adapted to VARs as a system of $ K $ equations \citep{Runkle1987}. The \textit{bootstrap-after-bootstrap} procedure corrects the small-sample bias of OLS in this case \citep{Kilian1998}. Recent studies exercise particular caution about temporal dependencies over $ t = 1,\ldots,T $ in the error terms such as conditional heteroskedasticity and stochastic volatility. The \textit{wild bootstrap} \citep{Liu1988,Mammen1993,DavidsonFlachaire2008} has been approved for univariate AR processes \citep{GoncalvesKilian2004,GoncalvesKilian2007}, for tests on the coefficients of VAR dynamics \citep{HafnerHerwartz2009}, and for structural VAR models identified by conditional heteroskedasticity \citep{LutkepohlSchlaak2019}. The wild bootstrap yet cannot correctly replicate the sampling distribution of residual covariance matrices and derivatives thereof such as structural IRFs. In contrast, the \textit{moving-block bootstrap} (MBB) \citep{Kunsch1989,Liu1992} has been validated for the IRFs of VARs \citep{BruggemannEtAl2016} and \textit{proxy SVARs} identified by external instruments \citep{JentschLunsford2019,JentschLunsford2021}. The related studies provide joint asymptotics for the dynamic VAR parameters and residual covariances, and they prove consistency of the residual-based MBB for SVARs and their IRFs under general assumptions on the error terms. The MBB thus constitutes a promising approach to be extended by the cross-section dimension for panel data.

The VAR analysis of panel data $ y_{itk}$ has developed extensively over the last decades. Combining several individuals $ i=1,\ldots,N $ such as countries, the model development ranges from \textit{large-scale VAR} and \textit{global VAR} \citep{DeesEtAl2007b} to homogeneous \textit{panel VAR} \citep{HoltzEtAl1988}. While the former two emphasize dynamic interactions among the individuals, the pooled sample of the latter provides increased degrees of freedom and sample variation. This induces a more precise estimation and higher test power.\footnote{For an overview, see \citet{Hsiao2007} on the advantages of panel data analysis and  \citet{CanovaCiccarelli2013} as well as \citet[ch.~5]{Hsiao2022} on the VAR analysis of panel data.} In the case of heterogeneous panel data, \textit{mean-group} (MG) estimation \citep{PesaranSmith1995} has been adapted for VAR coefficients \citep{Rebucci2010} or structural IRFs \citep{GambacortaEtAl2014,CesaEtAl2015,BernothHerwartz2021,Herwartz2017}, where the heterogeneous parameters are separately estimated by individual methods and subsequently combined by cross-sectional averages. Specifically, \cite{Herwartz2017} estimate a common rotation matrix of the SVAR identification problem by \textit{independent component analysis}. Their selective pooling of the shock series may be subsumed under \textit{pooled MG identification} in analogy to \textit{pooled MG estimation} of long-run coefficients in single equations \citep{PesaranEtAl1999} and bivariate systems \citep{ChudikEtAl2023}.
Overall, these developments in panel estimation demand panel inference methods analogous to the mentioned VAR procedures.

Despite this demand, an established panel VAR bootstrap that adapts the insights from the well-developed individual MBB is not yet available. Tellingly, panel VAR studies often only employ ad-hoc created bootstrap procedures while omitting the theoretical underpinnings. This aligns with their focus on panel modeling and point estimation but evades a systematic review, advancement, and evaluation of the panel bootstrap methodology. Different configurations of the MBB have been evaluated for single-equation panel regression models \citep{Kapetanios2008,Goncalves2011} using a fixed design. This design resamples the observations instead of the residuals as done by the recursive-design MBB of the VAR literature. Open questions are thus \textit{(i)} how configurations of the promising MBB can be adapted to panel SVARs, \textit{(ii)} how cross-sectional dependence (CSD) --as a typical panel data property \citep{SarafidisWansbeek2012}-- affects the bootstrap performance in addition to temporal dependence, and finally \textit{(iii)} which method a practitioner should prefer in finite samples. The latter issue in particular may necessitate adjustments of the estimation and identification.

\textbf{Objective.} In this article, we extend multivariate time series analysis by the cross-section dimension and develop inferential methods for panel vector autoregressive (VAR) models and their structural impulse response functions (IRFs). Starting from the data generating process of a heterogeneous panel VAR model with residual structure in the cross-sectional and temporal dimensions (i.e.~\textit{two-way dependence}), we arrive at two main contributions: 
\begin{itemize}
    \item We adapt mean-group (MG) estimation and construct pooled mean-group (PMG) identification for proxy SVARs and their structural IRFs $ \Theta_h = f_h{\left( \boldsymbol{\beta}, \boldsymbol{\sigma}, \boldsymbol{\varphi} \right)} $ at horizon $ h = 0,\ldots,H $. 
    The setup covers the uncertainty from estimating the VAR dynamics $ \boldsymbol{\beta} $, residual covariances $ \boldsymbol{\sigma} $, and structural parameters $ \boldsymbol{\varphi} $. Using asymptotic theory, we derive a joint central limit theorem for $ \left\{ \boldsymbol{\beta}, \boldsymbol{\sigma}, \boldsymbol{\varphi} \right\} $ and for $ \Theta_h $ subsequently. Specifically for the structural identification, our pooled estimator of $ \boldsymbol{\varphi} $ uses external data panels of proxy variables to determine the common rotation matrix. Unlike the MG of individual identification results, this PMG identification can greatly mitigate inferential issues of ``individually weak'' proxies with increasing precision from the cross-section dimension.

    \item We propose a recursive-design moving-block bootstrap procedure suitable for panel VARs which resamples the estimated errors either in \textit{(i)} the temporal, \textit{(ii)} the cross-sectional, or \textit{(iii)} both dimensions jointly. Under two-way dependence, we show that the panel-block bootstrap with joint resampling is asymptotically valid by combining results from \cite{BruggemannEtAl2016} and \cite{Menzel2021}. We further find superior performance in finite samples according to our comprehensive MC simulations. For small time series dimensions $T$, we incorporate a panel bootstrap-after-bootstrap procedure for bias correction in the MG estimation of $ \boldsymbol{\beta} $.
\end{itemize}

The remainder of this article is structured as follows: Section~\ref{sec:PVAR} presents the heterogeneous panel SVAR model \eqref{sec:Model}, its identification \eqref{sec:Identification}, mean-group estimation \eqref{sec:Estimation}, and asymptotic inference \eqref{sec:AsyInference}. In Section~\ref{sec:Bootstrap}, we propose the new panel-block bootstrap procedure for the panel VAR model and outline its asymptotic performance. Section~\ref{sec:Simulation} is a Monte Carlo study, which assesses and discusses the finite-sample performance of competing bootstrap resampling schemes. Finally, Section~\ref{sec:Conlusion} summarizes and gives an outlook. Supporting material such as proofs and simulation results can be found in our online appendix. The implementation of the assessed methods is available in the \textsf{R}-package \href{https://cran.r-project.org/package=pvars}{\textbf{pvars}} by \citet{Empting2021}.

\textbf{Notation.} The symbols $ \overset{d}{\longrightarrow} $ and $ \overset{p}{\longrightarrow} $ designate convergence in distribution and in probability, respectively. $ A $ is denoted by $ A_i $ if specific to the individual $ i = 1,\ldots,N $, by $ A_t $ if specific to the time period $ t=1,\ldots,T $, by plain $ A $ if homogeneous across all $i,t$ such as common parameters, by $ \widehat{A} $ if estimated by least squares, by $ \widebar{A} $ if estimated by the mean, and finally by $ A^* $ if obtained from a bootstrap procedure. The operator $ \vec(A) $ stacks the columns of a matrix $ A $, and $ \text{vech}(A) $ stacks only its lower-triangular elements, i.e.~those on and below the main diagonal. The set of $ K \times K $ orthonormal matrices is denoted by $ \mathbb{O}(K) = \left\{ Q \in \mathbb{R}^{K \times K} : Q'Q = QQ' = I_K \right\} $. For ease of notation, we write $ \boldsymbol{x} \leq \boldsymbol{y} $ for vectors $ \boldsymbol{x}, \boldsymbol{y} \in \mathbb{R}^J $ if $ x_j \leq y_j \ \forall j=1,\ldots,J $.

%% file: 2_PVAR.tex
\section{Panel vector autoregression} \label{sec:PVAR}

\subsection{Model} \label{sec:Model}
Our starting point for panel estimation and inference is the \textit{reduced-form VAR} model
\begin{align} \label{eq:PVAR}
\begin{split}
	\boldsymbol{y}_{it} = \Phi_i \boldsymbol{d}_{it} + A_{i,1} \boldsymbol{y}_{i,t-1} + ... + A_{i,p} \boldsymbol{y}_{i,t-p} + \boldsymbol{u}_{it} \quad \text{ with } \boldsymbol{u}_{it} \sim (0,\Sigma_{u,i})
\end{split}
\end{align}
for each individual $ i=1,\ldots,N $. Therein, $\boldsymbol{y}_{it}$ is a vector of $ K $ time series, $A_{i,j}, j=1,\ldots,p,$ are $ K \times K $ coefficient matrices for the VAR process of order $p$, which describe the propagation of the endogenous variables over the time periods $ t=1,\ldots,T $. Without loss of generality, we exclude deterministic terms by suppressing the $ K \times n_d $ coefficient matrix $ \Phi_i $ for the $ n_d $ deterministic regressors $ \boldsymbol{d}_{it} $ and impose $ \Phi_i \boldsymbol{d}_{it} = 0_K  $. Nonetheless, deterministic components can be easily reintroduced, e.g.~via the general formulation of the $ K \times (n_d + Kp) $ compact coefficient matrix $ A_i = \left[ \Phi_i : A_{i,1} : \ldots : A_{i,p} \right] $.

In \textit{structural VAR} analysis, the $ K $-variate reduced-form error series $ \boldsymbol{u}_{it} = \mathsf{B}_i \boldsymbol{\epsilon}_{it} $ is a linear combination of the $ K $ structural shock series $ \boldsymbol{\epsilon}_{it} $, where $ \mathsf{B}_i $ denotes the $ K \times K $ coefficient matrix of structural impact. The mutually independent shock series $ \boldsymbol{\epsilon}_{it} $  have a meaningful economic interpretation. In order to properly recover the associated structural IRFs from the reduced-form VAR model of Eq.~\eqref{eq:PVAR}, standard assumptions of the SVAR literature apply: 
\begin{assumption}[Panel SVAR assumption] \label{as:SVAR}
    \quad
    \begin{enumerate}[label=(\roman*)]
        \item The VAR process has a known lag order $ p $ and is stable (stationary and causal), i.e., $ \det(A_i(z)) $ of the characteristic polynomial $ A_i(z) = I_K - A_{i,1} z - \ldots - A_{i,p} z^{p} $ has all roots $ z $ outside the unit circle uniformly over $i$.
        
        \item The shock series $ \boldsymbol{\epsilon}_{it} $ obeys $ \E \left(\boldsymbol{\epsilon}_{it} \right) = \boldsymbol{0}_K $, $ \E \left(\boldsymbol{\epsilon}_{it}^{\ } \boldsymbol{\epsilon}_{it}' \right) = I_K $, and $ \E \left(\boldsymbol{\epsilon}_{it}^{\ } \boldsymbol{\epsilon}_{it^\bullet}' \right) = 0_{K \times K}, \ t \neq t^\bullet $. 
        
        \item The square coefficient matrix $ \mathsf{B}_i $ has full rank $ \text{rk}\left( \mathsf{B}_i \right) = K $ uniformly over $i$. 

        \item When factorizing $ \mathsf{B}_i = \mathsf{B}_i^\circ \mathsf{Q} $ into a $ K \times K $ lower-triangular matrix $ \mathsf{B}_i^\circ $ and a $ K \times K $ orthogonal matrix $ \mathsf{Q} \in \mathbb{O}(K) $, $ \mathsf{B}_i^\circ $ can be heterogeneous, but $ \mathsf{Q} $ is homogeneous over $ i $.
    \end{enumerate}
\end{assumption}
The SVAR assumptions guarantee that the model can be consistently estimated and identified. For example, Assumptions~\ref{as:SVAR}(ii) and (iii) imply that the $ K \times K $ error covariance matrix
\begin{align} \label{eq:ID_decomposition}
\begin{split}
	\E \left[ \u_{it}^{\ } \u_{it}' \right] = \mathsf{B}_i^{\ } \E \left[ \boldsymbol{\epsilon}_{it}^{\ } \boldsymbol{\epsilon}_{it}' \right] \mathsf{B}_i' = \mathsf{B}_i^{\ } \mathsf{B}_i' = \mathsf{B}_i^{\circ} \mathsf{Q} \mathsf{Q}' \mathsf{B}_i^{\circ '} =  \mathsf{B}_i^{\circ} \mathsf{B}_i^{\circ '} = \Sigma_{u,i}
\end{split}
\end{align}
has full rank $ \text{rk}\left( \Sigma_{u,i} \right) = K $ and thus admits standard inversion. Assumption~\ref{as:SVAR}(iv) describes the identification problem specifically for the panel SVAR model of $ \boldsymbol{y}_{it} $ \citep[see][]{Herwartz2017}, while Assumptions~\ref{as:SVAR}(i)-(iii) pertain to the literature on the individual SVAR of $ \boldsymbol{y}_{t} $ alike. From Eq.~\eqref{eq:ID_decomposition} it is clear that the SVAR Assumption~\ref{as:SVAR}(ii) holds under $ \mathsf{B}_i = \mathsf{B}_i^\circ \mathsf{Q} $ and $ \boldsymbol{\epsilon}_{it} = \mathsf{Q}' \boldsymbol{\epsilon}_{it}^\circ = \mathsf{B}_i^{-1} \boldsymbol{u}_{it} $ for any candidate $ \mathsf{Q} \in \mathbb{O}(K) $ and even under the baseline decomposition $ \mathsf{B}_i^\circ = \text{chol}(\Sigma_{u,i}) $ and $ \boldsymbol{\epsilon}_{it}^{\circ} = \mathsf{B}_i^{\circ -1} \boldsymbol{u}_{it} $. Finding the unique $ \mathsf{Q} $ that additionally permits a meaningful economic interpretation is a matter of structural identification. 

Assumption~\ref{as:SVAR}(i) on stability guarantees that the VAR model in Eq.~\eqref{eq:PVAR} has a \textit{vector moving average} (VMA) Wold representation. Leaving the deterministic terms $ \Phi_i \boldsymbol{d}_{it} = 0_K $ aside, this is given respectively in the reduced form and in the structural form by
\begin{align} \label{eq:VMA}
\begin{split}
	\boldsymbol{y}_{it} = \sum_{h=0}^{\infty} \Xi_{i,h} \boldsymbol{u}_{i,t-h} = \sum_{h=0}^{\infty} \Xi_{i,h} \mathsf{B}_i \boldsymbol{\epsilon}_{i,t-h} = \sum_{h=0}^{\infty} \Theta_{i,h}  \boldsymbol{\epsilon}_{i,t-h}.
\end{split}
\end{align}

Therein, the element $ \left[ \Theta_{i,h} \right]_{ks} $ of the $ K \times K $ structural VMA coefficient matrix $ \Theta_{i,h} = \Xi_{i,h} \mathsf{B}_i $ can be interpreted as the structural response of $ y_{i,t+h,k} $ to an isolated unit shock in $ \epsilon_{its}, \ s=1,\ldots,K $. Characteristically, $ \left[ \Theta_{i,h} \right]_{ks} $ is plotted over $ h=0,\ldots,H $ to visualize this structural IRF from the initial impulse at $ h=0 $ up to a finite horizon $ H $. Note that the initial responses are $ \Theta_{i,0} = \mathsf{B}_i $ by $ \Xi_{i,0} = I_K $, that $ \Xi_{i,h} = \sum_{j=1}^{h} \Xi_{i,h-j} A_{i,j} $, and that $ \Xi_{i,h} \to 0_{K \times K} $ for $ h \to \infty $ in the stable VAR model. We do not consider any subsequent normalization of the initial responses.

In Eq.~\eqref{eq:PVAR}, the VAR parameters $ A_i, \Sigma_{u,i}, \mathsf{B_i} $ have been introduced as individual-specific. However, they may actually have a common base suitable for the econometric methods of heterogeneous panel data. Extending from the single-equation panel model by \citet{PesaranSmith1995} and the reduced-form panel VAR model by \citet{Rebucci2010}, this is described by 
\begin{assumption}[Random parameters] \label{as:random_parameters}
	If some parameters in $ \text{vec}\left( A_i \right) = \boldsymbol{\beta}_i $, $ \text{vech}\left( \Sigma_{u,i} \right) = \boldsymbol{\sigma}_i $ are heterogeneous over $ i $, they vary around common constants $ \text{vec}\left( A \right) = \boldsymbol{\beta} $ , $ \text{vech}\left( \Sigma_{u} \right) = \boldsymbol{\sigma} $ by
	\begin{enumerate}[label=(\roman*)]
		\item $ \boldsymbol{\beta}_i = \boldsymbol{\beta} + \boldsymbol{\eta}_{\beta i}, \quad \boldsymbol{\eta}_{\beta i} \sim \left(0, \Sigma_{\eta_\beta} \right) $,
		
		\item $ \boldsymbol{\sigma}_i = \boldsymbol{\sigma} + \boldsymbol{\eta}_{\sigma i}, \quad \boldsymbol{\eta}_{\sigma i} \sim \left(0, \Sigma_{\eta_\sigma} \right) $.
		
		\item The three groups $ \left\{ \boldsymbol{\eta}_{\beta i} \right\}, \ \left\{ \boldsymbol{\eta}_{\sigma i} \right\}, $ and $ \left\{ \boldsymbol{\epsilon}_{it}^{ \ } \right\} $ are mutually independent for all $ i,t $. At most weak dependence within each group is allowed.
	\end{enumerate}
\end{assumption}
Assumption~\ref{as:random_parameters}(ii) specifies the randomness of the individual-specific, non-zero lower-triangular coefficients of the baseline $ \mathsf{B}_i^\circ $, while the structural $ \mathsf{B}_i = \mathsf{B}_i^\circ \mathsf{Q} $ also contains the panel-wide constant matrix $ \mathsf{Q} $ by Assumption~\ref{as:SVAR}(iv) for the panel SVAR model. All three common parameter vectors then define the baseline IRFs $ \Theta_h = f_h{\left( \boldsymbol{\beta}, \boldsymbol{\sigma}, \boldsymbol{\varphi} \right)} $ over $ h = 0,\ldots,H $.

For the asymptotic analysis, it is useful to distinguish between assumptions required for the large-sample properties of the estimator and those additionally needed for bootstrap validity. While asymptotic normality is established under mixing conditions in both the temporal and cross-sectional dimensions (see Section \ref{sec:AsyInference} below), bootstrap consistency requires further restrictions on the degree of heterogeneity when the cross-sectional dimension grows (Section \ref{sec:BootAsy}). The MC study in Section \ref{sec:Simulation} nevertheless shows that, within the random-parameter framework, the procedure remains robust to alternative forms of cross-sectional dependence beyond those covered by the theoretical assumptions.

\subsection{Identification} \label{sec:Identification}
To find the unique $ \mathsf{Q} $, we construct \textit{pooled mean-group} (PMG) identification for proxy SVARs. The principles of PMG identification have been established by \citet{Herwartz2017}, who use independent component analysis (ICA) for a statistical identification of $ \mathsf{Q} $. Our PMG identification extends the moment-based approach of the individual proxy SVAR literature \citep[e.g.][]{MertensRavn2013,JentschLunsford2019,JentschLunsford2021} and uses panel data of external proxies $ \boldsymbol{m}_{it} $. In this section, we first define the candidate set $ \mathbb{O}(K) $ for $ \mathsf{Q} $, then reduce the moment-based approach on  $ {\mathsf{B}}_i $ to a partial identification of $ \mathsf{Q} $, and finally reconstruct an even simpler solution for full identification of $ \mathsf{Q} $.

\subsubsection{The candidate set}
The reduction of the identification problem of $ {\mathsf{B}}_i $ to $ \mathsf{Q} $, i.e.~removing the heterogeneous variances and correlations in $ \mathsf{B}^\circ_i $ from $ {\mathsf{B}}_i = {\mathsf{B}}_i^\circ {\mathsf{Q}} $, has several advantages. \textit{(i)} We can focus on the sampling properties of estimator $ \widehat{\mathsf{Q}} $, where only $ n_{K} = K(K-1)/2 $ rotation angles describe its $ K^2 $ elements \citep{HoffmanEtAl1972} and keep $ \boldsymbol{\widehat{\epsilon}}_{it} $ uncorrelated for all candidates. \textit{(ii)} We can easily adapt and compare additional identifying restrictions discussed in the literature. \textit{(iii)} As done in Section~\ref{sec:Estimation}, we can simply pool the panel of proxies $ \boldsymbol{m}_{it} $ and baseline shocks $ \boldsymbol{\widehat{\epsilon}}_{it}^\circ $ in the vein of \citet{Herwartz2017} to improve the estimation performance.

The set of suitable candidates for $ \mathsf{Q} $ is given by the group of orthogonal matrices $ \mathbb{O}(K) $ \citep{Rubio-RamirezEtAl2010}, which is \textit{closed} under the operations of matrix multiplication and inversion \citep{Neusser2016}. For the identification based on proxies, we thus arrive at
\begin{definition}[Candidate set] \label{def:CandidateSet} 
    Consider Assumption~\ref{as:SVAR}(iv) on $ {\mathsf{B}}_i = {\mathsf{B}}_i^\circ {\mathsf{Q}} $ in the SVAR model. \textit{Global full exact} identification is given by the unique $ \mathsf{Q} = \mathsf{U} \mathsf{W} \mathsf{P}^{\pm} \in \mathbb{O}(K) $, while the identification is just (i) \textit{local} for any signed permutation matrix $ \mathsf{P}^{\pm} $, (ii) \textit{partial} for a unique $ K \times K_1 $ block $ {\mathsf{Q}}_1 $ in $ \mathsf{Q} = \left[ {\mathsf{Q}}_1 : {\mathsf{Q}}_{2} \right] \mathsf{P}^{\pm} $, and (iii) a \textit{set}-identification of $ \mathsf{U}_{1} = \mathsf{Q}_1 \mathsf{W}_{11} $ for any $ \mathsf{W}_{11} \in \mathbb{O}(K_1) $.
\end{definition}

Since there are typically not enough proxy variables to identify all $ K $ components of the SVAR model, individual methods of the moment-based approach resort to a partial identification of $ \mathsf{B}_i $. With $ K_1 $ instrumented and $ K_2 = K - K_1 $ non-instrumented components, it is common practice to partition $ \boldsymbol{\epsilon}_{it} $ into $ \left( \boldsymbol{\epsilon}_{1,it}', \boldsymbol{\epsilon}_{2,it}' \right)' $ and $ \mathsf{B}_{i} $ into $ \left[ \mathsf{B}_{i,1} : \mathsf{B}_{i,2} \right] $, where the $ K_1 $ identified components are ordered first. Meaningful signs are chosen for the coefficients and shock series after estimation. Overall, this constitutes a signed permutation in $ \boldsymbol{u}_{it} = \mathsf{B}_i \mathsf{P}^{\pm '} \mathsf{P}^{\pm} \boldsymbol{\epsilon}_{it} = \left[ \mathsf{B}_{i,1} : \mathsf{B}_{i,2} \right] \left( \boldsymbol{\epsilon}_{1,it}', \boldsymbol{\epsilon}_{2,it}' \right)' $, wherein the proxies render local partial identification.

\subsubsection{The moment-based approach}
We reduce the moment-based approach in the case of multiple proxy variables \citep[e.g.][sec.~2.2.2]{JentschLunsford2021} to the partial identification of $ \mathsf{Q}$. 
The approach builds upon the proxy assumption, where the $ L $ proxies $ \boldsymbol{m}_{it} $ need to be \textit{relevant} for the $ K_1 $ instrumented shocks $ \boldsymbol{\epsilon}_{1,it} $ and \textit{exogenous} to the $ K_2 $ non-instrumented $ \boldsymbol{\epsilon}_{2,it} $. This is formalized in 
\begin{assumption}[Proxy assumption] \label{as:Proxy}
    \quad
    \begin{enumerate}[label=(\roman*)]
        \item $ \text{Cov} \left( \boldsymbol{m}_{it}, \boldsymbol{\epsilon}_{1,it} \right) = \E \left( \boldsymbol{m}_{it} \boldsymbol{\epsilon}_{1,it}' \right) = \Sigma_{m\epsilon_1,i} $. 
        \item $ \text{Cov} \left( \boldsymbol{m}_{it}, \boldsymbol{\epsilon}_{2,it} \right) = \E \left( \boldsymbol{m}_{it} \boldsymbol{\epsilon}_{2,it}' \right) = 0_{L \times K_2} $.
        \item The $ L \times K_1 $ matrix $ \E \left( \Sigma_{m\epsilon_1,i} \right) = \Sigma_{m\epsilon_1} $ has rank $ \text{rk} \left( \Sigma_{m\epsilon_1} \right) = K_1 $.
    \end{enumerate}
\end{assumption}
The relevance Assumption~\ref{as:Proxy}(i) stipulates correlation between the proxies and the instrumented shocks. Specifically, for the $ L $ proxies to provide distinct information about the $ K_1 $ instrumented shocks, an individual identification requires $ \text{rk} \left( \Sigma_{m\epsilon_1 i} \right) = K_1 $ and a pooled identification requires Assumption~\ref{as:Proxy}(iii) equivalently. The exogeneity Assumption~\ref{as:Proxy}(ii) stipulates no correlation between proxies and non-instrumented shocks and therewith suppresses a contamination of the proxies from the non-instrumented shocks. Both moment conditions (i) and (ii) can be summarized in the $ L \times K $ proxy moment matrix $ \Sigma_{m\epsilon,i} = \left[ \Sigma_{m\epsilon_1,i} : 0_{L \times K_2} \right] \mathsf{P}^{\pm} $. Using $ \boldsymbol{u}_{it} = \mathsf{B}_{i} \boldsymbol{\epsilon}_{it} = \mathsf{B}_{i}^\circ \mathsf{Q} \mathsf{Q}' \boldsymbol{\epsilon}_{it}^\circ = \mathsf{B}_{i}^\circ \boldsymbol{\epsilon}_{it}^\circ $ from Assumption~\ref{as:SVAR}, we can then separate the individual-specific $ \mathsf{B}_{i}^\circ $ from the estimable proxy moment matrix $ \Sigma_{um,i} $ according to
\begin{align} \label{eq:ID_moments}
	\begin{split}		
		\mathsf{B}_i \E \left[ \boldsymbol{\epsilon}_{it} \boldsymbol{m}_{it}' \right] = \mathsf{B}_i \Sigma_{\epsilon m, i} = \E \left[ \boldsymbol{u}_{it} \boldsymbol{m}_{it}' \right] = \Sigma_{um,i} = \mathsf{B}_i^\circ \E \left[ \boldsymbol{\epsilon}_{it}^\circ \boldsymbol{m}_{it}' \right] = \mathsf{B}_i^\circ \Sigma_{\epsilon^\circ m, i}. 
	\end{split}
\end{align}

With the ``cleaned'' population moment $ \Sigma_{\epsilon^\circ m, i} $, the derivation of a moment-based estimator proceeds just as in the proxy SVAR literature. Combining Eqs.~\eqref{eq:ID_decomposition} and \eqref{eq:ID_moments} yields
\begin{align} \label{eq:ID_square}
	\begin{split}
		{\Sigma}^{\ }_{mu,i} {\Sigma}^{-1}_{u,i} {\Sigma}^{\ }_{um,i} = {\Sigma}_{m \epsilon^\circ, i} {\Sigma}_{\epsilon^\circ m, i} = {\Sigma}_{m \epsilon, i} {\Sigma}_{\epsilon m, i}.
	\end{split}
\end{align}
Under $ L =K_1  $, the inverted $ L \times L $ matrix square root of either side of Eq.~\eqref{eq:ID_square} eliminates $ {\Sigma}_{\epsilon m, i} $ from the term $ \mathsf{B}_i \Sigma_{\epsilon m, i} $ of Eq.~\eqref{eq:ID_moments}. We therewith obtain for both partitions, the individual $ \mathsf{B}_i = \left[ \mathsf{B}_{i,1} : \mathsf{B}_{i,2} \right] \mathsf{P}^{\pm} $ and the common $ \mathsf{Q} = \left[ \mathsf{Q}_{1} : \mathsf{Q}_{2} \right] \mathsf{P}^{\pm} $, the estimable identification relations
\begin{align} 
		{\mathsf{B}}_{i,1} = {\mathsf{B}}_{i}^\circ {\mathsf{Q}}_{1}^{\ } 
        & = {\Sigma}^{\ }_{um,i} \left[ {\Sigma}_{mu,i}^{\ } {\Sigma}^{-1}_{u,i} {\Sigma}^{\ }_{um,i} \right]^{-1/2} \label{eq:ID_Bi} \\
		& = {\mathsf{B}}_{i}^\circ {\Sigma}_{\epsilon^\circ m, i} \left[ {\Sigma}_{m \epsilon^\circ, i} {\mathsf{B}}_{i}^{\circ'} \left( {\mathsf{B}}_{i}^\circ {\mathsf{B}}_{i}^{\circ'} \right)^{-1} {\mathsf{B}}_{i}^\circ {\Sigma}_{\epsilon^\circ m, i} \right]^{-1/2} \\
		& = {\mathsf{B}}_{i}^\circ {\Sigma}_{\epsilon^\circ m, i} \left[ {\Sigma}_{m \epsilon^\circ, i} {\Sigma}_{\epsilon^\circ m, i} \right]^{-1/2}. \label{eq:ID_Qi}
\end{align}

Eq.~\eqref{eq:ID_Qi} reveals that Eq.~\eqref{eq:ID_Bi}, often used in the proxy SVAR literature, already contains the closed-form solution for finding the $ K_1 $ orthonormal column vectors nearest to $ {\Sigma}_{\epsilon^\circ m, i} $. This proximity $ \| {\Sigma}_{\epsilon^\circ m, i} - {\mathsf{Q}}_{1} \| = \min \left\{ \| {\Sigma}_{\epsilon^\circ m, i} - Q \| : Q'Q = I_{K_1}  \right\}  $ applies to any unitarily invariant norm, such as the Frobenius norm. Its minimizing solution given in Eq.~\eqref{eq:ID_Qi} results from the (non-square) polar decomposition (PD) of $ {\Sigma}_{\epsilon^\circ m, i} = {\mathsf{Q}}_{1} \mathsf{H} $, where the real symmetric matrix $ \mathsf{H} = \left( {\Sigma}_{\epsilon^\circ m, i}' {\Sigma}_{\epsilon^\circ m, i}^{\ } \right)^{1/2} $ solves for the orthonormal $ {\mathsf{Q}}_{1} = {\Sigma}_{\epsilon^\circ m, i} \mathsf{H}^{-1} $.\footnote{See \citet[ch.~8]{Higham2008} on ``Theorem~8.4 \citep{FanHoffman1955}''.}

\subsubsection{The general solution by singular value decomposition}
We extend the moment-based solution of Eq.~\eqref{eq:ID_Qi} in three directions, namely \textit{(i)} pooled, \textit{(ii)} full, and \textit{(iii)} over-identification. Specifically, we bind a zero block to the panel proxy moment $ \Sigma_\circ = \left[ \Sigma_{\epsilon^\circ m}  : 0_{K \times K_2} \right] $, which serves as a placeholder for the yet missing information on the $ K_2 $ non-instrumented $ \boldsymbol{\epsilon}_{2,it} $. To find its $ K \times K $ \textit{nearest orthogonal matrix} $ \mathsf{Q} $, we adapt a more practical solution known in other disciplines\footnote{In psychometrics, \citet{Gibson1962}, \citet{Cliff1966}, and \citet{Schonemann1966} discuss the SVD solution \textit{(i)} for the \textit{nearest orthogonal matrix}, \textit{(ii)} for the problem under singularity of the decomposed matrix, and \textit{(iii)} for the closely related \textit{orthogonal Procrustes problem}. In physics and engineering, \citet{ArunEtAl1987} and \citet{HornEtAl1988} propose the SVD and PD solution respectively for a model with noise. This noisy model may be reformulated as the proxy process often used in the augmented-VAR approach of the proxy SVAR literature. However, the definition of the proxy moment $ \Sigma_{m\epsilon,i} $ by Assumption~\ref{as:Proxy} is sufficient for our moment-based approach.} and prefer the singular value decomposition (SVD).
The solution follows by inserting the SVD of $ \Sigma_\circ = \mathsf{U} \mathsf{D} \mathsf{V}' $ into the PD solution, i.e.
\begin{align} \label{eq:ID_svd}
\begin{split}
    \left[ \mathsf{Q}_{1} : \mathsf{Q}_{2} \right] & = \Sigma_\circ^{\ } \left( \Sigma_\circ' \Sigma_\circ^{\ } \right)^{-1/2} = \mathsf{U} \mathsf{D} \mathsf{V}' \left( \left( \mathsf{U} \mathsf{D} \mathsf{V}' \right)' \mathsf{U} \mathsf{D} \mathsf{V}' \right)^{-1/2} \\
    & = \mathsf{U} \mathsf{D} \mathsf{V}' (\mathsf{V} \mathsf{D}^2 \mathsf{V}')^{-1/2} = \mathsf{U} \mathsf{D} \mathsf{V}' \mathsf{V} \mathsf{D}^{-1}\mathsf{V}' \\
    & = \mathsf{U} \mathsf{V}'.
\end{split}
\end{align}
By the orthogonality $ \mathsf{U}' \mathsf{U} = \mathsf{V}' \mathsf{V} = I $ and by the matrix square root via eigendecomposition, the nearest orthogonal matrix $ \left[ {\mathsf{Q}}_1 : {\mathsf{Q}}_{2} \right] = \mathsf{U} \mathsf{V}' $ is directly obtained as the product of the $ K \times K $ orthogonal matrices $ \mathsf{U} $ and $ \mathsf{V}' $ collecting the left and right singular vectors respectively. In effect, the diagonal matrix $ \mathsf{D} $ of the $ K $ singular values $ \boldsymbol{\mathsf{d}} = \left( \boldsymbol{\mathsf{d}}_1', \boldsymbol{\mathsf{d}}_2' \right)' $ is just scaled to $ I_K $. Thereof, the first $ \boldsymbol{\mathsf{d}}_1 \neq \boldsymbol{0}_{K_1} $ is related to the $ K_1 $ orthonormal columns in $ \mathsf{Q}_{1} = \mathsf{U}_{1}^{\ } \mathsf{V}_{11}' $, which is the $ K \times K_1 $ SVD solution for Eq.~\eqref{eq:ID_Qi} with the plain $ \Sigma_{\epsilon^\circ m} $. Their non-zero size reflects Assumption~\ref{as:Proxy}(iii) and indicates the proxy strength. The second $ \boldsymbol{\mathsf{d}}_2 = \boldsymbol{0}_{K_2} $  is related to the $ K_2 $ columns in $ \mathsf{Q}_{1\perp} = \mathsf{U}_2^{\ } \mathsf{V}_{22}' $, which is the orthonormal complement of $ \mathsf{Q}_{1} $ induced by $ 0_{K \times K_2} $.\footnote{In general, a $ K \times K_2 $ matrix $ \mathsf{Q}_{1\perp} $ as an orthogonal complement of the $ K \times K_1 $ matrix $ \mathsf{Q}_{1} $ obeys $ \text{rk}\left( \left[ {\mathsf{Q}}_1 : {\mathsf{Q}}_{1\perp} \right]  \right) = K $ and $ {\mathsf{Q}}_{1}' {\mathsf{Q}}_{1\perp}^{\ } = 0_{K_1 \times K_2} $, see e.g.~\citet[App.~A.8.2]{Lutkepohl2005}. It collects the maximum number $ K_2 $ of columns that are mutually orthogonal and orthogonal to $ \mathsf{Q}_1 $. Here, they are unit-normalized due to the SVD such that $ {\mathsf{Q}}_{1\perp}' {\mathsf{Q}}_{1\perp}^{\ } = I_{K_2} $ and thus $ \left[ {\mathsf{Q}}_1 : {\mathsf{Q}}_{1\perp} \right]'\left[ {\mathsf{Q}}_1 : {\mathsf{Q}}_{1\perp} \right] = I_K $. Hence, $ {\mathsf{Q}}_{1\perp} $ accords with the orthogonality of $ \left[ {\mathsf{Q}}_1 : {\mathsf{Q}}_{1\perp} \right] $.} The accordant partition of the fully identified $ \left[ \mathsf{B}_{i,1} : \mathsf{B}_{i,2} \right] =  \mathsf{B}_{i}^\circ \left[ \mathsf{Q}_{1} : \mathsf{Q}_{2} \right] $ is thus
\begin{align} \label{eq:ID_subrotation}
	\begin{split}		
		\begin{bmatrix}
			{\mathsf{B}}_{i,11} & {\mathsf{B}}_{i,12} \\[-6pt]
			{\scriptstyle (K_1 \times K_1)} & {\scriptstyle (K_1 \times K_2)} \\
			{\mathsf{B}}_{i,21} & {\mathsf{B}}_{i,22} \\[-6pt]
			{\scriptstyle (K_2 \times K_1)} & {\scriptstyle (K_2 \times K_2)} \\
		\end{bmatrix}
		 =
		 \begin{bmatrix}
			{\mathsf{B}}_{i,11}^\circ & 0 \\[-6pt]
			{\scriptstyle (K_1 \times K_1)} & {\scriptstyle (K_1 \times K_2)} \\
			{\mathsf{B}}_{i,21}^\circ & {\mathsf{B}}_{i,22}^\circ \\[-6pt]
			{\scriptstyle (K_2 \times K_1)} & {\scriptstyle (K_2 \times K_2)} \\
		\end{bmatrix}
		\begin{bmatrix}
			{\mathsf{U}}_{11} & {\mathsf{U}}_{12} \\[-6pt]
			{\scriptstyle (K_1 \times K_1)} & {\scriptstyle (K_1 \times K_2)} \\
			{\mathsf{U}}_{21} & {\mathsf{U}}_{22} \\[-6pt]
			{\scriptstyle (K_2 \times K_1)} & {\scriptstyle (K_2 \times K_2)} \\
		\end{bmatrix}
		\begin{bmatrix}
			  {\mathsf{V}}_{11}' & 0 \\[-6pt]
			{\scriptstyle (K_1 \times K_1)} & {\scriptstyle (K_1 \times K_2)} \\
			0 & {\mathsf{V}}_{22}' \\[-6pt]
			{\scriptstyle (K_2 \times K_1)} & {\scriptstyle (K_2 \times K_2)} \\
		\end{bmatrix}.
	\end{split}
\end{align}

This explicit representation highlights the counting rule of $ n_{K} =  K_1 \cdot K_2 + n_{K1} + n_{K2} $ restrictions required for local full identification. The proxy Assumption~\ref{as:Proxy} provides $ K_1 \cdot K_2 $ restrictions by the separation of instrumented and non-instrumented components in $ \left[ {\mathsf{Q}}_1 : {\mathsf{Q}}_{2} \right] $, which is given by the baseline separation $ \mathsf{U} $ and preserved under $ \mathsf{V}_{12}' = \mathsf{V}_{21} = 0_{K_1 \times K_2} $. Matrix $ \mathsf{V}_{11} \in \mathbb{O}(K_1) $ has $ n_{K1} = K_1(K_1-1)/2 $ rotation angles for the identification within the instrumented components, and the yet arbitrary $ \mathsf{V}_{22} \in \mathbb{O}(K_2) $ has equivalent $ n_{K2} = K_2(K_2-1)/2 $ for the identification within the non-instrumented ones. Both blocks on the main diagonal preserve the orthogonality of $ \mathsf{V} $ and thus keep $ \mathsf{Q} = \left[ {\mathsf{U}}_1^{\ } \mathsf{V}_{11}' : {\mathsf{U}}_{2}^{\ } \mathsf{V}_{22}' \right] \mathsf{P}^\pm $ in its candidate set $ \mathbb{O}(K) $.

\textbf{Remark~2.1.} The SVD solution nests the available closed-form procedures of the moment-based approach, which rely on $ n_{K1} $ specific restrictions to exactly identify $ \mathsf{B}_{i,1} $ under $ L = K_1 $. The SVD-based $ \mathsf{V}_{11}' = \mathsf{W}_{11}^{\text{NQ}} $ aligns $ \mathsf{Q}_{1} = \mathsf{U}_{1}^{\ } \mathsf{W}_{11}^{\text{NQ}} $ closest to the non-orthonormal proxy moment,  which is a statistical principle and requires $ \Sigma_{m\epsilon_1} $ to hold meaningful information. In particular, the proxies could be mutually uncorrelated and one-to-one correlated with the shocks as in \citet{BrunsEtAl2025}. \citet[\textbf{MR}]{MertensRavn2013} impose that the Schur complement $ S_{i,1} = \mathsf{B}_{i,11}^{\ } - \mathsf{B}_{i,{12}}^{\ }  \mathsf{B}_{i,{22}}^{-1} \mathsf{B}_{i,21}^{\ } $ is lower-triangular with positive diagonal elements. \citet[\textbf{JL}]{JentschLunsford2021} use the Cholesky decomposition of Eq.~\eqref{eq:ID_square} to obtain the matrix square root for Eq.~\eqref{eq:ID_Bi}, placing the $ n_{K1} $ zero-restrictions on the lower-triangular $ \Sigma_{m\epsilon_1,i} $. In Appendix~\ref{app:S2ID}, we show that, after performing SVD under $ L \geq K_1 $ and $ N \geq 1 $, alternative restriction templates such as $ \bullet \in \left\{ \textbf{MR}, \textbf{JL} \right\} $ can be recovered as specific $ {\mathsf{W}}_{11}^\bullet $ (and similarly $ \mathsf{W}_{22}^\bullet $ for the $ K_2 $ non-instrumented components).

\textbf{Remark~2.2.} Singleton dimensions imply straightforward simplifications of Eq.~\eqref{eq:ID_subrotation}. If $ K_1=1 $ and thus $ n_{K1} = 0$, the set ambiguity collapses into $ \mathsf{W}_{11} = \pm 1 $. In this case, the individual $ {\mathsf{B}}_{i,1} $ determined by \citet{MertensRavn2013}, \citet{JentschLunsford2021}, or Eq.~\eqref{eq:ID_Qi} under $ L = N = 1 $ is identical up to the sign choice. If $ K_2=1 $ and thus $ n_{K2} = 0 $ additionally, the SVAR model is fully identified with $ \mathsf{Q}_2 = \pm \mathsf{Q}_{1\perp} $ just by respecting the candidate set $ \mathbb{O}(K) $. In Section~\ref{sec:SimDGP}, we adapt the bivariate VAR processes of \citet{JentschLunsford2021} for our MC simulations, where $ K = 2 $ leaves the two problems of set-identification aside.

\subsection{Estimation} \label{sec:Estimation}
To determine the parameter matrices for the VAR dynamics $ A_i $, the error covariance $ \Sigma_{u,i} $, and the structural impact $ \mathsf{B}_i $ from the observable data panel $ \boldsymbol{y}_{it} $, we consider moment-based estimators in this section. We define the $ K \times T $ regressand matrix $ Y_i = \left[ \boldsymbol{y}_{i1} : \ldots : \boldsymbol{y}_{iT} \right] $ and the $ K \times T $ matrices $ Z_{i,-h} = \left[ \boldsymbol{y}_{i,1-h} : \ldots : \boldsymbol{y}_{i,T-h} \right] $ of the regressors at a lag $ h \in \left\{ 1,\ldots, p \right\} $. Further, we define the $ K \times T $ error matrix $ U_i = \left[ \boldsymbol{u}_{i1} : \ldots : \boldsymbol{u}_{iT} \right] $ and the $ L \times T $ proxy matrix  $ M_i = \left[ \boldsymbol{m}_{i1} : \ldots : \boldsymbol{m}_{iT} \right] $. The concatenated coefficients $ A_i = \left[ A_{i,1} : \ldots : A_{i,p} \right] $ and regressors $ Z_i = \left[ Z_{i,-1}' : \ldots : Z_{i,-p}' \right]' $ form the compact model $ Y_i = A_i Z_i + U_i $. The individual sample moments are
\begin{align} \label{eq:sample_moments}
    \widehat{\Sigma}_{YZ,i} = \frac{{Y}_{i}^{\ } {Z}_{i}'}{T}, \quad \widehat{\Sigma}_{ZZ,i} = \frac{{Z}_{i}^{\ } {Z}_{i}'}{T}, \quad \widehat{\Sigma}_{u,i} = \frac{\widehat{U}_{i}^{\ } \widehat{U}_{i}'}{T}, \text{\quad and \quad} \widehat{\Sigma}_{m\epsilon^\circ,i} = \frac{{M}_{i}^{\ } \widehat{E}_{i}^{\circ '}}{T}.
\end{align}
The multivariate least squares estimator for the VAR coefficients is then $ \widehat{A}_i = \widehat{\Sigma}_{YZ,i}^{\ } \widehat{\Sigma}_{ZZ,i}^{-1} $, the estimated residuals are $ \widehat{U}_i = Y_i - \widehat{A}_i Z_i $, the baseline Cholesky decomposition is $ \widehat{\mathsf{B}}_i^\circ = \text{chol}(\widehat{\Sigma}_{u,i}) $, and orthogonal baseline shocks are $ \widehat{E}_i^{\circ} = \widehat{\mathsf{B}}_i^{\circ -1} \widehat{U}_i $. The underlying common parameters of Assumption~\ref{as:random_parameters} are MG-estimated by $ \widebar{A} = N^{-1} \sum_{i=1}^{N} \widehat{A}_i $ and $ \widebar{\mathsf{B}}^\circ = N^{-1} \sum_{i=1}^N \widehat{\mathsf{B}}^\circ_{i} $.

For the PMG identification, the population moment $ \Sigma_{m \epsilon^\circ} $ in the SVD solution  of Eq.~\eqref{eq:ID_svd} is estimated by the pooled sample moment
\begin{align} \label{eq:IDpooled_moment}
	\begin{split}		
		\widebar{\Sigma}_{m \epsilon^\circ} = \frac{\left[ M_1 : \ldots : M_N \right] \left[ \widehat{E}^\circ_1 : \ldots : \widehat{E}^\circ_N \right]'}{NT} = \frac{1}{N} \sum_{i=1}^{N} \widehat{\Sigma}_{m \epsilon^\circ, i}.
	\end{split}
\end{align}
Hence, we obtain the pooled estimator $ \widehat{\mathsf{Q}} = \left[ \widehat{\mathsf{Q}}_1 : \widehat{\mathsf{Q}}_2 \right] \mathsf{P}^{\pm} $ for all $ K = K_1 + K_2 $ components by 
\begin{align} \label{eq:IDpooled_estimator}
	\begin{split}		
		\widehat{\mathsf{Q}} = \widehat{\mathsf{U}} \widehat{\mathsf{V}}' \mathsf{P}^{\pm} \quad \text{from the SVD} \quad \left[ \widebar{\Sigma}_{\epsilon^\circ m}  : 0_{K \times K_2} \right] = \widehat{\mathsf{U}} \widehat{\mathsf{D}} \widehat{\mathsf{V}}'.
	\end{split}
\end{align}
Notably, alternative restrictions from Appendix~\ref{app:S2ID} on $ \widehat{\mathsf{W}} \mathsf{P}^{\pm} $ are simply right-multiplied in the individual $ \widehat{\mathsf{B}}_{i} = \widehat{\mathsf{B}}_{i}^\circ \widehat{\mathsf{Q}} $ and PMG estimates $ \widebar{\mathsf{B}} =  N^{-1} \sum_{i=1}^N \widehat{\mathsf{B}}_{i} = \widebar{\mathsf{B}}^\circ \widehat{\mathsf{Q}} $. This further extends to the individual $ \widehat{\Theta}_{ih} = \widehat{\Xi}_{ih} \widehat{\mathsf{B}}_{i}^\circ \mathsf{\widehat{Q}} $ and PMG estimates $ \widebar{\Theta}_h = \widebar{\Xi}_h \widebar{\mathsf{B}}^\circ \mathsf{\widehat{Q}} $ of the structural IRFs, where $ \widehat{\Xi}_{ih} $ and $ \widebar{\Xi}_{h} $ are derived from $ \widehat{A}_i $ resp.~$ \widebar{A} $ via $ f_h $ just like $ {\Xi}_{ih} $ from $ A_i $ in the VMA model of Eq.~\eqref{eq:VMA}. 

For the asymptotic analysis in  the subsequent Section~\ref{sec:AsyInference}, we follow the notation of \citet[ch.~3]{Lutkepohl2005}, \citet{BruggemannEtAl2016}, and \citet{JentschLunsford2021}. We vectorize (i.e.~use the $ \text{vec}(\cdot) $-operator to stack the columns of) the individual time series matrices into the $ KT \times 1 $ regressand vector $ \boldsymbol{y}_i = \vec(Y_i) $, the $ KT \times 1 $ error vector $ \boldsymbol{u}_i = \vec(U_i) $, and the $ Kp \times 1 $ regressor vectors $ \boldsymbol{z}_{it} = \vec \left( \left[ \boldsymbol{y}_{i,t-1}: \dots : \boldsymbol{y}_{i,t-p} \right] \right) $. The latter can be concatenated into the same $ Kp \times T $ regressor matrix $  Z_i = \left[ \boldsymbol{z}_{i1} : \dots : \boldsymbol{z}_{iT} \right] $ as already defined for the compact model. Further, we vectorize the individual parameters and their estimators
\begin{alignat*}{4}
    \boldsymbol{\beta}_i & = \vec(A_{i}), && \boldsymbol{\widehat{\beta}}_i = \vec(\widehat{A}_{i}), && (K^2p \times 1) \\
    \boldsymbol{\sigma}_i & = \text{vech}(\Sigma_{u,i}), && \boldsymbol{\widehat{\sigma}}_i = \text{vech}(\widehat{\Sigma}_{u,i}), && ((K^2+K)/2 \times 1) \\
    \boldsymbol{\varphi}_i & = \vec( \Sigma_{\epsilon^\circ m, i} ), \quad && \boldsymbol{\widehat{\varphi}}_i = \vec( \widehat{\Sigma}_{\epsilon^\circ m, i} ), \quad && (KL \times 1) 
\end{alignat*} 
where $ {\Sigma}_{\epsilon^\circ m, i} = \mathsf{Q}_1 \Sigma_{\epsilon_1 m, i} = \mathsf{B}_i^{\circ -1} \mathsf{B}_i \Sigma_{\epsilon m, i} $ and $ \widebar{\Sigma}_{\epsilon^\circ m} = \widebar{\mathsf{Q} \Sigma_{\epsilon m}} = \widebar{\mathsf{Q}_1 \Sigma_{\epsilon_1 m}} $. The vectorized least squares estimator is thus $ \boldsymbol{\widehat{\beta}}_i = \vec(\widehat{A}_{i}) = \left( \left( Z_i^{\ } Z_i' \right)^{-1} Z_i \otimes \text{I}_K \right) \boldsymbol{y}_i $. Accordant dimensions apply to the common parameters $ \boldsymbol{\beta} = \text{vec}(A), \ \boldsymbol{\sigma} = \text{vech}(\mathsf{B}^\circ \mathsf{B}^{\circ '} ), \ \boldsymbol{\varphi} = \text{vec}(\Sigma_{\epsilon^\circ m}) $ and their estimators
\begin{align*}
    \boldsymbol{\widebar{\beta}} = \vec(\widebar{A}) = \frac{\sum_{i=1}^{N} \boldsymbol{\widehat{\beta}}_i}{N}, \quad \boldsymbol{\widebar{\sigma}} = \text{vech}(\widebar{\Sigma}_u) = \frac{\sum_{i=1}^{N} \boldsymbol{\widehat{\sigma}}_i}{N}, \quad \boldsymbol{\widebar{\varphi}} = \text{vec}( \widebar{\Sigma}_{\epsilon^\circ m} ) =\frac{\sum_{i=1}^{N} \boldsymbol{\widehat{\varphi}}_i}{N}.
\end{align*}

\subsection{Asymptotic inference} \label{sec:AsyInference}
In the following, we derive the asymptotic properties of the PMG estimation described in Section~\ref{sec:Estimation}. To do so, we first define the variance and proxy factor with true error terms as \(\Tilde{\varphi} = \vec(\Tilde{\Sigma}_{\epsilon^{\circ} m})\) with  \(\Tilde{\Sigma}_{\epsilon^{\circ} m} = (NT)^{-1} \sum_{i=1}^N \sum_{t=1}^T \operatorname{chol}(\Tilde{\Sigma}_{u,i})^{-1} u_{it} m_{it}' \) and \(\Tilde{\Sigma}_{u,i} = T^{-1} \sum_{t=1}^T u_{it} u_{it}'\), and \(\Tilde{\sigma} = \vec(\Tilde{\Sigma}_u)\) with \( \Tilde{\Sigma}_u = (NT)^{-1} \sum_{i=1}^N \sum_{t=1}^T u_{it} u_{it}'\). We show unconditional joint asymptotic normality of $\boldsymbol{\widebar{\beta}}$, $\boldsymbol{\Tilde{\sigma}}$ and $\boldsymbol{\Tilde{\varphi}}$. Subsequently, we follow the equivalent result for $\boldsymbol{\widebar{\beta}}$, $\boldsymbol{\widebar{\sigma}}$ and $\boldsymbol{\widebar{\varphi}}$ and formulate a CLT for the structural IRFs. The results largely follow analogously to the statements made in \cite{BruggemannEtAl2016} and \cite{JentschLunsford2021}, but need additional careful derivations due to the complexity introduced by the panel dimension.

We embed the panel index set in the irregularly spaced lattice \((\mathbb N_+\times\mathbb Z)\subset\mathbb R^2\) and equip it with the supremum metric. Hence, for any two distinct elements \((i,t)\neq(j,s)\) in \((\mathbb N_+\times\mathbb Z)\),  \(\operatorname{dist}((i,t),(j,s))\geq 1\). Therefore, the lattice satisfies the minimum-spacing condition of \citet{jenish2012}.

Our (P)MG estimates involve double sums over cross-sectional and time dimensions. To derive asymptotic properties, we use random field theory and assume strong mixing conditions \citep[as imposed in][]{driscoll1998}, for the stacked residuals and proxies $ \mathbf{x}_{it} = \left( \boldsymbol{u}_{it}', \boldsymbol{m}_{it}' \right)' $. This controls for potential residual correlation in both temporal and cross-sectional dimensions. In the simulation study, we additionally assess the robustness of the bootstrap procedure to alternative (stronger) forms of cross-sectional dependence that are not covered by this assumption.

\begin{assumption}[Mixing and moment conditions]
\label{as:mixing_conditions}
    Let \( \mathbf{x}_{it} = \left( \boldsymbol{u}_{it}', \boldsymbol{m}_{it}' \right)' \in \mathbb{R}^{(K + L)} \) denote the random field indexed by \( (i,t) \in (\mathbb{N}_+ \times \mathbb{Z}) \). Assume that 
    \begin{enumerate}[label=(\roman*)]
        \item \( \left( \mathbf{x}_{it} \right)_{t \in \mathbb{Z}} \) is strictly stationary in \( t \) for every \( i \),
        
        \item \( \left( \mathbf{x}_{it} \right)_{i\in \mathbb{N}_+, t\in \mathbb{Z}} \) is \(\alpha\)-mixing with coefficients 
        \begin{equation} \label{eq:alpha-mixing coefs}
            \begin{aligned}
                \widebar{\alpha}(u,v,r) &= \sup \left\{ \alpha\left( \sigma(U), \sigma(V) \right) : |U| \leq u,\ |V| \leq v, \operatorname{ dist}(U,V) \geq r \right\}, \\
                \widebar{\alpha}(u,v,r) &\leq \varphi(u,v) \hat{\alpha}(r),
            \end{aligned}
        \end{equation}
        where $\alpha(\mathcal{A}, \mathcal{B})\coloneqq \sup \{ |\mathbb P(A\cap B) - \mathbb P(A) \mathbb P(B)|: A\in\mathcal{A},B\in\mathcal{B}\}$ is the standard $\alpha$-mixing coefficient \citep[ch. 1.1]{doukhan1994} and $\varphi(u,v)$ is non-decreasing in each argument. Furthermore, there exists \( \delta > 0 \) such that \(\sup_{i\in\mathbb N} \mathbb E \| \mathbf{x}_{it} \|^{4(2 + \delta)} < \infty, \) as well as
        \[\hat{\alpha}(r) = O(r^{-\gamma}), \quad \text{with} \quad \gamma > \frac{4(2+\delta)}{\delta}. \]

        \item For the long-run covariance matrices $\Sigma_{NT}$, as defined in Appendix \ref{app:AsyMG}, it holds
        \[
        \liminf_{N,T\to\infty}
        \lambda_{\min}\left(\Sigma_{NT}\right)>0.
        \]
	\end{enumerate}
\end{assumption}

In contrast to standard mixing coefficients of time-series processes, the mixing coefficients for random fields as in \ref{as:mixing_conditions}~(ii) depend not only on the distance between two datasets but also on their sizes. To explicitly account for such dependence, we assume that $\varphi(u,v) = u + v$ is of a certain form \citep[p.~182]{jenish2012}. Assumption \ref{as:mixing_conditions}~(iii) states that the joint limit in the upcoming theorem is non-degenerate.

Under these assumptions, joint asymptotic normality can be stated for $\boldsymbol{\widebar{\beta}}$, $\boldsymbol{\Tilde{\sigma}}$ and $\boldsymbol{\Tilde{\varphi}}$.

\begin{theorem}[CLT for MG and PMG estimators with true errors] \label{th:AsyMG}
    Under Assumptions~ \ref{as:SVAR}, \ref{as:Proxy}, \ref{as:mixing_conditions} and the joint limit of $ (N,T) \rightarrow ( \infty, \infty ) $, 
     where $N/T\to 0$, 
    \begin{align} \label{eq:AsyMG}
    \begin{split}
    	\sqrt{NT} & 
        \begin{pmatrix}
            \boldsymbol{\widebar{\beta}} - \boldsymbol{\beta} \\
            \boldsymbol{\Tilde{\sigma}} - \boldsymbol{\sigma} \\
            \boldsymbol{\Tilde{\varphi}} - \boldsymbol{\varphi}
        \end{pmatrix}
        \overset{d}{\longrightarrow} \mathcal{N} \left(
        \begin{pmatrix}
            0 \\
            0 \\  
            0
        \end{pmatrix}
        , 
        \begin{pmatrix}
            \Tilde{V}_{11} & \Tilde{V}_{21}' & \Tilde{V}_{31}'\\
            \Tilde{V}_{21} & \Tilde{V}_{22} & \Tilde{V}_{32}' \\  
            \Tilde{V}_{31} & \Tilde{V}_{32} & \Tilde{V}_{33}
        \end{pmatrix}
        \right).
        \end{split}
    \end{align}
\end{theorem}

To prove this result, we apply a CLT for triangular arrays of near-epoch dependent (NED) processes on $\alpha$-mixing random fields \citep[][]{jenish2012}. Besides the convergence result itself, the unconditional CLT provides the joint limiting covariance matrix $\Tilde{V}$. The definition of the blocks in $ \Tilde{V} $ and the detailed proof are given in Appendix \ref{sec:AppAsymp}.

\begin{corollary}[Unconditional CLT for MG and PMG estimators] \label{cor:clt_pmg}
    Under the assumptions of Theorem \ref{th:AsyMG} and the joint limit of $ (N,T) \rightarrow ( \infty, \infty ) $, 
    \begin{align} \label{eq:UncondAsyMG}
    \begin{split}
    	\sqrt{NT} & 
        \begin{pmatrix}
            \boldsymbol{\widebar{\beta}} - \boldsymbol{\beta} \\
            \boldsymbol{\widebar{\sigma}} - \boldsymbol{\sigma} \\
            \boldsymbol{\widebar{\varphi}} - \boldsymbol{\varphi}
        \end{pmatrix}
        \overset{d}{\longrightarrow} \mathcal{N} \left(
        \begin{pmatrix}
            0 \\
            0 \\  
            0
        \end{pmatrix}
        , 
        \begin{pmatrix}
            {V}_{11} & {V}_{21}' & V_{31}'\\
            V_{21} & V_{22} & V_{32}' \\  
            V_{31} & V_{32} & V_{33}
        \end{pmatrix}
        \right).
        \end{split}
    \end{align}
\end{corollary}

This result is a direct consequence of Theorem \ref{th:AsyMG}. Analogously to \cite{JentschLunsford2021} it suffices to consider the differences between $ \boldsymbol{\widebar{\sigma}} $ and $ \boldsymbol{\tilde{\sigma}} $ and $ \boldsymbol{\widebar{\varphi}} $ and $ \boldsymbol{\tilde{\varphi}} $, which vanish under the long-panel regime $N/T\to 0$, as shown in Appendix~\ref{sec:AppAsymp}. Note that $ \Tilde{V} $ of Theorem \ref{th:AsyMG} is used for the asymptotic results of the bootstrap estimators in Section \ref{sec:BootAsy}. With Corollary \ref{cor:clt_pmg}, we state finiteness of the $ 3 \times 3 $ blocks in the covariance matrix $ V $ but omit their exact definition for space reasons.

\begin{corollary}[CLT for structural IRF] \label{cor:CLT_SIRF}
    Under the assumptions of Theorem \ref{th:AsyMG} and the joint limit of $ (N,T) \rightarrow ( \infty, \infty ) $, 
    \begin{align} \label{eq:Asy_IRF}
    \begin{split}
    	\sqrt{NT} & \textup{ vec} \left( \widebar{\Theta}_h - \Theta_h \right) \overset{d}{\longrightarrow} \mathcal{N} \left(0, \Sigma_{\widebar{\Theta},h}  \right) \textup{\quad for } h=1,\ldots,H. \\
    \end{split}
    \end{align}
\end{corollary}
The estimates of the structural IRF are \( \widebar{\Theta}_h = \widebar{\Xi}_h \widebar{\mathsf{B}}^\circ\widehat{\mathsf{Q}} \), i.e., functional forms of the MG estimates $ \widebar{A}$, $ \widebar{\mathsf{B}}^\circ $ and pooled $ \widehat{\mathsf{Q}} $. Thus, the result follows from the application of the delta method with a continuously differentiable function $g$. The asymptotic covariance matrix is  
\[
\Sigma_{\widebar{\Theta},h} = Dg(\theta_0) \Sigma Dg(\theta_0)'.
\]

%% file: 3_Bootstrap.tex
\section{Bootstrap inference for panel SVARs} \label{sec:Bootstrap}

To estimate the limiting distributions of Section~\ref{sec:AsyInference} consistently, this section considers block bootstrapping. The idea of block bootstrapping is to arrange blocks of consecutive data points to preserve temporal or cross-sectional dependencies during the resampling. We propose and later evaluate a recursive-design moving-block bootstrap procedure where the resampling scheme takes particular account of residual structure in both dimensions of the panel. Specifically in Section~\ref{sec:BootBC}, we moreover consider adjustments for panel SVAR models estimated from small samples, most notably, a panel-block bootstrap-after-bootstrap. 

\subsection{Residual-based panel-block bootstrap} \label{sec:BootScheme}
The necessary steps to implement our adaptation of the moving-block bootstrap procedure are as follows. Starting from this general form, the empirical practitioner can configure a wide variety of resampling schemes, as shown in {Remark~3.1} and listed in Table~\ref{tab:Methods}.
\begin{algorithm}[Panel-block bootstrap, PBB] \label{algo:pmb} \end{algorithm}
\begin{enumerate}[label=\textbf{Step \arabic*},align=left]
    \item Estimate the $ N $ individual VAR($ p $) models of Eq.~\eqref{eq:PVAR} by least squares using the data panel $ \boldsymbol{y}_{it} $ to get their \textit{(i)} $ K \times (n_d + Kp) $ coefficient matrices $ \widehat{A}_{i} $ and \textit{(ii)} $ K \times T $ residual matrices $ \widehat{U}_i = Y_i - \widehat{A}_i Z_i $ for $ i=1,\ldots,N $. In the course of point estimation, also identify $ \widehat{\textsf{B}}_i $ (e.g.~by the PMG of Eq.~\eqref{eq:IDpooled_estimator}), calculate the panel mean $ \widebar{A}, \widebar{\mathsf{B}} $, and derive the structural IRFs $ \widebar{\Theta}_h = \widebar{\Xi}_h \widebar{\mathsf{B}} $, although the subsequent iterations of \textbf{Steps~2--5} do not require this.

    \item Resample the residuals $ \widehat{\boldsymbol{u}}_{it} $ into $ \widecheck{\boldsymbol{u}}_{it}^* $ using panel moving-blocks with the pre-selected temporal and cross-sectional size $ \left( b_{(t)}, b_{(i)} \right) $. For this, draw independently with equal probability and replacement \textit{(i)} $ T^\circ = [T/b_{(t)}] $ integers $ t^\circ $ from $ \{ 0,\ldots,T-b_{(t)} \} $ and \textit{(ii)} $ N^\circ = [N/b_{(i)}] $ integers $ i^\circ $ from $ \{ 0,\ldots,N-1 \} $. Here, `$ [\cdot] $' rounds up to the nearest integer so that $ T^\circ b_{(t)} \geq T $ and $ N^\circ b_{(i)} \geq N $. Then, assemble and trim the selection vectors
    \begin{align} \label{eq:Resampling}
    \begin{split}
        \boldsymbol{t^*} & = \left[ t^\circ_1 + \boldsymbol{s_{(t)}}, \ldots, t^\circ_{T^\circ} + \boldsymbol{s_{(t)}} \right]_{1:T}, \ \boldsymbol{s_{(t)}} = \left( 1, \ldots, b_{(t)} \right), \text{\quad and} \\
        \boldsymbol{i^*} & = \left[ i^\circ_1 + \boldsymbol{s_{(i)}}, \ldots, i^\circ_{N^\circ} +\boldsymbol{s_{(i)}} \right]_{1:N}, \ \boldsymbol{s_{(i)}} = \left( 1, \ldots, b_{(i)} \right), \text{\quad with } i^*_i = i^*_i - N \text{ if } i^*_i > N. 
    \end{split}
    \end{align}
    Finally, shuffle the residuals and, if selected, the proxies for each $ i $ by $ i^*_i $ according to
    \begin{align} \label{eq:Shuffle}
    \begin{split}
         [\widehat{U}_{i^*_i}]_{\boldsymbol{t^*}} = & [\widehat{\boldsymbol{u}}_{i^*_i, t^*_1} : \ldots : \widehat{\boldsymbol{u}}_{i^*_i, t^*_T}] = [\widecheck{\boldsymbol{u}}_{i1}^* : \ldots : \widecheck{\boldsymbol{u}}_{iT}^* ] = \widecheck{U}^*_i \text{\quad and} \\
         [{M}_{i^*_i}]_{\boldsymbol{t^*}} = & [{\boldsymbol{m}}_{i^*_i, t^*_1} : \ldots : {\boldsymbol{m}}_{i^*_i, t^*_T}] = [{\boldsymbol{m}}_{i1}^* : \ldots : {\boldsymbol{m}}_{iT}^* ] = {M}^*_i.
    \end{split}
    \end{align}

    \item Re-center $ \widecheck{\boldsymbol{u}}_{it}^* $ by subtracting its bootstrap expectation $ \text{E}^*( \widecheck{\boldsymbol{u}}_{it}^* ) $ so that the recentered errors satisfy $ \text{E}^*( \boldsymbol{u}_{it}^* ) = 0 $ for all $ i,t $. Specifically, estimate $ \text{E}^*( \widecheck{\boldsymbol{u}}_{it}^* ) $ by its sample mean and subtract this mean for $ i =1,\ldots,N, \ j = 0,\ldots,T^\circ-1, $ and $ s_{(t)} = 1,\ldots,b_{(t)} $ as follows: 
    \begin{align} \label{eq:Recentering}
    \begin{split}
    	\boldsymbol{u}_{i, jb_{(t)} + s_{(t)}}^* = \widecheck{\boldsymbol{u}}_{i, jb_{(t)} + s_{(t)}}^* - \frac{1}{T - b_{(t)} + 1} \sum_{r = 0}^{T - b_{(t)}} \widehat{\boldsymbol{u}}_{i^*_i, r + s_{(t)}}. 
    \end{split}
    \end{align}

    \item Generate the bootstrap sample $ \boldsymbol{y}_{it}^* $ recursively from $ \boldsymbol{u}_{it}^* $. For each $ i $, initialize the pre-sample by $ [\boldsymbol{y}_{i, -p+1}^* : \ldots : \boldsymbol{y}_{i,0}^*] = [\boldsymbol{y}_{i, -p+1} : \ldots : \boldsymbol{y}_{i,0}] $ and reconstruct recursively
    \begin{align} \label{eq:Recursive}
    \begin{split}
    	\boldsymbol{y}_{it}^* = \widehat{\Phi}_i \boldsymbol{d}_{it} + \widehat{A}_{i,1} \boldsymbol{y}_{i,t-1}^* + ... + \widehat{A}_{i,p} \boldsymbol{y}_{i,t-p}^* + \boldsymbol{u}_{it}^* \text{\quad over } t=1,\ldots,T .
    \end{split}
    \end{align}

    \item Re-estimate the bootstrapped parameters using the bootstrap sample $ \boldsymbol{y}_{it}^* $. By the very same methods of \textbf{Step~1}, estimate $ \widehat{A}_i^*, \widehat{U}_i^*, \widehat{\Sigma}_{u,i}^* $, identify $ \widehat{\textsf{B}}_i^* $, and calculate their panel mean $ \widebar{A}^*, \widebar{\mathsf{B}}^* $. Finally, derive the structural IRFs $ \widebar{\Theta}_h^* = \widebar{\Xi}_h^* \widebar{\mathsf{B}}^* $.
    
    \item Estimate the metrics of the sampling variance after $ R_{\text{bt}} $ iterations of \textbf{Steps~2--5}. For example, construct confidence intervals surrounding $ \widebar{\Theta}_h $ at a level of $ \alpha \in (0,1) $ using standard \citep{Efron1979} or Hall's \citeyearpar{Hall1992} $ \alpha/2 $ and $ (1 - \alpha/2) $ percentiles of $ \{ \widebar{\Theta}^*_{h,1}, \ldots, \widebar{\Theta}^*_{h,R_{\text{bt}}} \} $.~\footnote{See \citet[App.~D.3]{Lutkepohl2005} or \citet[ch.~12.2.6]{KilianLutkepohl2017}, who give a detailed overview on the different construction methods for confidence intervals.}
\end{enumerate}

\textbf{Remark~3.1.} $ \left( b_{(t)}, b_{(i)} \right) $ allows to configure the resampling scheme in \textbf{Step~2}. For example, block sizes of $ b_{(t)} = b_{(i)} = 1 $ imply an \textit{iid.}~resampling in both dimensions. With fixed $ \boldsymbol{i^*} = (1,\ldots,N) $ as a single block, only the temporal dimension is resampled like in the moving-block bootstrap by \citet{BruggemannEtAl2016}, but the $ NK \times b_{(t)} $ blocks now stretch over the panel. With $ \boldsymbol{t^*} = (1,\ldots,T) $ on the other hand, the resampling scheme collapses into a cross-sectional resampling as favored by \citet{Kapetanios2008} for single-equation panel models. Note that, despite the variety, all of these schemes preserve the panel structure because the indices $i$ and $t$ separately have their specific $ i^*_i $ and $ t^*_t $ in each iteration. Assessments of these schemes and illustrations of their implied panel structures are given in Section~\ref{sec:SimSchemes}.

\textbf{Remark~3.2.} The temporal resampling and recentering of \textbf{Steps~2} and \textbf{3} proceed just as in \citet{BruggemannEtAl2016}. Although we do not explicitly define and glue residual blocks here, the consecutive resampling indices in $ \boldsymbol{t^*} $ imply the same temporal moving-block structure to be present in $ \widecheck{U}^*_i $ as illustrated in Section~\ref{sec:SimSchemes}. Within moving-block bootstraps, the time periods at the beginning and end of the time series sample have fewer neighbors (e.g.~only one direct neighbor in case of $\{ 1,T \} $) than the other residuals. In effect, their probabilities to be drawn into $ \boldsymbol{t^*} $ are lower than those of interior periods such that the estimator for $ \text{E}^*( \widecheck{\boldsymbol{u}}_{it}^* ) $ of \citet{BruggemannEtAl2016} applies in Eq.~\eqref{eq:Recentering}. Notably, if the temporal dimension is resampled \textit{iid.}, there is just $ s_{(t)}=1 $ single mean subtracted from all residuals of a given individual.

\textbf{Remark~3.3.} The cross-sectional resampling and recentering of \textbf{Steps~2} and \textbf{3} has been simplified. Without a pre-defined structure, the cross-sectional dimension offers some freedom to pre-select a convenient ordering. Firstly, the cross-section has thus been re-ordered to maximize correlation between neighbors. This mimics a spatial ordering equivalent to the serial ordering of time series, where the error terms that are ``close'' to each other have a stronger correlation. The weak CSD is thus locally present in neighborhoods and can be captured in the blocks of consecutive residuals.
Secondly, the last and the first individuals have been connected as neighbors in a circular moving-block resampling.\footnote{\citet{PolitisRomano1992} propose a circular resampling scheme for non-overlapping blocks of observations in their fixed-design bootstrap procedure. See also \citet[ch.~2.7.1]{Lahiri2003}.} This allows to capture also their potential residual correlation and assigns equal draw probability to all individuals. Notably, the mean in \textbf{Step~3} as an estimator uses residual time series of individual $ i^*_i $ as the sample so that the recentering is observation-specific $ \text{E}^*( \boldsymbol{u}_{it}^* ) = 0 $ after resampling in any dimension. Here, the sub-selection \textit{(i)} of sample blocks shifted by $ s_{(t)} $ accounts for the unequal draw probability in the temporal dimension and \textit{(ii)} of individual $ i^*_i $ accounts for the heterogeneity in the cross-sectional dimension.

\subsection{Asymptotic bootstrap theory} \label{sec:BootAsy}
Similarly to \cite{JentschLunsford2021}, we presume that the process $ \mathbf{x}_{it} = \left( \boldsymbol{u}_{it}', \boldsymbol{m}_{it}' \right)' $ obeys 

\begin{assumption}[Cumulant summability]\label{as:cumulants}
Denoting $\mathbf{x}_{it}=(x_{it,1},\ldots,x_{it,K+L})'$ for
$i\in\mathbb N_+, t\in\mathbb Z$ the random field
$(\mathbf{x}_{it})_{i\in\mathcal \mathbb N_+, t\in\mathbb Z}$ has absolutely summable
joint cumulants up to order eight. That is, for every
$q=2,\ldots,8$ and every coordinate vector
$a=(a_1,\ldots,a_q)\in\{1,\ldots,K+L\}^q$,
\[
\sup_{i\in\mathbb N_+, t\in\mathbb Z}
\sum_{\substack{h_2,\ldots,h_q\in\mathbb Z^2\\
(i,t)+h_r\in\mathcal (\mathbb N_+\times\mathbb Z),\ r=2,\ldots,q}}
\left|
\operatorname{cum}
\left(
x_{it,a_1},
x_{(i,t)+h_2,a_2},
\ldots,
x_{(i,t)+h_q,a_q}
\right)
\right|
<\infty .
\]
Here, for $r=2,\ldots,q$ and
$h_r=(h_{r,i},h_{r,t})\in\mathbb Z^2$,
$x_{(i,t)+h_r,a_r}
=
x_{i+h_{r,t},\,t+h_{r,t},\,a_r},$
and $\operatorname{cum}(\cdot)$ denotes the joint cumulant as defined in \citet{Lukkarinen2018}. 
\end{assumption}

A joint cumulant summability condition for random fields works similarly to the time series case \citep[see e.g.,]{Brillinger1981}, as pointed out by \citet{Grainger2026} for the case of higher order cumulants. Together with the assumptions imposed in Section~\ref{sec:AsyInference}, we show that the conditional distributions of the centered bootstrap estimators approximate the sampling distribution of the corresponding original estimators. Let
\[
    P^* \equiv P^*_{NT}(\cdot)
    =
    \mathbb P\!\left(
    \cdot \mid 
    \sigma\!\left\{\boldsymbol y_{it}:
    i\leq N,\ t\leq T
    \right\} \right)
\]
denote the bootstrap probability measure induced by the resampling scheme conditional on the observed sample. Bootstrap consistency then means that the conditional law under \(P^*_{NT}\) converges, in probability, to the sampling law of the original statistic.

\begin{theorem}[Residual-based PBB consistency]\label{th:pbb_consist}
    Under (i) Assumptions~\ref{as:SVAR}, \ref{as:Proxy}, \ref{as:mixing_conditions}, and \ref{as:cumulants}, 
    (ii) Assumption \ref{as:random_parameters} with $ \left\|\Sigma_{\eta_\sigma} \right\| = o\left(T^{-1}\right)$, 
    (iii) $ b_{(i)}, b_{(t)} \rightarrow \infty $ such that $b_{(i)}/N \rightarrow 0$, $b_{(t)}/T\rightarrow 0$, $ b_{(i)}b_{(t)}^3/T \rightarrow 0  $ as $ N,T \rightarrow \infty $ and $N/T\to 0$ ,
    \begin{align} \label{eq:AsyPBB_Esti}
    \begin{split}
    	\underset{ \boldsymbol{x} \in \mathbb{R}^{K^\dagger}}{\textup{sup}} \left| P^* \left( \sqrt{NT}
        \begin{pmatrix}
           \boldsymbol{\widebar{\beta}}^*  - \boldsymbol{\widebar{\beta}} \\
           \boldsymbol{\widebar{\sigma}}^* - \boldsymbol{\widebar{\sigma}} \\    
           \boldsymbol{\widebar{\varphi}}^* - \boldsymbol{\widebar{\varphi}} \\  
        \end{pmatrix} 
        \leq \boldsymbol{x} \right) - \mathbb P \left( \sqrt{NT} 
        \begin{pmatrix}
           \boldsymbol{\widebar{\beta}}  - \boldsymbol{\beta} \\
           \boldsymbol{\Tilde{\sigma}} - \boldsymbol{\sigma} \\    
           \boldsymbol{\Tilde{\varphi}} - \boldsymbol{\varphi} \\  
        \end{pmatrix}
        \leq \boldsymbol{x} \right) \right| \overset{p}{\longrightarrow} 0,
    \end{split}
    \end{align}
    with $ K^\dagger = K^2p + (K^2 + K)/2 + KL $.
\end{theorem}
The proof is given in Appendix \ref{sec:AppAsymp}. We state two remarks on the asymptotics:

\textbf{Remark~3.4.} The assumption on $\eta_{\sigma_i}$ listed under (ii) relates to cross-sectional heteroskedasticity in $\Sigma_u$. It accommodates standard patterns of near-homogeneous heterogeneity, e.g., local-to-homogeneous deviations, in the spirit of the local-alternative framework in \citet{PesaranYamagata2008}, as well as fixed-size deviations affecting at most $ o(N) $ units. Covering patterns with asymptotically non-vanishing cross-sectional heterogeneity, such as spatial heterogeneity or unrestricted random-coefficient heterogeneity \citep[e.g.][]{Pesaran2006}, would require adapting the bootstrap algorithm, e.g., applying localized cross-sectional resampling, or adapting the estimation technique. In contrast to these large sample results, the MC study below assesses finite-sample performance under general heteroskedasticity assumptions. 

\textbf{Remark~3.5.} When restricting resampling to one dimension only, i.e., either cross-sectional or time resampling, the asymptotics hold when assuming mixing in the respective dimension only and independence of the residuals in the other. Note that this can be regarded as a special case of Theorem \ref{th:pbb_consist}. Under dependence in the non-sampled dimension, different asymptotic theory is needed.

Asymptotic validity of the panel-block bootstrap of $ \widebar{\Theta}_h = f_h{\left( \boldsymbol{\widebar{\beta}}, \boldsymbol{\widebar{\sigma}}, \boldsymbol{\widebar{\varphi}} \right)} $ follows directly by
\begin{corollary}[Asymptotic validity of bootstrap structural impulse responses] \label{cor:asym_PBBSIRF}
    Under the assumptions of Theorem \ref{th:pbb_consist},
    \begin{align} \label{eq:AsyPBB_IRF}
    \begin{split}
    	\underset{ \boldsymbol{x} \in \mathbb{R}^{K^2}}{\textup{sup}} \left| P^* \left( \sqrt{NT} \textup{vec}\left(\widebar{\Theta}_h^* - \widebar{\Theta}_h \right) \leq \boldsymbol{x} \right) - \mathbb P \left( \sqrt{NT} \textup{vec}\left(\widebar{\Theta}_h - \Theta_h \right) \leq \boldsymbol{x} \right) \right| \overset{p}{\longrightarrow} 0.
    \end{split}
    \end{align}
\end{corollary}

\subsection{Small-T adjustments for MG estimation} \label{sec:BootBC}
Although consistently estimated and correctly bootstrapped as $ T \to \infty $ and $ \ N \to \infty $, the three parameters in $ \Theta_h = f_h{\left( \boldsymbol{\beta}, \boldsymbol{\sigma}, \boldsymbol{\varphi} \right)} $ are subject to bias in small samples. Using the same scenarios of data generation as the MC study in Section~\ref{sec:Simulation}, our introductory simulations in Appendix~\ref{sec:SimRes_point} compare the performance of point estimators with and without small-sample adjustments. The results illustrate small-$ T $ biases in $ \Theta_h = f_h{\left( \boldsymbol{\beta}, \boldsymbol{\sigma}, \boldsymbol{\varphi} \right)} $ and also confirm effective remedies. 

On the individual level, the estimators of Section~\ref{sec:Estimation} exhibit small-$ T $ biases already known from the time series literature. The OLS estimates $ \boldsymbol{\widehat{\beta}}_i $ for a persistent VAR process have a bias that diminishes the root size of the VAR-companion matrix.\footnote{See \citet{NichollsPope1988} and \citet{Pope1990}. In general, the magnitude and direction of the bias depend on the deterministic terms and root size of the VAR model \citep[ch.~2.3.3]{KilianLutkepohl2017}.} Consequently, the IRFs have distributions that are skewed towards the zero-line and understate the response persistence \citep{Kilian1998}. This is most prominent at mid-range horizons $ h $, where the reduced-form $ \widehat{\Xi}_{i,h} $ is the exponentiation of biased $ \boldsymbol{\widehat{\beta}}_i $. For $ h \to \infty $, the biased and actual IRFs both collapse to the zero-line under the assumed (and estimated) stability of the VAR process. For small $ h $, a twofold bias of $ \mathsf{\widehat{B}}_i $ is more prominent in $ \widehat{\Theta}_{i,h} = \widehat{\Xi}_{i,h} \mathsf{\widehat{B}}_i $: Firstly, $ \widehat{\Sigma}_{u,i} $ in Eq.~\eqref{eq:sample_moments} underestimates $ {\Sigma}_{u,i} $ by a factor of $ T/(T - n_Z) $, where $ n_Z = n_d + Kp $ denotes the number of regressors per equation. Secondly, the individually identified $ \mathsf{B}_i = \mathsf{B}_i \mathsf{Q}_i $ rely on imperfect correlation between $ \boldsymbol{m}_{it}$ and $ \boldsymbol{\epsilon}_{1,it} $. If the proxy strength $ \boldsymbol{\mathsf{d}}_1 $ in Eq.~\eqref{eq:ID_svd} is low compared to the sampling variability of $ \boldsymbol{\widehat{\mathsf{d}}}_1 $ (as is typical in small samples), the unfavorable ratio induces an instrumental variable bias on $ \mathsf{\widehat{Q}}_i $ and confidence intervals may become very wide or even unbounded.\footnote{Compare \citet{AndrewsEtAl2019} on instrumental variables (IV) regression. Using the \textit{local-to-zero} asymptotics proposed by \citet{StaigerStock1997} for IV regression, \citet{OleaEtAl2021} discuss weak proxies of asymptotically vanishing proxy relevance in individual proxy SVARs. In Appendix~\ref{sec:SimRes_point}, we also show results on these individual proxy SVARs next to the panel results in order to illustrate the issues of weak proxies.}

On the panel level, these individual small-$ T $ biases are averaged and thus persist in the MG estimates.  
The individual $ \widehat{\Theta}_{i,h} $ often hide $ \text{Bias} \left[ \text{vec} \left( \widehat{\Theta}_{i,h} \right) \right] = \E \left[  \text{vec} \left( \widehat{\Theta}_{i,h} - \Theta_{i,h} \right) \right] $ under a large $  \text{Cov} \left[ \text{vec} \left( \widehat{\Theta}_{i,h}  \right) \right] $ such that this bias has a minor share within the mean squared error
\begin{align} \label{eq:MSE_IRF}
\begin{split}
    \text{MSE} \left[ \text{vec} \left( \widehat{\Theta}_{i,h} \right) \right]
    & = \text{Cov} \left[ \text{vec} \left( \widehat{\Theta}_{i,h}  \right) \right] + \text{Bias} \left[ \text{vec} \left( \widehat{\Theta}_{i,h}  \right) \right] \text{Bias} \left[ \text{vec} \left( \widehat{\Theta}_{i,h}  \right) \right]'.
\end{split}
\end{align}
MG estimation only decreases the estimator variance under $ N \to \infty $ such that the three small-$ T $ biases become more pronounced within the MSE. The simulations in Appendix~\ref{sec:SimRes_point} even indicate that the 90\%-quantile bands of the uncorrected MG estimates $ \widebar{\Theta}_{h} $ miss the actual $ \Theta_{h} $.

Against this, we first consider small-$ T $ adjustments for $ \left\{ \boldsymbol{\sigma}, \boldsymbol{\varphi} \right\} $. As is standard for VAR estimation \citep[p.~31]{KilianLutkepohl2017}, we adjust for the degrees of freedom $ T - n_Z $ and use $ \widehat{\boldsymbol{\sigma}}_{\text{BC},i} = T/(T - n_Z) \cdot \widehat{\boldsymbol{\sigma}}_{i} $. Notably, $ \mathsf{\widehat{Q}} $ is invariant to this correction because the bias is ``scaled out'' irrespective of the denominator used for $ \mathsf{\widehat{B}}_i^\circ = \text{chol} \left( \widehat{\Sigma}_{u,i} \right) $ and $ \widebar{\Sigma}_{\epsilon^\circ m} $. This becomes evident as the scaling $ \mathsf{D} $ is dropped from the SVD of $ \left[ \widebar{\Sigma}_{\epsilon^\circ m}  : 0_{K \times K_2} \right] = \widehat{\mathsf{U}} \widehat{\mathsf{D}} \widehat{\mathsf{V}}' $. Moreover, the PMG-identification equipped with the pooled $ \widehat{\mathsf{Q}} $ increases the sample by a factor of $ N $ and thus the precision of $ \boldsymbol{\widehat{\mathsf{d}}}_1 $. Although not a panacea, this pooled estimator avoids notable distortions under \textit{(i)} standard panel sizes of empirical applications, \textit{(ii)} the ``individually weak'' proxies considered in the MC simulations of \cite{JentschLunsford2021}, and \textit{(iii)} elemental heterogeneity in the identification problem (unlike MG estimation) as shown in Appendix~\ref{sec:SimRes_local}.

In contrast to the pair of $ \left\{ \boldsymbol{\sigma}, \boldsymbol{\varphi} \right\} $, the estimator for $ \boldsymbol{\beta} $ is subjected to an iterative procedure. The pooled estimation of the dynamic panel model could imply a sufficiently large panel sample of size $ NT $, but would be biased under fixed effects \citep{Nickell1981} and inconsistent under coefficient heterogeneity \citep{PesaranSmith1995}. In direct reference to Algorithm~\ref{algo:pmb}, we resort to a bootstrap-based technique instead and expand the bootstrap-after-bootstrap from \citet{Kilian1998} by the cross-sectional dimension. We obtain an equivalent procedure for the heterogeneous panel VAR model, which is carried out in the following steps.
\begin{algorithm}[Panel bootstrap-after-bootstrap]  \label{algo:bab} \end{algorithm}
\begin{itemize}[align=left]
    \item[\textbf{Step~1a}] Estimate the individual bias $ \Psi_i = \E \left[ \widehat{\boldsymbol{\beta}}_i - \boldsymbol{\beta}_i \right] \ \forall i $. For this, apply the bootstrap Algorithm~\ref{algo:pmb} to generate $ R_{\text{bt}} $ replications of $ \widehat{\boldsymbol{\beta}}_i^* $ and calculate $ \widehat{\Psi}_i = \left( R_{\text{bt}}^{-1} \sum_j^{R_{\text{bt}}} \widehat{\boldsymbol{\beta}}_{ij}^* \right) - \widehat{\boldsymbol{\beta}}_i $.

    \item[\textbf{Step~1b}] Calculate the bias-corrected coefficients $ \widehat{\boldsymbol{\beta}}_{\text{BC},i} \ \forall i $. 
    To keep the VAR processes stationary, the diminution bias of $ \widehat{\boldsymbol{\beta}}_i $ is only corrected if the modulus of the largest root of the companion matrix obeys $ \left| \lambda \left( \widehat{\boldsymbol{\beta}}_i \right) \right| < 1 $. In this case, calculate $ \widehat{\boldsymbol{\beta}}_{\text{BC},ij} = \widehat{\boldsymbol{\beta}}_i - \delta_j \widehat{\Psi}_i $ within an iterative stationarity correction that initializes the weight with $ \delta_1 = 1 $, reduces $ \delta_j \rightarrow 0 $ over increasing iterations $ j $ as long as $ \left| \lambda \left( \widehat{\boldsymbol{\beta}}_{\text{BC},ij} \right) \right| \geq 1 $, and stops if $ \left| \lambda \left( \widehat{\boldsymbol{\beta}}_{\text{BC},ij} \right) \right| < 1 $. Choose $ \widehat{\boldsymbol{\beta}}_{\text{BC},i} = \widehat{\boldsymbol{\beta}}_{\text{BC},ij} $ in case of stationarity and $ \widehat{\boldsymbol{\beta}}_{\text{BC},i} = \widehat{\boldsymbol{\beta}}_i $ in case of non-stationarity.

    \item[\textbf{Step~2a}] Re-conduct \textbf{Steps~2--5} from the panel-block bootstrap Algorithm~\ref{algo:pmb} using $ \widehat{\boldsymbol{\beta}}_{\text{BC},i} $. In \textbf{Step~4} of this second-step bootstrap, $ \widehat{\boldsymbol{\beta}}_{\text{BC},i} $ replaces the biased first-step estimates $ \widehat{\boldsymbol{\beta}}_i $ in Eq.~\eqref{eq:Recursive}, which recursively generates the bootstrap samples. 

    \item[\textbf{Step~2b}] Re-estimate the bias-corrected coefficients $ \widehat{\boldsymbol{\beta}}_{\text{BC},i}^* $ from $ \widehat{\boldsymbol{\beta}}_{\text{BC},i} $. Here, reuse\footnote{The \textit{double bootstrap} would re-estimate and correct the bias $ \widehat{\Psi}_i^* $ based on a third layer of nested bootstrap loops in \textbf{Step 2b}. However, this option has much smaller room for improved performance and is computationally expensive, especially in the panel VAR setting, where the factor $ N $ is effective at each procedural step.} $ \widehat{\Psi}_i $ from \textbf{Step~1a} with the correction procedure of \textbf{Step~1b} to obtain $ \widehat{\boldsymbol{\beta}}_{\text{BC},i}^* \ \forall i $, their panel mean $ \boldsymbol{\widebar{\beta}}_{\text{BC}}^* $, and $ \widebar{\Theta}_{\text{BC},h}^* = \widebar{\Xi}_{\text{BC},h}^* \widebar{\mathsf{B}}_{\text{BC}}^* = f_h{\left( \boldsymbol{\widebar{\beta}}_{\text{BC}}^*, \boldsymbol{\widebar{\sigma}}_{\text{BC}}^*, \boldsymbol{\widebar{\varphi}}^* \right)} $ with $\boldsymbol{\widebar{\sigma}}_{\text{BC}}^* = T/(T - n_Z) \cdot \boldsymbol{\widebar{\sigma}}^* $.

    \item[\textbf{Step~3\ }] Estimate the metrics of the sampling variance after $ R_{\text{bt}} $ iterations of \textbf{Steps~2a+b}. For example, construct confidence intervals surrounding $ \widebar{\Theta}_{\text{BC},h} $ at a level of $ \alpha \in (0,1) $ using standard \citep{Efron1979} or Hall's \citeyearpar{Hall1992} $ \alpha/2 $ and $ (1 - \alpha/2) $ percentiles of $ \{ \widebar{\Theta}_{\text{BC},h,1}^*, \ldots, \widebar{\Theta}_{\text{BC},h,R_{\text{bt}}}^* \} $.    
\end{itemize}

\textbf{Remark~3.6.} For a more flexible specification and straightforward computation, the stationarity correction in \textbf{Step~1b} slightly deviates from the implementation of \citet{Kilian1998}. The weights $ \delta_j $ control how carefully this iterative procedure pushes $ \lambda \left( \widehat{\boldsymbol{\beta}}_{\text{BC},i} \right) $ towards the unit circle. For a bounded number of iterations $ J $, they can be collected in the vector $ \boldsymbol{\delta} = \left( \delta_1, \ldots , \delta_J \right) $. The iterative procedure in \citet{Kilian1998} can be recursively solved and then corresponds to our scheme with the specification $ \boldsymbol{\delta} = \{ \left( \delta_1,  \ldots,\ \delta_{101} \right) \mid \delta_j = \prod_{n=0}^{j-1} \frac{100-n}{100} \} $. In the MC study of Section~\ref{sec:Simulation}, we just specify $ \boldsymbol{\delta} = 1 $ for both corrections of \textbf{Steps~1b} and \textbf{2b} to deactivate the stationarity correction. As it turns out, $ \boldsymbol{\widebar{\beta}}_{\text{BC}} $ is robust to the non-stationary individual extremes in the distribution of $ \widehat{\boldsymbol{\beta}}_{\text{BC},i} $. By contrast, the individual stationarity correction would dampen those extremes in a lopsided way such that the MG then understates the persistence.

\textbf{Remark~3.7.} The small-sample bias corrections do not affect the asymptotic derivations and their validity results. As \citet{Kilian1998} shows, the estimated bias $ \Psi (\boldsymbol{\widehat{\beta}}_i) $ decays at a rate of $ O_p(T^{-1}) $, i.e., faster than the OLS estimator converges (at a rate of $ T^{-1/2}$). For the degrees of freedom, it holds that $ T/(T - n_Z) \to 1 $ and thus $ \widehat{\boldsymbol{\sigma}}_{\text{BC},i} \to \widehat{\boldsymbol{\sigma}}_{i} $ for $ T \to \infty $.

%% file: 4_Simulation.tex
\section{Monte Carlo study} \label{sec:Simulation}
To evaluate the finite-sample performance of the inferential methods, we conduct Monte Carlo (MC) simulations using various methods and panel data settings. In the following, we first describe competing resampling schemes from Table~\ref{tab:Methods} for the panel bootstrap, then define data generating processes, and finally present results from the MC experiments. 
\begin{table}[ht]
    \centering
    \caption{Configurations of the inferential methods}
    \resizebox{0.90\textwidth}{!}{

\input{Tab_Methods}}
    \label{tab:Methods}
\end{table}

\subsection{Resampling schemes} \label{sec:SimSchemes}
Table~\ref{tab:Methods} lists the evaluated methods along their abbreviation and a brief description. Therein, a broad distinction is made between the four main categories of \textit{(i)} cross-sectional resampling, \textit{(ii)} temporal resampling, \textit{(iii)} panel resampling, and \textit{(iv)} mean-group (MG) inference. The prevalent MG-IRF estimation by $ N^{-1} \sum_{i=1}^{N}{ \widehat{\Theta}_{i,h} } = N^{-1} \sum_{i=1}^{N} f_h{\left( \boldsymbol{\widehat{\beta}}_i, \boldsymbol{\widehat{\sigma}}_i, \boldsymbol{\widebar{\varphi}} \right)} $ and the related MG inference by the empirical $ \alpha/2 $ and $ (1 - \alpha/2) $ percentiles of individual $ \{ \widehat{\Theta}_{h,1}, \ldots, \widehat{\Theta}_{h,N} \} $ are easily available ad-hoc methods. They are loosely adapted from \citet{PesaranSmith1995} as a benchmark that emphasizes methodological simplicity. The resampling schemes configured in \textbf{Step 2} of Algorithm~\ref{algo:pmb} accord with \citet[p.~380]{Kapetanios2008}, but the reduced-form residuals are exclusively the elements to be resampled as the recursive design is used. In the MC experiments, we let a generous rule of thumb determine block sizes of $ b_{(i)} \approx N/10 $ and $ b_{(t)} \approx T/10 $ for the block bootstraps. The rule is suggested in the time series literature \citep[e.g.][]{LangeEtAl_fc,BruggemannEtAl2016},  which we simply adopt for cross-sectional resampling, too. Next to each moving-block resampling, sizes of $ b_{(i)} = 1 $ or $ b_{(t)} = 1 $ configure an \textit{iid.} resampling in the respective dimension.

The different configurations of \textbf{Step~2} result in the resampling schemes listed in Table~\ref{tab:Methods}, which shall be described and exemplarily applied to the $ NK \times T $ residual panel matrix
\input{Fig_Uorig}
for a single bootstrap iteration. Each row of the matrix refers to the residual vector of an individual $ i = 1, \ldots, 5 $ (marked by the hue), and each column represents a time period $ t = 1, \ldots, 7 $ (marked by the saturation). In the cases of moving-block bootstraps, sizes of $ b_{(i)} = b_{(t)} = 2 $ are applied such that $ N^\circ = 3 $ cross-sectional and $ T^\circ = 4 $ temporal overlapping blocks (delineated by white borders) are drawn, assembled, and trimmed in each iteration.

\textbf{Cross-sectional resampling.} This scheme selects $ NK $ rows from the $ NK \times T $ matrix randomly with replacement either as \textit{iid.} or in moving blocks, while the columns as the temporal dimension remain fixed. For example, $ N $ individuals are drawn \textit{iid.} ($ \boldsymbol{i^*} = 2, 4, 3, 1, 3 $) or as consecutive rows ($ \boldsymbol{i^*} = 4, 5, 2, 3, 1 $) to construct respectively
\input{Fig_Ucsr}

\textbf{Temporal resampling.} This scheme selects $ T $ columns from the $ NK \times T $ matrix randomly with replacement either as \textit{iid.} or in moving blocks, while the rows as the cross-sectional dimension remain fixed. For example, $ T $ periods are drawn \textit{iid.} ($ \boldsymbol{t^*} = 7, 3, 4, 7, 1, 5, 2 $) or as consecutive columns ($ \boldsymbol{t^*} = 2, 3, 5, 6, 4, 5, 1 $) to construct respectively
\input{Fig_Utmpr}

\textbf{Panel resampling.} This scheme selects $ NK $ rows and $ T $ columns from the $ NK \times T $ matrix randomly with replacement, each either as \textit{iid.} or in moving blocks. Although not illustrated here, the scheme can be configured for both dimensions separately (e.g.~the cross-sectional rows can be resampled as \textit{iid.}, while the temporal columns are resampled in blocks). In the example, $ N $ individuals and $ T $ periods are drawn \textit{iid.} ($ \boldsymbol{i^*} = 2, 4, 3, 1, 3; \ \boldsymbol{t^*} = 7, 3, 4, 7, 1, 5, 2 $) or as consecutive rows and columns ($ \boldsymbol{i^*} = 4, 5, 2, 3, 1; \ \boldsymbol{t^*} = 2, 3, 5, 6, 4, 5, 1 $) to construct respectively
\input{Fig_Upanelr}

As evident from the illustrations, none of the resampling schemes breaks the grid of the dimensions $i,t,$ and $ k$. More specifically, none mixes $ \widehat{u}_{it,k} $ and $ \widehat{u}_{i^\bullet t^\bullet,k^\bullet} $ across variables $ k\neq k^\bullet, $ even if the cross-section is resampled. Further, none dissolves the panel structure of $ \{ i,t \} $ such that the resampled $ i^*_i $ (the hue) is uniform in each row $ i $, and $ t^*_t $ (the saturation) is uniform in each column $ t $. Hence, all presented schemes preserve residual structure specific to rows $ i $ or columns $ t $ beyond the blocks, such as potential cross-sectional heteroskedasticity or common factors respectively. Due to the grid, the panel resampling scheme is invariant to the implementation choice of which dimension $ \{ i,t \} $ is resampled first.

\subsection{Scenarios of data generating processes} \label{sec:SimDGP}
For our MC study, the data generating process (DGP) in \citet{JentschLunsford2021} is transferred to a panel VAR setting. To generate $ \boldsymbol{y}_{it} $, we specify Eq.~\eqref{eq:PVAR} as a bivariate PVAR model without deterministic terms $ \Phi_i \boldsymbol{d}_{it} $ and with a lag order $ p = 2 $ of the baseline VAR coefficients
\begin{align*}
    A_{1} = 
    \begin{pmatrix}
         0.44 & 0.66 \\
        -0.11 & 1.32
    \end{pmatrix} \text{\quad and \quad} A_{2} = 
    \begin{pmatrix}
        -0.18 & 0 \\
        -0.18 & -0.09
    \end{pmatrix}. 
\end{align*}
These parameters imply fairly persistent, yet stable dynamics with moduli of $ 0.930 $ and $ 0.634 $, pertaining to the roots of their characteristic polynomial. In all MC experiments, we introduce slope heterogeneity $ A_{i,1} = \sigma_{(A1),i} \cdot A_1 $ and $  A_{i,2} = \sigma_{(A2),i} \cdot A_2 $ by two scaling factors $ \sigma_{(Aj),i} = 1 + w_j \cdot u_{(A),i} $ with $ w_1 = 0.073 $,  $ w_2 = 0.263 $, and joint $ u_{(A),i} \sim U(-1, +1) $, see Appendix~\ref{sec:HetIRF}. The reduced-form errors are generated by $ \boldsymbol{u}_{it} = \mathsf{B}_i \boldsymbol{\epsilon}_{it} = \mathsf{B}_i^\circ \mathsf{Q}(\theta) \boldsymbol{\epsilon}_{it}^{\ } $ with $ \theta = \pi / 4 $ and baseline
\begin{align*}
    \mathsf{B} & =  
    \begin{pmatrix}
        1 & 0 \\
        0.5 & \sqrt{1-0.5^2}
    \end{pmatrix}  
    \begin{pmatrix}
         \cos \theta & \sin \theta \\
        -\sin \theta & \cos \theta
    \end{pmatrix}
    =
    \begin{pmatrix}
         0.707 & 0.707 \\
        -0.259 & 0.966
    \end{pmatrix}
    \\
    & \text{\quad such that \ } \Sigma_{u} = \mathsf{B}^{\ } \mathsf{B}' = 
    \begin{pmatrix}
        1 & 0.5 \\
        0.5 & 1 
    \end{pmatrix}. 
\end{align*}
This assigns a moderately strong correlation within $ \boldsymbol{u}_{it} $ of an individual unit. The DGP of \citet{JentschLunsford2021} is similar to \cite{BruggemannEtAl2016}, but the rotated $ \mathsf{B} $ is no longer lower-triangular, which thus exhibits no recursive causality and rules out identification by Cholesky decomposition. The single $ L = 1 $ proxy variable necessary for identification is generated by $ m_{it} = \psi_{(m)} \epsilon_{it,1} + e_{(m),it} $ with $ e_{(m),it} \overset{iid}{\sim} \mathcal{N}(0,1) $ and ``individually weak'' proxy strength of $ \psi_{(m)} = 0.17 $. Overall, this DGP parameterization produces hump-shaped IRFs as displayed in Figure~\ref{fig:MeanCI_SN4N20T100}.

We are particularly interested in evaluating how the resampling schemes capture different residual structures in the panel DGP. Hence, we incrementally introduce cross-sectional and temporal dependence in $  \boldsymbol{u}_{it} = \mathsf{B}_i \boldsymbol{\epsilon}_{it} $ via the structural shocks $ \boldsymbol{\epsilon}_{it} $ in four scenarios:
\begin{enumerate}[label=\textbf{SN\arabic*}]	
	\item \textbf{\textit{iid.}} The first scenario specifies the DGP as \textit{iid.}~$ \boldsymbol{\epsilon}_{it} \sim \mathcal{N} \left(0, I_K \right) $ for all dimensions $ i,t,k $.
    
	\item \textbf{\textit{Heteroskedasticity}.} In this scenario, we impose heteroskedasticity in both dimensions of $ \boldsymbol{u}_{it} $. Cross-sectional heteroskedasticity in $ \boldsymbol{u}_{it} $ emerges from parameter heterogeneity in $ \mathsf{B}_i^\circ = \sigma_{(CH),i} \cdot \mathsf{B}^\circ$ by a scaling factor $ \sigma_{(CH),i} \sim U(0.75, 1.25) $. Temporal heteroskedasticity is given as stochastic volatility in a univariate process $ \epsilon_{jt} = \sigma_{(SV),jt} \cdot \varepsilon_{jt} $ for each series $ j = 1,\ldots,NK $. The ``primary shocks'' $ \varepsilon_{jt} \sim \mathcal{N}(0,1) $ are \textit{iid.}~across both $ t,j $, and the volatility process is $ \ln \sigma_{(SV),jt} = \nu_{(SV)} + \rho_{(SV)} \ln \sigma_{(SV),j,t-1} + e_{(SV),jt} $ with volatility innovations $ e_{(SV),jt} \sim \mathcal{N}(0, \sigma_{e(SV)}^2) $, $ \sigma_{e(SV)} = 0.15 $, and a volatility persistence of $ \rho_{(SV)} = 0.85 $ as in \citet[p.~13]{JentschLunsford2021}. Using $ \nu_{(SV)} = (1-\rho_{(SV)}) (- \sigma_{e(SV)}^2 /  (1-\rho_{(SV)}^2)) $, this form of heteroskedasticity still aligns with the unit-variance normalization of $ \epsilon_{jt} $.
    
	\item \textbf{\textit{Cross-sectional dependence}.} Here, cross-correlation is confined to each panel $ k=1,\ldots,K $ of shocks $ \epsilon_{itk} $. We generate $ K $ panels of $ \boldsymbol{\epsilon}_{tk} = \left(\epsilon_{1tk}, \ldots, \epsilon_{Ntk}  \right)' \sim \mathcal{N} \left(0, \Sigma_{\epsilon} \right) $, which are \textit{iid.}~across $ t,k $ and dependent across $ i $ by the $ N \times N $ correlation matrix $ \Sigma_{\epsilon} = 0.25 \cdot \Sigma_{(CF)} + 0.75 \cdot \Sigma_{(CC)} $. The rank-deficient matrix $ \Sigma_{(CF)} = 1_{N \times N} $ introduces a common factor for strong cross-sectional dependence. Similarly to \citet[p.~627]{Herwartz2017}, $ \Sigma_{(CC)} $ is a Toeplitz matrix collecting the elements $ \sigma^2_{(CC)ii^\bullet} = \rho^{\delta_{ii^\bullet}} $ with a moderate spatial correlation of $ \rho = 0.6 $. We yet specify the distances as $ \delta_{ii^\bullet} = \text{min}(|i-i^\bullet|, \ N-|i-i^\bullet|) $ so that the correlation increases again after $ N/2 $ individuals and the first and last individuals enter the cyclical neighborhood. As a result, $ \Sigma_{\epsilon} $ contains $ \sigma^2_{\epsilon,12} = \sigma^2_{\epsilon,1N} = 0.7 $ and becomes in the exemplary case of $ N = 5 $ the circulant matrix
    \begin{align*}
        \Sigma_{\epsilon} =
        0.25 \cdot 1_{N \times N} + 0.75 \cdot
        \begin{pmatrix}
              1  & 0.60 & 0.36 & 0.36 & 0.60 \\
            0.60 &   1  & 0.60 & 0.36 & 0.36 \\
            0.36 & 0.60 &   1  & 0.60 & 0.36 \\
            0.36 & 0.36 & 0.60 &   1  & 0.60 \\
            0.60 & 0.36 & 0.36 & 0.60 & 1
        \end{pmatrix}
        = 
        \begin{pmatrix}
            1    & 0.70 & 0.52 & 0.52 & 0.70 \\
            0.70 &   1  & 0.70 & 0.52 & 0.52 \\
            0.52 & 0.70 &   1  & 0.70 & 0.52 \\
            0.52 & 0.52 & 0.70 &  1   & 0.70 \\
            0.70 & 0.52 & 0.52 & 0.70 & 1
        \end{pmatrix}.
    \end{align*}
	
	\item \textbf{\textit{Two-way dependence}.} Finally, we combine the error generation of SN2 and SN3 to assess the asymptotic performance of the research approaches under this most complex residual structure. For this, we iteratively draw primary shocks $ \left\{ \left\{ \varepsilon_{jt} \right\}_{t=1}^T \right\}_{j=1}^{NK} $ by SN3 instead of \textit{iid.} $ \varepsilon_{jt} \sim \mathcal{N}(0,1) $ and then introduce heteroskedasticity by SN2.
\end{enumerate}

With these DGPs, we generate samples of $ \boldsymbol{y}_{it} $ with four combinations of varying panel sizes $ N \in \{20,100\} \times  T \in \{100, 500\} $. Specifically, we consider a \textit{small-sample case} $ \left\{ N=20, T=100 \right\} $, two \textit{medium-sample cases} $ \left\{ N=20, T=500 \right\}$ and $ \left\{ N=100, T=100 \right\} $, and a \textit{large-sample case} $ \left\{ N=100, T=500 \right\} $. The smaller (resp.~larger) sizes $T,N$ of both medium-sample cases correspond to the small (resp.~large)-sample case so that these incremental increases reveal the effects of a specific dimension size on the inferential performance. All cases obey $ N \leq T $ and thus relate to the typical data panels of macroeconomic applications. Sizes of the total samples found specifically in panel VAR applications are $ \left\{ N = 8, T = 42 \right\} $ in \citet{GambacortaEtAl2014}, $ \left\{ N = 57, T = 52 \right\} $ in \citet{CesaEtAl2015}, $ \left\{ N = 16, T = 147 \right\} $ in \citet{BernothHerwartz2021}, and $ \left\{ N = 14, T = 76 \right\} $ in \cite{Herwartz2017}, but the cross-section is often split into subsets for assessing heterogeneity between groups of individuals.

In our simulations, the panel time series are generated with additional 200 periods to ``burn in''. For this, the effective size $ T $ of the temporal dimension is extended by 200 initial periods, which are discarded after the data generation. The initial periods allow the stochastic process to approach its equilibrium, thus reduce the impact of the initial state of the recursive data generation, and therewith achieve more reliable MC results.

\subsection{Simulation results} \label{sec:SimRes}
We conduct the MC experiments for each inferential method of Section~\ref{sec:SimSchemes} with each DGP setup of Section~\ref{sec:SimDGP} as follows. We generate $ R_{\text{MC}} = 600 $ panel samples from the DGP. For each sample, the individual VAR models are estimated with an intercept just as in \cite{JentschLunsford2021}, but $ \widehat{\Sigma}_{u,i} $ is now based on the degrees of freedom $ T-n_Z $ and $ \widehat{\mathsf{B}}_i $ stems from the pooled SVD estimate of $ \mathsf{Q} $. IRFs of the panel SVAR are derived over a horizon of $ h=0,\ldots,20 $. For the mean-group inference, the confidence bands are constructed as the 90\%-percentile intervals from the set of individual $ \widehat{\Theta}_{ih} $. For each resampling scheme, $ R_{\text{bt}} = 500 $ iterations in the first and second-step bootstrap of Algorithm~\ref{algo:bab} create a set of $\widebar{\Theta}_{h}^{*}$. Standard percentiles\footnote{To keep the dimensionality of MC results reasonable, we focus on the standard percentile intervals by \citet{Efron1979} just like \citet{Kilian1998}, \citet{LutkepohlSchlaak2019}, and \citet{JentschLunsford2021}. \citet[p.~75]{BruggemannEtAl2016} use Hall's \citeyearpar{Hall1992} percentile intervals, but find that their results are relatively similar to those of standard percentile intervals. For the empirical practitioner confronted with typical sample sizes $ T $, this decision is of minor importance after all \citep[compare][ch.~12.5]{KilianLutkepohl2017}. \citet{JentschLunsford2021} also consider \textit{grid MBB AR} confidence sets \citep{AndersonRubin1949,OleaEtAl2021} for normalized IRFs of $ K_1 = L = 1$, which cannot be constructed for standard-deviation IRFs as used in our setup.} with nominal $ \alpha = 0.1$ are constructed from this set. Although possible, we do not mix different resampling schemes between the first-step and second-step bootstrap. The simulation results are finally evaluated by the following performance criteria. Additional performance results, e.g.~on point estimation and local identification, are documented in Appendix~\ref{app:Simultaion}.

\subsubsection{Average point estimates and confidence intervals}  \label{sec:SimRes_MeanCI}
The \textit{average point estimates} and \textit{average confidence intervals} indicate the expected point estimates and their CIs formed by the inferential methods over the $ R_{\text{MC}} = 600 $ replications. They provide a visual overview of the finite-sample performance, in particular, on the potential bias as the deviation from the ``true'' $ \left[ \Theta_h  \right]_{ks} $ and on the width $ {\ell}_{h,ks} $ of the CIs over $ h $. 

\begin{figure}[ht] 
	\centering
	\resizebox{0.98\textwidth}{!}{
		\input{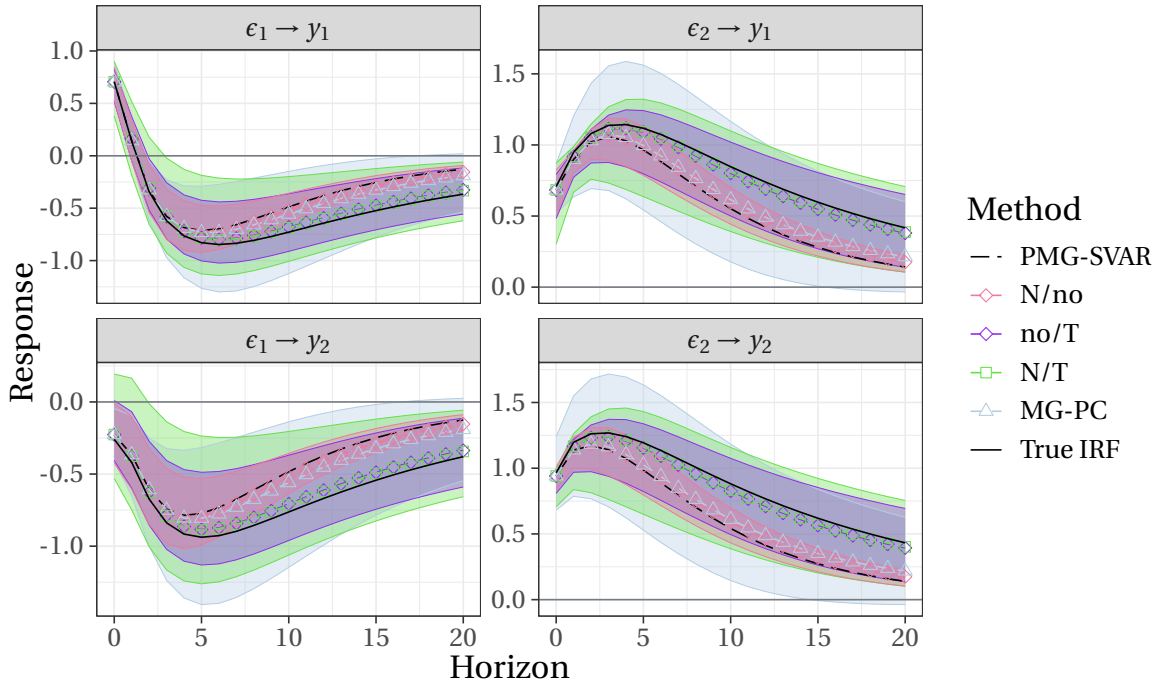}}
	\caption[Average confidence intervals.]{Average confidence intervals $ \widehat{CI}_{h,ks} $ of the inferential methods for the PMG-identified proxy SVARs according to Monte Carlo simulations. $ R_{\text{MC}} = 600 $ replications with sample size $ \left\{ N=20,  T=100 \right\} $ have been used to assess the SN4 \textit{small sample} performance.}
	\label{fig:MeanCI_SN4N20T100}
\end{figure}
Figure~\ref{fig:MeanCI_SN4N20T100} displays the $ 2 \times 2 $ IRF results exemplarily for the SN4 \textit{small sample} with residual two-way dependence as the most challenging DGP setup. Overall, the IRF point estimates $ \widebar{\Theta}_h = f_h{\left( \boldsymbol{\widebar{\beta}}, \boldsymbol{\widebar{\sigma}}, \boldsymbol{\widebar{\varphi}} \right)} $ from PMG-SVAR (dashed black lines) are subject to the small-$ T $ bias of $ \boldsymbol{\widehat{\beta}}_i $ particularly at mid-range horizons after $ h=5 $ before approaching the zero-line in the long run due to the stability of the VAR process. The bias is slightly offset in the (turquoise) MG-IRF estimates by averaging over the skewed sampling distribution of $ \widehat{\Theta}_{i,h} $, where the DGP limits influences from parameter heterogeneity by $ \Theta_h \approx \E \left[ \Theta_{i,h} \right] $, see Appendix~\ref{sec:HetIRF}. The remaining bias is not prominent in the MG-inference due to the wide CIs.

Figure~\ref{fig:MeanCI_SN4N20T100} further illustrates how Algorithm~\ref{algo:bab} corrects the bias only under those schemes that resample in the temporal dimension. The cross-sectional block bootstrap (pink blocks) slightly reduces the deviation from (the black solid lines of) $ \Theta_h $ but is insufficient in view of its very narrow CIs. Its (pink) 90\%-quantiles even miss $ \Theta_h $ at the problematic mid-range horizons. In contrast, the (purple) temporal and (green) panel block bootstraps remove the bias effectively. Their CIs turn out to be centered around $ \Theta_h $ and, in terms of width, take the middle ground between the CIs of the MG-inference and of the cross-sectional bootstrap.

Evidently, our large-scale MC study carries the ``curse of dimensionality'' and thus compels us to condense the simulation results. In total, we have 16 DGP setups (for 4 scenarios $ \times $ 4 sample sizes) producing 84 IRF values (for $ 2 \times 2 $ IRFs $ \Theta_h $ over a horizon of $ H+1 = 21 $), which are subjected to 8 inferential methods from Table~\ref{tab:Methods} and then evaluated by 3 performance criteria $ f_\bullet(\cdot) $. In Figures~\ref{fig:MeanLengths_SN4N20T100} and \ref{fig:CovRates_SN4N20T100}, separate $ 2 \times 2 $ IRF results on the selected SN4 \textit{small sample} scenario are plotted over the full horizon up to $ H=20 $ to give an intuitive and complementary picture of the dynamics. To reduce the dimensionality of Tables~\ref{tab:MeanLengths} and \ref{tab:CoverageProb}, we report averaged results $ \left( \frac{100}{K^2} \cdot \sum_{k,s=1}^{K,K} f_\bullet(\cdot)_{h,ks} \right) $ of all $ 2 \times 2 $ IRF elements in $ \Theta_h $ at $ h = \left\{ 0, 5, 15 \right\} $ for each of the 16 DGP setups. The horizons of $ h = \left\{ 0, 5, 15 \right\} $ are particularly interesting since $ \widebar{\Theta}_0 = \widebar{\mathsf{B}} $ holds the structural coefficients directly, $ {\Theta}_5 $ contains the largest effects among $ {\Theta}_h $, and $ \widebar{\Theta}_{15} $ is highly prone to bias from $ \widehat{\boldsymbol{\beta}}_i $. Differentiating between IRF shocks and target variables $ \left\{s,k\right\} $ is less interesting since their performance behavior is similar, as visible e.g. in Figure~\ref{fig:MeanCI_SN4N20T100}.

\subsubsection{Mean length} \label{sec:SimRes_MeanLengths}
The \textit{mean length} $ \widehat{\ell}_{h,ks} $ of the CIs over all $ R_{\text{MC}} = 600 $ replications can be recognized as the width of the mean CIs in Figure~\ref{fig:MeanCI_SN4N20T100} and shall be considered for all 8 inferential methods and 16 DGP setups. The criterion $ \widehat{\ell}_{h,ks} $ indicates the estimation uncertainty derived by each inferential method. Wider CIs are more likely to contain the zero-line, such as those of the MG-inference at $ h > 15 $ in Figure~\ref{fig:MeanCI_SN4N20T100}, and thus indicate more conservative inference. Further, in line with the econometric theory of Section~\ref{sec:AsyInference}, $ \widehat{\ell} _{h,ks} $ should decrease with increasing sample size and increase with the aggravating residual dependence from SN1 to SN4.

\begin{figure}[ht] 
	\centering
	\resizebox{0.98\textwidth}{!}{
		\input{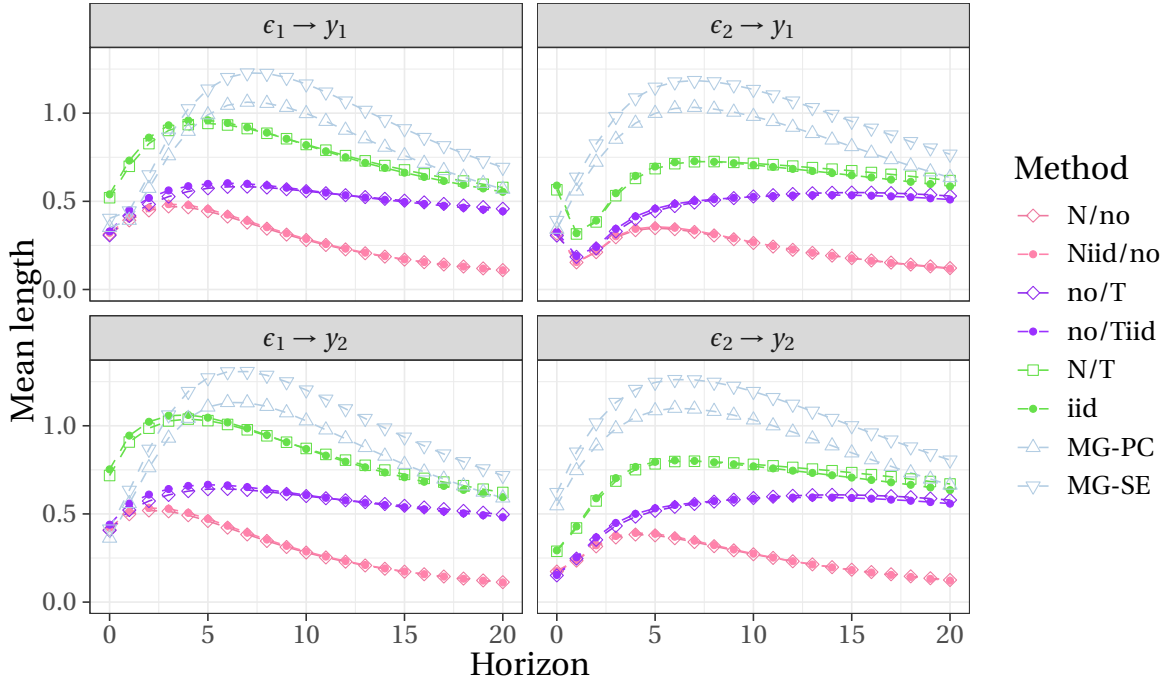}}
	\caption[Mean lengths.]{Mean lengths $ \widehat{\ell}_{h,ks} $ of the inferential methods for the PMG-identified proxy SVARs according to Monte Carlo simulations. $ R_{\text{MC}} = 600 $ replications with sample size $ \left\{ N=20,  T=100 \right\} $ have been used to assess the SN4 \textit{small sample} performance.}
	\label{fig:MeanLengths_SN4N20T100}
\end{figure}
Figure~\ref{fig:MeanLengths_SN4N20T100} displays the $ 2 \times 2 $ IRF results for the SN4 \textit{small sample} DGP setup. The ordering of $ \widehat{\ell}_{h,ks} $ becomes more apparent, particularly at mid-range horizons, where the MG inference surpasses all bootstraps and cross-sectional resampling yields smallest  $ \widehat{\ell}_{h,ks} $. This pattern holds for both sub-configurations because \textit{iid.} and moving-block bootstraps of the same resampling scheme provide almost congruent CIs.
Moreover, Table~\ref{tab:MeanLengths} juxtaposes the mean lengths for all $ 4 \times 4 $ DGP setups. Indeed, larger sample sizes of $ N $ and particularly of $ T $ lead to narrower CIs, while stronger residual dependence leads to wider CIs.
\begin{table}[ht]
    \centering
    \caption{Mean lengths $ \left( \frac{100}{K^2} \cdot \sum_{k,s=1}^{K,K} \widehat{\ell}_{h,ks} \right) $}
    \resizebox{0.99\textwidth}{!}{
        \input{Tab_MeanLengths}}
    \label{tab:MeanLengths}
\end{table}

\subsubsection{Coverage probability}  \label{sec:SimRes_CoverageProb}
The \textit{coverage probability} $ \widehat{P}_{h,ks} $ indicates the share of the $ R_{\text{MC}} = 600 $ confidence intervals that contain the ``true'' IRF coefficient $ \left[ {\Theta}_{h} \right]_{ks} $. A confidence interval $ CI_{1-\alpha}^{*} $ based on $ \left[ \widebar{\Theta}_{h}^* \right]_{ks} $ is consistent if plim$ \left( \left[ {\Theta}_{h} \right]_{ks} \in CI_{1-\alpha}^{*} \right) = 1 - \alpha $. The ideal CI of our MC simulations would thus always align with the nominal coverage rate of $ 1-\alpha = 0.9 $ over the course of $ h $.

Figure~\ref{fig:CovRates_SN4N20T100} displays the $ 2 \times 2 $ IRF results for the SN4 \textit{small sample} DGP setup. 
As expected, the very wide CIs of the MG-inference almost always capture $ \Theta_{h} $, but the substantial diminution bias of $ \boldsymbol{\widebar{\beta}} $ due to the small $T$ still  causes the coverage to fall below the nominal level to 80\% at longer horizons. Small CI width and insufficient bias correction let the coverage rate of the cross-sectional resampling immediately fall after the initial shock period down to 20\%. The temporal and the panel resampling both show good performance at longer horizons, but the panel resampling is closer to the nominal 90\% in the short-run. Notably for all four inferential methods, the two sub-configurations again do not produce any relevant differences.
\begin{figure}[ht] 
	\centering
	\resizebox{0.98\textwidth}{!}{
		\input{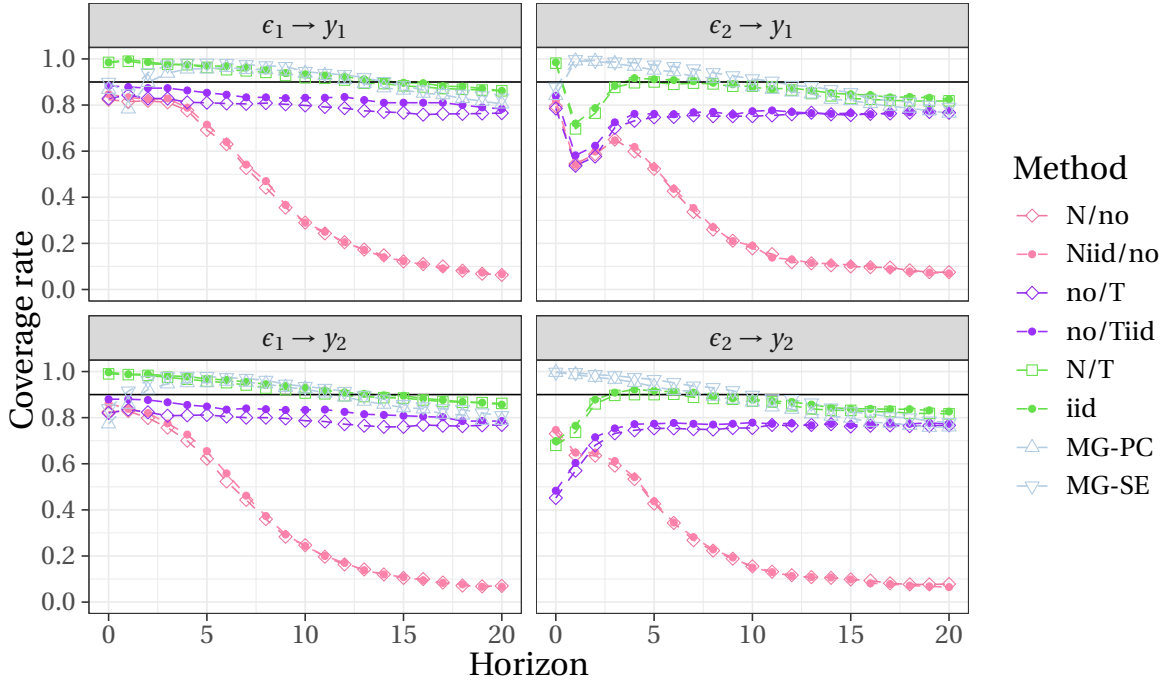}}
	\caption[Coverage probabilities.]{Coverage probabilities $ \widehat{P}_{h,ks} $ of the inferential methods for the PMG-identified proxy SVARs according to Monte Carlo simulations. $ R_{\text{MC}} = 600 $ replications with sample size $ \left\{ N=20,  T=100 \right\} $ have been used to assess the SN4 \textit{small sample} performance.}
	\label{fig:CovRates_SN4N20T100}
\end{figure}

Table~\ref{tab:CoverageProb} juxtaposes the mean absolute coverage errors for all $ 4 \times 4 $ DGP setups as the mean absolute deviations from the nominal confidence level (expressed in percentage points). In general, the bootstrap procedures deliver their ``personal best'' (i.e. results close to zero) in case of large $ T $. Although a large $ N $ reduces sampling variability and accordingly the length of CIs, a large ratio of $ N/T $ in the medium sample poses a challenge for all bootstrap procedures,  particularly, for the cross-sectional resampling at $ h \ge 5 $ and for the temporal resampling at $ h=0 $. Effects from residual structure are less clear; at least heteroskedasticity of SN2 has a noticeable effect on the performance of temporal resampling, and cross-sectional dependence of SN3 worsens the performance of cross-sectional resampling bootstraps. 
Note that, if overstated, $ \widehat{P}_{h,ks} $ can reach a deviation of 10 percentage points at most. This is the case for the 100\%-coverage of the two methods for MG-inference, which surpass the nominal 90\% by 10 percentage points in large-$T$ samples due to their large intervals.
\begin{table}[ht]
    \centering
    \caption{Mean absolute coverage errors $ \left( \frac{100}{K^2} \cdot \sum_{k,s=1}^{K,K} \left| \widehat{P}_{h,ks} - (1-\alpha) \right| \right) $}
    \resizebox{0.99\textwidth}{!}{
        \input{Tab_CovRates}}
    \label{tab:CoverageProb}
\end{table}

\subsection{Summarizing recommendations}  \label{sec:recommendations}
Summarizing the finite-sample results of Section~\ref{sec:Simulation} and Appendix~\ref{app:Simultaion}, we recommend the following configurations to the empirical practitioner. Those may also serve as a basis for future extensions of the panel econometric methodology.

\textit{\textbf{(1)}} Consider small-$ T $ adjustments for the MG estimation of the three heterogeneous parameters in $ \Theta_h = f_h{\left( \boldsymbol{\beta}, \boldsymbol{\sigma}, \boldsymbol{\varphi} \right)} $. Panel samples with sizes of $ N \gtrsim T $ in particular require the bootstrap-after-bootstrap for the MG estimator $ \widebar{A} $, consideration of the degrees of freedom $ T-n_Z $ in $ \widehat{\Sigma}_{u,i} $, and pooled identification of $ \mathsf{Q} $. The latter may encompass different approaches, strengthening the identification in our study against weak proxies and in \citet{Herwartz2017} against insufficient non-Gaussianity \citep{OleaEtAl2022} of individual VARs.

The simulations indicate good performance under common factors without dynamics, but persistent multi-factor errors fall outside this scope. These would correlate with the lagged endogenous regressors such that the individual $ \boldsymbol{\widehat{\beta}}_i $ and their MG $ \boldsymbol{\widebar{\beta}} $ lose consistency. Common solutions for single-equation models address \textit{common correlated effects} \citep{Pesaran2006} or \textit{interactive fixed effects} \citep{Bai2009}. Related bootstrap procedures yet build on variants of the wild bootstrap by \citet{GoncalvesPerron2014,GoncalvesPerron2020}, which would not correctly replicate the sampling distribution of $ \boldsymbol{\widebar{\sigma}} $ in panel SVARs, or use cross-sectional resampling \citep{VosOvidijus2024,VosOvidijus2026}, which the subsequent recommendation advises against.

\textit{\textbf{(2)}} Select the panel resampling scheme. The panel resampling performs better than the temporal one at short horizons $ h $. Due to insufficient bias correction at mid-range horizons, the cross-sectional resampling should be avoided for panel SVARs. This differs from the recommendations given by \citet{Kapetanios2008} for single-equation models and by \citet{SigmundFerstl2019} for (homogeneous) panel VAR models estimated by (pooled) system-GMM.

\textit{\textbf{(3)}} The size of blocks turns out to be less relevant than the choice of resampling scheme. Even for the boundary specifications of \textit{(i)} \textit{iid.} draws of atomic residuals and \textit{(ii)} fairly large blocks of $ b_{(i)} \approx N/10 $ and $ b_{(t)} \approx T/10 $, the procedures display almost congruent performance for the SN4 small sample simulations. Largest differences for the panel resampling are found at $ h > 0 $ in the setups with medium sample sizes, where block assembling thus matters most. 

Although ad-hoc rules are prevalent, the VAR literature \citep[e.g.][]{JentschLunsford2019} acknowledges a trade-off when fragmenting samples into blocks. While larger block sizes preserve residual dependence, larger block numbers increase bootstrap variation. The trade-off is particularly relevant in the temporal resampling of small time series but mitigated in the panel resampling as the maximum number of blocks increases by a factor of $ N $. To determine optimal panel block sizes that balance over both dimensions, future research may combine selection rules for temporal \citep{HallEtAl1995} and spatial blocks \citep{NordmanEtAl2007}.

%% file: Tab_Methods.tex
\begin{tabular}{p{0.01\textwidth} p{0.09\textwidth} p{0.40\textwidth} p{0.50\textwidth} }
    \hline \hline
    \multicolumn{2}{l}{\textbf{Symbol}} & \textbf{Inferential method} & \textbf{Short description} \\
    
    \hline
    \multicolumn{3}{l}{\textit{\textbf{Cross-sectional resampling}}} & \\
    \textcolor{pink}{ \rotatebox{45}{$\boldsymbol{\square}$} } & N/no & Cross-sectional moving-block bootstrap & Cross-sectional dimension is resampled in blocks, temporal dimension is fixed. \\ [3ex]
    \textcolor{pink}{\large $ {\bullet} $ } & N\textit{iid}/no & Cross-sectional \textit{iid.} bootstrap & Cross-sectional dimension is \textit{iid.} resampled, temporal dimension is fixed. \\[3ex]
    
    \hline
    \multicolumn{3}{l}{\textit{\textbf{Temporal resampling}}} & \\
    \textcolor{violet}{ \rotatebox{45}{$\boldsymbol{\square}$} } & no/T & Temporal moving-block bootstrap & Cross-sectional dimension is fixed, temporal dimension is resampled in blocks. \\ [3ex]
    \textcolor{violet}{\large $ \boldsymbol{\bullet} $ } & no/T\textit{iid}& Temporal \textit{iid.} bootstrap & Cross-sectional dimension is fixed, temporal dimension is \textit{iid.} resampled. \\[3ex]
    
    \hline
    \multicolumn{3}{l}{\textit{\textbf{Panel resampling}}} & \\
    \textcolor{green}{ $ \boldsymbol{\square} $ } & N/T & Panel moving-block bootstrap & Both dimensions are resampled in blocks. \\ [3ex]
    \textcolor{green}{\large $ \boldsymbol{\bullet} $ } & \textit{iid} & Panel \textit{iid.} bootstrap & Both dimensions are \textit{iid.} resampled. \\
    
    $ \times $ & N/T\textit{iid} & Panel cross-sectional moving-block temporal \textit{iid.} bootstrap & Cross-sectional dimension is resampled in blocks, temporal dimension is \textit{iid.} resampled. \\
    $ \times $ & N\textit{iid}/T & Panel cross-sectional \textit{iid.} temporal moving-block bootstrap & Cross-sectional dimension is \textit{iid.} resampled, temporal dimension is resampled in blocks. \\ [3ex]
    
    \hline
    \multicolumn{3}{l}{\textit{\textbf{Mean-group inference}}} & \\
    \textcolor{teal}{ $ \boldsymbol{\triangle} $ } & MG-PC & Mean-group percentile & The individual IRF estimates depict the sampling variation of their common parameter. \\
    \textcolor{teal}{\large $ \boldsymbol{\triangledown} $ } & MG-SE & Mean-group standard error & Standard errors for the confidence intervals are calculated from individual IRF estimates. \\
    \hline \hline
\end{tabular}

%% file: Fig_Uorig.tex
\begin{align*}
    \begin{bmatrix}
        \widehat{U}_{1} \\ \widehat{U}_{2} \\ \widehat{U}_{3} \\ \widehat{U}_{4} \\ \widehat{U}_{5} \\ 
    \end{bmatrix}
    =
    \left[
    \begin{array}{ccccccc}
        \cellcolor{red!100!lightgray} \textcolor{white}{\boldsymbol{\widehat{u}}_{11}} & \cellcolor{red!88!lightgray} \textcolor{white}{\boldsymbol{\widehat{u}}_{12}} & \cellcolor{red!76!lightgray} \textcolor{white}{\boldsymbol{\widehat{u}}_{13}} & \cellcolor{red!64!lightgray} \textcolor{white}{\boldsymbol{\widehat{u}}_{14}} & \cellcolor{red!52!lightgray} \textcolor{white}{\boldsymbol{\widehat{u}}_{15}} & \cellcolor{red!40!lightgray} \textcolor{white}{\boldsymbol{\widehat{u}}_{16}} & \cellcolor{red!28!lightgray} \textcolor{white}{\boldsymbol{\widehat{u}}_{17}} \\
        \cellcolor{orange!100!lightgray} \textcolor{black}{\boldsymbol{\widehat{u}}_{21}} & \cellcolor{orange!88!lightgray} \textcolor{black}{\boldsymbol{\widehat{u}}_{22}} & \cellcolor{orange!76!lightgray} \textcolor{black}{\boldsymbol{\widehat{u}}_{23}} & \cellcolor{orange!64!lightgray} \textcolor{black}{\boldsymbol{\widehat{u}}_{24}} & \cellcolor{orange!52!lightgray} \textcolor{black}{\boldsymbol{\widehat{u}}_{25}} & \cellcolor{orange!40!lightgray} \textcolor{black}{\boldsymbol{\widehat{u}}_{26}} & \cellcolor{orange!28!lightgray} \textcolor{black}{\boldsymbol{\widehat{u}}_{27}} \\
        \cellcolor{yellow!100!lightgray} \textcolor{black}{\boldsymbol{\widehat{u}}_{31}} & \cellcolor{yellow!88!lightgray} \textcolor{black}{\boldsymbol{\widehat{u}}_{32}} & \cellcolor{yellow!76!lightgray} \textcolor{black}{\boldsymbol{\widehat{u}}_{33}} & \cellcolor{yellow!64!lightgray} \textcolor{black}{\boldsymbol{\widehat{u}}_{34}} & \cellcolor{yellow!52!lightgray} \textcolor{black}{\boldsymbol{\widehat{u}}_{35}} & \cellcolor{yellow!40!lightgray} \textcolor{black}{\boldsymbol{\widehat{u}}_{36}} & \cellcolor{yellow!28!lightgray} \textcolor{black}{\boldsymbol{\widehat{u}}_{37}} \\
        \cellcolor{green!100!lightgray} \textcolor{black}{\boldsymbol{\widehat{u}}_{41}} & \cellcolor{green!88!lightgray} \textcolor{black}{\boldsymbol{\widehat{u}}_{42}} & \cellcolor{green!76!lightgray} \textcolor{black}{\boldsymbol{\widehat{u}}_{43}} & \cellcolor{green!64!lightgray} \textcolor{black}{\boldsymbol{\widehat{u}}_{44}} & \cellcolor{green!52!lightgray} \textcolor{black}{\boldsymbol{\widehat{u}}_{45}} & \cellcolor{green!40!lightgray} \textcolor{black}{\boldsymbol{\widehat{u}}_{46}} & \cellcolor{green!28!lightgray} \textcolor{black}{\boldsymbol{\widehat{u}}_{47}} \\
        \cellcolor{blue!100!lightgray} \textcolor{white}{\boldsymbol{\widehat{u}}_{51}} & \cellcolor{blue!88!lightgray} \textcolor{white}{\boldsymbol{\widehat{u}}_{52}} & \cellcolor{blue!76!lightgray} \textcolor{white}{\boldsymbol{\widehat{u}}_{53}} & \cellcolor{blue!64!lightgray} \textcolor{white}{\boldsymbol{\widehat{u}}_{54}} & \cellcolor{blue!52!lightgray} \textcolor{white}{\boldsymbol{\widehat{u}}_{55}} & \cellcolor{blue!40!lightgray} \textcolor{white}{\boldsymbol{\widehat{u}}_{56}} & \cellcolor{blue!28!lightgray} \textcolor{white}{\boldsymbol{\widehat{u}}_{57}} 
    \end{array}
    \right]
\end{align*}

%% file: Fig_Ucsr.tex
\begin{align*}
\arrayrulewidth=2pt
\arrayrulecolor{white}
\begin{split}
    \begin{bmatrix}
        \widecheck{U}^*_{1} \\ \widecheck{U}^*_{2} \\ \widecheck{U}^*_{3} \\ \widecheck{U}^*_{4} \\ \widecheck{U}^*_{5} \\ 
    \end{bmatrix}
    =
    \left[
    \begin{array}{ccccccc}
        \cellcolor{orange!100!lightgray} \textcolor{black}{\boldsymbol{\widehat{u}}_{21}} & \cellcolor{orange!88!lightgray} \textcolor{black}{\boldsymbol{\widehat{u}}_{22}} & \cellcolor{orange!76!lightgray} \textcolor{black}{\boldsymbol{\widehat{u}}_{23}} & \cellcolor{orange!64!lightgray} \textcolor{black}{\boldsymbol{\widehat{u}}_{24}} & \cellcolor{orange!52!lightgray} \textcolor{black}{\boldsymbol{\widehat{u}}_{25}} & \cellcolor{orange!40!lightgray} \textcolor{black}{\boldsymbol{\widehat{u}}_{26}} & \cellcolor{orange!28!lightgray} \textcolor{black}{\boldsymbol{\widehat{u}}_{27}} \\
        \hline
        \cellcolor{green!100!lightgray} \textcolor{black}{\boldsymbol{\widehat{u}}_{41}} & \cellcolor{green!88!lightgray} \textcolor{black}{\boldsymbol{\widehat{u}}_{42}} & \cellcolor{green!76!lightgray} \textcolor{black}{\boldsymbol{\widehat{u}}_{43}} & \cellcolor{green!64!lightgray} \textcolor{black}{\boldsymbol{\widehat{u}}_{44}} & \cellcolor{green!52!lightgray} \textcolor{black}{\boldsymbol{\widehat{u}}_{45}} & \cellcolor{green!40!lightgray} \textcolor{black}{\boldsymbol{\widehat{u}}_{46}} & \cellcolor{green!28!lightgray} \textcolor{black}{\boldsymbol{\widehat{u}}_{47}} \\
        \hline
        \cellcolor{yellow!100!lightgray} \textcolor{black}{\boldsymbol{\widehat{u}}_{31}} & \cellcolor{yellow!88!lightgray} \textcolor{black}{\boldsymbol{\widehat{u}}_{32}} & \cellcolor{yellow!76!lightgray} \textcolor{black}{\boldsymbol{\widehat{u}}_{33}} & \cellcolor{yellow!64!lightgray} \textcolor{black}{\boldsymbol{\widehat{u}}_{34}} & \cellcolor{yellow!52!lightgray} \textcolor{black}{\boldsymbol{\widehat{u}}_{35}} & \cellcolor{yellow!40!lightgray} \textcolor{black}{\boldsymbol{\widehat{u}}_{36}} & \cellcolor{yellow!28!lightgray} \textcolor{black}{\boldsymbol{\widehat{u}}_{37}} \\
        \hline
        \cellcolor{red!100!lightgray} \textcolor{white}{\boldsymbol{\widehat{u}}_{11}} & \cellcolor{red!88!lightgray} \textcolor{white}{\boldsymbol{\widehat{u}}_{12}} & \cellcolor{red!76!lightgray} \textcolor{white}{\boldsymbol{\widehat{u}}_{13}} & \cellcolor{red!64!lightgray} \textcolor{white}{\boldsymbol{\widehat{u}}_{14}} & \cellcolor{red!52!lightgray} \textcolor{white}{\boldsymbol{\widehat{u}}_{15}} & \cellcolor{red!40!lightgray} \textcolor{white}{\boldsymbol{\widehat{u}}_{16}} & \cellcolor{red!28!lightgray} \textcolor{white}{\boldsymbol{\widehat{u}}_{17}} \\
        \hline
        \cellcolor{yellow!100!lightgray} \textcolor{black}{\boldsymbol{\widehat{u}}_{31}} & \cellcolor{yellow!88!lightgray} \textcolor{black}{\boldsymbol{\widehat{u}}_{32}} & \cellcolor{yellow!76!lightgray} \textcolor{black}{\boldsymbol{\widehat{u}}_{33}} & \cellcolor{yellow!64!lightgray} \textcolor{black}{\boldsymbol{\widehat{u}}_{34}} & \cellcolor{yellow!52!lightgray} \textcolor{black}{\boldsymbol{\widehat{u}}_{35}} & \cellcolor{yellow!40!lightgray} \textcolor{black}{\boldsymbol{\widehat{u}}_{36}} & \cellcolor{yellow!28!lightgray} \textcolor{black}{\boldsymbol{\widehat{u}}_{37}} \\
    \end{array}
    \right]
\end{split},
\arrayrulewidth=2pt
\arrayrulecolor{white}
\begin{split}
    \left[
    \begin{array}{ccccccc}
        \cellcolor{green!100!lightgray} \textcolor{black}{\boldsymbol{\widehat{u}}_{41}} & \cellcolor{green!88!lightgray} \textcolor{black}{\boldsymbol{\widehat{u}}_{42}} & \cellcolor{green!76!lightgray} \textcolor{black}{\boldsymbol{\widehat{u}}_{43}} & \cellcolor{green!64!lightgray} \textcolor{black}{\boldsymbol{\widehat{u}}_{44}} & \cellcolor{green!52!lightgray} \textcolor{black}{\boldsymbol{\widehat{u}}_{45}} & \cellcolor{green!40!lightgray} \textcolor{black}{\boldsymbol{\widehat{u}}_{46}} & \cellcolor{green!28!lightgray} \textcolor{black}{\boldsymbol{\widehat{u}}_{47}} \\
        \cellcolor{blue!100!lightgray} \textcolor{white}{\boldsymbol{\widehat{u}}_{51}} & \cellcolor{blue!88!lightgray} \textcolor{white}{\boldsymbol{\widehat{u}}_{52}} & \cellcolor{blue!76!lightgray} \textcolor{white}{\boldsymbol{\widehat{u}}_{53}} & \cellcolor{blue!64!lightgray} \textcolor{white}{\boldsymbol{\widehat{u}}_{54}} & \cellcolor{blue!52!lightgray} \textcolor{white}{\boldsymbol{\widehat{u}}_{55}} & \cellcolor{blue!40!lightgray} \textcolor{white}{\boldsymbol{\widehat{u}}_{56}} & \cellcolor{blue!28!lightgray} \textcolor{white}{\boldsymbol{\widehat{u}}_{57}} \\
        \hline
        \cellcolor{orange!100!lightgray} \textcolor{black}{\boldsymbol{\widehat{u}}_{21}} & \cellcolor{orange!88!lightgray} \textcolor{black}{\boldsymbol{\widehat{u}}_{22}} & \cellcolor{orange!76!lightgray} \textcolor{black}{\boldsymbol{\widehat{u}}_{23}} & \cellcolor{orange!64!lightgray} \textcolor{black}{\boldsymbol{\widehat{u}}_{24}} & \cellcolor{orange!52!lightgray} \textcolor{black}{\boldsymbol{\widehat{u}}_{25}} & \cellcolor{orange!40!lightgray} \textcolor{black}{\boldsymbol{\widehat{u}}_{26}} & \cellcolor{orange!28!lightgray} \textcolor{black}{\boldsymbol{\widehat{u}}_{27}} \\
        \cellcolor{yellow!100!lightgray} \textcolor{black}{\boldsymbol{\widehat{u}}_{31}} & \cellcolor{yellow!88!lightgray} \textcolor{black}{\boldsymbol{\widehat{u}}_{32}} & \cellcolor{yellow!76!lightgray} \textcolor{black}{\boldsymbol{\widehat{u}}_{33}} & \cellcolor{yellow!64!lightgray} \textcolor{black}{\boldsymbol{\widehat{u}}_{34}} & \cellcolor{yellow!52!lightgray} \textcolor{black}{\boldsymbol{\widehat{u}}_{35}} & \cellcolor{yellow!40!lightgray} \textcolor{black}{\boldsymbol{\widehat{u}}_{36}} & \cellcolor{yellow!28!lightgray} \textcolor{black}{\boldsymbol{\widehat{u}}_{37}} \\
        \hline
        \cellcolor{red!100!lightgray} \textcolor{white}{\boldsymbol{\widehat{u}}_{11}} & \cellcolor{red!88!lightgray} \textcolor{white}{\boldsymbol{\widehat{u}}_{12}} & \cellcolor{red!76!lightgray} \textcolor{white}{\boldsymbol{\widehat{u}}_{13}} & \cellcolor{red!64!lightgray} \textcolor{white}{\boldsymbol{\widehat{u}}_{14}} & \cellcolor{red!52!lightgray} \textcolor{white}{\boldsymbol{\widehat{u}}_{15}} & \cellcolor{red!40!lightgray} \textcolor{white}{\boldsymbol{\widehat{u}}_{16}} & \cellcolor{red!28!lightgray} \textcolor{white}{\boldsymbol{\widehat{u}}_{17}} \\
    \end{array}
    \right]
\end{split}.
\end{align*}
\global\arrayrulewidth=0.4pt

%% file: Fig_Utmpr.tex
\begin{align*}
\arrayrulewidth=2pt
\arrayrulecolor{white}
\begin{split}
    \begin{bmatrix}
        \widecheck{U}^*_{1} \\ \widecheck{U}^*_{2} \\ \widecheck{U}^*_{3} \\ \widecheck{U}^*_{4} \\ \widecheck{U}^*_{5} \\ 
    \end{bmatrix}
    =
    \left[
    \begin{array}{c|c|c|c|c|c|c}
        \cellcolor{red!28!lightgray} \textcolor{white}{\boldsymbol{\widehat{u}}_{17}} & \cellcolor{red!76!lightgray} \textcolor{white}{\boldsymbol{\widehat{u}}_{13}} & \cellcolor{red!64!lightgray} \textcolor{white}{\boldsymbol{\widehat{u}}_{14}} & \cellcolor{red!28!lightgray} \textcolor{white}{\boldsymbol{\widehat{u}}_{17}} & \cellcolor{red!100!lightgray} \textcolor{white}{\boldsymbol{\widehat{u}}_{11}} & \cellcolor{red!52!lightgray} \textcolor{white}{\boldsymbol{\widehat{u}}_{15}} & \cellcolor{red!88!lightgray} \textcolor{white}{\boldsymbol{\widehat{u}}_{12}} \\
        \cellcolor{orange!28!lightgray} \textcolor{black}{\boldsymbol{\widehat{u}}_{27}} & \cellcolor{orange!76!lightgray} \textcolor{black}{\boldsymbol{\widehat{u}}_{23}} & \cellcolor{orange!64!lightgray} \textcolor{black}{\boldsymbol{\widehat{u}}_{24}} & \cellcolor{orange!28!lightgray} \textcolor{black}{\boldsymbol{\widehat{u}}_{27}} & \cellcolor{orange!100!lightgray} \textcolor{black}{\boldsymbol{\widehat{u}}_{21}} & \cellcolor{orange!52!lightgray} \textcolor{black}{\boldsymbol{\widehat{u}}_{25}} & \cellcolor{orange!88!lightgray} \textcolor{black}{\boldsymbol{\widehat{u}}_{22}} \\
        \cellcolor{yellow!28!lightgray} \textcolor{black}{\boldsymbol{\widehat{u}}_{37}} & \cellcolor{yellow!76!lightgray} \textcolor{black}{\boldsymbol{\widehat{u}}_{33}} & \cellcolor{yellow!64!lightgray} \textcolor{black}{\boldsymbol{\widehat{u}}_{34}}& \cellcolor{yellow!28!lightgray} \textcolor{black}{\boldsymbol{\widehat{u}}_{37}}  & \cellcolor{yellow!100!lightgray} \textcolor{black}{\boldsymbol{\widehat{u}}_{31}} & \cellcolor{yellow!52!lightgray} \textcolor{black}{\boldsymbol{\widehat{u}}_{35}} & \cellcolor{yellow!88!lightgray} \textcolor{black}{\boldsymbol{\widehat{u}}_{32}} \\
        \cellcolor{green!28!lightgray} \textcolor{black}{\boldsymbol{\widehat{u}}_{47}} & \cellcolor{green!76!lightgray} \textcolor{black}{\boldsymbol{\widehat{u}}_{43}} & \cellcolor{green!64!lightgray} \textcolor{black}{\boldsymbol{\widehat{u}}_{44}} & \cellcolor{green!28!lightgray} \textcolor{black}{\boldsymbol{\widehat{u}}_{47}} & \cellcolor{green!100!lightgray} \textcolor{black}{\boldsymbol{\widehat{u}}_{41}} & \cellcolor{green!52!lightgray} \textcolor{black}{\boldsymbol{\widehat{u}}_{45}} & \cellcolor{green!88!lightgray} \textcolor{black}{\boldsymbol{\widehat{u}}_{42}} \\
        \cellcolor{blue!28!lightgray} \textcolor{white}{\boldsymbol{\widehat{u}}_{57}} & \cellcolor{blue!76!lightgray} \textcolor{white}{\boldsymbol{\widehat{u}}_{53}} & \cellcolor{blue!64!lightgray} \textcolor{white}{\boldsymbol{\widehat{u}}_{54}} & \cellcolor{blue!28!lightgray} \textcolor{white}{\boldsymbol{\widehat{u}}_{57}} & \cellcolor{blue!100!lightgray} \textcolor{white}{\boldsymbol{\widehat{u}}_{51}} & \cellcolor{blue!52!lightgray} \textcolor{white}{\boldsymbol{\widehat{u}}_{55}} & \cellcolor{blue!88!lightgray} \textcolor{white}{\boldsymbol{\widehat{u}}_{52}} \\
    \end{array}
    \right]
\end{split},
\arrayrulewidth=2pt
\arrayrulecolor{white}
\begin{split}
    \left[
    \begin{array}{cc|cc|cc|c}
        \cellcolor{red!88!lightgray} \textcolor{white}{\boldsymbol{\widehat{u}}_{12}} & \cellcolor{red!76!lightgray} \textcolor{white}{\boldsymbol{\widehat{u}}_{13}} & \cellcolor{red!52!lightgray} \textcolor{white}{\boldsymbol{\widehat{u}}_{15}} & \cellcolor{red!40!lightgray} \textcolor{white}{\boldsymbol{\widehat{u}}_{16}} & \cellcolor{red!64!lightgray} \textcolor{white}{\boldsymbol{\widehat{u}}_{14}} & \cellcolor{red!52!lightgray} \textcolor{white}{\boldsymbol{\widehat{u}}_{15}} & \cellcolor{red!100!lightgray} \textcolor{white}{\boldsymbol{\widehat{u}}_{11}} \\
        \cellcolor{orange!88!lightgray} \textcolor{black}{\boldsymbol{\widehat{u}}_{22}} & \cellcolor{orange!76!lightgray} \textcolor{black}{\boldsymbol{\widehat{u}}_{23}} & \cellcolor{orange!52!lightgray} \textcolor{black}{\boldsymbol{\widehat{u}}_{25}} & \cellcolor{orange!40!lightgray} \textcolor{black}{\boldsymbol{\widehat{u}}_{26}} & \cellcolor{orange!64!lightgray} \textcolor{black}{\boldsymbol{\widehat{u}}_{24}} & \cellcolor{orange!52!lightgray} \textcolor{black}{\boldsymbol{\widehat{u}}_{25}} & \cellcolor{orange!100!lightgray} \textcolor{black}{\boldsymbol{\widehat{u}}_{21}} \\
        \cellcolor{yellow!88!lightgray} \textcolor{black}{\boldsymbol{\widehat{u}}_{32}} & \cellcolor{yellow!76!lightgray} \textcolor{black}{\boldsymbol{\widehat{u}}_{33}} & \cellcolor{yellow!52!lightgray} \textcolor{black}{\boldsymbol{\widehat{u}}_{35}} & \cellcolor{yellow!40!lightgray} \textcolor{black}{\boldsymbol{\widehat{u}}_{36}} & \cellcolor{yellow!64!lightgray} \textcolor{black}{\boldsymbol{\widehat{u}}_{34}} & \cellcolor{yellow!52!lightgray} \textcolor{black}{\boldsymbol{\widehat{u}}_{35}} & \cellcolor{yellow!100!lightgray} \textcolor{black}{\boldsymbol{\widehat{u}}_{31}} \\
        \cellcolor{green!88!lightgray} \textcolor{black}{\boldsymbol{\widehat{u}}_{42}} & \cellcolor{green!76!lightgray} \textcolor{black}{\boldsymbol{\widehat{u}}_{43}} & \cellcolor{green!52!lightgray} \textcolor{black}{\boldsymbol{\widehat{u}}_{45}} & \cellcolor{green!40!lightgray} \textcolor{black}{\boldsymbol{\widehat{u}}_{46}} & \cellcolor{green!64!lightgray} \textcolor{black}{\boldsymbol{\widehat{u}}_{44}} & \cellcolor{green!52!lightgray} \textcolor{black}{\boldsymbol{\widehat{u}}_{45}} & \cellcolor{green!100!lightgray} \textcolor{black}{\boldsymbol{\widehat{u}}_{41}} \\
        \cellcolor{blue!88!lightgray} \textcolor{white}{\boldsymbol{\widehat{u}}_{52}} & \cellcolor{blue!76!lightgray} \textcolor{white}{\boldsymbol{\widehat{u}}_{53}} & \cellcolor{blue!52!lightgray} \textcolor{white}{\boldsymbol{\widehat{u}}_{55}} & \cellcolor{blue!40!lightgray} \textcolor{white}{\boldsymbol{\widehat{u}}_{56}} & \cellcolor{blue!64!lightgray} \textcolor{white}{\boldsymbol{\widehat{u}}_{54}} & \cellcolor{blue!52!lightgray} \textcolor{white}{\boldsymbol{\widehat{u}}_{55}} & \cellcolor{blue!100!lightgray} \textcolor{white}{\boldsymbol{\widehat{u}}_{51}} \\
    \end{array}
    \right]
\end{split}.
\end{align*}

%% file: Fig_Upanelr.tex
\begin{align*}
\arrayrulewidth=2pt
\arrayrulecolor{white}
\begin{split}
    \begin{bmatrix}
        \widecheck{U}^*_{1} \\ \widecheck{U}^*_{2} \\ \widecheck{U}^*_{3} \\ \widecheck{U}^*_{4} \\ \widecheck{U}^*_{5} \\ 
    \end{bmatrix}
    =
    \left[
    \begin{array}{c|c|c|c|c|c|c}
        \cellcolor{orange!28!lightgray} \textcolor{black}{\boldsymbol{\widehat{u}}_{27}} & \cellcolor{orange!76!lightgray} \textcolor{black}{\boldsymbol{\widehat{u}}_{23}} & \cellcolor{orange!64!lightgray} \textcolor{black}{\boldsymbol{\widehat{u}}_{24}} & \cellcolor{orange!28!lightgray} \textcolor{black}{\boldsymbol{\widehat{u}}_{27}} & \cellcolor{orange!100!lightgray} \textcolor{black}{\boldsymbol{\widehat{u}}_{21}} & \cellcolor{orange!52!lightgray} \textcolor{black}{\boldsymbol{\widehat{u}}_{25}} & \cellcolor{orange!88!lightgray} \textcolor{black}{\boldsymbol{\widehat{u}}_{22}} \\
        \hline
        \cellcolor{green!28!lightgray} \textcolor{black}{\boldsymbol{\widehat{u}}_{47}} & \cellcolor{green!76!lightgray} \textcolor{black}{\boldsymbol{\widehat{u}}_{43}} & \cellcolor{green!64!lightgray} \textcolor{black}{\boldsymbol{\widehat{u}}_{44}} & \cellcolor{green!28!lightgray} \textcolor{black}{\boldsymbol{\widehat{u}}_{47}} & \cellcolor{green!100!lightgray} \textcolor{black}{\boldsymbol{\widehat{u}}_{41}} & \cellcolor{green!52!lightgray} \textcolor{black}{\boldsymbol{\widehat{u}}_{45}} & \cellcolor{green!88!lightgray} \textcolor{black}{\boldsymbol{\widehat{u}}_{42}} \\
        \hline
        \cellcolor{yellow!28!lightgray} \textcolor{black}{\boldsymbol{\widehat{u}}_{37}} & \cellcolor{yellow!76!lightgray} \textcolor{black}{\boldsymbol{\widehat{u}}_{33}} & \cellcolor{yellow!64!lightgray} \textcolor{black}{\boldsymbol{\widehat{u}}_{34}} & \cellcolor{yellow!28!lightgray} \textcolor{black}{\boldsymbol{\widehat{u}}_{37}} & \cellcolor{yellow!100!lightgray} \textcolor{black}{\boldsymbol{\widehat{u}}_{31}} & \cellcolor{yellow!52!lightgray} \textcolor{black}{\boldsymbol{\widehat{u}}_{35}} & \cellcolor{yellow!88!lightgray} \textcolor{black}{\boldsymbol{\widehat{u}}_{32}} \\
        \hline
        \cellcolor{red!28!lightgray} \textcolor{white}{\boldsymbol{\widehat{u}}_{17}} & \cellcolor{red!76!lightgray} \textcolor{white}{\boldsymbol{\widehat{u}}_{13}} & \cellcolor{red!64!lightgray} \textcolor{white}{\boldsymbol{\widehat{u}}_{14}} & \cellcolor{red!28!lightgray} \textcolor{white}{\boldsymbol{\widehat{u}}_{17}} & \cellcolor{red!100!lightgray} \textcolor{white}{\boldsymbol{\widehat{u}}_{11}} & \cellcolor{red!52!lightgray} \textcolor{white}{\boldsymbol{\widehat{u}}_{15}} & \cellcolor{red!88!lightgray} \textcolor{white}{\boldsymbol{\widehat{u}}_{12}} \\
        \hline
        \cellcolor{yellow!28!lightgray} \textcolor{black}{\boldsymbol{\widehat{u}}_{37}} & \cellcolor{yellow!76!lightgray} \textcolor{black}{\boldsymbol{\widehat{u}}_{33}} & \cellcolor{yellow!64!lightgray} \textcolor{black}{\boldsymbol{\widehat{u}}_{34}} & \cellcolor{yellow!28!lightgray} \textcolor{black}{\boldsymbol{\widehat{u}}_{37}} & \cellcolor{yellow!100!lightgray} \textcolor{black}{\boldsymbol{\widehat{u}}_{31}} & \cellcolor{yellow!52!lightgray} \textcolor{black}{\boldsymbol{\widehat{u}}_{35}} & \cellcolor{yellow!88!lightgray} \textcolor{black}{\boldsymbol{\widehat{u}}_{32}} \\
    \end{array}
    \right]
\end{split},
\arrayrulewidth=2pt
\arrayrulecolor{white}
\begin{split}
    \left[
    \begin{array}{cc|cc|cc|c}
        \cellcolor{green!88!lightgray} \textcolor{black}{\boldsymbol{\widehat{u}}_{42}} & \cellcolor{green!76!lightgray} \textcolor{black}{\boldsymbol{\widehat{u}}_{43}} & \cellcolor{green!52!lightgray} \textcolor{black}{\boldsymbol{\widehat{u}}_{45}} & \cellcolor{green!40!lightgray} \textcolor{black}{\boldsymbol{\widehat{u}}_{46}} & \cellcolor{green!64!lightgray} \textcolor{black}{\boldsymbol{\widehat{u}}_{44}} & \cellcolor{green!52!lightgray} \textcolor{black}{\boldsymbol{\widehat{u}}_{45}} & \cellcolor{green!100!lightgray} \textcolor{black}{\boldsymbol{\widehat{u}}_{41}} \\
        \cellcolor{blue!88!lightgray} \textcolor{white}{\boldsymbol{\widehat{u}}_{52}} & \cellcolor{blue!76!lightgray} \textcolor{white}{\boldsymbol{\widehat{u}}_{53}} & \cellcolor{blue!52!lightgray} \textcolor{white}{\boldsymbol{\widehat{u}}_{55}} & \cellcolor{blue!40!lightgray} \textcolor{white}{\boldsymbol{\widehat{u}}_{56}} & \cellcolor{blue!64!lightgray} \textcolor{white}{\boldsymbol{\widehat{u}}_{54}} & \cellcolor{blue!52!lightgray} \textcolor{white}{\boldsymbol{\widehat{u}}_{55}} & \cellcolor{blue!100!lightgray} \textcolor{white}{\boldsymbol{\widehat{u}}_{51}} \\
        \hline
        \cellcolor{orange!88!lightgray} \textcolor{black}{\boldsymbol{\widehat{u}}_{22}} & \cellcolor{orange!76!lightgray} \textcolor{black}{\boldsymbol{\widehat{u}}_{23}} & \cellcolor{orange!52!lightgray} \textcolor{black}{\boldsymbol{\widehat{u}}_{25}} & \cellcolor{orange!40!lightgray} \textcolor{black}{\boldsymbol{\widehat{u}}_{26}} & \cellcolor{orange!64!lightgray} \textcolor{black}{\boldsymbol{\widehat{u}}_{24}} & \cellcolor{orange!52!lightgray} \textcolor{black}{\boldsymbol{\widehat{u}}_{25}} & \cellcolor{orange!100!lightgray} \textcolor{black}{\boldsymbol{\widehat{u}}_{21}} \\
        \cellcolor{yellow!88!lightgray} \textcolor{black}{\boldsymbol{\widehat{u}}_{32}} & \cellcolor{yellow!76!lightgray} \textcolor{black}{\boldsymbol{\widehat{u}}_{33}} & \cellcolor{yellow!52!lightgray} \textcolor{black}{\boldsymbol{\widehat{u}}_{35}} & \cellcolor{yellow!40!lightgray} \textcolor{black}{\boldsymbol{\widehat{u}}_{36}} & \cellcolor{yellow!64!lightgray} \textcolor{black}{\boldsymbol{\widehat{u}}_{34}} & \cellcolor{yellow!52!lightgray} \textcolor{black}{\boldsymbol{\widehat{u}}_{35}} & \cellcolor{yellow!100!lightgray} \textcolor{black}{\boldsymbol{\widehat{u}}_{31}} \\
        \hline
        \cellcolor{red!88!lightgray} \textcolor{white}{\boldsymbol{\widehat{u}}_{12}} & \cellcolor{red!76!lightgray} \textcolor{white}{\boldsymbol{\widehat{u}}_{13}} & \cellcolor{red!52!lightgray} \textcolor{white}{\boldsymbol{\widehat{u}}_{15}} & \cellcolor{red!40!lightgray} \textcolor{white}{\boldsymbol{\widehat{u}}_{16}} & \cellcolor{red!64!lightgray} \textcolor{white}{\boldsymbol{\widehat{u}}_{14}} & \cellcolor{red!52!lightgray} \textcolor{white}{\boldsymbol{\widehat{u}}_{15}} & \cellcolor{red!100!lightgray} \textcolor{white}{\boldsymbol{\widehat{u}}_{11}} \\
    \end{array}
    \right]
\end{split}.
\end{align*}
\arrayrulecolor{black}

%% file: Tab_MeanLengths.tex
\begin{tabular}{lrrrrrcrrrrcrrrrcrrrr} 
 \hline \hline 
 & & \multicolumn{4}{c}{\textbf{Small sample}} && \multicolumn{4}{c}{\textbf{Medium sample}} && \multicolumn{4}{c}{\textbf{Medium sample}} && \multicolumn{4}{c}{\textbf{Large sample}} \\  
 & & \multicolumn{4}{c}{$ N=20, \ T=100 $} && \multicolumn{4}{c}{$ N=20, \ T=500 $} && \multicolumn{4}{c}{$ N=100, \ T=100 $} && \multicolumn{4}{c}{$ N=100, \ T=500 $} \\  
 \cmidrule(lr){3-6} \cmidrule(lr){8-11} \cmidrule(lr){13-16} \cmidrule(lr){18-21} 
 \textbf{Method} & $ h $ & SN1 & SN2 & SN3 & SN4 && SN1 & SN2 & SN3 & SN4 && SN1 & SN2 & SN3 & SN4 && SN1 & SN2 & SN3 & SN4 \\ 
 \hline 
\multicolumn{21}{l}{\textbf{$ \bullet $ Cross-sectional resampling}} \\ 
	N/no & 0 & 28.0 & 30.0 & 28.4 & 30.3 &  & 11.9 & 14.7 & 12.2 & 14.9 &  & 12.7 & 13.5 & 13.0 & 13.7 &  &  5.4 &  6.6 &  5.6 &  6.7 \\
	  & 5 & 38.8 & 40.4 & 39.4 & 40.8 &  & 18.0 & 21.2 & 18.3 & 21.6 &  & 17.7 & 18.3 & 18.0 & 18.4 &  &  8.1 &  9.4 &  8.4 &  9.7 \\
	  & 15 & 18.1 & 18.3 & 17.2 & 17.6 &  &  9.6 & 11.2 &  9.7 & 11.3 &  &  8.2 &  8.3 &  8.5 &  8.6 &  &  4.3 &  4.9 &  4.5 &  5.1 \\
	Niid/no & 0 & 29.1 & 31.0 & 29.1 & 31.1 &  & 12.5 & 15.3 & 12.5 & 15.5 &  & 13.4 & 14.3 & 13.3 & 14.3 &  &  5.8 &  7.0 &  5.8 &  7.0 \\
	  & 5 & 40.3 & 41.8 & 40.4 & 41.8 &  & 18.8 & 22.0 & 18.9 & 22.3 &  & 18.7 & 19.2 & 18.5 & 19.2 &  &  8.7 & 10.1 &  8.7 & 10.1 \\
	  & 15 & 18.6 & 18.7 & 16.9 & 17.6 &  & 10.0 & 11.6 & 10.0 & 11.7 &  &  8.5 &  8.5 &  7.9 &  8.2 &  &  4.6 &  5.3 &  4.6 &  5.3 \\
\multicolumn{21}{l}{\textbf{$ \bullet $ Temporal resampling}} \\ 
	no/T & 0 & 28.2 & 28.2 & 29.5 & 29.3 &  & 12.1 & 12.3 & 12.6 & 13.1 &  & 12.5 & 12.9 & 14.1 & 14.0 &  &  5.4 &  5.6 &  6.3 &  6.4 \\
	  & 5 & 47.4 & 47.0 & 56.4 & 54.4 &  & 20.1 & 20.2 & 23.8 & 24.0 &  & 22.8 & 23.2 & 31.0 & 29.9 &  &  9.2 &  9.5 & 13.4 & 13.1 \\
	  & 15 & 42.8 & 41.8 & 58.1 & 55.2 &  & 16.3 & 16.5 & 23.5 & 22.7 &  & 25.8 & 25.0 & 36.8 & 35.2 &  &  9.3 &  9.5 & 15.1 & 14.7 \\
	no/Tiid & 0 & 29.7 & 29.8 & 31.3 & 31.2 &  & 13.0 & 13.1 & 13.8 & 13.8 &  & 13.2 & 13.3 & 14.9 & 14.9 &  &  5.8 &  5.8 &  6.7 &  6.6 \\
	  & 5 & 48.3 & 48.1 & 58.3 & 56.4 &  & 21.4 & 21.5 & 25.8 & 25.2 &  & 22.2 & 21.9 & 31.4 & 30.1 &  &  9.6 &  9.6 & 14.1 & 13.6 \\
	  & 15 & 39.4 & 38.6 & 57.3 & 54.0 &  & 16.1 & 16.1 & 24.6 & 23.4 &  & 19.1 & 18.5 & 34.0 & 31.8 &  &  7.4 &  7.4 & 15.2 & 14.3 \\
\multicolumn{21}{l}{\textbf{$ \bullet $ Panel resampling}} \\ 
	N/T & 0 & 49.9 & 50.7 & 51.0 & 52.3 &  & 20.9 & 23.0 & 21.5 & 23.7 &  & 21.5 & 22.3 & 22.8 & 23.2 &  &  9.3 & 10.4 & 10.1 & 11.0 \\
	  & 5 & 79.2 & 79.2 & 86.7 & 86.0 &  & 33.8 & 35.8 & 36.6 & 38.7 &  & 35.6 & 36.0 & 42.4 & 41.7 &  & 15.3 & 16.4 & 18.7 & 19.2 \\
	  & 15 & 58.6 & 57.4 & 72.9 & 70.3 &  & 24.3 & 25.3 & 30.2 & 30.3 &  & 30.8 & 30.0 & 42.2 & 40.5 &  & 12.3 & 12.8 & 17.9 & 17.7 \\
	iid & 0 & 52.1 & 53.4 & 52.7 & 54.4 &  & 21.9 & 24.0 & 22.4 & 24.6 &  & 22.7 & 23.2 & 23.7 & 24.2 &  & 10.0 & 10.9 & 10.5 & 11.3 \\
	  & 5 & 81.2 & 81.5 & 87.8 & 87.5 &  & 35.4 & 37.4 & 38.2 & 39.9 &  & 36.0 & 35.9 & 42.3 & 41.5 &  & 16.2 & 16.9 & 19.2 & 19.5 \\
	  & 15 & 56.4 & 55.6 & 70.9 & 68.1 &  & 24.7 & 25.5 & 30.7 & 30.6 &  & 25.8 & 25.2 & 38.2 & 36.2 &  & 11.3 & 11.6 & 17.5 & 16.9 \\
\multicolumn{21}{l}{\textbf{$ \bullet $ Mean-group inference}} \\ 
	MG-PC & 0 &  19.1 &  40.9 &  17.6 &  40.1 &  &   8.3 &  30.9 &   7.7 &  30.7 &  &  21.3 &  45.8 &  20.5 &  45.5 &  &   9.3 &  34.0 &   8.8 &  33.9 \\
	  & 5 &  90.1 & 107.6 &  86.1 & 104.6 &  &  58.5 &  76.9 &  57.0 &  76.0 &  & 101.1 & 120.9 &  98.4 & 119.1 &  &  64.8 &  85.5 &  63.8 &  85.1 \\
	  & 15 &  81.1 &  84.3 &  76.5 &  80.0 &  &  55.4 &  62.4 &  53.9 &  60.5 &  &  91.0 &  93.4 &  88.1 &  92.1 &  &  62.2 &  68.9 &  60.8 &  68.0 \\
	MG-SE & 0 &  21.5 &  46.6 &  19.9 &  45.8 &  &   9.5 &  35.5 &   8.7 &  35.4 &  &  21.9 &  47.6 &  21.1 &  47.3 &  &   9.6 &  35.9 &   9.1 &  35.7 \\
	  & 5 & 102.7 & 123.7 &  98.0 & 120.0 &  &  66.7 &  87.8 &  65.0 &  86.7 &  & 104.3 & 125.9 & 101.2 & 123.9 &  &  67.6 &  88.5 &  66.6 &  87.9 \\
	  & 15 &  94.8 & 100.4 &  88.7 &  95.0 &  &  63.3 &  71.6 &  61.4 &  69.5 &  &  97.6 & 102.1 &  93.7 & 100.2 &  &  64.5 &  72.2 &  63.0 &  71.0 \\
\hline \hline 
\end{tabular}

%% file: Tab_CovRates.tex
\begin{tabular}{lrrrrrcrrrrcrrrrcrrrr} 
 \hline \hline 
 & & \multicolumn{4}{c}{\textbf{Small sample}} && \multicolumn{4}{c}{\textbf{Medium sample}} && \multicolumn{4}{c}{\textbf{Medium sample}} && \multicolumn{4}{c}{\textbf{Large sample}} \\  
 & & \multicolumn{4}{c}{$ N=20, \ T=100 $} && \multicolumn{4}{c}{$ N=20, \ T=500 $} && \multicolumn{4}{c}{$ N=100, \ T=100 $} && \multicolumn{4}{c}{$ N=100, \ T=500 $} \\  
 \cmidrule(lr){3-6} \cmidrule(lr){8-11} \cmidrule(lr){13-16} \cmidrule(lr){18-21} 
 \textbf{Method} & $ h $ & SN1 & SN2 & SN3 & SN4 && SN1 & SN2 & SN3 & SN4 && SN1 & SN2 & SN3 & SN4 && SN1 & SN2 & SN3 & SN4 \\ 
 \hline 
\multicolumn{21}{l}{\textbf{$ \bullet $ Cross-sectional resampling}} \\ 
	N/no & 0 &  6.2 &  6.8 &  7.8 & 10.7 &  &  5.1 &  5.1 &  5.3 &  5.6 &  &  6.7 &  9.3 & 12.7 & 12.0 &  &  6.3 &  5.4 &  8.5 &  6.8 \\
	  & 5 & 30.7 & 29.2 & 33.5 & 33.4 &  & 16.8 & 14.1 & 22.3 & 18.6 &  & 58.8 & 65.9 & 57.6 & 61.9 &  & 26.9 & 23.0 & 34.8 & 29.3 \\
	  & 15 & 85.5 & 86.6 & 80.1 & 79.4 &  & 52.5 & 48.0 & 52.1 & 50.8 &  & 90.0 & 90.0 & 89.0 & 88.8 &  & 84.8 & 83.8 & 73.3 & 70.7 \\
	Niid/no & 0 &  3.4 &  4.6 &  7.9 &  8.1 &  &  2.5 &  3.3 &  4.0 &  4.0 &  &  3.2 &  6.5 & 12.0 &  8.7 &  &  2.7 &  1.5 &  6.9 &  3.5 \\
	  & 5 & 29.0 & 26.9 & 33.0 & 31.6 &  & 14.5 & 11.6 & 19.8 & 16.0 &  & 57.6 & 63.9 & 56.5 & 60.0 &  & 23.2 & 18.0 & 32.8 & 26.9 \\
	  & 15 & 84.8 & 86.2 & 80.4 & 79.1 &  & 51.8 & 46.6 & 51.1 & 49.2 &  & 90.0 & 90.0 & 89.1 & 89.0 &  & 83.5 & 83.5 & 73.9 & 70.5 \\
\multicolumn{21}{l}{\textbf{$ \bullet $ Temporal resampling}} \\ 
	no/T & 0 & 16.2 & 19.7 & 13.0 & 17.9 &  &  5.7 & 17.8 &  4.2 & 15.0 &  & 27.1 & 31.4 & 19.7 & 25.5 &  & 11.9 & 18.5 &  7.0 & 14.3 \\
	  & 5 &  9.6 & 13.2 &  9.3 & 12.0 &  & 11.6 & 16.0 &  9.4 & 14.4 &  & 15.0 & 28.2 & 14.1 & 21.9 &  & 13.2 & 15.2 & 10.9 & 12.5 \\
	  & 15 &  7.4 & 12.9 &  8.9 & 13.9 &  &  7.8 &  9.7 &  7.5 & 10.2 &  & 10.0 & 22.0 & 13.1 & 17.7 &  &  3.5 &  3.7 &  8.1 &  8.6 \\
	no/Tiid & 0 & 11.9 & 16.0 &  7.6 & 12.8 &  &  1.2 & 13.9 &  1.8 & 10.9 &  & 24.5 & 27.3 & 15.4 & 20.0 &  &  5.7 & 15.4 &  2.9 & 10.4 \\
	  & 5 &  7.2 & 12.4 &  4.8 &  9.1 &  &  5.8 & 11.4 &  4.1 &  7.8 &  & 22.9 & 36.4 & 13.9 & 21.7 &  &  9.3 & 13.4 &  4.8 &  8.1 \\
	  & 15 & 10.2 & 17.0 &  7.3 & 11.2 &  &  6.6 & 10.0 &  2.5 &  5.7 &  & 34.5 & 46.8 & 14.3 & 20.0 &  &  7.5 & 10.9 &  4.6 &  4.0 \\
\multicolumn{21}{l}{\textbf{$ \bullet $ Panel resampling}} \\ 
	N/T & 0 & 11.5 & 10.8 & 11.1 & 11.9 &  &  8.6 &  6.4 &  8.6 &  6.2 &  & 22.7 & 19.9 & 13.6 & 14.6 &  &  6.3 &  5.8 &  6.4 &  5.2 \\
	  & 5 &  6.5 &  3.7 &  4.9 &  3.0 &  &  6.6 &  4.9 &  6.2 &  4.9 &  &  3.0 &  4.1 &  2.5 &  4.7 &  &  4.7 &  2.7 &  2.9 &  2.2 \\
	  & 15 &  1.9 &  0.7 &  2.0 &  4.3 &  &  4.7 &  3.6 &  2.2 &  1.4 &  &  1.5 &  9.3 &  6.9 & 10.3 &  &  5.3 &  4.3 &  0.6 &  1.4 \\
	iid & 0 & 10.5 & 11.0 & 11.0 & 11.7 &  &  9.2 &  7.2 &  9.0 &  6.9 &  & 23.4 & 20.6 & 14.6 & 15.7 &  &  7.9 &  5.9 &  7.1 &  6.3 \\
	  & 5 &  7.2 &  5.0 &  6.0 &  4.2 &  &  8.1 &  6.6 &  7.5 &  5.9 &  &  2.8 &  6.3 &  2.0 &  5.2 &  &  6.7 &  5.5 &  4.7 &  4.7 \\
	  & 15 &  0.8 &  2.8 &  1.3 &  3.1 &  &  6.7 &  5.5 &  3.7 &  3.7 &  & 13.4 & 24.6 &  9.1 & 12.9 &  &  6.0 &  5.3 &  0.9 &  2.0 \\
\multicolumn{21}{l}{\textbf{$ \bullet $ Mean-group inference}} \\ 
	MG-PC & 0 & 25.5 &  5.9 & 29.6 &  7.8 &  & 26.3 &  8.0 & 29.0 &  8.0 &  &  5.6 & 10.0 &  4.3 &  9.8 &  &  4.8 & 10.0 &  4.2 & 10.0 \\
	  & 5 &  6.3 &  7.8 &  3.9 &  4.9 &  &  9.7 &  9.9 &  9.2 &  9.9 &  & 10.0 & 10.0 & 10.0 & 10.0 &  & 10.0 & 10.0 & 10.0 & 10.0 \\
	  & 15 &  1.2 &  1.5 &  7.6 &  6.8 &  &  9.5 &  9.8 &  9.0 &  9.5 &  & 10.0 & 10.0 &  9.4 &  9.7 &  & 10.0 & 10.0 & 10.0 & 10.0 \\
	MG-SE & 0 & 21.7 &  3.6 & 25.6 &  4.6 &  & 21.6 &  7.9 & 24.1 &  7.6 &  &  6.1 & 10.0 &  3.9 &  9.8 &  &  5.9 & 10.0 &  4.4 & 10.0 \\
	  & 5 &  8.4 &  8.6 &  5.7 &  7.5 &  &  9.9 & 10.0 &  9.5 & 10.0 &  & 10.0 & 10.0 & 10.0 & 10.0 &  & 10.0 & 10.0 & 10.0 & 10.0 \\
	  & 15 &  3.8 &  3.5 &  5.4 &  4.5 &  & 10.0 & 10.0 &  9.6 &  9.9 &  & 10.0 & 10.0 &  8.7 &  9.3 &  & 10.0 & 10.0 & 10.0 & 10.0 \\
\hline \hline 
\end{tabular}

%% file: 5_Conclusion.tex
\section{Conclusion} \label{sec:Conlusion}
In this article, we develop inferential methods for panel vector autoregressive (VAR) models and their structural impulse response functions (IRFs). Accounting for parameter heterogeneity and \textit{two-way dependence}, we make two main contributions: First, we adapt mean-group estimation \citep{PesaranSmith1995} and construct closed-form moment-based \textit{pooled mean-group} identification for proxy SVARs and their IRFs $ \Theta_h = f_h{\left( \boldsymbol{\beta}, \boldsymbol{\sigma}, \boldsymbol{\varphi} \right)} $, for which we establish a joint central limit theorem under standard regularity conditions. Second, we construct a recursive-design \textit{panel moving-block bootstrap} procedure and prove its consistency with the derived central limit theorem by combining results from \cite{BruggemannEtAl2016} and \cite{Menzel2021}. Comprehensive finite-sample simulations indicate that the practitioner should prefer the panel-block bootstrap with joint resampling in both dimensions, as covered by our asymptotic results under two-way dependence.

%% file: Appendix.tex
    \thispagestyle{plain} 
    \begin{center}
        {\huge ONLINE APPENDIX} \\[10pt]
        {\huge Inference in Panel SVARs \\[10pt]
            with Two-Way Dependence}\\[30pt]
        \begin{tabular}{cc}
            Lennart Empting~\orcidlink{0009-0004-5068-4639} & Saskia Öztürk~\orcidlink{0009-0008-1347-3313} \\[-0.5ex]
            {\footnotesize University of Duisburg-Essen} & {\footnotesize Georg August University of Göttingen} \\[1.5ex]

            Simone Maxand~\orcidlink{0000-0002-3153-7922} & Konstantin Wagner~\orcidlink{0009-0006-0216-0755} \\[-0.5ex]
            {\footnotesize Leipzig University} & {\footnotesize Humboldt University of Berlin}
        \end{tabular}
    \end{center}

    \renewcommand\thefigure{App.\arabic{figure}}
    \renewcommand\thetable{App.\arabic{table}}
    \setcounter{figure}{0}
    \setcounter{table}{0}
    \input{A_AppS2ID}

\input{B_AppSimRes}

    \clearpage
    \input{C_AppAsymp}

    \addcontentsline{toc}{section}{References}
    \bibliographystyle{ecta}  
    \bibliography{Lit_Boot.bib,Lit_Eco.bib,Lit_SoftRefs.bib,Lit_Asym.bib}

%% file: A_AppS2ID.tex
\section*{Overview}
This online appendix collects the supporting material referenced in the main text. \textbf{Appendix~\ref{app:S2ID}} provides the technical details behind Remark~2.1 to impose the additional restrictions from the proxy SVAR literature. \textbf{Appendix~\ref{app:Simultaion}} complements the Monte Carlo study of Section~\ref{sec:Simulation}: Appendix~\ref{sec:HetIRF} documents the scaling scheme that introduces slope heterogeneity into the panel DGP. Appendix~\ref{sec:SimRes_point} reports the finite-sample performance of IRF point estimation, which motivates the small-$T$ adjustments of Section~\ref{sec:BootBC}. Appendix~\ref{sec:SimRes_local} isolates the performance of local identification for the simulations of Appendix~\ref{sec:SimRes_point} and for additional simulations under heterogeneity in the identification problem. In \textbf{Appendix~\ref{sec:AppAsymp}}, Appendix~\ref{app:AsyMG} proves the joint central limit theorem of Theorem~\ref{th:AsyMG} and defines the blocks of its limiting covariance matrix $\tilde{V}$. Appendix~\ref{app:pbb_consist} proves the bootstrap consistency of Theorem~\ref{th:pbb_consist}.

\section{Technical details on recovering exact identification} \label{app:S2ID}
The SVAR literature proposes a variety of additional restrictions to overcome the set ambiguity from unidentified $ \mathsf{W}_{11} $ and achieve exact identification. In this technical appendix, we briefly show how the empirical practitioner can recover exemplary cases $ \bullet \in \left\{ \textbf{MR}, \textbf{JL}, \textbf{LT} \right\} $ of these restrictions as specific $ \mathsf{W}_{11}^\bullet \in \mathbb{O}(K_1) $ after performing the SVD. Usual interpretations of the moment-based approach are thus adopted under the pooled full over-identification from SVD. The $ n_{K1} $ restrictions in the form of lower-triangularity are imposed on either $ \mathsf{B} $, $ \Sigma_{m\epsilon} $, or --as we propose-- $ \mathsf{Q} $ directly. In each case, a (signed) LQ decomposition $ M = LQ $ thus provides $ Q \in \mathbb{O}(K_1) $ to obtain lower-triangular $ L = MQ' $ (with positive elements on the main diagonal). An implementation of these cases is available in the \textsf{R}-package \href{https://cran.r-project.org/package=pvars}{\textbf{pvars}} by \citet{Empting2021}.
\begin{itemize}
    \item[\textbf{MR}] \citet[p.~1224]{MertensRavn2013} impose restrictions on $ \mathsf{B} $ such that its Schur complement $ {S}_{1}^{\text{MR}} $ under the partition of Eq.~\eqref{eq:ID_subrotation} is lower-triangular with positive diagonal elements. For estimated $ \widebar{\mathsf{B}}_{1}^{\text{MR}} = \widebar{\mathsf{B}}_{1}^{\ } \widehat{\mathsf{W}}_{11}^{\text{MR}} $, this is recovered by the LQ decomposition of
    \begin{align*} 
        \widebar{S}_{1}^{\ } 
        & = \widebar{\mathsf{B}}_{11}^{\ } - \widebar{\mathsf{B}}_{12}^{\ }  \widebar{\mathsf{B}}_{22}^{-1} \widebar{\mathsf{B}}_{21}^{\ } \\
        & = \widebar{\mathsf{B}}_{11}^{\text{MR}'} \widehat{\mathsf{W}}_{11}^{\text{MR}'} - \widebar{\mathsf{B}}_{12}^{\text{MR}} \widehat{\mathsf{W}}_{22}^{\text{MR}'} \left( \widebar{\mathsf{B}}_{22}^{\text{MR}} \widehat{\mathsf{W}}_{22}^{\text{MR}'} \right)^{-1} \widebar{\mathsf{B}}_{21}^{\text{MR}} \widehat{\mathsf{W}}_{11}^{\text{MR}'} \\
        & = \left[ \widebar{\mathsf{B}}_{11}^{\text{MR}} - \widebar{\mathsf{B}}_{12}^{\text{MR}}  \widebar{\mathsf{B}}_{22}^{\text{MR}-1} \widebar{\mathsf{B}}_{21}^{\text{MR}} \right] \widehat{\mathsf{W}}_{11}^{\text{MR}'} \\
        & = \widebar{S}_{1}^{\text{MR}} \widehat{\mathsf{W}}_{11}^{\text{MR}'}
    \end{align*}
    Without actually estimating $ \mathsf{B}_{2} $, \citet{MertensRavn2013} also consider $ {S}_{2}^{\text{MR}} $, which could be analogously derived via $ \widebar{S}_{2}^{\ } = \widebar{S}_{2}^{\text{MR}} \widehat{\mathsf{W}}_{22}^{\text{MR}'} $ and complement $ \widehat{\mathsf{W}}_{11}^{\text{MR}} $ for $ \widebar{\mathsf{B}}^{\text{MR}} = \widebar{\mathsf{B}} \widehat{\mathsf{W}}^{\text{MR}} $.
    
    \item[\textbf{JL}] \citet{JentschLunsford2021} presume lower-triangularity of the $ K_1 \times K_1 $ proxy moment matrix $ \Sigma_{m\epsilon_1}^{JL} = {\Sigma}_{m\epsilon^\circ} {\mathsf{Q}}_{1}^{\text{JL}} $ when solving the matrix square root in Eq.~\eqref{eq:ID_Bi} by the Cholesky decomposition of Eq.~\eqref{eq:ID_square}. For estimated $ \widehat{\mathsf{Q}}_{1}^{\text{JL}} = \widehat{\mathsf{U}}_{1} \widehat{\mathsf{W}}_{11}^{\text{JL}} $ under $ L \geq K_1 $, this is recovered by the signed LQ decomposition of 
    \begin{align*}
        \widehat{\mathsf{M}}
        & = \widehat{\Sigma}_{m\epsilon^\circ} \widehat{\mathsf{U}}_1
          = \widehat{\mathsf{V}} \widehat{\mathsf{D}}' \widehat{\mathsf{U}}' \widehat{\mathsf{U}}_1
          = \widehat{\mathsf{V}}_1 \widehat{\mathsf{D}}_1 \\
        & = \widehat{\mathsf{M}}^{\text{JL}} \widehat{\mathsf{W}}_{11}^{\text{JL}'}.
    \end{align*}
    By the lower-trapezoidal $ \widehat{\mathsf{M}}^{\text{JL}} $, proxy $ m_{it,l} $ is positively correlated with its own shock $ \epsilon_{it,l} $ for $ l = 1,\dots,K_1 $ and uncorrelated with $ \epsilon_{it,k}, \ k = l+1,\dots,K_1 $. The proxies $ m_{it,l}, \ l = K_1,\dots,L $, can correlate with all $ \epsilon_{it,k}, \ k=1,\dots, K_1 $ because $ \widehat{\mathsf{W}}_{11}^{\text{JL}} $ is already determined by the leading $ K_1 \times K_1 $ block $ \widehat{\mathsf{M}}_{1} = \widehat{\mathsf{V}}_{11} \widehat{\mathsf{D}}_1 $. Under $ L = K_1 $ and $ \boldsymbol{\widehat{\mathsf{d}}}_1 \neq \boldsymbol{0}_{K_1} $, the matrix square root of Eq.~\eqref{eq:ID_square} is recovered as $ \widehat{\mathsf{H}} = \widehat{\mathsf{M}}_{1} \widehat{\mathsf{V}}_{11}' = \widehat{\mathsf{V}}_{11} \widehat{\mathsf{D}}_1 \widehat{\mathsf{V}}_{11}' $ and $ \widehat{\mathsf{M}}_{1}^{\text{JL}} = \widehat{\Sigma}_{m\epsilon_1}^{JL} = \operatorname{chol} \big( \widehat{\Sigma}_{m \epsilon^\circ} \widehat{\Sigma}_{\epsilon^\circ m} \big) $.
    
    \item[\textbf{LT}] We propose restrictions on $ \widehat{\mathsf{Q}}^{\text{LT}} = \widehat{\mathsf{U}} \widehat{\mathsf{W}}^{\text{LT}} $ that imply $ \widebar{\mathsf{B}}_{1}^{\text{LT}} = \widebar{\mathsf{B}}^\circ \widehat{\mathsf{U}}_{1}^{\ } \widehat{\mathsf{W}}_{11}^{\text{LT}} $ and $ \widebar{\mathsf{B}}_{2}^{\text{LT}} = \widebar{\mathsf{B}}^\circ \widehat{\mathsf{U}}_{2}^{\ } \widehat{\mathsf{W}}_{22}^{\text{LT}} $ to be lower-trapezoidal. More specifically, this lower-triangularity of the leading square block is imposed on
    \begin{align*}
        \widebar{\mathsf{B}}_{1}^{\text{LT}} = 
		\begin{bmatrix}
			\underset{(K_1 \times K_1)}{\widebar{\mathsf{B}}_{11}^{\text{LT}}} \\
			\underset{(K_2 \times K_1)}{\widebar{\mathsf{B}}_{21}^{\text{LT}}} \\
		\end{bmatrix}
        = \widebar{\mathsf{B}}^\circ 
		\begin{bmatrix}
			\underset{(K_1 \times K_1)}{\widehat{\mathsf{U}}_{11}} \\
			\underset{(K_2 \times K_1)}{\widehat{\mathsf{U}}_{21}} \\
		\end{bmatrix}
         \widehat{\mathsf{W}}_{11}^{\text{LT}}
        \quad \text{ and } \quad
        \widebar{\mathsf{B}}_{2}^{\text{LT}} = 
		\begin{bmatrix}
			\underset{(K_2 \times K_2)}{\widebar{\mathsf{B}}_{12}^{\text{LT}}} \\
			\underset{(K_1 \times K_2)}{\widebar{\mathsf{B}}_{22}^{\text{LT}}} \\
		\end{bmatrix}
        = \widebar{\mathsf{B}}^\circ 
		\begin{bmatrix}
			\underset{(K_2 \times K_2)}{\widehat{\mathsf{U}}_{12}} \\
			\underset{(K_1 \times K_2)}{\widehat{\mathsf{U}}_{22}} \\
		\end{bmatrix}
         \widehat{\mathsf{W}}_{22}^{\text{LT}},
    \end{align*}
    where the partitions of $ \widebar{\mathsf{B}}_{2}^{\text{LT}} $ and $ \widehat{\mathsf{U}}_{2} $ differ from Eq.~\eqref{eq:ID_subrotation}. The two LQ decompositions of $ \widehat{\mathsf{U}}_{11} = \widehat{\mathsf{Q}}_{11}^{\text{LT}} \widehat{\mathsf{W}}_{11}^{\text{LT}'} $ and $ \widehat{\mathsf{U}}_{12} = \widehat{\mathsf{Q}}_{12}^{\text{LT}} \widehat{\mathsf{W}}_{22}^{\text{LT}'} $ provide the orthogonal $ \widehat{\mathsf{W}}_{11}^{\text{LT}}, \widehat{\mathsf{W}}_{22}^{\text{LT}} $ and lower-triangular $ \widehat{\mathsf{Q}}_{11}^{\text{LT}}, \widehat{\mathsf{Q}}_{12}^{\text{LT}} $. Positive elements on the main diagonal of $ \widehat{\mathsf{Q}}^{\text{LT}} $ fix the column signs. Being all lower-triangular, $ \widehat{\mathsf{Q}}_{11}^{\text{LT}}, \widehat{\mathsf{Q}}_{12}^{\text{LT}} $ and $ \widehat{\mathsf{B}}_{i}^\circ \ \forall i $ yield $ n_{K1} $ zero-restrictions on $ \widehat{\mathsf{B}}_{i,11}^{\text{LT}} = \widehat{\mathsf{B}}_{i,11}^\circ \widehat{\mathsf{Q}}_{11}^{\text{LT}} $ and $ n_{K2} $ zero-restrictions on $ \widehat{\mathsf{B}}_{i,12}^{\text{LT}} = \widehat{\mathsf{B}}_{i,11}^\circ \widehat{\mathsf{Q}}_{12}^{\text{LT}} $ for each individual and for the panel mean $ \widebar{\mathsf{B}}^{\text{LT}} = \widebar{\mathsf{B}}^\circ \widehat{\mathsf{Q}}^{\text{LT}} $. In reference to $ \widebar{\mathsf{B}}^{\text{LT}} = \widebar{\mathsf{B}}^\circ $ as the edge case $ K_1 = 0 $, this may establish a convenient baseline for subsequent identification procedures that build on the mechanics of \cite{Rubio-RamirezEtAl2010}, such as sign-restrictions or ICA. Accordant block-diagonal choices of $ \mathbb{O}(K) $ are simply multiplied to the right side of $ \widebar{\mathsf{B}}^{\text{LT}} $ or $ \widebar{\Theta}_h^{\text{LT}} = \widebar{\Xi}_h \widebar{\mathsf{B}}^{\text{LT}} $.
    
    \item[\textbf{NQ}] The solution of the nearest orthogonal matrix $ \mathsf{W}_{11}^{\text{NQ}} = \mathsf{V}_{11}' $ in Eq.~\eqref{eq:ID_svd} simply passes identifying properties from the proxy moment matrix under $ L = K_1 $. In particular, the proxy variables could be mutually uncorrelated and instrumenting the shocks one-by-one as in \cite{BrunsEtAl2025} and \citet{GregoryEtAl2024}. In this case, $ \widehat{\mathsf{V}}_{11} $ from the SVD-solution re-permutes and re-signs $ \widehat{\mathsf{U}}_{1} $, thus reverting the order and sign conventions of SVD imposed on the singular vectors. In the general case, the leading $ K_1 \times K_1 $ block of $ \Sigma_{m\epsilon1} = \Sigma_{m\epsilon^\circ} \mathsf{Q}_{1}^{\text{NQ}} = \mathsf{M} \mathsf{W}_{11}^{\text{NQ}} $ is presumed to be symmetric positive definite, i.e.~$ \mathsf{M}_{1} \mathsf{W}_{11}^{\text{NQ}} = \mathsf{P} $ for the leading block $ \mathsf{M}_{1} = \mathsf{V}_{11} \mathsf{D}_1 $ of $ \mathsf{M} = \Sigma_{m\epsilon^\circ} \mathsf{U}_1 $. If $ L > K_1 $, the leading block $ \widehat{\mathsf{V}}_{11} $ is no longer orthogonal and the estimated $  \widehat{\mathsf{W}}_{11}^{\text{NQ}} $ is thus obtained as the orthogonal factor of the PD of $ \widehat{\mathsf{M}}_{1} = \widehat{\mathsf{P}}  \widehat{\mathsf{W}}_{11}^{\text{NQ}'} $, which nests $  \widehat{\mathsf{W}}_{11}^{\text{NQ}} = \widehat{\mathsf{V}}_{11}' $ under $ L = K_1 $. The arbitrary SVD result $ \widehat{\mathsf{V}}_{22} $ needs to be replaced by a unique choice from $ \mathbb{O}(K_2) $ such as \textbf{LT}.
\end{itemize}

%% file: B_AppSimRes.tex
\section{Details on Monte Carlo study} \label{app:Simultaion}

\subsection{Parameter heterogeneity in VAR coefficients} \label{sec:HetIRF}
To maintain the wave pattern of the IRFs in \citet{JentschLunsford2021}, the heterogeneous VAR coefficients are aligned to the strong initial propagation of $ A_1 $ and subsequent offsetting effect of $ A_2 $. We use $ A_i = \left[ \sigma_{(A1),i} \cdot A_1 : \sigma_{(A2),i} \cdot A_2 \right] $ with weights $ w_1 = 0.073 $ and  $ w_2 = 0.263 $ and joint randomness $ u_{(A),i} \sim U(-1, +1) $ in the two scaling factors $ \sigma_{(Aj),i} = 1 + w_j \cdot u_{(A),i} $.
\begin{figure}[ht] 
	\centering
	\resizebox{0.95\textwidth}{!}{
		\input{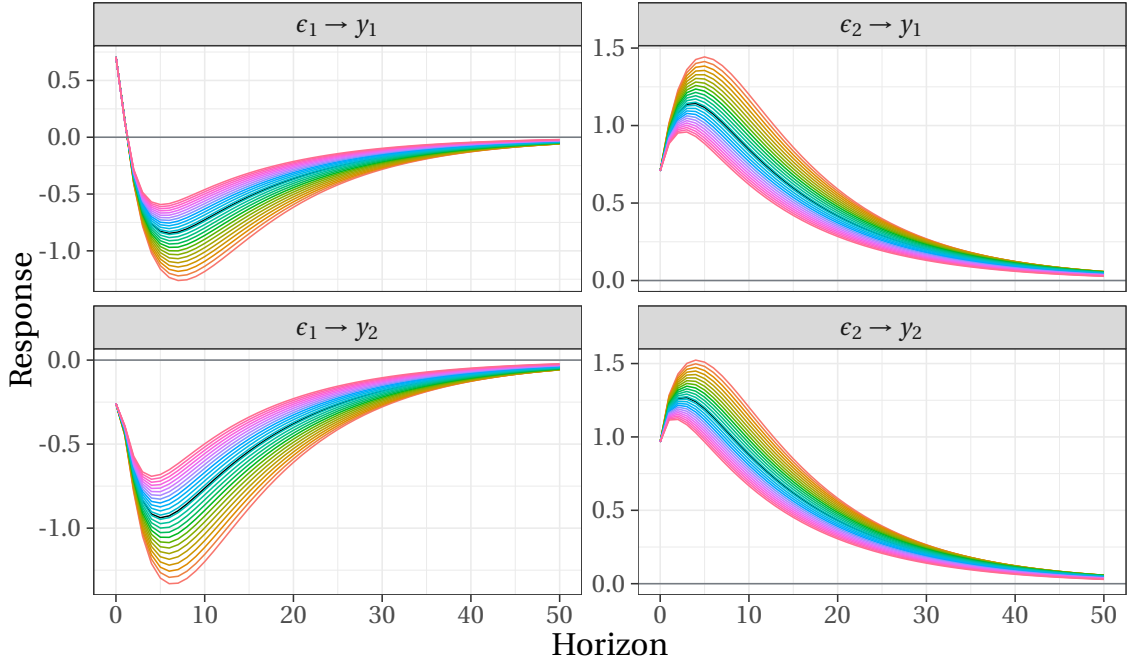}}
	\caption[Parameter heterogeneity of IRF.]{Parameter heterogeneity of individual ``true'' IRFs under common $ \mathsf{B}_i = \mathsf{B} $ as in SN1 and SN3. Heterogeneity emerges solely from $ A_i = \left[ \sigma_{(A1),i} \cdot A_1 : \sigma_{(A2),i} \cdot A_2 \right] $, here with individual realizations of $ u_{(A),i} = -1,\ldots,+1 $ at equal distances for illustrative purposes only. The baseline IRF (black line) is derived from constant baseline $ A = \left[ A_1 : A_2 \right] $ and $ \mathsf{B} $.}
	\label{fig:HetIRF}
\end{figure}

This weighted scaling scheme facilitates a ``controlled escalation'' in heterogeneity. The joint $ u_{(A),i} $ exert perfect group-dependence in $ \left\{ \boldsymbol{\eta}_{\beta i} \right\} $, which is covered by Assumption~\ref{as:random_parameters}(iii) and reflected in the spread of (non-intersecting) individual IRFs of Figure~\ref{fig:HetIRF}. In the scenarios SN1 and SN3, this spread is widest at the turning point at around $ h = 5 $ but of zero width at $ h = 0 $. Cross-sectional heteroskedasticity in SN2 or SN4, however, further amplifies the spreading by a factor of $ \sigma_{(CH),i} \sim U(0.75, 1.25) $ already from $ h = 0 $ onward until it collapses to the zero-line in the long run just like in SN1 and SN3. Despite these interactions, the scheme with two offsetting weights leaves the largest eigen root of the individual VAR processes almost unaffected in the range of $ \lambda_{i,1} \in \left(0.923,\ 0.931\right) $. It thus maintains the stability Assumption~\ref{as:SVAR}(i) by construction of the panel DGP and avoids the need for dedicated ex-post checks on stability and potential exclusions of non-stable VAR processes. 
\begin{figure}[ht] 
	\centering
	\resizebox{0.95\textwidth}{!}{
		\input{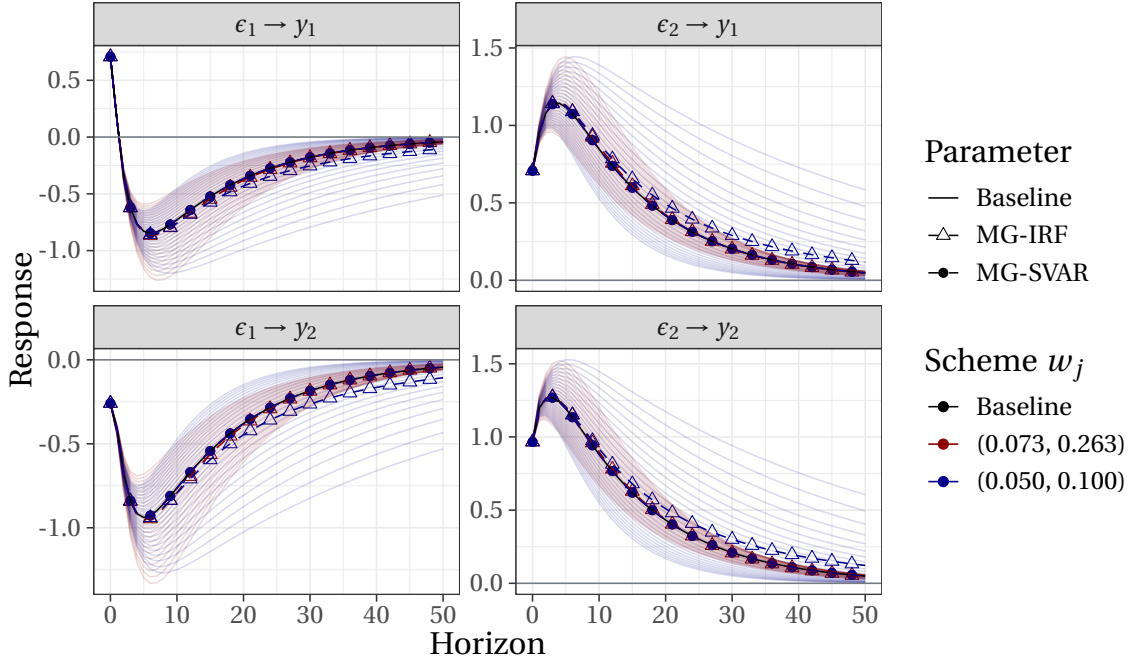}}
	\caption[Parameter heterogeneity and MG.]{Individual IRFs and their MG parameters under two different scaling schemes. Again, heterogeneity emerges solely from $ A_i = \left[ \sigma_{(A1),i} \cdot A_1 : \sigma_{(A2),i} \cdot A_2 \right] $ with individual realizations of $ u_{(A),i} = -1,\ldots,+1 $ at equal distances for illustrative purposes only. The baseline IRF (black line) is the median of both schemes and is derived from $ A = \left[ A_1 : A_2 \right] $ and $ \mathsf{B} $.}
	\label{fig:HetIRF2}
\end{figure}

Figure~\ref{fig:HetIRF2} compares the selected scaling scheme (here in red) to a more typical pattern in a statistical sense. The (blue) scaling scheme with $ w_1 = 0.05 $ and  $ w_2 = 0.10 $ illustrates exemplarily that $ f_h $ maps from $ A_i $ (not depicted) generated with symmetric marginal distributions to $ \Theta_{i,h} $ (light blue lines) having skewed marginal distributions at horizons $ h > 1 $. Here, the largest root is in the range of $ \lambda_{i,1} \in \left(0.885,\ 0.978 \right) $, leading to a typical long tail from more persistent individuals. Since the IRF transformation $ f_h $ involves exponentiation of $ \boldsymbol{\beta} $ to the power of $ h $,  $ f_h $ at $ h > 1 $ is not affine and thus typically exhibits
\begin{align} \label{eq:HetIRF2}
	\begin{split}		
		 \Theta_h
         = f_h{\left( \boldsymbol{\beta}, \boldsymbol{\sigma}, \boldsymbol{\varphi} \right)} 
         = f_h{\left( \E \left[ \boldsymbol{\beta}_i \right], \E \left[ \boldsymbol{\sigma}_i  \right], \boldsymbol{\varphi} \right)} 
         \neq \E \left[ f_h{\left( \boldsymbol{\beta}_i, \boldsymbol{\sigma}_i, \boldsymbol{\varphi} \right)} \right] 
         = \E \left[ \Theta_{i,h} \right].
	\end{split}
\end{align} 
Figure~\ref{fig:HetIRF2} depicts this by the congruent IRF of the baseline SVAR process ($ \Theta_h $, black line) and of the MG-SVAR ($ f_h{\left( \boldsymbol{\widebar{\beta}}, \boldsymbol{\widebar{\sigma}}, \boldsymbol{\widebar{\varphi}} \right)} $, blue dots) and the diverging MG-IRFs ($ N^{-1} \sum_{i=1}^{N} f_h{\left( \boldsymbol{\widehat{\beta}}_i, \boldsymbol{\widehat{\sigma}}_i, \boldsymbol{\widebar{\varphi}} \right)} $, blue triangles) due to increasing skewness over $ h $. As can be inferred, consistent estimates of individual VARs and IRFs converge to their parameter distribution but still result in an inconsistency between the IRFs of MG-SVARs and the MG-IRFs under slope heterogeneity.

Even under slope homogeneity, an equivalent relation remains for the map $ f_h $ from the estimated $ \widehat{A}_i $ to the estimated $ \widehat{\Theta}_{i,h} $. In finite samples, the sampling distribution of $ \widehat{\Theta}_{i,h} $ is known to be skewed \citep{Kilian1998} such that $ f_h{\left( \boldsymbol{\widebar{\beta}}, \boldsymbol{\widebar{\sigma}}, \boldsymbol{\widebar{\varphi}} \right)} \neq N^{-1} \sum_{i=1}^{N} f_h{\left( \boldsymbol{\widehat{\beta}}_i, \boldsymbol{\widehat{\sigma}}_i, \boldsymbol{\widebar{\varphi}} \right)} $ at $ h > 1 $. Compared to the IRFs of pooled or MG-SVAR estimates, MG-IRF estimates typically overemphasize the more persistent individual realizations in the skewed distribution such that larger sampling variances lead to larger MG-IRF estimates. The MG-IRF estimates in the literature remain a location parameter of the empirical distribution of individual effect estimates rather than referring those effect estimates back to the data generating statistical model. At least, if $ \widehat{A}_i $ degenerates asymptotically to the common point mass of homogeneous $ A_i = A $, randomness collapses and Eq.~\eqref{eq:HetIRF2} thus exhibits equality as a trivial case.

Comparing both scaling schemes, Figure~\ref{fig:HetIRF2} illustrates some favorable particularities of the (red) scheme with $ w_1 = 0.073 $ and  $ w_2 = 0.263 $ for MC simulations. 
Firstly, the approximate equality $ \Theta_h \approx \E \left[ \Theta_{i,h} \right] $ allows us to limit influences from the (now almost symmetric) parameter distribution of IRFs and thus to isolate influences from the sampling distribution on MG-IRFs. Secondly, the almost homogeneous persistence by $ \lambda_{i,1} \approx \lambda_{1} = 0.93 $ imposes not only stability $ \lambda_{i,1} < 1 \ \forall i $ but a long-run homogeneity that lets the individual IRFs already converge at mid-range horizons $ h $. Both features accord with economic theory, which often emphasizes heterogeneous short-run dynamics and has therewith inspired the use of MG-IRFs, e.g.~in \citet{GambacortaEtAl2014} for responses to unconventional monetary policy shocks.

As an aside, we remark that larger $ w_j $ do not necessarily translate into wider spreads of $ \Theta_{i,h} $ across $ i $ and thus cannot verify robustness of MG-IRFs against stronger heterogeneity either. Although $ w_1 = 0.073 $ and  $ w_2 = 0.263 $ imply larger VAR parameter variance, $ w_1 = 0.05 $ and  $ w_2 = 0.10 $ exert less control on the persistence and thus lead to a larger skew in the parameter distribution of $ \Theta_{i,h} $. Without any control on heterogeneity, even smaller variances in a panel DGP of all-independent $ a_{kk^\bullet,ij} $ in $ A_{ij} $ can occasionally escalate into large $ \lambda_{i,1} $ of unit- or explosive roots. These extreme individuals then dominate the MG-IRFs in the long run of $ h $ although the mass of individual IRFs is close to the baseline in this panel DGP. Nonetheless, those heterogeneity configurations which generate the typical pattern in a statistical sense do not pose a problem for MG-SVAR estimation because extreme individuals already average out by $ \widebar{a}_{kk^\bullet,j} $ at the element level of $ \widebar{A} $. Our single choice of scheme is simply motivated by the focus on the four residual structure scenarios under a typical pattern in an economic sense.

\subsection{Performance of IRF point estimation} \label{sec:SimRes_point}
While the MC study of Section~\ref{sec:Simulation} focuses on the inference validated by Corollary~\ref{cor:asym_PBBSIRF}, we report here the performance of the point estimation of $ \Theta_h $ in reference to Corollary~\ref{cor:CLT_SIRF}. The performance results are throughout based on MC simulations using the same 16 DGP setups of Section~\ref{sec:SimDGP} again, but now with an increased number of $ R_{\text{MC}} = 2000 $ replications.

\begin{figure}[ht] 
	\centering
	\resizebox{0.98\textwidth}{!}{
		\input{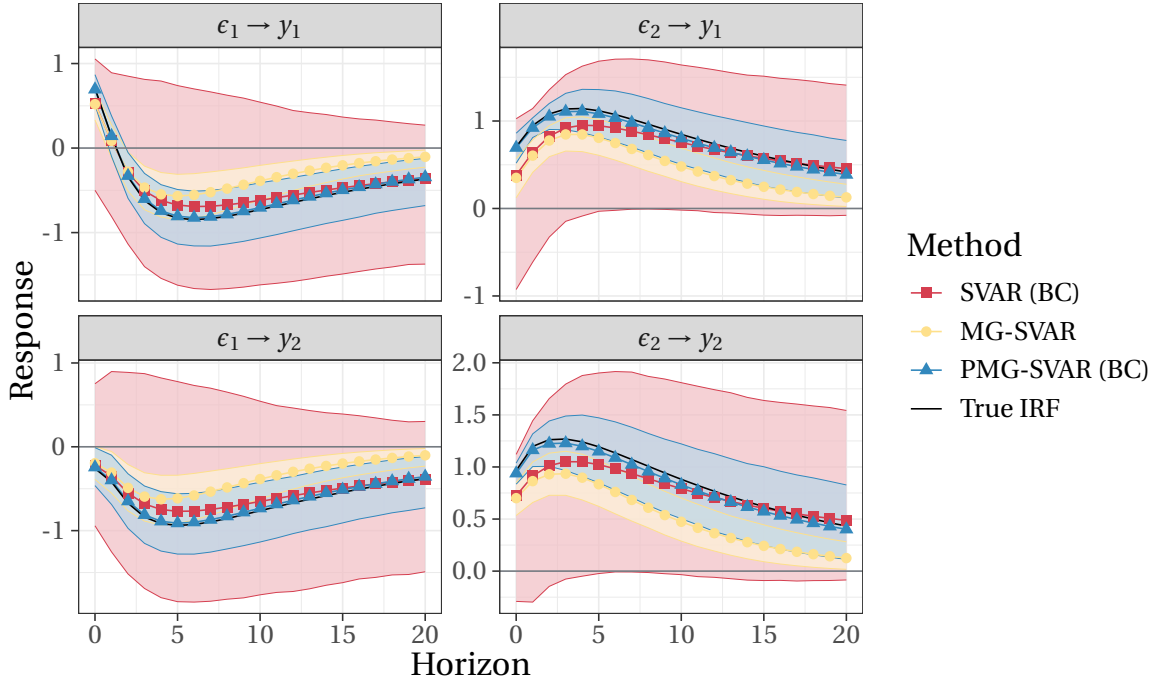}}
	\caption[IRF performance.]{Average IRF and 90\%-quantile bands of the three moment-based identification methods (\textit{individual} with BC by \citet{Kilian1998}, \textit{mean-group} without any BC, and \textit{pooled MG} with BC) for proxy SVARs according to Monte Carlo simulations. $ R_{\text{MC}} = 2000 $ replications with sample size $ \left\{ N=20,  T=100 \right\} $ have been used to assess the SN4 \textit{small sample} performance.}
	\label{fig:cbs_SN1N20T500}
\end{figure}
Figure~\ref{fig:cbs_SN1N20T500} displays the $ 2 \times 2 $ IRF results for SN4 \textit{small sample} with residual two-way dependence as the most challenging DGP setup. Therein, \textit{SVAR (BC)} refers to the individual estimation and proxy-identification of a single $ N = 1$ unit with the individual bias correction of $ \boldsymbol{\widehat{\beta}}_{\text{BC},i} $ using \citet{Kilian1998} and with the sample covariance estimator $ \boldsymbol{\widehat{\sigma}}_{\text{BC},i} $ using the degrees of freedom $ T - (K \cdot p + 1) $. The individual results give insights into the estimator performance before combining them by MG and allow to compare the result from the individual and pooled identification. Unlike the panel estimates from heterogeneous panels, the individual SVAR estimates are generated from the ideal parameters $ A_i = A $ and $ \mathsf{B}_i = \mathsf{B} $. \textit{MG-SVAR} refers to the MG estimation of $ N $ individual SVAR without small-$T$-corrections on $ \boldsymbol{\widehat{\beta}}_i $ and $ \boldsymbol{\widehat{\sigma}}_i $ and with individual estimates $ \widehat{\mathsf{Q}}_i $. The procedure thus constitutes the panel MG-extension of the setup in \citet{JentschLunsford2021}. \textit{PMG-SVAR (BC)} has been outlined in Section~\ref{sec:BootBC} and is the estimation procedure of main interest. The procedure uses both small-$ T $-corrections and the pooled SVD estimator $ \widehat{\mathsf{Q}} $ to derive $ \widebar{\Theta}_{\text{BC},h} = f_h{\left( \boldsymbol{\widebar{\beta}}_{\text{BC}}, \boldsymbol{\widebar{\sigma}}_{\text{BC}}, \boldsymbol{\widebar{\varphi}} \right)} $.

\begin{table}[ht] 
    \centering
    \caption{Performance $ \left( \frac{100}{K^2} \cdot \sum_{k,s=1}^{K,K} f_\bullet(\cdot)_{h,ks} \right) $ of IRF point estimation}
    \resizebox{0.99\textwidth}{!}{
        \def\arraystretch{.85}
        \input{Tab_irfPOINT.tex}}
    \label{tab:irfPOINT}
\end{table}
Table~\ref{tab:irfPOINT} juxtaposes the aggregated performance of the individual, MG, and PMG estimation for all $ 4 \times 4 $ DGP setups. 
Comparing \textit{(i)} columns of $ (20,100) $ and $ (100,100) $ or $ (20,500) $ and $ (100,500) $ or \textit{(ii)} SN1 and SN3 for the individual approach gives hint on the precision of the simulation, since the results should be the same. Overall, MG estimation retains the small-$ T $ biases of the individual methods but makes those more visible. Within the RMSE, the bias persists while the estimator variance shrinks as $ N \to \infty $. Conversely, bias-corrected (BC) PMG estimation effectively remedies this issue and is thus considered in Section~\ref{sec:SimRes} for evaluating the finite-sample performance of the inferential procedures.

\subsection{Performance of local identification} \label{sec:SimRes_local}
To isolate the performance of local identification, we consider the Amari distance between the actual $ \mathsf{Q} $ and the estimated $ \mathsf{\widehat{Q}} $. Originating from the evaluation of ICA,\footnote{See e.g. \citet[Fig.~1]{Herwartz2017}, who identify panel SVARs by pooled ICA.} this performance metric is invariant to $ \mathsf{P}^{\pm} $ and the related periodicity of $ \mathsf{Q}(\boldsymbol{\theta}) $. It is calculated and scaled by
\begin{align} \label{eq:ari}
	\begin{split}
	    \left\Vert \mathsf{Q}, \mathsf{\widehat{Q}} \right\Vert_{\mathrm{Amari}} & =
        \sum_{k=1}^{K} \left( \frac{\sum_{s=1}^{K} |r_{ks}|}{\max_s |r_{ks}|} - 1 \right) +
        \sum_{s=1}^{K} \left( \frac{\sum_{k=1}^{K} |r_{ks}|}{\max_k |r_{ks}|} - 1 \right), 
        \quad r_{ks} = \left[ \mathsf{Q}' \mathsf{\widehat{Q}}  \right]_{ks}, \text{ and} \\
        \frac{ \left\Vert \mathsf{Q}, \mathsf{\widehat{Q}} \right\Vert_{\mathrm{Amari}} }{2K(K-1)} & \in \left[ 0, 1 \right].
	\end{split}
\end{align}
In the following, we first report the sampling distributions of this scaled Amari distance for all point estimation setups of Section~\ref{sec:SimRes_point}, which are now subsumed under the bivariate {DGP~I}. Subsequently, we introduce a trivariate {DGP~II} additionally and illustrate the local identification performance for both DGPs under heterogeneity in the identification problem.

\begin{figure}[ht] 
	\centering
	\resizebox{0.98\textwidth}{!}{
		\input{Fig_ari.tex}}
	\caption[Identification performance.]{Scaled Amari distance on $ \mathsf{Q} $ of the three identification methods (\textit{individual}, \textit{mean-group}, and \textit{pooled MG}) for proxy SVARs according to Monte Carlo simulations with $ R_{\text{MC}} = 2000 $ replications of DGP~I. The boxplots display median bars and min-max whiskers.}
	\label{fig:ari}
\end{figure}
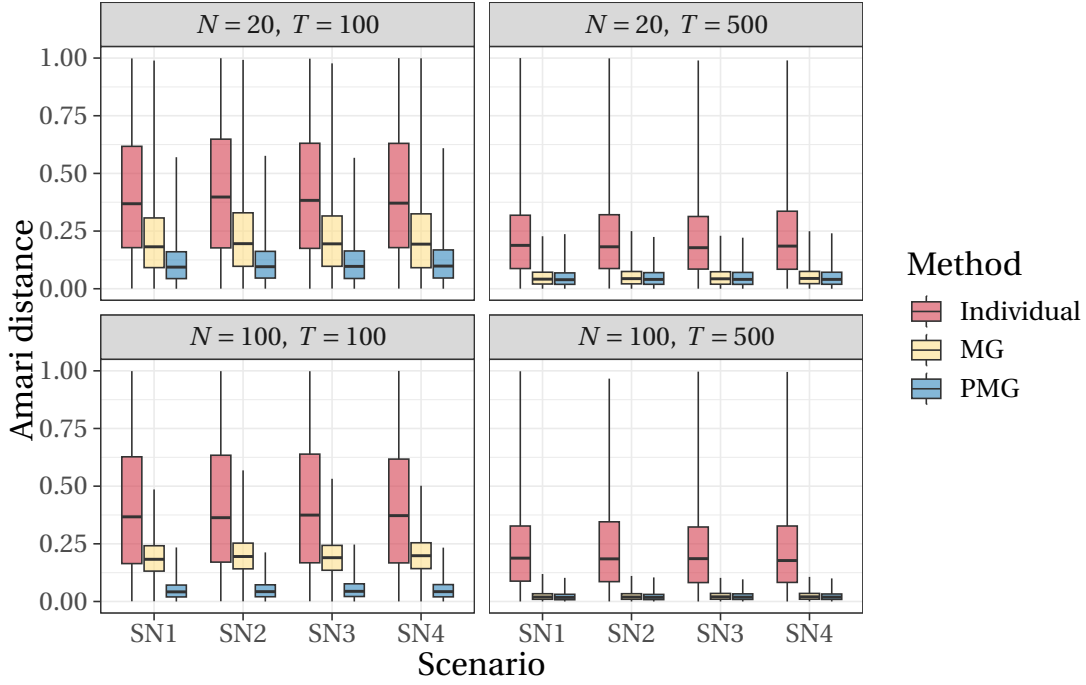
Figure~\ref{fig:ari} displays simulated distributions of the Amari distance for the $ 4 \times 4 $ DGP setups of Section~\ref{sec:SimDGP} under DGP~I adapted from \citet{JentschLunsford2021}. Again, the individual results admit an assessment of simulation precision since they are based on the same number of $ T=100 $ periods in the left and $ T=500 $ in the right column respectively. Extreme realizations of the individual SVARs under the $ L = 1 $ weak proxy can be seen to span the complete interval of $ \left[ 0, 1 \right] $. 
Both panel estimators reflect consistent behavior as $ \left\Vert \mathsf{Q}, \mathsf{\widehat{Q}} \right\Vert_{\mathrm{Amari}} $ shrinks under increasing sample size irrespective of residual structure. While the estimation precision of MG rather profits from larger $ T $, the precision of PMG increases with the panel sample size $ NT $. An inspection of the systematic deviations from $ \mathsf{Q} $ attributes this to issues from the ``individually weak'' proxy, which the PMG identification reduces greatly but the MG retains as the cross-sectional average of individual biases.

For DGP II, we adapt the trivariate VAR process of lag-order $ p=1 $ from \citet{BrunsEtAl2025} to create a homogeneous panel VAR process (again without deterministic terms). Accordingly, we specify Eq.~\eqref{eq:PVAR} using the common parameter matrices
\begin{align*}
    A_{1} = 
    \begin{pmatrix}
        0.9 &   0 &   0 \\ 
        1/3 & 1/3 & 1/3 \\
        1/3 & 1/3 & 1/3
    \end{pmatrix}
  \text{\quad and \quad}
    \mathsf{B} =  
    \begin{pmatrix}
        1   & 0.2 & 0.2 \\
        0.2 & 1   & 0.2 \\
        0.2 & 0.2 & 1   \\
    \end{pmatrix} 
\end{align*}
such that $ \boldsymbol{\theta} \approx (0.197, 0.194, 0.041) $ in $ \mathsf{B} = \mathsf{B}^\circ \mathsf{Q}( \boldsymbol{\theta} ) $. Focusing on the SN4 \textit{small sample} performance, we only use $ T=100 $ as in \citet{BrunsEtAl2025}, $ N=20 $, and residual two-way dependence (SN4) as described in Section~\ref{sec:SimDGP}. The $ L $ proxy variables are generated by
\begin{align*}
    \boldsymbol{m}_{it} =  
    \begin{bmatrix}
        {\psi}_{(m),i} I_{K_1} & 0_{K_1 \times K_2} \\
        {\psi}_{(m),i} I_{K_1} & 0_{K_1 \times K_2}
    \end{bmatrix} 
    \begin{pmatrix}
        \boldsymbol{\epsilon}_{1,it} \\
        \boldsymbol{\epsilon}_{2,it}
    \end{pmatrix} 
    + \boldsymbol{e}_{(m),it} \text{\quad with \quad } \boldsymbol{e}_{(m),it} \overset{iid}{\sim} \mathcal{N}(0, I_L) \text{\ and \ } \boldsymbol{\epsilon}_{it} \overset{SN4}{\sim} (0, I_{K}).
\end{align*} 
They are mutually uncorrelated and of identical strength by $ {\psi}_{(m)} I_{K_1} $ and $ \text{Cov}\left( \boldsymbol{e}_{(m),it} \right) = I_L $. In two different configurations, we set a baseline proxy strength \textit{(i)} with $ \psi_{(m)} = 0.17 $ again as in the simulation for Figure~\ref{fig:ari} and \textit{(ii)} with $ \psi_{(m)} = 0.58 $ that implies a correlation of $\rho_{m, \epsilon_1} = 0.5 $ between each proxy variable and its attributed shock series as in \citet{BrunsEtAl2025}.

Moreover, we consider not only the three methods under just-identification again ($ L = K_1 = 2 $ for DGP~II), but additionally the pooled MG under over-identification by $ L = 2 > K_1 = 1 $ for DGP~I and $ L = 4 > K_1 = 2 $ for DGP~II respectively. The additional $ L - K_1 $ proxies are generated just as the first $ K_1 $ proxies, but they have their own uncorrelated error terms $ {e}_{(m),lit} $ and do not enter the just-identification. The remaining parameters for the $ K_2 = 1 $ shock series $ \boldsymbol{\epsilon}_{2,it} $ are exactly identified by the orthogonal complement in $ \left[ {\mathsf{Q}}_1 : {\mathsf{Q}}_{1\perp} \right] $.

We are interested in evaluating how merely elemental forms of parameter heterogeneity already affect the local identification performance of the identification methods. Specifically, we vary the individualized factors of the SVD selectively and without any interaction between them. The effects from this cross-sectional dispersion can be isolated if the baseline parameters are held constant over the $ R_{\text{MC}} = 2000 $ Monte Carlo replications. To do so, we introduce the elemental heterogeneity using antithetic draws of individual parameters in four cases: 
\begin{enumerate}[label=\textbf{CS\arabic*}]	
    \item \textbf{None.} This case specifies homogeneity of all parameters as the benchmark. As done under DGP~I, this ideal setup is also used for individual estimation and identification.
    
    \item \textbf{$ \mathsf{U}_i $.} Heterogeneity in the set-identified $ \mathsf{U}_i $ is directly introduced into $ \mathsf{Q}_i = \mathsf{Q} \mathsf{R}_i $ in order to briefly check on the violated homogeneity Assumption~\ref{as:SVAR}(iv) on $ \mathsf{Q} $. 
    For DGP~II, $ \mathsf{R}_i = \mathsf{R}\left( \alpha_i, n_i \right) $ rotates by the angle $ \alpha_i $ about an axis $ n_i = \left( \cos \omega_i, \sin \omega_i, 0 \right)'$. The axis $ n_i $ keeps the $ n_{K1} = 1 $ NQ-restriction by $ \mathsf{V}_{i,11} $ correctly aligned relative to its varying frame $ \mathsf{U}_{1,i} $. The angle  $ \omega_i \sim U \left( 0, 2\pi \right) $ distributes the tilt from $ \alpha_i $ (in expectation) evenly over the $ K_1 = 2 $ columns of $ \mathsf{Q}_1 $. For DGP~I, the rotation matrix $ \mathsf{R}_i $ simply reduces to $ \mathsf{R}\left( \alpha_i \right) $ due to $ \mathsf{V}_{i,11} = 1 $. For both DGPs, the angle $ \alpha_i \sim U \left(  -\frac{\pi}{6}, \frac{\pi}{6} \right) $ varies the subspace spanned by $ \mathsf{U}_{1,i} $ and, orthogonally to it, the one spanned by $ \mathsf{U}_{2,i} $. In the bivariate case, it covers two-thirds of the maximum admissible angle of $ \pm \frac{\pi}{4} $, where larger angles would introduce signed permutation via $ \mathsf{R}_i $. Each pair $ \{ - \alpha_i, +\alpha_{i} \} $ of antithetic draw preserves the baseline $ \mathsf{Q} $ as the orthogonal factor of the polar decomposition $ N^{-1} \sum_i \mathsf{Q}_i = \mathsf{Q} \mathsf{H} $. Notably, the symmetric factor $ \mathsf{H} $ is diagonal in expectation here and dampens the proxy strength $ \boldsymbol{\mathsf{d}}_1 $ at the panel level.
    
    \item \textbf{$ \psi_{(m),i} $.} Heterogeneity in the proxy strength from $ \psi_{(m),i} = \sigma_{(\psi),i} \psi_{(m)} $ is introduced via the scaling factor $ \sigma_{(\psi),i} = 1 + u_{(\psi),i} $ and antithetic draws of $ u_{(\psi),i} \sim U \left( -1, +1 \right) $. Each pair $ \{ -u_{(\psi),i}, +u_{(\psi),i} \} $ of the $ \frac{N}{2} $ antithetic draws perfectly averages to zero such that each panel draw of the $ R_{MC} = 2000 $ replications obeys $ N^{-1} \sum_i \psi_{(m),i} = \psi_{(m)} $ by construction. 
    
    \item \textbf{Rank.} Heterogeneity in the rank of the proxy moment $ \Sigma_{m\epsilon_1,i} $ is introduced as random failures of the rank condition in $ 50 \% $ of the individuals. Antithetic draws of $ \sigma_{(\psi),i} \in \{0, 2\} $ for $ \psi_{(m),i} = \sigma_{(\psi),i} \psi_{(m)} $ keep $ N^{-1} \sum_i \psi_{(m),i} = \psi_{(m)} $ constant again for each panel draw.
\end{enumerate}
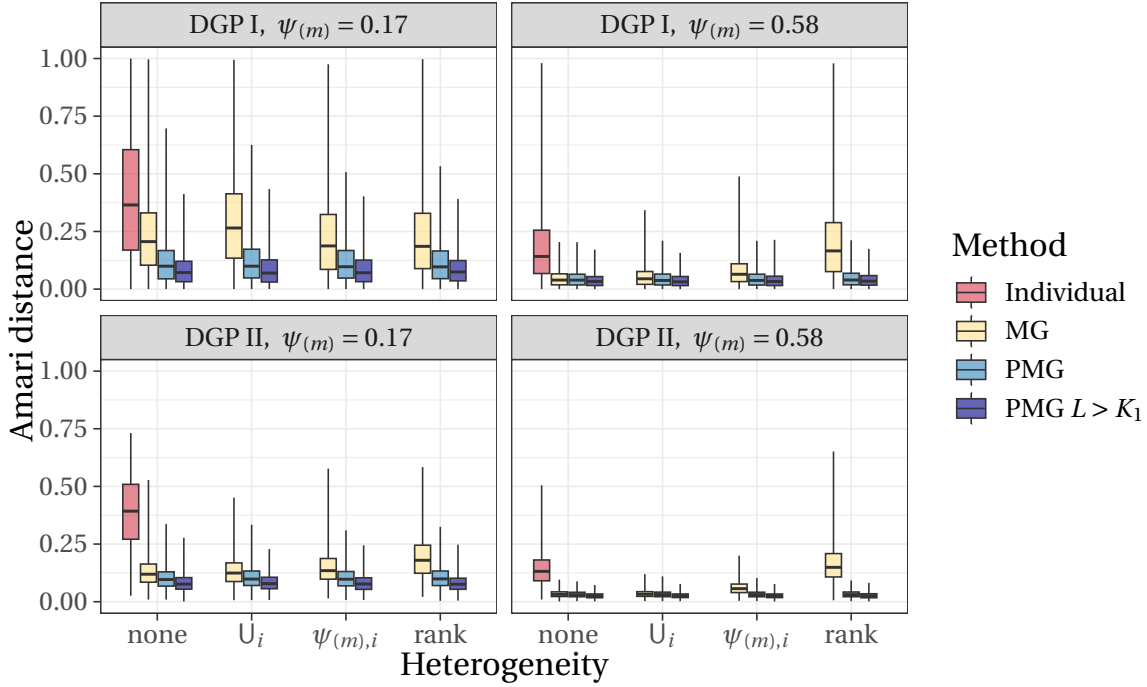
\begin{figure}[ht] 
	\centering
	\resizebox{0.98\textwidth}{!}{
		\input{Fig_ari2.tex}}
	\caption[Identification performance.]{Scaled Amari distance on $ \mathsf{Q} $ of the three identification methods (\textit{individual}, \textit{mean-group}, and \textit{pooled MG}) for proxy SVARs according to Monte Carlo simulations with $ R_{\text{MC}} = 2000 $ replications. The boxplots display median bars and min-max whiskers.}
	\label{fig:ari2}
\end{figure}

Figure~\ref{fig:ari2} displays the simulated distributions of the Amari distance for the $ 2 \times 4 $ setups under DGP~I and II. Compared to the homogeneity benchmark, MG turns out susceptible to cases of weakly identified individuals under heterogeneous proxy strength. Although the information set is held constant at the panel level, the extreme case of rank failures exhibits this in particular. By contrast, CS3 and CS4 do not affect the performance of PMG as only the panel information set (i.e. total strength and number of proxies) is decisive here. The dampening effect introduced in CS2 by the heterogeneous $ \mathsf{U}_i $ is not visible due to the dominant sampling variation.

%% file: Tab_irfPOINT.tex
\begin{tabular}{lrrrrrcrrrrcrrrrcrrrr} 
 \hline \hline 
 & & \multicolumn{4}{c}{\textbf{Small sample}} && \multicolumn{4}{c}{\textbf{Medium sample}} && \multicolumn{4}{c}{\textbf{Medium sample}} && \multicolumn{4}{c}{\textbf{Large sample}} \\  
 & & \multicolumn{4}{c}{$ N=20, \ T=100 $} && \multicolumn{4}{c}{$ N=20, \ T=500 $} && \multicolumn{4}{c}{$ N=100, \ T=100 $} && \multicolumn{4}{c}{$ N=100, \ T=500 $} \\  
 \cmidrule(lr){3-6} \cmidrule(lr){8-11} \cmidrule(lr){13-16} \cmidrule(lr){18-21} 
 \textbf{Method} & $ h $ & SN1 & SN2 & SN3 & SN4 && SN1 & SN2 & SN3 & SN4 && SN1 & SN2 & SN3 & SN4 && SN1 & SN2 & SN3 & SN4 \\ 
 \hline 
\multicolumn{21}{l}{\textbf{$ \bullet $ Absolute bias}} \\ 
	SVAR (BC) & 0 & 18.3 & 21.2 & 20.4 & 19.7 &  &  2.9 &  3.0 &  2.5 &  3.6 &  & 20.2 & 22.2 & 19.3 & 21.7 &  &  2.9 &  3.2 &  2.8 &  3.5 \\
	  & 5 & 15.6 & 19.2 & 18.1 & 16.6 &  &  3.9 &  4.4 &  3.6 &  5.2 &  & 17.2 & 20.4 & 18.4 & 19.2 &  &  4.0 &  4.2 &  4.0 &  5.1 \\
	  & 15 &  2.8 &  4.6 &  5.3 &  3.7 &  &  1.5 &  1.4 &  0.6 &  2.2 &  &  3.3 &  5.5 &  4.9 &  4.6 &  &  1.6 &  1.6 &  1.3 &  2.4 \\
	MG-SVAR & 0 & 20.8 & 21.7 & 21.0 & 21.8 &  &  3.3 &  3.5 &  3.3 &  3.5 &  & 20.8 & 21.7 & 20.9 & 21.7 &  &  3.3 &  3.5 &  3.3 &  3.6 \\
	  & 5 & 30.2 & 31.7 & 30.0 & 31.3 &  &  7.1 &  7.4 &  7.0 &  7.3 &  & 30.4 & 31.6 & 30.1 & 31.6 &  &  7.1 &  7.4 &  7.0 &  7.5 \\
	  & 15 & 34.7 & 35.3 & 34.0 & 34.7 &  &  8.8 &  9.1 &  8.6 &  9.0 &  & 35.1 & 35.5 & 34.7 & 35.2 &  &  8.9 &  9.2 &  8.8 &  9.2 \\
	PMG-SVAR (BC) & 0 &  1.1 &  1.8 &  1.1 &  1.7 &  &  0.2 &  0.4 &  0.2 &  0.4 &  &  0.7 &  1.3 &  0.7 &  1.3 &  &  0.1 &  0.3 &  0.1 &  0.3 \\
	  & 5 &  1.2 &  3.5 &  1.1 &  3.2 &  &  0.3 &  0.7 &  0.3 &  0.6 &  &  0.8 &  2.8 &  0.5 &  2.7 &  &  0.2 &  0.6 &  0.1 &  0.6 \\
	  & 15 &  3.3 &  4.8 &  1.9 &  3.7 &  &  0.4 &  0.8 &  0.1 &  0.6 &  &  3.8 &  5.1 &  2.9 &  4.2 &  &  0.5 &  0.8 &  0.2 &  0.7 \\
\multicolumn{21}{l}{\textbf{$ \bullet $ Length of 90\%-quantiles}} \\ 
	SVAR (BC) & 0 & 160.1 & 165.1 & 163.1 & 165.3 &  &  59.0 &  59.7 &  58.0 &  62.4 &  & 163.0 & 167.5 & 160.4 & 166.4 &  &  61.2 &  62.1 &  59.4 &  61.9 \\
	  & 5 & 206.3 & 221.6 & 211.6 & 215.7 &  &  96.5 &  99.4 &  95.7 & 101.8 &  & 211.1 & 222.5 & 210.6 & 223.3 &  &  99.5 & 100.0 &  97.6 & 102.7 \\
	  & 15 & 178.4 & 181.9 & 180.2 & 179.6 &  &  72.1 &  74.3 &  73.2 &  75.7 &  & 174.6 & 183.8 & 180.8 & 186.6 &  &  73.2 &  73.9 &  73.7 &  75.0 \\
	MG-SVAR & 0 &  34.8 &  36.4 &  35.6 &  37.8 &  &  14.7 &  17.2 &  15.2 &  17.5 &  &  15.7 &  16.1 &  17.0 &  17.3 &  &   6.6 &   7.6 &   7.2 &   8.1 \\
	  & 5 &  40.6 &  41.6 &  46.8 &  49.0 &  &  25.4 &  28.7 &  29.0 &  31.4 &  &  18.2 &  18.8 &  24.7 &  25.3 &  &  11.4 &  12.4 &  15.0 &  15.2 \\
	  & 15 &  21.0 &  21.1 &  32.5 &  31.3 &  &  17.1 &  18.6 &  23.8 &  24.4 &  &   9.4 &   9.2 &  20.2 &  19.3 &  &   7.7 &   8.0 &  14.6 &  13.9 \\
	PMG-SVAR (BC) & 0 &  29.9 &  32.1 &  31.9 &  33.9 &  &  13.2 &  15.8 &  14.2 &  16.6 &  &  13.7 &  14.3 &  15.2 &  15.8 &  &   6.0 &   7.0 &   6.7 &   7.7 \\
	  & 5 &  52.4 &  53.9 &  64.6 &  63.2 &  &  25.0 &  28.0 &  29.1 &  31.5 &  &  24.0 &  24.3 &  34.5 &  34.4 &  &  11.3 &  12.1 &  15.6 &  16.1 \\
	  & 15 &  47.1 &  46.2 &  71.2 &  68.0 &  &  19.1 &  21.1 &  27.4 &  28.0 &  &  20.9 &  20.3 &  43.5 &  41.1 &  &   8.8 &   9.1 &  16.5 &  16.1 \\
\multicolumn{21}{l}{\textbf{$ \bullet $ Root mean square error}} \\ 
	SVAR (BC) & 0 & 52.0 & 55.8 & 54.5 & 54.4 &  & 20.3 & 20.6 & 18.6 & 21.9 &  & 54.1 & 56.7 & 52.5 & 55.7 &  & 19.1 & 20.5 & 19.9 & 21.6 \\
	  & 5 & 64.5 & 70.4 & 67.0 & 67.7 &  & 31.7 & 32.8 & 30.0 & 34.1 &  & 65.9 & 69.9 & 67.3 & 69.8 &  & 30.8 & 32.4 & 31.6 & 33.6 \\
	  & 15 & 55.9 & 58.1 & 57.1 & 57.7 &  & 22.8 & 23.7 & 22.5 & 24.5 &  & 56.2 & 57.8 & 57.6 & 58.9 &  & 22.5 & 23.6 & 23.0 & 24.0 \\
	MG-SVAR & 0 & 24.0 & 25.0 & 24.3 & 25.4 &  &  5.8 &  6.4 &  5.9 &  6.6 &  & 21.6 & 22.5 & 21.7 & 22.6 &  &  4.0 &  4.3 &  4.1 &  4.4 \\
	  & 5 & 32.7 & 34.3 & 33.3 & 34.8 &  & 10.5 & 11.5 & 11.3 & 12.1 &  & 30.9 & 32.1 & 31.1 & 32.6 &  &  7.9 &  8.4 &  8.5 &  9.0 \\
	  & 15 & 35.3 & 35.9 & 35.5 & 36.1 &  & 10.3 & 10.8 & 11.3 & 11.7 &  & 35.2 & 35.6 & 35.3 & 35.7 &  &  9.2 &  9.6 &  9.8 & 10.1 \\
	PMG-SVAR (BC) & 0 &  9.3 & 10.0 &  9.9 & 10.6 &  &  4.0 &  4.9 &  4.3 &  5.1 &  &  4.3 &  4.6 &  4.8 &  5.1 &  &  1.8 &  2.2 &  2.0 &  2.4 \\
	  & 5 & 16.1 & 16.8 & 19.4 & 19.9 &  &  7.6 &  8.7 &  8.9 &  9.6 &  &  7.3 &  7.9 & 10.6 & 10.9 &  &  3.4 &  3.7 &  4.8 &  5.0 \\
	  & 15 & 14.5 & 14.8 & 21.8 & 21.4 &  &  5.9 &  6.5 &  8.4 &  8.7 &  &  7.4 &  8.1 & 13.7 & 13.6 &  &  2.7 &  2.9 &  5.2 &  5.0 \\
\hline \hline 
\end{tabular}

%% file: Fig_ari.tex
\begin{tikzpicture}[x=1pt,y=1pt]
\definecolor{fillColor}{RGB}{255,255,255}
\path[use as bounding box,fill=fillColor,fill opacity=0.00] (0,0) rectangle (441.02,264.61);
\begin{scope}
\path[clip] (  0.00,  0.00) rectangle (441.02,264.61);
\definecolor{drawColor}{RGB}{255,255,255}
\definecolor{fillColor}{RGB}{255,255,255}

\path[draw=drawColor,line width= 0.6pt,line join=round,line cap=round,fill=fillColor] (  0.00,  0.00) rectangle (441.02,264.61);
\end{scope}
\begin{scope}
\path[clip] ( 38.56,147.65) rectangle (178.39,242.54);
\definecolor{fillColor}{RGB}{255,255,255}

\path[fill=fillColor] ( 38.56,147.65) rectangle (178.39,242.54);
\definecolor{drawColor}{gray}{0.92}

\path[draw=drawColor,line width= 0.3pt,line join=round] ( 38.56,162.74) --
	(178.39,162.74);

\path[draw=drawColor,line width= 0.3pt,line join=round] ( 38.56,184.31) --
	(178.39,184.31);

\path[draw=drawColor,line width= 0.3pt,line join=round] ( 38.56,205.88) --
	(178.39,205.88);

\path[draw=drawColor,line width= 0.3pt,line join=round] ( 38.56,227.44) --
	(178.39,227.44);

\path[draw=drawColor,line width= 0.6pt,line join=round] ( 38.56,151.96) --
	(178.39,151.96);

\path[draw=drawColor,line width= 0.6pt,line join=round] ( 38.56,173.53) --
	(178.39,173.53);

\path[draw=drawColor,line width= 0.6pt,line join=round] ( 38.56,195.09) --
	(178.39,195.09);

\path[draw=drawColor,line width= 0.6pt,line join=round] ( 38.56,216.66) --
	(178.39,216.66);

\path[draw=drawColor,line width= 0.6pt,line join=round] ( 38.56,238.23) --
	(178.39,238.23);

\path[draw=drawColor,line width= 0.6pt,line join=round] ( 58.53,147.65) --
	( 58.53,242.54);

\path[draw=drawColor,line width= 0.6pt,line join=round] ( 91.83,147.65) --
	( 91.83,242.54);

\path[draw=drawColor,line width= 0.6pt,line join=round] (125.12,147.65) --
	(125.12,242.54);

\path[draw=drawColor,line width= 0.6pt,line join=round] (158.42,147.65) --
	(158.42,242.54);
\definecolor{drawColor}{gray}{0.20}

\path[draw=drawColor,line width= 0.6pt,line join=round] ( 50.21,205.22) -- ( 50.21,238.10);

\path[draw=drawColor,line width= 0.6pt,line join=round] ( 50.21,167.33) -- ( 50.21,152.04);
\definecolor{fillColor}{RGB}{213,62,79}

\path[draw=drawColor,line width= 0.6pt,fill=fillColor,fill opacity=0.60] ( 46.46,205.22) --
	( 46.46,167.33) --
	( 53.95,167.33) --
	( 53.95,205.22) --
	( 46.46,205.22) --
	cycle;

\path[draw=drawColor,line width= 1.1pt] ( 46.46,183.75) -- ( 53.95,183.75);

\path[draw=drawColor,line width= 0.6pt,line join=round] ( 58.53,178.44) -- ( 58.53,237.29);

\path[draw=drawColor,line width= 0.6pt,line join=round] ( 58.53,159.86) -- ( 58.53,151.97);
\definecolor{fillColor}{RGB}{254,224,139}

\path[draw=drawColor,line width= 0.6pt,fill=fillColor,fill opacity=0.60] ( 54.79,178.44) --
	( 54.79,159.86) --
	( 62.28,159.86) --
	( 62.28,178.44) --
	( 54.79,178.44) --
	cycle;

\path[draw=drawColor,line width= 1.1pt] ( 54.79,167.65) -- ( 62.28,167.65);

\path[draw=drawColor,line width= 0.6pt,line join=round] ( 66.86,165.79) -- ( 66.86,201.15);

\path[draw=drawColor,line width= 0.6pt,line join=round] ( 66.86,155.75) -- ( 66.86,151.96);
\definecolor{fillColor}{RGB}{50,136,189}

\path[draw=drawColor,line width= 0.6pt,fill=fillColor,fill opacity=0.60] ( 63.11,165.79) --
	( 63.11,155.75) --
	( 70.60,155.75) --
	( 70.60,165.79) --
	( 63.11,165.79) --
	cycle;

\path[draw=drawColor,line width= 1.1pt] ( 63.11,160.05) -- ( 70.60,160.05);

\path[draw=drawColor,line width= 0.6pt,line join=round] ( 83.50,207.90) -- ( 83.50,238.17);

\path[draw=drawColor,line width= 0.6pt,line join=round] ( 83.50,167.24) -- ( 83.50,151.99);
\definecolor{fillColor}{RGB}{213,62,79}

\path[draw=drawColor,line width= 0.6pt,fill=fillColor,fill opacity=0.60] ( 79.76,207.90) --
	( 79.76,167.24) --
	( 87.25,167.24) --
	( 87.25,207.90) --
	( 79.76,207.90) --
	cycle;

\path[draw=drawColor,line width= 1.1pt] ( 79.76,186.26) -- ( 87.25,186.26);

\path[draw=drawColor,line width= 0.6pt,line join=round] ( 91.83,180.36) -- ( 91.83,237.54);

\path[draw=drawColor,line width= 0.6pt,line join=round] ( 91.83,160.35) -- ( 91.83,151.99);
\definecolor{fillColor}{RGB}{254,224,139}

\path[draw=drawColor,line width= 0.6pt,fill=fillColor,fill opacity=0.60] ( 88.08,180.36) --
	( 88.08,160.35) --
	( 95.57,160.35) --
	( 95.57,180.36) --
	( 88.08,180.36) --
	cycle;

\path[draw=drawColor,line width= 1.1pt] ( 88.08,168.82) -- ( 95.57,168.82);

\path[draw=drawColor,line width= 0.6pt,line join=round] (100.15,165.94) -- (100.15,201.67);

\path[draw=drawColor,line width= 0.6pt,line join=round] (100.15,155.95) -- (100.15,151.98);
\definecolor{fillColor}{RGB}{50,136,189}

\path[draw=drawColor,line width= 0.6pt,fill=fillColor,fill opacity=0.60] ( 96.40,165.94) --
	( 96.40,155.95) --
	(103.90,155.95) --
	(103.90,165.94) --
	( 96.40,165.94) --
	cycle;

\path[draw=drawColor,line width= 1.1pt] ( 96.40,160.20) -- (103.90,160.20);

\path[draw=drawColor,line width= 0.6pt,line join=round] (116.80,206.37) -- (116.80,238.00);

\path[draw=drawColor,line width= 0.6pt,line join=round] (116.80,167.06) -- (116.80,151.98);
\definecolor{fillColor}{RGB}{213,62,79}

\path[draw=drawColor,line width= 0.6pt,fill=fillColor,fill opacity=0.60] (113.05,206.37) --
	(113.05,167.06) --
	(120.54,167.06) --
	(120.54,206.37) --
	(113.05,206.37) --
	cycle;

\path[draw=drawColor,line width= 1.1pt] (113.05,185.01) -- (120.54,185.01);

\path[draw=drawColor,line width= 0.6pt,line join=round] (125.12,179.18) -- (125.12,236.25);

\path[draw=drawColor,line width= 0.6pt,line join=round] (125.12,160.35) -- (125.12,152.01);
\definecolor{fillColor}{RGB}{254,224,139}

\path[draw=drawColor,line width= 0.6pt,fill=fillColor,fill opacity=0.60] (121.38,179.18) --
	(121.38,160.35) --
	(128.87,160.35) --
	(128.87,179.18) --
	(121.38,179.18) --
	cycle;

\path[draw=drawColor,line width= 1.1pt] (121.38,168.75) -- (128.87,168.75);

\path[draw=drawColor,line width= 0.6pt,line join=round] (133.44,166.09) -- (133.44,200.93);

\path[draw=drawColor,line width= 0.6pt,line join=round] (133.44,155.79) -- (133.44,151.97);
\definecolor{fillColor}{RGB}{50,136,189}

\path[draw=drawColor,line width= 0.6pt,fill=fillColor,fill opacity=0.60] (129.70,166.09) --
	(129.70,155.79) --
	(137.19,155.79) --
	(137.19,166.09) --
	(129.70,166.09) --
	cycle;

\path[draw=drawColor,line width= 1.1pt] (129.70,160.32) -- (137.19,160.32);

\path[draw=drawColor,line width= 0.6pt,line join=round] (150.09,206.33) -- (150.09,238.19);

\path[draw=drawColor,line width= 0.6pt,line join=round] (150.09,167.34) -- (150.09,152.01);
\definecolor{fillColor}{RGB}{213,62,79}

\path[draw=drawColor,line width= 0.6pt,fill=fillColor,fill opacity=0.60] (146.35,206.33) --
	(146.35,167.34) --
	(153.84,167.34) --
	(153.84,206.33) --
	(146.35,206.33) --
	cycle;

\path[draw=drawColor,line width= 1.1pt] (146.35,183.97) -- (153.84,183.97);

\path[draw=drawColor,line width= 0.6pt,line join=round] (158.42,179.97) -- (158.42,238.14);

\path[draw=drawColor,line width= 0.6pt,line join=round] (158.42,159.83) -- (158.42,152.02);
\definecolor{fillColor}{RGB}{254,224,139}

\path[draw=drawColor,line width= 0.6pt,fill=fillColor,fill opacity=0.60] (154.67,179.97) --
	(154.67,159.83) --
	(162.16,159.83) --
	(162.16,179.97) --
	(154.67,179.97) --
	cycle;

\path[draw=drawColor,line width= 1.1pt] (154.67,168.64) -- (162.16,168.64);

\path[draw=drawColor,line width= 0.6pt,line join=round] (166.74,166.48) -- (166.74,204.52);

\path[draw=drawColor,line width= 0.6pt,line join=round] (166.74,155.97) -- (166.74,151.97);
\definecolor{fillColor}{RGB}{50,136,189}

\path[draw=drawColor,line width= 0.6pt,fill=fillColor,fill opacity=0.60] (162.99,166.48) --
	(162.99,155.97) --
	(170.49,155.97) --
	(170.49,166.48) --
	(162.99,166.48) --
	cycle;

\path[draw=drawColor,line width= 1.1pt] (162.99,160.44) -- (170.49,160.44);

\path[draw=drawColor,line width= 0.6pt,line join=round,line cap=round] ( 38.56,147.65) rectangle (178.39,242.54);
\end{scope}
\begin{scope}
\path[clip] ( 38.56, 30.69) rectangle (178.39,125.58);
\definecolor{fillColor}{RGB}{255,255,255}

\path[fill=fillColor] ( 38.56, 30.69) rectangle (178.39,125.58);
\definecolor{drawColor}{gray}{0.92}

\path[draw=drawColor,line width= 0.3pt,line join=round] ( 38.56, 45.78) --
	(178.39, 45.78);

\path[draw=drawColor,line width= 0.3pt,line join=round] ( 38.56, 67.35) --
	(178.39, 67.35);

\path[draw=drawColor,line width= 0.3pt,line join=round] ( 38.56, 88.91) --
	(178.39, 88.91);

\path[draw=drawColor,line width= 0.3pt,line join=round] ( 38.56,110.48) --
	(178.39,110.48);

\path[draw=drawColor,line width= 0.6pt,line join=round] ( 38.56, 35.00) --
	(178.39, 35.00);

\path[draw=drawColor,line width= 0.6pt,line join=round] ( 38.56, 56.57) --
	(178.39, 56.57);

\path[draw=drawColor,line width= 0.6pt,line join=round] ( 38.56, 78.13) --
	(178.39, 78.13);

\path[draw=drawColor,line width= 0.6pt,line join=round] ( 38.56, 99.70) --
	(178.39, 99.70);

\path[draw=drawColor,line width= 0.6pt,line join=round] ( 38.56,121.26) --
	(178.39,121.26);

\path[draw=drawColor,line width= 0.6pt,line join=round] ( 58.53, 30.69) --
	( 58.53,125.58);

\path[draw=drawColor,line width= 0.6pt,line join=round] ( 91.83, 30.69) --
	( 91.83,125.58);

\path[draw=drawColor,line width= 0.6pt,line join=round] (125.12, 30.69) --
	(125.12,125.58);

\path[draw=drawColor,line width= 0.6pt,line join=round] (158.42, 30.69) --
	(158.42,125.58);
\definecolor{drawColor}{gray}{0.20}

\path[draw=drawColor,line width= 0.6pt,line join=round] ( 50.21, 89.11) -- ( 50.21,121.15);

\path[draw=drawColor,line width= 0.6pt,line join=round] ( 50.21, 49.17) -- ( 50.21, 35.08);
\definecolor{fillColor}{RGB}{213,62,79}

\path[draw=drawColor,line width= 0.6pt,fill=fillColor,fill opacity=0.60] ( 46.46, 89.11) --
	( 46.46, 49.17) --
	( 53.95, 49.17) --
	( 53.95, 89.11) --
	( 46.46, 89.11) --
	cycle;

\path[draw=drawColor,line width= 1.1pt] ( 46.46, 66.66) -- ( 53.95, 66.66);

\path[draw=drawColor,line width= 0.6pt,line join=round] ( 58.53, 55.85) -- ( 58.53, 76.90);

\path[draw=drawColor,line width= 0.6pt,line join=round] ( 58.53, 46.36) -- ( 58.53, 35.08);
\definecolor{fillColor}{RGB}{254,224,139}

\path[draw=drawColor,line width= 0.6pt,fill=fillColor,fill opacity=0.60] ( 54.79, 55.85) --
	( 54.79, 46.36) --
	( 62.28, 46.36) --
	( 62.28, 55.85) --
	( 54.79, 55.85) --
	cycle;

\path[draw=drawColor,line width= 1.1pt] ( 54.79, 50.78) -- ( 62.28, 50.78);

\path[draw=drawColor,line width= 0.6pt,line join=round] ( 66.86, 41.14) -- ( 66.86, 55.21);

\path[draw=drawColor,line width= 0.6pt,line join=round] ( 66.86, 36.70) -- ( 66.86, 35.00);
\definecolor{fillColor}{RGB}{50,136,189}

\path[draw=drawColor,line width= 0.6pt,fill=fillColor,fill opacity=0.60] ( 63.11, 41.14) --
	( 63.11, 36.70) --
	( 70.60, 36.70) --
	( 70.60, 41.14) --
	( 63.11, 41.14) --
	cycle;

\path[draw=drawColor,line width= 1.1pt] ( 63.11, 38.60) -- ( 70.60, 38.60);

\path[draw=drawColor,line width= 0.6pt,line join=round] ( 83.50, 89.69) -- ( 83.50,121.14);

\path[draw=drawColor,line width= 0.6pt,line join=round] ( 83.50, 49.72) -- ( 83.50, 35.09);
\definecolor{fillColor}{RGB}{213,62,79}

\path[draw=drawColor,line width= 0.6pt,fill=fillColor,fill opacity=0.60] ( 79.76, 89.69) --
	( 79.76, 49.72) --
	( 87.25, 49.72) --
	( 87.25, 89.69) --
	( 79.76, 89.69) --
	cycle;

\path[draw=drawColor,line width= 1.1pt] ( 79.76, 66.34) -- ( 87.25, 66.34);

\path[draw=drawColor,line width= 0.6pt,line join=round] ( 91.83, 56.83) -- ( 91.83, 84.03);

\path[draw=drawColor,line width= 0.6pt,line join=round] ( 91.83, 47.21) -- ( 91.83, 35.05);
\definecolor{fillColor}{RGB}{254,224,139}

\path[draw=drawColor,line width= 0.6pt,fill=fillColor,fill opacity=0.60] ( 88.08, 56.83) --
	( 88.08, 47.21) --
	( 95.57, 47.21) --
	( 95.57, 56.83) --
	( 88.08, 56.83) --
	cycle;

\path[draw=drawColor,line width= 1.1pt] ( 88.08, 51.82) -- ( 95.57, 51.82);

\path[draw=drawColor,line width= 0.6pt,line join=round] (100.15, 41.24) -- (100.15, 53.36);

\path[draw=drawColor,line width= 0.6pt,line join=round] (100.15, 36.77) -- (100.15, 35.00);
\definecolor{fillColor}{RGB}{50,136,189}

\path[draw=drawColor,line width= 0.6pt,fill=fillColor,fill opacity=0.60] ( 96.40, 41.24) --
	( 96.40, 36.77) --
	(103.90, 36.77) --
	(103.90, 41.24) --
	( 96.40, 41.24) --
	cycle;

\path[draw=drawColor,line width= 1.1pt] ( 96.40, 38.68) -- (103.90, 38.68);

\path[draw=drawColor,line width= 0.6pt,line join=round] (116.80, 90.11) -- (116.80,121.12);

\path[draw=drawColor,line width= 0.6pt,line join=round] (116.80, 49.47) -- (116.80, 35.09);
\definecolor{fillColor}{RGB}{213,62,79}

\path[draw=drawColor,line width= 0.6pt,fill=fillColor,fill opacity=0.60] (113.05, 90.11) --
	(113.05, 49.47) --
	(120.54, 49.47) --
	(120.54, 90.11) --
	(113.05, 90.11) --
	cycle;

\path[draw=drawColor,line width= 1.1pt] (113.05, 67.31) -- (120.54, 67.31);

\path[draw=drawColor,line width= 0.6pt,line join=round] (125.12, 55.98) -- (125.12, 80.89);

\path[draw=drawColor,line width= 0.6pt,line join=round] (125.12, 46.68) -- (125.12, 35.00);
\definecolor{fillColor}{RGB}{254,224,139}

\path[draw=drawColor,line width= 0.6pt,fill=fillColor,fill opacity=0.60] (121.38, 55.98) --
	(121.38, 46.68) --
	(128.87, 46.68) --
	(128.87, 55.98) --
	(121.38, 55.98) --
	cycle;

\path[draw=drawColor,line width= 1.1pt] (121.38, 51.38) -- (128.87, 51.38);

\path[draw=drawColor,line width= 0.6pt,line join=round] (133.44, 41.61) -- (133.44, 56.30);

\path[draw=drawColor,line width= 0.6pt,line join=round] (133.44, 36.86) -- (133.44, 35.00);
\definecolor{fillColor}{RGB}{50,136,189}

\path[draw=drawColor,line width= 0.6pt,fill=fillColor,fill opacity=0.60] (129.70, 41.61) --
	(129.70, 36.86) --
	(137.19, 36.86) --
	(137.19, 41.61) --
	(129.70, 41.61) --
	cycle;

\path[draw=drawColor,line width= 1.1pt] (129.70, 38.80) -- (137.19, 38.80);

\path[draw=drawColor,line width= 0.6pt,line join=round] (150.09, 88.25) -- (150.09,121.24);

\path[draw=drawColor,line width= 0.6pt,line join=round] (150.09, 49.43) -- (150.09, 35.03);
\definecolor{fillColor}{RGB}{213,62,79}

\path[draw=drawColor,line width= 0.6pt,fill=fillColor,fill opacity=0.60] (146.35, 88.25) --
	(146.35, 49.43) --
	(153.84, 49.43) --
	(153.84, 88.25) --
	(146.35, 88.25) --
	cycle;

\path[draw=drawColor,line width= 1.1pt] (146.35, 67.11) -- (153.84, 67.11);

\path[draw=drawColor,line width= 0.6pt,line join=round] (158.42, 56.96) -- (158.42, 78.29);

\path[draw=drawColor,line width= 0.6pt,line join=round] (158.42, 47.29) -- (158.42, 35.01);
\definecolor{fillColor}{RGB}{254,224,139}

\path[draw=drawColor,line width= 0.6pt,fill=fillColor,fill opacity=0.60] (154.67, 56.96) --
	(154.67, 47.29) --
	(162.16, 47.29) --
	(162.16, 56.96) --
	(154.67, 56.96) --
	cycle;

\path[draw=drawColor,line width= 1.1pt] (154.67, 52.12) -- (162.16, 52.12);

\path[draw=drawColor,line width= 0.6pt,line join=round] (166.74, 41.32) -- (166.74, 55.14);

\path[draw=drawColor,line width= 0.6pt,line join=round] (166.74, 36.73) -- (166.74, 35.00);
\definecolor{fillColor}{RGB}{50,136,189}

\path[draw=drawColor,line width= 0.6pt,fill=fillColor,fill opacity=0.60] (162.99, 41.32) --
	(162.99, 36.73) --
	(170.49, 36.73) --
	(170.49, 41.32) --
	(162.99, 41.32) --
	cycle;

\path[draw=drawColor,line width= 1.1pt] (162.99, 38.68) -- (170.49, 38.68);

\path[draw=drawColor,line width= 0.6pt,line join=round,line cap=round] ( 38.56, 30.69) rectangle (178.39,125.58);
\end{scope}
\begin{scope}
\path[clip] (183.89,147.65) rectangle (323.73,242.54);
\definecolor{fillColor}{RGB}{255,255,255}

\path[fill=fillColor] (183.89,147.65) rectangle (323.73,242.54);
\definecolor{drawColor}{gray}{0.92}

\path[draw=drawColor,line width= 0.3pt,line join=round] (183.89,162.74) --
	(323.73,162.74);

\path[draw=drawColor,line width= 0.3pt,line join=round] (183.89,184.31) --
	(323.73,184.31);

\path[draw=drawColor,line width= 0.3pt,line join=round] (183.89,205.88) --
	(323.73,205.88);

\path[draw=drawColor,line width= 0.3pt,line join=round] (183.89,227.44) --
	(323.73,227.44);

\path[draw=drawColor,line width= 0.6pt,line join=round] (183.89,151.96) --
	(323.73,151.96);

\path[draw=drawColor,line width= 0.6pt,line join=round] (183.89,173.53) --
	(323.73,173.53);

\path[draw=drawColor,line width= 0.6pt,line join=round] (183.89,195.09) --
	(323.73,195.09);

\path[draw=drawColor,line width= 0.6pt,line join=round] (183.89,216.66) --
	(323.73,216.66);

\path[draw=drawColor,line width= 0.6pt,line join=round] (183.89,238.23) --
	(323.73,238.23);

\path[draw=drawColor,line width= 0.6pt,line join=round] (203.87,147.65) --
	(203.87,242.54);

\path[draw=drawColor,line width= 0.6pt,line join=round] (237.16,147.65) --
	(237.16,242.54);

\path[draw=drawColor,line width= 0.6pt,line join=round] (270.46,147.65) --
	(270.46,242.54);

\path[draw=drawColor,line width= 0.6pt,line join=round] (303.75,147.65) --
	(303.75,242.54);
\definecolor{drawColor}{gray}{0.20}

\path[draw=drawColor,line width= 0.6pt,line join=round] (195.55,179.45) -- (195.55,238.22);

\path[draw=drawColor,line width= 0.6pt,line join=round] (195.55,159.51) -- (195.55,151.99);
\definecolor{fillColor}{RGB}{213,62,79}

\path[draw=drawColor,line width= 0.6pt,fill=fillColor,fill opacity=0.60] (191.80,179.45) --
	(191.80,159.51) --
	(199.29,159.51) --
	(199.29,179.45) --
	(191.80,179.45) --
	cycle;

\path[draw=drawColor,line width= 1.1pt] (191.80,168.18) -- (199.29,168.18);

\path[draw=drawColor,line width= 0.6pt,line join=round] (203.87,158.14) -- (203.87,171.59);

\path[draw=drawColor,line width= 0.6pt,line join=round] (203.87,153.72) -- (203.87,151.97);
\definecolor{fillColor}{RGB}{254,224,139}

\path[draw=drawColor,line width= 0.6pt,fill=fillColor,fill opacity=0.60] (200.12,158.14) --
	(200.12,153.72) --
	(207.62,153.72) --
	(207.62,158.14) --
	(200.12,158.14) --
	cycle;

\path[draw=drawColor,line width= 1.1pt] (200.12,155.55) -- (207.62,155.55);

\path[draw=drawColor,line width= 0.6pt,line join=round] (212.19,157.89) -- (212.19,172.38);

\path[draw=drawColor,line width= 0.6pt,line join=round] (212.19,153.61) -- (212.19,151.97);
\definecolor{fillColor}{RGB}{50,136,189}

\path[draw=drawColor,line width= 0.6pt,fill=fillColor,fill opacity=0.60] (208.45,157.89) --
	(208.45,153.61) --
	(215.94,153.61) --
	(215.94,157.89) --
	(208.45,157.89) --
	cycle;

\path[draw=drawColor,line width= 1.1pt] (208.45,155.34) -- (215.94,155.34);

\path[draw=drawColor,line width= 0.6pt,line join=round] (228.84,179.64) -- (228.84,238.12);

\path[draw=drawColor,line width= 0.6pt,line join=round] (228.84,159.52) -- (228.84,152.03);
\definecolor{fillColor}{RGB}{213,62,79}

\path[draw=drawColor,line width= 0.6pt,fill=fillColor,fill opacity=0.60] (225.09,179.64) --
	(225.09,159.52) --
	(232.59,159.52) --
	(232.59,179.64) --
	(225.09,179.64) --
	cycle;

\path[draw=drawColor,line width= 1.1pt] (225.09,167.65) -- (232.59,167.65);

\path[draw=drawColor,line width= 0.6pt,line join=round] (237.16,158.41) -- (237.16,173.51);

\path[draw=drawColor,line width= 0.6pt,line join=round] (237.16,153.79) -- (237.16,151.96);
\definecolor{fillColor}{RGB}{254,224,139}

\path[draw=drawColor,line width= 0.6pt,fill=fillColor,fill opacity=0.60] (233.42,158.41) --
	(233.42,153.79) --
	(240.91,153.79) --
	(240.91,158.41) --
	(233.42,158.41) --
	cycle;

\path[draw=drawColor,line width= 1.1pt] (233.42,155.74) -- (240.91,155.74);

\path[draw=drawColor,line width= 0.6pt,line join=round] (245.49,157.99) -- (245.49,171.33);

\path[draw=drawColor,line width= 0.6pt,line join=round] (245.49,153.62) -- (245.49,151.96);
\definecolor{fillColor}{RGB}{50,136,189}

\path[draw=drawColor,line width= 0.6pt,fill=fillColor,fill opacity=0.60] (241.74,157.99) --
	(241.74,153.62) --
	(249.23,153.62) --
	(249.23,157.99) --
	(241.74,157.99) --
	cycle;

\path[draw=drawColor,line width= 1.1pt] (241.74,155.45) -- (249.23,155.45);

\path[draw=drawColor,line width= 0.6pt,line join=round] (262.14,178.98) -- (262.14,237.33);

\path[draw=drawColor,line width= 0.6pt,line join=round] (262.14,159.31) -- (262.14,151.98);
\definecolor{fillColor}{RGB}{213,62,79}

\path[draw=drawColor,line width= 0.6pt,fill=fillColor,fill opacity=0.60] (258.39,178.98) --
	(258.39,159.31) --
	(265.88,159.31) --
	(265.88,178.98) --
	(258.39,178.98) --
	cycle;

\path[draw=drawColor,line width= 1.1pt] (258.39,167.31) -- (265.88,167.31);

\path[draw=drawColor,line width= 0.6pt,line join=round] (270.46,158.32) -- (270.46,171.77);

\path[draw=drawColor,line width= 0.6pt,line join=round] (270.46,153.65) -- (270.46,151.96);
\definecolor{fillColor}{RGB}{254,224,139}

\path[draw=drawColor,line width= 0.6pt,fill=fillColor,fill opacity=0.60] (266.71,158.32) --
	(266.71,153.65) --
	(274.20,153.65) --
	(274.20,158.32) --
	(266.71,158.32) --
	cycle;

\path[draw=drawColor,line width= 1.1pt] (266.71,155.68) -- (274.20,155.68);

\path[draw=drawColor,line width= 0.6pt,line join=round] (278.78,158.07) -- (278.78,171.05);

\path[draw=drawColor,line width= 0.6pt,line join=round] (278.78,153.60) -- (278.78,151.97);
\definecolor{fillColor}{RGB}{50,136,189}

\path[draw=drawColor,line width= 0.6pt,fill=fillColor,fill opacity=0.60] (275.04,158.07) --
	(275.04,153.60) --
	(282.53,153.60) --
	(282.53,158.07) --
	(275.04,158.07) --
	cycle;

\path[draw=drawColor,line width= 1.1pt] (275.04,155.46) -- (282.53,155.46);

\path[draw=drawColor,line width= 0.6pt,line join=round] (295.43,180.95) -- (295.43,237.35);

\path[draw=drawColor,line width= 0.6pt,line join=round] (295.43,159.25) -- (295.43,151.98);
\definecolor{fillColor}{RGB}{213,62,79}

\path[draw=drawColor,line width= 0.6pt,fill=fillColor,fill opacity=0.60] (291.68,180.95) --
	(291.68,159.25) --
	(299.18,159.25) --
	(299.18,180.95) --
	(291.68,180.95) --
	cycle;

\path[draw=drawColor,line width= 1.1pt] (291.68,167.91) -- (299.18,167.91);

\path[draw=drawColor,line width= 0.6pt,line join=round] (303.75,158.44) -- (303.75,173.48);

\path[draw=drawColor,line width= 0.6pt,line join=round] (303.75,153.85) -- (303.75,151.97);
\definecolor{fillColor}{RGB}{254,224,139}

\path[draw=drawColor,line width= 0.6pt,fill=fillColor,fill opacity=0.60] (300.01,158.44) --
	(300.01,153.85) --
	(307.50,153.85) --
	(307.50,158.44) --
	(300.01,158.44) --
	cycle;

\path[draw=drawColor,line width= 1.1pt] (300.01,155.81) -- (307.50,155.81);

\path[draw=drawColor,line width= 0.6pt,line join=round] (312.08,158.13) -- (312.08,172.69);

\path[draw=drawColor,line width= 0.6pt,line join=round] (312.08,153.63) -- (312.08,151.96);
\definecolor{fillColor}{RGB}{50,136,189}

\path[draw=drawColor,line width= 0.6pt,fill=fillColor,fill opacity=0.60] (308.33,158.13) --
	(308.33,153.63) --
	(315.82,153.63) --
	(315.82,158.13) --
	(308.33,158.13) --
	cycle;

\path[draw=drawColor,line width= 1.1pt] (308.33,155.46) -- (315.82,155.46);

\path[draw=drawColor,line width= 0.6pt,line join=round,line cap=round] (183.89,147.65) rectangle (323.73,242.54);
\end{scope}
\begin{scope}
\path[clip] (183.89, 30.69) rectangle (323.73,125.58);
\definecolor{fillColor}{RGB}{255,255,255}

\path[fill=fillColor] (183.89, 30.69) rectangle (323.73,125.58);
\definecolor{drawColor}{gray}{0.92}

\path[draw=drawColor,line width= 0.3pt,line join=round] (183.89, 45.78) --
	(323.73, 45.78);

\path[draw=drawColor,line width= 0.3pt,line join=round] (183.89, 67.35) --
	(323.73, 67.35);

\path[draw=drawColor,line width= 0.3pt,line join=round] (183.89, 88.91) --
	(323.73, 88.91);

\path[draw=drawColor,line width= 0.3pt,line join=round] (183.89,110.48) --
	(323.73,110.48);

\path[draw=drawColor,line width= 0.6pt,line join=round] (183.89, 35.00) --
	(323.73, 35.00);

\path[draw=drawColor,line width= 0.6pt,line join=round] (183.89, 56.57) --
	(323.73, 56.57);

\path[draw=drawColor,line width= 0.6pt,line join=round] (183.89, 78.13) --
	(323.73, 78.13);

\path[draw=drawColor,line width= 0.6pt,line join=round] (183.89, 99.70) --
	(323.73, 99.70);

\path[draw=drawColor,line width= 0.6pt,line join=round] (183.89,121.26) --
	(323.73,121.26);

\path[draw=drawColor,line width= 0.6pt,line join=round] (203.87, 30.69) --
	(203.87,125.58);

\path[draw=drawColor,line width= 0.6pt,line join=round] (237.16, 30.69) --
	(237.16,125.58);

\path[draw=drawColor,line width= 0.6pt,line join=round] (270.46, 30.69) --
	(270.46,125.58);

\path[draw=drawColor,line width= 0.6pt,line join=round] (303.75, 30.69) --
	(303.75,125.58);
\definecolor{drawColor}{gray}{0.20}

\path[draw=drawColor,line width= 0.6pt,line join=round] (195.55, 63.23) -- (195.55,121.10);

\path[draw=drawColor,line width= 0.6pt,line join=round] (195.55, 42.61) -- (195.55, 35.01);
\definecolor{fillColor}{RGB}{213,62,79}

\path[draw=drawColor,line width= 0.6pt,fill=fillColor,fill opacity=0.60] (191.80, 63.23) --
	(191.80, 42.61) --
	(199.29, 42.61) --
	(199.29, 63.23) --
	(191.80, 63.23) --
	cycle;

\path[draw=drawColor,line width= 1.1pt] (191.80, 51.21) -- (199.29, 51.21);

\path[draw=drawColor,line width= 0.6pt,line join=round] (203.87, 37.90) -- (203.87, 45.27);

\path[draw=drawColor,line width= 0.6pt,line join=round] (203.87, 35.82) -- (203.87, 35.00);
\definecolor{fillColor}{RGB}{254,224,139}

\path[draw=drawColor,line width= 0.6pt,fill=fillColor,fill opacity=0.60] (200.12, 37.90) --
	(200.12, 35.82) --
	(207.62, 35.82) --
	(207.62, 37.90) --
	(200.12, 37.90) --
	cycle;

\path[draw=drawColor,line width= 1.1pt] (200.12, 36.69) -- (207.62, 36.69);

\path[draw=drawColor,line width= 0.6pt,line join=round] (212.19, 37.69) -- (212.19, 43.84);

\path[draw=drawColor,line width= 0.6pt,line join=round] (212.19, 35.66) -- (212.19, 35.00);
\definecolor{fillColor}{RGB}{50,136,189}

\path[draw=drawColor,line width= 0.6pt,fill=fillColor,fill opacity=0.60] (208.45, 37.69) --
	(208.45, 35.66) --
	(215.94, 35.66) --
	(215.94, 37.69) --
	(208.45, 37.69) --
	cycle;

\path[draw=drawColor,line width= 1.1pt] (208.45, 36.53) -- (215.94, 36.53);

\path[draw=drawColor,line width= 0.6pt,line join=round] (228.84, 64.80) -- (228.84,118.34);

\path[draw=drawColor,line width= 0.6pt,line join=round] (228.84, 42.42) -- (228.84, 35.03);
\definecolor{fillColor}{RGB}{213,62,79}

\path[draw=drawColor,line width= 0.6pt,fill=fillColor,fill opacity=0.60] (225.09, 64.80) --
	(225.09, 42.42) --
	(232.59, 42.42) --
	(232.59, 64.80) --
	(225.09, 64.80) --
	cycle;

\path[draw=drawColor,line width= 1.1pt] (225.09, 50.94) -- (232.59, 50.94);

\path[draw=drawColor,line width= 0.6pt,line join=round] (237.16, 37.91) -- (237.16, 44.55);

\path[draw=drawColor,line width= 0.6pt,line join=round] (237.16, 35.82) -- (237.16, 35.00);
\definecolor{fillColor}{RGB}{254,224,139}

\path[draw=drawColor,line width= 0.6pt,fill=fillColor,fill opacity=0.60] (233.42, 37.91) --
	(233.42, 35.82) --
	(240.91, 35.82) --
	(240.91, 37.91) --
	(233.42, 37.91) --
	cycle;

\path[draw=drawColor,line width= 1.1pt] (233.42, 36.69) -- (240.91, 36.69);

\path[draw=drawColor,line width= 0.6pt,line join=round] (245.49, 37.65) -- (245.49, 44.00);

\path[draw=drawColor,line width= 0.6pt,line join=round] (245.49, 35.70) -- (245.49, 35.00);
\definecolor{fillColor}{RGB}{50,136,189}

\path[draw=drawColor,line width= 0.6pt,fill=fillColor,fill opacity=0.60] (241.74, 37.65) --
	(241.74, 35.70) --
	(249.23, 35.70) --
	(249.23, 37.65) --
	(241.74, 37.65) --
	cycle;

\path[draw=drawColor,line width= 1.1pt] (241.74, 36.55) -- (249.23, 36.55);

\path[draw=drawColor,line width= 0.6pt,line join=round] (262.14, 62.86) -- (262.14,120.96);

\path[draw=drawColor,line width= 0.6pt,line join=round] (262.14, 42.07) -- (262.14, 35.00);
\definecolor{fillColor}{RGB}{213,62,79}

\path[draw=drawColor,line width= 0.6pt,fill=fillColor,fill opacity=0.60] (258.39, 62.86) --
	(258.39, 42.07) --
	(265.88, 42.07) --
	(265.88, 62.86) --
	(258.39, 62.86) --
	cycle;

\path[draw=drawColor,line width= 1.1pt] (258.39, 51.02) -- (265.88, 51.02);

\path[draw=drawColor,line width= 0.6pt,line join=round] (270.46, 38.02) -- (270.46, 43.83);

\path[draw=drawColor,line width= 0.6pt,line join=round] (270.46, 35.84) -- (270.46, 35.00);
\definecolor{fillColor}{RGB}{254,224,139}

\path[draw=drawColor,line width= 0.6pt,fill=fillColor,fill opacity=0.60] (266.71, 38.02) --
	(266.71, 35.84) --
	(274.20, 35.84) --
	(274.20, 38.02) --
	(266.71, 38.02) --
	cycle;

\path[draw=drawColor,line width= 1.1pt] (266.71, 36.76) -- (274.20, 36.76);

\path[draw=drawColor,line width= 0.6pt,line join=round] (278.78, 37.86) -- (278.78, 43.28);

\path[draw=drawColor,line width= 0.6pt,line join=round] (278.78, 35.77) -- (278.78, 35.00);
\definecolor{fillColor}{RGB}{50,136,189}

\path[draw=drawColor,line width= 0.6pt,fill=fillColor,fill opacity=0.60] (275.04, 37.86) --
	(275.04, 35.77) --
	(282.53, 35.77) --
	(282.53, 37.86) --
	(275.04, 37.86) --
	cycle;

\path[draw=drawColor,line width= 1.1pt] (275.04, 36.64) -- (282.53, 36.64);

\path[draw=drawColor,line width= 0.6pt,line join=round] (295.43, 63.23) -- (295.43,120.85);

\path[draw=drawColor,line width= 0.6pt,line join=round] (295.43, 42.13) -- (295.43, 35.00);
\definecolor{fillColor}{RGB}{213,62,79}

\path[draw=drawColor,line width= 0.6pt,fill=fillColor,fill opacity=0.60] (291.68, 63.23) --
	(291.68, 42.13) --
	(299.18, 42.13) --
	(299.18, 63.23) --
	(291.68, 63.23) --
	cycle;

\path[draw=drawColor,line width= 1.1pt] (291.68, 50.31) -- (299.18, 50.31);

\path[draw=drawColor,line width= 0.6pt,line join=round] (303.75, 38.02) -- (303.75, 44.20);

\path[draw=drawColor,line width= 0.6pt,line join=round] (303.75, 35.84) -- (303.75, 35.00);
\definecolor{fillColor}{RGB}{254,224,139}

\path[draw=drawColor,line width= 0.6pt,fill=fillColor,fill opacity=0.60] (300.01, 38.02) --
	(300.01, 35.84) --
	(307.50, 35.84) --
	(307.50, 38.02) --
	(300.01, 38.02) --
	cycle;

\path[draw=drawColor,line width= 1.1pt] (300.01, 36.76) -- (307.50, 36.76);

\path[draw=drawColor,line width= 0.6pt,line join=round] (312.08, 37.81) -- (312.08, 43.63);

\path[draw=drawColor,line width= 0.6pt,line join=round] (312.08, 35.78) -- (312.08, 35.00);
\definecolor{fillColor}{RGB}{50,136,189}

\path[draw=drawColor,line width= 0.6pt,fill=fillColor,fill opacity=0.60] (308.33, 37.81) --
	(308.33, 35.78) --
	(315.82, 35.78) --
	(315.82, 37.81) --
	(308.33, 37.81) --
	cycle;

\path[draw=drawColor,line width= 1.1pt] (308.33, 36.66) -- (315.82, 36.66);

\path[draw=drawColor,line width= 0.6pt,line join=round,line cap=round] (183.89, 30.69) rectangle (323.73,125.58);
\end{scope}
\begin{scope}
\path[clip] ( 38.56,125.58) rectangle (178.39,142.15);
\definecolor{drawColor}{gray}{0.20}
\definecolor{fillColor}{gray}{0.85}

\path[draw=drawColor,line width= 0.6pt,line join=round,line cap=round,fill=fillColor] ( 38.56,125.58) rectangle (178.39,142.15);
\definecolor{drawColor}{gray}{0.10}

\node[text=drawColor,anchor=base,inner sep=0pt, outer sep=0pt, scale=  0.88] at (108.47,130.83) {$ N=100, \ T=100 $};
\end{scope}
\begin{scope}
\path[clip] (183.89,125.58) rectangle (323.73,142.15);
\definecolor{drawColor}{gray}{0.20}
\definecolor{fillColor}{gray}{0.85}

\path[draw=drawColor,line width= 0.6pt,line join=round,line cap=round,fill=fillColor] (183.89,125.58) rectangle (323.73,142.15);
\definecolor{drawColor}{gray}{0.10}

\node[text=drawColor,anchor=base,inner sep=0pt, outer sep=0pt, scale=  0.88] at (253.81,130.83) {$ N=100, \ T=500 $};
\end{scope}
\begin{scope}
\path[clip] ( 38.56,242.54) rectangle (178.39,259.11);
\definecolor{drawColor}{gray}{0.20}
\definecolor{fillColor}{gray}{0.85}

\path[draw=drawColor,line width= 0.6pt,line join=round,line cap=round,fill=fillColor] ( 38.56,242.54) rectangle (178.39,259.11);
\definecolor{drawColor}{gray}{0.10}

\node[text=drawColor,anchor=base,inner sep=0pt, outer sep=0pt, scale=  0.88] at (108.47,247.79) {$ N=20, \ T=100 $};
\end{scope}
\begin{scope}
\path[clip] (183.89,242.54) rectangle (323.73,259.11);
\definecolor{drawColor}{gray}{0.20}
\definecolor{fillColor}{gray}{0.85}

\path[draw=drawColor,line width= 0.6pt,line join=round,line cap=round,fill=fillColor] (183.89,242.54) rectangle (323.73,259.11);
\definecolor{drawColor}{gray}{0.10}

\node[text=drawColor,anchor=base,inner sep=0pt, outer sep=0pt, scale=  0.88] at (253.81,247.79) {$ N=20, \ T=500 $};
\end{scope}
\begin{scope}
\path[clip] (  0.00,  0.00) rectangle (441.02,264.61);
\definecolor{drawColor}{gray}{0.20}

\path[draw=drawColor,line width= 0.6pt,line join=round] ( 58.53, 27.94) --
	( 58.53, 30.69);

\path[draw=drawColor,line width= 0.6pt,line join=round] ( 91.83, 27.94) --
	( 91.83, 30.69);

\path[draw=drawColor,line width= 0.6pt,line join=round] (125.12, 27.94) --
	(125.12, 30.69);

\path[draw=drawColor,line width= 0.6pt,line join=round] (158.42, 27.94) --
	(158.42, 30.69);
\end{scope}
\begin{scope}
\path[clip] (  0.00,  0.00) rectangle (441.02,264.61);
\definecolor{drawColor}{gray}{0.30}

\node[text=drawColor,anchor=base,inner sep=0pt, outer sep=0pt, scale=  0.88] at ( 58.53, 19.68) {SN1};

\node[text=drawColor,anchor=base,inner sep=0pt, outer sep=0pt, scale=  0.88] at ( 91.83, 19.68) {SN2};

\node[text=drawColor,anchor=base,inner sep=0pt, outer sep=0pt, scale=  0.88] at (125.12, 19.68) {SN3};

\node[text=drawColor,anchor=base,inner sep=0pt, outer sep=0pt, scale=  0.88] at (158.42, 19.68) {SN4};
\end{scope}
\begin{scope}
\path[clip] (  0.00,  0.00) rectangle (441.02,264.61);
\definecolor{drawColor}{gray}{0.20}

\path[draw=drawColor,line width= 0.6pt,line join=round] (203.87, 27.94) --
	(203.87, 30.69);

\path[draw=drawColor,line width= 0.6pt,line join=round] (237.16, 27.94) --
	(237.16, 30.69);

\path[draw=drawColor,line width= 0.6pt,line join=round] (270.46, 27.94) --
	(270.46, 30.69);

\path[draw=drawColor,line width= 0.6pt,line join=round] (303.75, 27.94) --
	(303.75, 30.69);
\end{scope}
\begin{scope}
\path[clip] (  0.00,  0.00) rectangle (441.02,264.61);
\definecolor{drawColor}{gray}{0.30}

\node[text=drawColor,anchor=base,inner sep=0pt, outer sep=0pt, scale=  0.88] at (203.87, 19.68) {SN1};

\node[text=drawColor,anchor=base,inner sep=0pt, outer sep=0pt, scale=  0.88] at (237.16, 19.68) {SN2};

\node[text=drawColor,anchor=base,inner sep=0pt, outer sep=0pt, scale=  0.88] at (270.46, 19.68) {SN3};

\node[text=drawColor,anchor=base,inner sep=0pt, outer sep=0pt, scale=  0.88] at (303.75, 19.68) {SN4};
\end{scope}
\begin{scope}
\path[clip] (  0.00,  0.00) rectangle (441.02,264.61);
\definecolor{drawColor}{gray}{0.30}

\node[text=drawColor,anchor=base east,inner sep=0pt, outer sep=0pt, scale=  0.88] at ( 33.61,148.93) {0.00};

\node[text=drawColor,anchor=base east,inner sep=0pt, outer sep=0pt, scale=  0.88] at ( 33.61,170.50) {0.25};

\node[text=drawColor,anchor=base east,inner sep=0pt, outer sep=0pt, scale=  0.88] at ( 33.61,192.06) {0.50};

\node[text=drawColor,anchor=base east,inner sep=0pt, outer sep=0pt, scale=  0.88] at ( 33.61,213.63) {0.75};

\node[text=drawColor,anchor=base east,inner sep=0pt, outer sep=0pt, scale=  0.88] at ( 33.61,235.20) {1.00};
\end{scope}
\begin{scope}
\path[clip] (  0.00,  0.00) rectangle (441.02,264.61);
\definecolor{drawColor}{gray}{0.20}

\path[draw=drawColor,line width= 0.6pt,line join=round] ( 35.81,151.96) --
	( 38.56,151.96);

\path[draw=drawColor,line width= 0.6pt,line join=round] ( 35.81,173.53) --
	( 38.56,173.53);

\path[draw=drawColor,line width= 0.6pt,line join=round] ( 35.81,195.09) --
	( 38.56,195.09);

\path[draw=drawColor,line width= 0.6pt,line join=round] ( 35.81,216.66) --
	( 38.56,216.66);

\path[draw=drawColor,line width= 0.6pt,line join=round] ( 35.81,238.23) --
	( 38.56,238.23);
\end{scope}
\begin{scope}
\path[clip] (  0.00,  0.00) rectangle (441.02,264.61);
\definecolor{drawColor}{gray}{0.30}

\node[text=drawColor,anchor=base east,inner sep=0pt, outer sep=0pt, scale=  0.88] at ( 33.61, 31.97) {0.00};

\node[text=drawColor,anchor=base east,inner sep=0pt, outer sep=0pt, scale=  0.88] at ( 33.61, 53.53) {0.25};

\node[text=drawColor,anchor=base east,inner sep=0pt, outer sep=0pt, scale=  0.88] at ( 33.61, 75.10) {0.50};

\node[text=drawColor,anchor=base east,inner sep=0pt, outer sep=0pt, scale=  0.88] at ( 33.61, 96.67) {0.75};

\node[text=drawColor,anchor=base east,inner sep=0pt, outer sep=0pt, scale=  0.88] at ( 33.61,118.23) {1.00};
\end{scope}
\begin{scope}
\path[clip] (  0.00,  0.00) rectangle (441.02,264.61);
\definecolor{drawColor}{gray}{0.20}

\path[draw=drawColor,line width= 0.6pt,line join=round] ( 35.81, 35.00) --
	( 38.56, 35.00);

\path[draw=drawColor,line width= 0.6pt,line join=round] ( 35.81, 56.57) --
	( 38.56, 56.57);

\path[draw=drawColor,line width= 0.6pt,line join=round] ( 35.81, 78.13) --
	( 38.56, 78.13);

\path[draw=drawColor,line width= 0.6pt,line join=round] ( 35.81, 99.70) --
	( 38.56, 99.70);

\path[draw=drawColor,line width= 0.6pt,line join=round] ( 35.81,121.26) --
	( 38.56,121.26);
\end{scope}
\begin{scope}
\path[clip] (  0.00,  0.00) rectangle (441.02,264.61);
\definecolor{drawColor}{RGB}{0,0,0}

\node[text=drawColor,anchor=base,inner sep=0pt, outer sep=0pt, scale=  1.10] at (181.14,  7.64) {Scenario};
\end{scope}
\begin{scope}
\path[clip] (  0.00,  0.00) rectangle (441.02,264.61);
\definecolor{drawColor}{RGB}{0,0,0}

\node[text=drawColor,rotate= 90.00,anchor=base,inner sep=0pt, outer sep=0pt, scale=  1.10] at ( 13.08,136.61) {Amari distance};
\end{scope}
\begin{scope}
\path[clip] (  0.00,  0.00) rectangle (441.02,264.61);
\definecolor{fillColor}{RGB}{255,255,255}

\path[fill=fillColor] (334.73,101.82) rectangle (435.52,171.40);
\end{scope}
\begin{scope}
\path[clip] (  0.00,  0.00) rectangle (441.02,264.61);
\definecolor{drawColor}{RGB}{0,0,0}

\node[text=drawColor,anchor=base west,inner sep=0pt, outer sep=0pt, scale=  1.10] at (340.23,157.26) {Method};
\end{scope}
\begin{scope}
\path[clip] (  0.00,  0.00) rectangle (441.02,264.61);
\definecolor{fillColor}{RGB}{255,255,255}

\path[fill=fillColor] (340.23,136.23) rectangle (354.68,150.69);
\definecolor{drawColor}{gray}{0.20}

\path[draw=drawColor,line width= 0.6pt] (347.46,137.68) --
	(347.46,139.85);

\path[draw=drawColor,line width= 0.6pt] (347.46,147.07) --
	(347.46,149.24);
\definecolor{fillColor}{RGB}{213,62,79}

\path[draw=drawColor,line width= 0.6pt,fill=fillColor,fill opacity=0.60] (342.04,139.85) rectangle (352.88,147.07);

\path[draw=drawColor,line width= 0.6pt] (342.04,143.46) --
	(352.88,143.46);

\path[draw=drawColor,line width= 0.6pt] (347.46,137.68) --
	(347.46,137.68);

\path[draw=drawColor,line width= 0.6pt] (347.46,149.24) --
	(347.46,149.24);
\end{scope}
\begin{scope}
\path[clip] (  0.00,  0.00) rectangle (441.02,264.61);
\definecolor{fillColor}{RGB}{255,255,255}

\path[fill=fillColor] (340.23,121.78) rectangle (354.68,136.23);
\definecolor{drawColor}{gray}{0.20}

\path[draw=drawColor,line width= 0.6pt] (347.46,123.22) --
	(347.46,125.39);

\path[draw=drawColor,line width= 0.6pt] (347.46,132.62) --
	(347.46,134.79);
\definecolor{fillColor}{RGB}{254,224,139}

\path[draw=drawColor,line width= 0.6pt,fill=fillColor,fill opacity=0.60] (342.04,125.39) rectangle (352.88,132.62);

\path[draw=drawColor,line width= 0.6pt] (342.04,129.01) --
	(352.88,129.01);

\path[draw=drawColor,line width= 0.6pt] (347.46,123.22) --
	(347.46,123.22);

\path[draw=drawColor,line width= 0.6pt] (347.46,134.79) --
	(347.46,134.79);
\end{scope}
\begin{scope}
\path[clip] (  0.00,  0.00) rectangle (441.02,264.61);
\definecolor{fillColor}{RGB}{255,255,255}

\path[fill=fillColor] (340.23,107.32) rectangle (354.68,121.78);
\definecolor{drawColor}{gray}{0.20}

\path[draw=drawColor,line width= 0.6pt] (347.46,108.77) --
	(347.46,110.94);

\path[draw=drawColor,line width= 0.6pt] (347.46,118.16) --
	(347.46,120.33);
\definecolor{fillColor}{RGB}{50,136,189}

\path[draw=drawColor,line width= 0.6pt,fill=fillColor,fill opacity=0.60] (342.04,110.94) rectangle (352.88,118.16);

\path[draw=drawColor,line width= 0.6pt] (342.04,114.55) --
	(352.88,114.55);

\path[draw=drawColor,line width= 0.6pt] (347.46,108.77) --
	(347.46,108.77);

\path[draw=drawColor,line width= 0.6pt] (347.46,120.33) --
	(347.46,120.33);
\end{scope}
\begin{scope}
\path[clip] (  0.00,  0.00) rectangle (441.02,264.61);
\definecolor{drawColor}{RGB}{0,0,0}

\node[text=drawColor,anchor=base west,inner sep=0pt, outer sep=0pt, scale=  0.88] at (360.18,140.43) {Individual};
\end{scope}
\begin{scope}
\path[clip] (  0.00,  0.00) rectangle (441.02,264.61);
\definecolor{drawColor}{RGB}{0,0,0}

\node[text=drawColor,anchor=base west,inner sep=0pt, outer sep=0pt, scale=  0.88] at (360.18,125.97) {MG};
\end{scope}
\begin{scope}
\path[clip] (  0.00,  0.00) rectangle (441.02,264.61);
\definecolor{drawColor}{RGB}{0,0,0}

\node[text=drawColor,anchor=base west,inner sep=0pt, outer sep=0pt, scale=  0.88] at (360.18,111.52) {PMG};
\end{scope}
\end{tikzpicture}

%% file: Fig_ari2.tex
\begin{tikzpicture}[x=1pt,y=1pt]
\definecolor{fillColor}{RGB}{255,255,255}
\path[use as bounding box,fill=fillColor,fill opacity=0.00] (0,0) rectangle (441.02,264.61);
\begin{scope}
\path[clip] (  0.00,  0.00) rectangle (441.02,264.61);
\definecolor{drawColor}{RGB}{255,255,255}
\definecolor{fillColor}{RGB}{255,255,255}

\path[draw=drawColor,line width= 0.6pt,line join=round,line cap=round,fill=fillColor] (  0.00,  0.00) rectangle (441.02,264.61);
\end{scope}
\begin{scope}
\path[clip] ( 38.56,147.65) rectangle (186.79,242.54);
\definecolor{fillColor}{RGB}{255,255,255}

\path[fill=fillColor] ( 38.56,147.65) rectangle (186.79,242.54);
\definecolor{drawColor}{gray}{0.92}

\path[draw=drawColor,line width= 0.3pt,line join=round] ( 38.56,162.75) --
	(186.79,162.75);

\path[draw=drawColor,line width= 0.3pt,line join=round] ( 38.56,184.31) --
	(186.79,184.31);

\path[draw=drawColor,line width= 0.3pt,line join=round] ( 38.56,205.88) --
	(186.79,205.88);

\path[draw=drawColor,line width= 0.3pt,line join=round] ( 38.56,227.45) --
	(186.79,227.45);

\path[draw=drawColor,line width= 0.6pt,line join=round] ( 38.56,151.96) --
	(186.79,151.96);

\path[draw=drawColor,line width= 0.6pt,line join=round] ( 38.56,173.53) --
	(186.79,173.53);

\path[draw=drawColor,line width= 0.6pt,line join=round] ( 38.56,195.10) --
	(186.79,195.10);

\path[draw=drawColor,line width= 0.6pt,line join=round] ( 38.56,216.67) --
	(186.79,216.67);

\path[draw=drawColor,line width= 0.6pt,line join=round] ( 38.56,238.24) --
	(186.79,238.24);

\path[draw=drawColor,line width= 0.6pt,line join=round] ( 59.73,147.65) --
	( 59.73,242.54);

\path[draw=drawColor,line width= 0.6pt,line join=round] ( 95.03,147.65) --
	( 95.03,242.54);

\path[draw=drawColor,line width= 0.6pt,line join=round] (130.32,147.65) --
	(130.32,242.54);

\path[draw=drawColor,line width= 0.6pt,line join=round] (165.62,147.65) --
	(165.62,242.54);
\definecolor{drawColor}{gray}{0.20}

\path[draw=drawColor,line width= 0.6pt,line join=round] ( 49.81,204.13) -- ( 49.81,238.23);

\path[draw=drawColor,line width= 0.6pt,line join=round] ( 49.81,166.56) -- ( 49.81,151.98);
\definecolor{fillColor}{RGB}{213,62,79}

\path[draw=drawColor,line width= 0.6pt,fill=fillColor,fill opacity=0.60] ( 46.83,204.13) --
	( 46.83,166.56) --
	( 52.78,166.56) --
	( 52.78,204.13) --
	( 46.83,204.13) --
	cycle;

\path[draw=drawColor,line width= 1.1pt] ( 46.83,183.44) -- ( 52.78,183.44);

\path[draw=drawColor,line width= 0.6pt,line join=round] ( 56.42,180.48) -- ( 56.42,237.90);

\path[draw=drawColor,line width= 0.6pt,line join=round] ( 56.42,160.91) -- ( 56.42,151.97);
\definecolor{fillColor}{RGB}{254,224,139}

\path[draw=drawColor,line width= 0.6pt,fill=fillColor,fill opacity=0.60] ( 53.45,180.48) --
	( 53.45,160.91) --
	( 59.40,160.91) --
	( 59.40,180.48) --
	( 53.45,180.48) --
	cycle;

\path[draw=drawColor,line width= 1.1pt] ( 53.45,169.73) -- ( 59.40,169.73);

\path[draw=drawColor,line width= 0.6pt,line join=round] ( 63.04,166.39) -- ( 63.04,212.07);

\path[draw=drawColor,line width= 0.6pt,line join=round] ( 63.04,155.81) -- ( 63.04,151.96);
\definecolor{fillColor}{RGB}{50,136,189}

\path[draw=drawColor,line width= 0.6pt,fill=fillColor,fill opacity=0.60] ( 60.06,166.39) --
	( 60.06,155.81) --
	( 66.02,155.81) --
	( 66.02,166.39) --
	( 60.06,166.39) --
	cycle;

\path[draw=drawColor,line width= 1.1pt] ( 60.06,160.54) -- ( 66.02,160.54);

\path[draw=drawColor,line width= 0.6pt,line join=round] ( 69.66,162.38) -- ( 69.66,187.57);

\path[draw=drawColor,line width= 0.6pt,line join=round] ( 69.66,154.73) -- ( 69.66,151.96);
\definecolor{fillColor}{RGB}{0,0,139}

\path[draw=drawColor,line width= 0.6pt,fill=fillColor,fill opacity=0.60] ( 66.68,162.38) --
	( 66.68,154.73) --
	( 72.64,154.73) --
	( 72.64,162.38) --
	( 66.68,162.38) --
	cycle;

\path[draw=drawColor,line width= 1.1pt] ( 66.68,158.15) -- ( 72.64,158.15);

\path[draw=drawColor,line width= 0.6pt,line join=round] ( 88.41,187.62) -- ( 88.41,237.72);

\path[draw=drawColor,line width= 0.6pt,line join=round] ( 88.41,163.54) -- ( 88.41,151.97);
\definecolor{fillColor}{RGB}{254,224,139}

\path[draw=drawColor,line width= 0.6pt,fill=fillColor,fill opacity=0.60] ( 85.43,187.62) --
	( 85.43,163.54) --
	( 91.39,163.54) --
	( 91.39,187.62) --
	( 85.43,187.62) --
	cycle;

\path[draw=drawColor,line width= 1.1pt] ( 85.43,174.83) -- ( 91.39,174.83);

\path[draw=drawColor,line width= 0.6pt,line join=round] ( 95.03,166.89) -- ( 95.03,205.92);

\path[draw=drawColor,line width= 0.6pt,line join=round] ( 95.03,156.14) -- ( 95.03,151.97);
\definecolor{fillColor}{RGB}{50,136,189}

\path[draw=drawColor,line width= 0.6pt,fill=fillColor,fill opacity=0.60] ( 92.05,166.89) --
	( 92.05,156.14) --
	( 98.00,156.14) --
	( 98.00,166.89) --
	( 92.05,166.89) --
	cycle;

\path[draw=drawColor,line width= 1.1pt] ( 92.05,160.54) -- ( 98.00,160.54);

\path[draw=drawColor,line width= 0.6pt,line join=round] (101.64,162.90) -- (101.64,189.41);

\path[draw=drawColor,line width= 0.6pt,line join=round] (101.64,154.62) -- (101.64,151.98);
\definecolor{fillColor}{RGB}{0,0,139}

\path[draw=drawColor,line width= 0.6pt,fill=fillColor,fill opacity=0.60] ( 98.67,162.90) --
	( 98.67,154.62) --
	(104.62,154.62) --
	(104.62,162.90) --
	( 98.67,162.90) --
	cycle;

\path[draw=drawColor,line width= 1.1pt] ( 98.67,157.97) -- (104.62,157.97);

\path[draw=drawColor,line width= 0.6pt,line join=round] (123.70,179.89) -- (123.70,236.08);

\path[draw=drawColor,line width= 0.6pt,line join=round] (123.70,159.33) -- (123.70,151.97);
\definecolor{fillColor}{RGB}{254,224,139}

\path[draw=drawColor,line width= 0.6pt,fill=fillColor,fill opacity=0.60] (120.73,179.89) --
	(120.73,159.33) --
	(126.68,159.33) --
	(126.68,179.89) --
	(120.73,179.89) --
	cycle;

\path[draw=drawColor,line width= 1.1pt] (120.73,168.12) -- (126.68,168.12);

\path[draw=drawColor,line width= 0.6pt,line join=round] (130.32,166.41) -- (130.32,195.73);

\path[draw=drawColor,line width= 0.6pt,line join=round] (130.32,156.08) -- (130.32,151.97);
\definecolor{fillColor}{RGB}{50,136,189}

\path[draw=drawColor,line width= 0.6pt,fill=fillColor,fill opacity=0.60] (127.34,166.41) --
	(127.34,156.08) --
	(133.30,156.08) --
	(133.30,166.41) --
	(127.34,166.41) --
	cycle;

\path[draw=drawColor,line width= 1.1pt] (127.34,160.33) -- (133.30,160.33);

\path[draw=drawColor,line width= 0.6pt,line join=round] (136.94,162.83) -- (136.94,186.67);

\path[draw=drawColor,line width= 0.6pt,line join=round] (136.94,154.74) -- (136.94,151.96);
\definecolor{fillColor}{RGB}{0,0,139}

\path[draw=drawColor,line width= 0.6pt,fill=fillColor,fill opacity=0.60] (133.96,162.83) --
	(133.96,154.74) --
	(139.92,154.74) --
	(139.92,162.83) --
	(133.96,162.83) --
	cycle;

\path[draw=drawColor,line width= 1.1pt] (133.96,158.08) -- (139.92,158.08);

\path[draw=drawColor,line width= 0.6pt,line join=round] (159.00,180.33) -- (159.00,238.00);

\path[draw=drawColor,line width= 0.6pt,line join=round] (159.00,159.63) -- (159.00,151.96);
\definecolor{fillColor}{RGB}{254,224,139}

\path[draw=drawColor,line width= 0.6pt,fill=fillColor,fill opacity=0.60] (156.02,180.33) --
	(156.02,159.63) --
	(161.98,159.63) --
	(161.98,180.33) --
	(156.02,180.33) --
	cycle;

\path[draw=drawColor,line width= 1.1pt] (156.02,167.96) -- (161.98,167.96);

\path[draw=drawColor,line width= 0.6pt,line join=round] (165.62,166.24) -- (165.62,197.94);

\path[draw=drawColor,line width= 0.6pt,line join=round] (165.62,155.87) -- (165.62,151.97);
\definecolor{fillColor}{RGB}{50,136,189}

\path[draw=drawColor,line width= 0.6pt,fill=fillColor,fill opacity=0.60] (162.64,166.24) --
	(162.64,155.87) --
	(168.59,155.87) --
	(168.59,166.24) --
	(162.64,166.24) --
	cycle;

\path[draw=drawColor,line width= 1.1pt] (162.64,160.28) -- (168.59,160.28);

\path[draw=drawColor,line width= 0.6pt,line join=round] (172.23,162.64) -- (172.23,185.70);

\path[draw=drawColor,line width= 0.6pt,line join=round] (172.23,155.04) -- (172.23,151.97);
\definecolor{fillColor}{RGB}{0,0,139}

\path[draw=drawColor,line width= 0.6pt,fill=fillColor,fill opacity=0.60] (169.26,162.64) --
	(169.26,155.04) --
	(175.21,155.04) --
	(175.21,162.64) --
	(169.26,162.64) --
	cycle;

\path[draw=drawColor,line width= 1.1pt] (169.26,158.39) -- (175.21,158.39);

\path[draw=drawColor,line width= 0.6pt,line join=round,line cap=round] ( 38.56,147.65) rectangle (186.79,242.54);
\end{scope}
\begin{scope}
\path[clip] ( 38.56, 30.69) rectangle (186.79,125.58);
\definecolor{fillColor}{RGB}{255,255,255}

\path[fill=fillColor] ( 38.56, 30.69) rectangle (186.79,125.58);
\definecolor{drawColor}{gray}{0.92}

\path[draw=drawColor,line width= 0.3pt,line join=round] ( 38.56, 45.78) --
	(186.79, 45.78);

\path[draw=drawColor,line width= 0.3pt,line join=round] ( 38.56, 67.35) --
	(186.79, 67.35);

\path[draw=drawColor,line width= 0.3pt,line join=round] ( 38.56, 88.92) --
	(186.79, 88.92);

\path[draw=drawColor,line width= 0.3pt,line join=round] ( 38.56,110.49) --
	(186.79,110.49);

\path[draw=drawColor,line width= 0.6pt,line join=round] ( 38.56, 35.00) --
	(186.79, 35.00);

\path[draw=drawColor,line width= 0.6pt,line join=round] ( 38.56, 56.57) --
	(186.79, 56.57);

\path[draw=drawColor,line width= 0.6pt,line join=round] ( 38.56, 78.14) --
	(186.79, 78.14);

\path[draw=drawColor,line width= 0.6pt,line join=round] ( 38.56, 99.71) --
	(186.79, 99.71);

\path[draw=drawColor,line width= 0.6pt,line join=round] ( 38.56,121.27) --
	(186.79,121.27);

\path[draw=drawColor,line width= 0.6pt,line join=round] ( 59.73, 30.69) --
	( 59.73,125.58);

\path[draw=drawColor,line width= 0.6pt,line join=round] ( 95.03, 30.69) --
	( 95.03,125.58);

\path[draw=drawColor,line width= 0.6pt,line join=round] (130.32, 30.69) --
	(130.32,125.58);

\path[draw=drawColor,line width= 0.6pt,line join=round] (165.62, 30.69) --
	(165.62,125.58);
\definecolor{drawColor}{gray}{0.20}

\path[draw=drawColor,line width= 0.6pt,line join=round] ( 49.81, 78.92) -- ( 49.81, 98.11);

\path[draw=drawColor,line width= 0.6pt,line join=round] ( 49.81, 58.37) -- ( 49.81, 37.23);
\definecolor{fillColor}{RGB}{213,62,79}

\path[draw=drawColor,line width= 0.6pt,fill=fillColor,fill opacity=0.60] ( 46.83, 78.92) --
	( 46.83, 58.37) --
	( 52.78, 58.37) --
	( 52.78, 78.92) --
	( 46.83, 78.92) --
	cycle;

\path[draw=drawColor,line width= 1.1pt] ( 46.83, 68.88) -- ( 52.78, 68.88);

\path[draw=drawColor,line width= 0.6pt,line join=round] ( 56.42, 49.09) -- ( 56.42, 80.54);

\path[draw=drawColor,line width= 0.6pt,line join=round] ( 56.42, 42.31) -- ( 56.42, 35.81);
\definecolor{fillColor}{RGB}{254,224,139}

\path[draw=drawColor,line width= 0.6pt,fill=fillColor,fill opacity=0.60] ( 53.45, 49.09) --
	( 53.45, 42.31) --
	( 59.40, 42.31) --
	( 59.40, 49.09) --
	( 53.45, 49.09) --
	cycle;

\path[draw=drawColor,line width= 1.1pt] ( 53.45, 45.29) -- ( 59.40, 45.29);

\path[draw=drawColor,line width= 0.6pt,line join=round] ( 63.04, 46.18) -- ( 63.04, 64.04);

\path[draw=drawColor,line width= 0.6pt,line join=round] ( 63.04, 40.83) -- ( 63.04, 35.81);
\definecolor{fillColor}{RGB}{50,136,189}

\path[draw=drawColor,line width= 0.6pt,fill=fillColor,fill opacity=0.60] ( 60.06, 46.18) --
	( 60.06, 40.83) --
	( 66.02, 40.83) --
	( 66.02, 46.18) --
	( 60.06, 46.18) --
	cycle;

\path[draw=drawColor,line width= 1.1pt] ( 60.06, 43.30) -- ( 66.02, 43.30);

\path[draw=drawColor,line width= 0.6pt,line join=round] ( 69.66, 44.01) -- ( 69.66, 58.93);

\path[draw=drawColor,line width= 0.6pt,line join=round] ( 69.66, 39.68) -- ( 69.66, 35.08);
\definecolor{fillColor}{RGB}{0,0,139}

\path[draw=drawColor,line width= 0.6pt,fill=fillColor,fill opacity=0.60] ( 66.68, 44.01) --
	( 66.68, 39.68) --
	( 72.64, 39.68) --
	( 72.64, 44.01) --
	( 66.68, 44.01) --
	cycle;

\path[draw=drawColor,line width= 1.1pt] ( 66.68, 41.57) -- ( 72.64, 41.57);

\path[draw=drawColor,line width= 0.6pt,line join=round] ( 88.41, 49.53) -- ( 88.41, 73.96);

\path[draw=drawColor,line width= 0.6pt,line join=round] ( 88.41, 42.55) -- ( 88.41, 35.56);
\definecolor{fillColor}{RGB}{254,224,139}

\path[draw=drawColor,line width= 0.6pt,fill=fillColor,fill opacity=0.60] ( 85.43, 49.53) --
	( 85.43, 42.55) --
	( 91.39, 42.55) --
	( 91.39, 49.53) --
	( 85.43, 49.53) --
	cycle;

\path[draw=drawColor,line width= 1.1pt] ( 85.43, 45.72) -- ( 91.39, 45.72);

\path[draw=drawColor,line width= 0.6pt,line join=round] ( 95.03, 46.50) -- ( 95.03, 63.82);

\path[draw=drawColor,line width= 0.6pt,line join=round] ( 95.03, 41.07) -- ( 95.03, 35.59);
\definecolor{fillColor}{RGB}{50,136,189}

\path[draw=drawColor,line width= 0.6pt,fill=fillColor,fill opacity=0.60] ( 92.05, 46.50) --
	( 92.05, 41.07) --
	( 98.00, 41.07) --
	( 98.00, 46.50) --
	( 92.05, 46.50) --
	cycle;

\path[draw=drawColor,line width= 1.1pt] ( 92.05, 43.47) -- ( 98.00, 43.47);

\path[draw=drawColor,line width= 0.6pt,line join=round] (101.64, 44.15) -- (101.64, 54.68);

\path[draw=drawColor,line width= 0.6pt,line join=round] (101.64, 39.85) -- (101.64, 35.65);
\definecolor{fillColor}{RGB}{0,0,139}

\path[draw=drawColor,line width= 0.6pt,fill=fillColor,fill opacity=0.60] ( 98.67, 44.15) --
	( 98.67, 39.85) --
	(104.62, 39.85) --
	(104.62, 44.15) --
	( 98.67, 44.15) --
	cycle;

\path[draw=drawColor,line width= 1.1pt] ( 98.67, 41.76) -- (104.62, 41.76);

\path[draw=drawColor,line width= 0.6pt,line join=round] (123.70, 51.17) -- (123.70, 84.79);

\path[draw=drawColor,line width= 0.6pt,line join=round] (123.70, 43.43) -- (123.70, 36.22);
\definecolor{fillColor}{RGB}{254,224,139}

\path[draw=drawColor,line width= 0.6pt,fill=fillColor,fill opacity=0.60] (120.73, 51.17) --
	(120.73, 43.43) --
	(126.68, 43.43) --
	(126.68, 51.17) --
	(120.73, 51.17) --
	cycle;

\path[draw=drawColor,line width= 1.1pt] (120.73, 46.63) -- (126.68, 46.63);

\path[draw=drawColor,line width= 0.6pt,line join=round] (130.32, 46.30) -- (130.32, 61.72);

\path[draw=drawColor,line width= 0.6pt,line join=round] (130.32, 40.93) -- (130.32, 35.66);
\definecolor{fillColor}{RGB}{50,136,189}

\path[draw=drawColor,line width= 0.6pt,fill=fillColor,fill opacity=0.60] (127.34, 46.30) --
	(127.34, 40.93) --
	(133.30, 40.93) --
	(133.30, 46.30) --
	(127.34, 46.30) --
	cycle;

\path[draw=drawColor,line width= 1.1pt] (127.34, 43.41) -- (133.30, 43.41);

\path[draw=drawColor,line width= 0.6pt,line join=round] (136.94, 43.94) -- (136.94, 56.10);

\path[draw=drawColor,line width= 0.6pt,line join=round] (136.94, 39.61) -- (136.94, 35.65);
\definecolor{fillColor}{RGB}{0,0,139}

\path[draw=drawColor,line width= 0.6pt,fill=fillColor,fill opacity=0.60] (133.96, 43.94) --
	(133.96, 39.61) --
	(139.92, 39.61) --
	(139.92, 43.94) --
	(133.96, 43.94) --
	cycle;

\path[draw=drawColor,line width= 1.1pt] (133.96, 41.62) -- (139.92, 41.62);

\path[draw=drawColor,line width= 0.6pt,line join=round] (159.00, 56.14) -- (159.00, 85.39);

\path[draw=drawColor,line width= 0.6pt,line join=round] (159.00, 45.65) -- (159.00, 36.80);
\definecolor{fillColor}{RGB}{254,224,139}

\path[draw=drawColor,line width= 0.6pt,fill=fillColor,fill opacity=0.60] (156.02, 56.14) --
	(156.02, 45.65) --
	(161.98, 45.65) --
	(161.98, 56.14) --
	(156.02, 56.14) --
	cycle;

\path[draw=drawColor,line width= 1.1pt] (156.02, 50.50) -- (161.98, 50.50);

\path[draw=drawColor,line width= 0.6pt,line join=round] (165.62, 46.50) -- (165.62, 63.02);

\path[draw=drawColor,line width= 0.6pt,line join=round] (165.62, 41.04) -- (165.62, 35.32);
\definecolor{fillColor}{RGB}{50,136,189}

\path[draw=drawColor,line width= 0.6pt,fill=fillColor,fill opacity=0.60] (162.64, 46.50) --
	(162.64, 41.04) --
	(168.59, 41.04) --
	(168.59, 46.50) --
	(162.64, 46.50) --
	cycle;

\path[draw=drawColor,line width= 1.1pt] (162.64, 43.54) -- (168.59, 43.54);

\path[draw=drawColor,line width= 0.6pt,line join=round] (172.23, 43.79) -- (172.23, 56.34);

\path[draw=drawColor,line width= 0.6pt,line join=round] (172.23, 39.65) -- (172.23, 35.41);
\definecolor{fillColor}{RGB}{0,0,139}

\path[draw=drawColor,line width= 0.6pt,fill=fillColor,fill opacity=0.60] (169.26, 43.79) --
	(169.26, 39.65) --
	(175.21, 39.65) --
	(175.21, 43.79) --
	(169.26, 43.79) --
	cycle;

\path[draw=drawColor,line width= 1.1pt] (169.26, 41.53) -- (175.21, 41.53);

\path[draw=drawColor,line width= 0.6pt,line join=round,line cap=round] ( 38.56, 30.69) rectangle (186.79,125.58);
\end{scope}
\begin{scope}
\path[clip] (192.29,147.65) rectangle (340.53,242.54);
\definecolor{fillColor}{RGB}{255,255,255}

\path[fill=fillColor] (192.29,147.65) rectangle (340.53,242.54);
\definecolor{drawColor}{gray}{0.92}

\path[draw=drawColor,line width= 0.3pt,line join=round] (192.29,162.75) --
	(340.53,162.75);

\path[draw=drawColor,line width= 0.3pt,line join=round] (192.29,184.31) --
	(340.53,184.31);

\path[draw=drawColor,line width= 0.3pt,line join=round] (192.29,205.88) --
	(340.53,205.88);

\path[draw=drawColor,line width= 0.3pt,line join=round] (192.29,227.45) --
	(340.53,227.45);

\path[draw=drawColor,line width= 0.6pt,line join=round] (192.29,151.96) --
	(340.53,151.96);

\path[draw=drawColor,line width= 0.6pt,line join=round] (192.29,173.53) --
	(340.53,173.53);

\path[draw=drawColor,line width= 0.6pt,line join=round] (192.29,195.10) --
	(340.53,195.10);

\path[draw=drawColor,line width= 0.6pt,line join=round] (192.29,216.67) --
	(340.53,216.67);

\path[draw=drawColor,line width= 0.6pt,line join=round] (192.29,238.24) --
	(340.53,238.24);

\path[draw=drawColor,line width= 0.6pt,line join=round] (213.47,147.65) --
	(213.47,242.54);

\path[draw=drawColor,line width= 0.6pt,line join=round] (248.76,147.65) --
	(248.76,242.54);

\path[draw=drawColor,line width= 0.6pt,line join=round] (284.06,147.65) --
	(284.06,242.54);

\path[draw=drawColor,line width= 0.6pt,line join=round] (319.35,147.65) --
	(319.35,242.54);
\definecolor{drawColor}{gray}{0.20}

\path[draw=drawColor,line width= 0.6pt,line join=round] (203.54,174.01) -- (203.54,236.51);

\path[draw=drawColor,line width= 0.6pt,line join=round] (203.54,157.80) -- (203.54,151.99);
\definecolor{fillColor}{RGB}{213,62,79}

\path[draw=drawColor,line width= 0.6pt,fill=fillColor,fill opacity=0.60] (200.57,174.01) --
	(200.57,157.80) --
	(206.52,157.80) --
	(206.52,174.01) --
	(200.57,174.01) --
	cycle;

\path[draw=drawColor,line width= 1.1pt] (200.57,164.19) -- (206.52,164.19);

\path[draw=drawColor,line width= 0.6pt,line join=round] (210.16,157.67) -- (210.16,169.55);

\path[draw=drawColor,line width= 0.6pt,line join=round] (210.16,153.56) -- (210.16,151.96);
\definecolor{fillColor}{RGB}{254,224,139}

\path[draw=drawColor,line width= 0.6pt,fill=fillColor,fill opacity=0.60] (207.18,157.67) --
	(207.18,153.56) --
	(213.14,153.56) --
	(213.14,157.67) --
	(207.18,157.67) --
	cycle;

\path[draw=drawColor,line width= 1.1pt] (207.18,155.35) -- (213.14,155.35);

\path[draw=drawColor,line width= 0.6pt,line join=round] (216.78,157.53) -- (216.78,169.54);

\path[draw=drawColor,line width= 0.6pt,line join=round] (216.78,153.51) -- (216.78,151.96);
\definecolor{fillColor}{RGB}{50,136,189}

\path[draw=drawColor,line width= 0.6pt,fill=fillColor,fill opacity=0.60] (213.80,157.53) --
	(213.80,153.51) --
	(219.76,153.51) --
	(219.76,157.53) --
	(213.80,157.53) --
	cycle;

\path[draw=drawColor,line width= 1.1pt] (213.80,155.32) -- (219.76,155.32);

\path[draw=drawColor,line width= 0.6pt,line join=round] (223.40,156.65) -- (223.40,166.72);

\path[draw=drawColor,line width= 0.6pt,line join=round] (223.40,153.31) -- (223.40,151.96);
\definecolor{fillColor}{RGB}{0,0,139}

\path[draw=drawColor,line width= 0.6pt,fill=fillColor,fill opacity=0.60] (220.42,156.65) --
	(220.42,153.31) --
	(226.37,153.31) --
	(226.37,156.65) --
	(220.42,156.65) --
	cycle;

\path[draw=drawColor,line width= 1.1pt] (220.42,154.83) -- (226.37,154.83);

\path[draw=drawColor,line width= 0.6pt,line join=round] (242.15,158.55) -- (242.15,181.50);

\path[draw=drawColor,line width= 0.6pt,line join=round] (242.15,153.74) -- (242.15,151.96);
\definecolor{fillColor}{RGB}{254,224,139}

\path[draw=drawColor,line width= 0.6pt,fill=fillColor,fill opacity=0.60] (239.17,158.55) --
	(239.17,153.74) --
	(245.12,153.74) --
	(245.12,158.55) --
	(239.17,158.55) --
	cycle;

\path[draw=drawColor,line width= 1.1pt] (239.17,155.82) -- (245.12,155.82);

\path[draw=drawColor,line width= 0.6pt,line join=round] (248.76,157.59) -- (248.76,170.08);

\path[draw=drawColor,line width= 0.6pt,line join=round] (248.76,153.51) -- (248.76,151.96);
\definecolor{fillColor}{RGB}{50,136,189}

\path[draw=drawColor,line width= 0.6pt,fill=fillColor,fill opacity=0.60] (245.79,157.59) --
	(245.79,153.51) --
	(251.74,153.51) --
	(251.74,157.59) --
	(245.79,157.59) --
	cycle;

\path[draw=drawColor,line width= 1.1pt] (245.79,155.20) -- (251.74,155.20);

\path[draw=drawColor,line width= 0.6pt,line join=round] (255.38,156.68) -- (255.38,165.51);

\path[draw=drawColor,line width= 0.6pt,line join=round] (255.38,153.27) -- (255.38,151.96);
\definecolor{fillColor}{RGB}{0,0,139}

\path[draw=drawColor,line width= 0.6pt,fill=fillColor,fill opacity=0.60] (252.40,156.68) --
	(252.40,153.27) --
	(258.36,153.27) --
	(258.36,156.68) --
	(252.40,156.68) --
	cycle;

\path[draw=drawColor,line width= 1.1pt] (252.40,154.69) -- (258.36,154.69);

\path[draw=drawColor,line width= 0.6pt,line join=round] (277.44,161.44) -- (277.44,194.09);

\path[draw=drawColor,line width= 0.6pt,line join=round] (277.44,154.77) -- (277.44,151.96);
\definecolor{fillColor}{RGB}{254,224,139}

\path[draw=drawColor,line width= 0.6pt,fill=fillColor,fill opacity=0.60] (274.46,161.44) --
	(274.46,154.77) --
	(280.42,154.77) --
	(280.42,161.44) --
	(274.46,161.44) --
	cycle;

\path[draw=drawColor,line width= 1.1pt] (274.46,157.51) -- (280.42,157.51);

\path[draw=drawColor,line width= 0.6pt,line join=round] (284.06,157.53) -- (284.06,170.02);

\path[draw=drawColor,line width= 0.6pt,line join=round] (284.06,153.47) -- (284.06,151.96);
\definecolor{fillColor}{RGB}{50,136,189}

\path[draw=drawColor,line width= 0.6pt,fill=fillColor,fill opacity=0.60] (281.08,157.53) --
	(281.08,153.47) --
	(287.04,153.47) --
	(287.04,157.53) --
	(281.08,157.53) --
	cycle;

\path[draw=drawColor,line width= 1.1pt] (281.08,155.21) -- (287.04,155.21);

\path[draw=drawColor,line width= 0.6pt,line join=round] (290.68,156.79) -- (290.68,170.40);

\path[draw=drawColor,line width= 0.6pt,line join=round] (290.68,153.30) -- (290.68,151.96);
\definecolor{fillColor}{RGB}{0,0,139}

\path[draw=drawColor,line width= 0.6pt,fill=fillColor,fill opacity=0.60] (287.70,156.79) --
	(287.70,153.30) --
	(293.66,153.30) --
	(293.66,156.79) --
	(287.70,156.79) --
	cycle;

\path[draw=drawColor,line width= 1.1pt] (287.70,154.81) -- (293.66,154.81);

\path[draw=drawColor,line width= 0.6pt,line join=round] (312.74,176.84) -- (312.74,236.44);

\path[draw=drawColor,line width= 0.6pt,line join=round] (312.74,158.51) -- (312.74,152.00);
\definecolor{fillColor}{RGB}{254,224,139}

\path[draw=drawColor,line width= 0.6pt,fill=fillColor,fill opacity=0.60] (309.76,176.84) --
	(309.76,158.51) --
	(315.71,158.51) --
	(315.71,176.84) --
	(309.76,176.84) --
	cycle;

\path[draw=drawColor,line width= 1.1pt] (309.76,166.26) -- (315.71,166.26);

\path[draw=drawColor,line width= 0.6pt,line join=round] (319.35,157.89) -- (319.35,170.24);

\path[draw=drawColor,line width= 0.6pt,line join=round] (319.35,153.57) -- (319.35,151.97);
\definecolor{fillColor}{RGB}{50,136,189}

\path[draw=drawColor,line width= 0.6pt,fill=fillColor,fill opacity=0.60] (316.38,157.89) --
	(316.38,153.57) --
	(322.33,153.57) --
	(322.33,157.89) --
	(316.38,157.89) --
	cycle;

\path[draw=drawColor,line width= 1.1pt] (316.38,155.37) -- (322.33,155.37);

\path[draw=drawColor,line width= 0.6pt,line join=round] (325.97,157.02) -- (325.97,167.03);

\path[draw=drawColor,line width= 0.6pt,line join=round] (325.97,153.45) -- (325.97,151.97);
\definecolor{fillColor}{RGB}{0,0,139}

\path[draw=drawColor,line width= 0.6pt,fill=fillColor,fill opacity=0.60] (322.99,157.02) --
	(322.99,153.45) --
	(328.95,153.45) --
	(328.95,157.02) --
	(322.99,157.02) --
	cycle;

\path[draw=drawColor,line width= 1.1pt] (322.99,154.93) -- (328.95,154.93);

\path[draw=drawColor,line width= 0.6pt,line join=round,line cap=round] (192.29,147.65) rectangle (340.53,242.54);
\end{scope}
\begin{scope}
\path[clip] (192.29, 30.69) rectangle (340.53,125.58);
\definecolor{fillColor}{RGB}{255,255,255}

\path[fill=fillColor] (192.29, 30.69) rectangle (340.53,125.58);
\definecolor{drawColor}{gray}{0.92}

\path[draw=drawColor,line width= 0.3pt,line join=round] (192.29, 45.78) --
	(340.53, 45.78);

\path[draw=drawColor,line width= 0.3pt,line join=round] (192.29, 67.35) --
	(340.53, 67.35);

\path[draw=drawColor,line width= 0.3pt,line join=round] (192.29, 88.92) --
	(340.53, 88.92);

\path[draw=drawColor,line width= 0.3pt,line join=round] (192.29,110.49) --
	(340.53,110.49);

\path[draw=drawColor,line width= 0.6pt,line join=round] (192.29, 35.00) --
	(340.53, 35.00);

\path[draw=drawColor,line width= 0.6pt,line join=round] (192.29, 56.57) --
	(340.53, 56.57);

\path[draw=drawColor,line width= 0.6pt,line join=round] (192.29, 78.14) --
	(340.53, 78.14);

\path[draw=drawColor,line width= 0.6pt,line join=round] (192.29, 99.71) --
	(340.53, 99.71);

\path[draw=drawColor,line width= 0.6pt,line join=round] (192.29,121.27) --
	(340.53,121.27);

\path[draw=drawColor,line width= 0.6pt,line join=round] (213.47, 30.69) --
	(213.47,125.58);

\path[draw=drawColor,line width= 0.6pt,line join=round] (248.76, 30.69) --
	(248.76,125.58);

\path[draw=drawColor,line width= 0.6pt,line join=round] (284.06, 30.69) --
	(284.06,125.58);

\path[draw=drawColor,line width= 0.6pt,line join=round] (319.35, 30.69) --
	(319.35,125.58);
\definecolor{drawColor}{gray}{0.20}

\path[draw=drawColor,line width= 0.6pt,line join=round] (203.54, 50.62) -- (203.54, 78.54);

\path[draw=drawColor,line width= 0.6pt,line join=round] (203.54, 42.83) -- (203.54, 35.80);
\definecolor{fillColor}{RGB}{213,62,79}

\path[draw=drawColor,line width= 0.6pt,fill=fillColor,fill opacity=0.60] (200.57, 50.62) --
	(200.57, 42.83) --
	(206.52, 42.83) --
	(206.52, 50.62) --
	(200.57, 50.62) --
	cycle;

\path[draw=drawColor,line width= 1.1pt] (200.57, 46.36) -- (206.52, 46.36);

\path[draw=drawColor,line width= 0.6pt,line join=round] (210.16, 38.69) -- (210.16, 43.28);

\path[draw=drawColor,line width= 0.6pt,line join=round] (210.16, 36.94) -- (210.16, 35.11);
\definecolor{fillColor}{RGB}{254,224,139}

\path[draw=drawColor,line width= 0.6pt,fill=fillColor,fill opacity=0.60] (207.18, 38.69) --
	(207.18, 36.94) --
	(213.14, 36.94) --
	(213.14, 38.69) --
	(207.18, 38.69) --
	cycle;

\path[draw=drawColor,line width= 1.1pt] (207.18, 37.72) -- (213.14, 37.72);

\path[draw=drawColor,line width= 0.6pt,line join=round] (216.78, 38.54) -- (216.78, 42.60);

\path[draw=drawColor,line width= 0.6pt,line join=round] (216.78, 36.85) -- (216.78, 35.20);
\definecolor{fillColor}{RGB}{50,136,189}

\path[draw=drawColor,line width= 0.6pt,fill=fillColor,fill opacity=0.60] (213.80, 38.54) --
	(213.80, 36.85) --
	(219.76, 36.85) --
	(219.76, 38.54) --
	(213.80, 38.54) --
	cycle;

\path[draw=drawColor,line width= 1.1pt] (213.80, 37.63) -- (219.76, 37.63);

\path[draw=drawColor,line width= 0.6pt,line join=round] (223.40, 37.90) -- (223.40, 41.26);

\path[draw=drawColor,line width= 0.6pt,line join=round] (223.40, 36.49) -- (223.40, 35.10);
\definecolor{fillColor}{RGB}{0,0,139}

\path[draw=drawColor,line width= 0.6pt,fill=fillColor,fill opacity=0.60] (220.42, 37.90) --
	(220.42, 36.49) --
	(226.37, 36.49) --
	(226.37, 37.90) --
	(220.42, 37.90) --
	cycle;

\path[draw=drawColor,line width= 1.1pt] (220.42, 37.14) -- (226.37, 37.14);

\path[draw=drawColor,line width= 0.6pt,line join=round] (242.15, 38.75) -- (242.15, 45.33);

\path[draw=drawColor,line width= 0.6pt,line join=round] (242.15, 37.00) -- (242.15, 35.24);
\definecolor{fillColor}{RGB}{254,224,139}

\path[draw=drawColor,line width= 0.6pt,fill=fillColor,fill opacity=0.60] (239.17, 38.75) --
	(239.17, 37.00) --
	(245.12, 37.00) --
	(245.12, 38.75) --
	(239.17, 38.75) --
	cycle;

\path[draw=drawColor,line width= 1.1pt] (239.17, 37.80) -- (245.12, 37.80);

\path[draw=drawColor,line width= 0.6pt,line join=round] (248.76, 38.58) -- (248.76, 44.46);

\path[draw=drawColor,line width= 0.6pt,line join=round] (248.76, 36.88) -- (248.76, 35.19);
\definecolor{fillColor}{RGB}{50,136,189}

\path[draw=drawColor,line width= 0.6pt,fill=fillColor,fill opacity=0.60] (245.79, 38.58) --
	(245.79, 36.88) --
	(251.74, 36.88) --
	(251.74, 38.58) --
	(245.79, 38.58) --
	cycle;

\path[draw=drawColor,line width= 1.1pt] (245.79, 37.67) -- (251.74, 37.67);

\path[draw=drawColor,line width= 0.6pt,line join=round] (255.38, 37.89) -- (255.38, 41.67);

\path[draw=drawColor,line width= 0.6pt,line join=round] (255.38, 36.56) -- (255.38, 35.18);
\definecolor{fillColor}{RGB}{0,0,139}

\path[draw=drawColor,line width= 0.6pt,fill=fillColor,fill opacity=0.60] (252.40, 37.89) --
	(252.40, 36.56) --
	(258.36, 36.56) --
	(258.36, 37.89) --
	(252.40, 37.89) --
	cycle;

\path[draw=drawColor,line width= 1.1pt] (252.40, 37.22) -- (258.36, 37.22);

\path[draw=drawColor,line width= 0.6pt,line join=round] (277.44, 41.60) -- (277.44, 52.24);

\path[draw=drawColor,line width= 0.6pt,line join=round] (277.44, 38.42) -- (277.44, 35.33);
\definecolor{fillColor}{RGB}{254,224,139}

\path[draw=drawColor,line width= 0.6pt,fill=fillColor,fill opacity=0.60] (274.46, 41.60) --
	(274.46, 38.42) --
	(280.42, 38.42) --
	(280.42, 41.60) --
	(274.46, 41.60) --
	cycle;

\path[draw=drawColor,line width= 1.1pt] (274.46, 39.89) -- (280.42, 39.89);

\path[draw=drawColor,line width= 0.6pt,line join=round] (284.06, 38.55) -- (284.06, 43.85);

\path[draw=drawColor,line width= 0.6pt,line join=round] (284.06, 36.85) -- (284.06, 35.22);
\definecolor{fillColor}{RGB}{50,136,189}

\path[draw=drawColor,line width= 0.6pt,fill=fillColor,fill opacity=0.60] (281.08, 38.55) --
	(281.08, 36.85) --
	(287.04, 36.85) --
	(287.04, 38.55) --
	(281.08, 38.55) --
	cycle;

\path[draw=drawColor,line width= 1.1pt] (281.08, 37.66) -- (287.04, 37.66);

\path[draw=drawColor,line width= 0.6pt,line join=round] (290.68, 37.89) -- (290.68, 41.62);

\path[draw=drawColor,line width= 0.6pt,line join=round] (290.68, 36.54) -- (290.68, 35.07);
\definecolor{fillColor}{RGB}{0,0,139}

\path[draw=drawColor,line width= 0.6pt,fill=fillColor,fill opacity=0.60] (287.70, 37.89) --
	(287.70, 36.54) --
	(293.66, 36.54) --
	(293.66, 37.89) --
	(287.70, 37.89) --
	cycle;

\path[draw=drawColor,line width= 1.1pt] (287.70, 37.18) -- (293.66, 37.18);

\path[draw=drawColor,line width= 0.6pt,line join=round] (312.74, 53.01) -- (312.74, 91.23);

\path[draw=drawColor,line width= 0.6pt,line join=round] (312.74, 44.26) -- (312.74, 35.53);
\definecolor{fillColor}{RGB}{254,224,139}

\path[draw=drawColor,line width= 0.6pt,fill=fillColor,fill opacity=0.60] (309.76, 53.01) --
	(309.76, 44.26) --
	(315.71, 44.26) --
	(315.71, 53.01) --
	(309.76, 53.01) --
	cycle;

\path[draw=drawColor,line width= 1.1pt] (309.76, 47.86) -- (315.71, 47.86);

\path[draw=drawColor,line width= 0.6pt,line join=round] (319.35, 38.63) -- (319.35, 42.95);

\path[draw=drawColor,line width= 0.6pt,line join=round] (319.35, 36.87) -- (319.35, 35.19);
\definecolor{fillColor}{RGB}{50,136,189}

\path[draw=drawColor,line width= 0.6pt,fill=fillColor,fill opacity=0.60] (316.38, 38.63) --
	(316.38, 36.87) --
	(322.33, 36.87) --
	(322.33, 38.63) --
	(316.38, 38.63) --
	cycle;

\path[draw=drawColor,line width= 1.1pt] (316.38, 37.66) -- (322.33, 37.66);

\path[draw=drawColor,line width= 0.6pt,line join=round] (325.97, 37.98) -- (325.97, 42.05);

\path[draw=drawColor,line width= 0.6pt,line join=round] (325.97, 36.48) -- (325.97, 35.09);
\definecolor{fillColor}{RGB}{0,0,139}

\path[draw=drawColor,line width= 0.6pt,fill=fillColor,fill opacity=0.60] (322.99, 37.98) --
	(322.99, 36.48) --
	(328.95, 36.48) --
	(328.95, 37.98) --
	(322.99, 37.98) --
	cycle;

\path[draw=drawColor,line width= 1.1pt] (322.99, 37.17) -- (328.95, 37.17);

\path[draw=drawColor,line width= 0.6pt,line join=round,line cap=round] (192.29, 30.69) rectangle (340.53,125.58);
\end{scope}
\begin{scope}
\path[clip] ( 38.56,125.58) rectangle (186.79,142.15);
\definecolor{drawColor}{gray}{0.20}
\definecolor{fillColor}{gray}{0.85}

\path[draw=drawColor,line width= 0.6pt,line join=round,line cap=round,fill=fillColor] ( 38.56,125.58) rectangle (186.79,142.15);
\definecolor{drawColor}{gray}{0.10}

\node[text=drawColor,anchor=base,inner sep=0pt, outer sep=0pt, scale=  0.88] at (112.67,130.83) {DGP~II, \ $ \psi_{(m)}=0.17 $};
\end{scope}
\begin{scope}
\path[clip] (192.29,125.58) rectangle (340.53,142.15);
\definecolor{drawColor}{gray}{0.20}
\definecolor{fillColor}{gray}{0.85}

\path[draw=drawColor,line width= 0.6pt,line join=round,line cap=round,fill=fillColor] (192.29,125.58) rectangle (340.53,142.15);
\definecolor{drawColor}{gray}{0.10}

\node[text=drawColor,anchor=base,inner sep=0pt, outer sep=0pt, scale=  0.88] at (266.41,130.83) {DGP~II, \ $ \psi_{(m)}=0.58 $};
\end{scope}
\begin{scope}
\path[clip] ( 38.56,242.54) rectangle (186.79,259.11);
\definecolor{drawColor}{gray}{0.20}
\definecolor{fillColor}{gray}{0.85}

\path[draw=drawColor,line width= 0.6pt,line join=round,line cap=round,fill=fillColor] ( 38.56,242.54) rectangle (186.79,259.11);
\definecolor{drawColor}{gray}{0.10}

\node[text=drawColor,anchor=base,inner sep=0pt, outer sep=0pt, scale=  0.88] at (112.67,247.79) {DGP~I, \ $ \psi_{(m)}=0.17 $};
\end{scope}
\begin{scope}
\path[clip] (192.29,242.54) rectangle (340.53,259.11);
\definecolor{drawColor}{gray}{0.20}
\definecolor{fillColor}{gray}{0.85}

\path[draw=drawColor,line width= 0.6pt,line join=round,line cap=round,fill=fillColor] (192.29,242.54) rectangle (340.53,259.11);
\definecolor{drawColor}{gray}{0.10}

\node[text=drawColor,anchor=base,inner sep=0pt, outer sep=0pt, scale=  0.88] at (266.41,247.79) {DGP~I, \ $ \psi_{(m)}=0.58 $};
\end{scope}
\begin{scope}
\path[clip] (  0.00,  0.00) rectangle (441.02,264.61);
\definecolor{drawColor}{gray}{0.20}

\path[draw=drawColor,line width= 0.6pt,line join=round] ( 59.73, 27.94) --
	( 59.73, 30.69);

\path[draw=drawColor,line width= 0.6pt,line join=round] ( 95.03, 27.94) --
	( 95.03, 30.69);

\path[draw=drawColor,line width= 0.6pt,line join=round] (130.32, 27.94) --
	(130.32, 30.69);

\path[draw=drawColor,line width= 0.6pt,line join=round] (165.62, 27.94) --
	(165.62, 30.69);
\end{scope}
\begin{scope}
\path[clip] (  0.00,  0.00) rectangle (441.02,264.61);
\definecolor{drawColor}{gray}{0.30}

\node[text=drawColor,anchor=base,inner sep=0pt, outer sep=0pt, scale=  0.88] at ( 59.73, 19.68) {none};

\node[text=drawColor,anchor=base,inner sep=0pt, outer sep=0pt, scale=  0.88] at ( 95.03, 19.68) {$ \mathsf{U}_{i} $};

\node[text=drawColor,anchor=base,inner sep=0pt, outer sep=0pt, scale=  0.88] at (130.32, 19.68) {$ \psi_{(m),i} $};

\node[text=drawColor,anchor=base,inner sep=0pt, outer sep=0pt, scale=  0.88] at (165.62, 19.68) {rank};
\end{scope}
\begin{scope}
\path[clip] (  0.00,  0.00) rectangle (441.02,264.61);
\definecolor{drawColor}{gray}{0.20}

\path[draw=drawColor,line width= 0.6pt,line join=round] (213.47, 27.94) --
	(213.47, 30.69);

\path[draw=drawColor,line width= 0.6pt,line join=round] (248.76, 27.94) --
	(248.76, 30.69);

\path[draw=drawColor,line width= 0.6pt,line join=round] (284.06, 27.94) --
	(284.06, 30.69);

\path[draw=drawColor,line width= 0.6pt,line join=round] (319.35, 27.94) --
	(319.35, 30.69);
\end{scope}
\begin{scope}
\path[clip] (  0.00,  0.00) rectangle (441.02,264.61);
\definecolor{drawColor}{gray}{0.30}

\node[text=drawColor,anchor=base,inner sep=0pt, outer sep=0pt, scale=  0.88] at (213.47, 19.68) {none};

\node[text=drawColor,anchor=base,inner sep=0pt, outer sep=0pt, scale=  0.88] at (248.76, 19.68) {$ \mathsf{U}_{i} $};

\node[text=drawColor,anchor=base,inner sep=0pt, outer sep=0pt, scale=  0.88] at (284.06, 19.68) {$ \psi_{(m),i} $};

\node[text=drawColor,anchor=base,inner sep=0pt, outer sep=0pt, scale=  0.88] at (319.35, 19.68) {rank};
\end{scope}
\begin{scope}
\path[clip] (  0.00,  0.00) rectangle (441.02,264.61);
\definecolor{drawColor}{gray}{0.30}

\node[text=drawColor,anchor=base east,inner sep=0pt, outer sep=0pt, scale=  0.88] at ( 33.61,148.93) {0.00};

\node[text=drawColor,anchor=base east,inner sep=0pt, outer sep=0pt, scale=  0.88] at ( 33.61,170.50) {0.25};

\node[text=drawColor,anchor=base east,inner sep=0pt, outer sep=0pt, scale=  0.88] at ( 33.61,192.07) {0.50};

\node[text=drawColor,anchor=base east,inner sep=0pt, outer sep=0pt, scale=  0.88] at ( 33.61,213.64) {0.75};

\node[text=drawColor,anchor=base east,inner sep=0pt, outer sep=0pt, scale=  0.88] at ( 33.61,235.21) {1.00};
\end{scope}
\begin{scope}
\path[clip] (  0.00,  0.00) rectangle (441.02,264.61);
\definecolor{drawColor}{gray}{0.20}

\path[draw=drawColor,line width= 0.6pt,line join=round] ( 35.81,151.96) --
	( 38.56,151.96);

\path[draw=drawColor,line width= 0.6pt,line join=round] ( 35.81,173.53) --
	( 38.56,173.53);

\path[draw=drawColor,line width= 0.6pt,line join=round] ( 35.81,195.10) --
	( 38.56,195.10);

\path[draw=drawColor,line width= 0.6pt,line join=round] ( 35.81,216.67) --
	( 38.56,216.67);

\path[draw=drawColor,line width= 0.6pt,line join=round] ( 35.81,238.24) --
	( 38.56,238.24);
\end{scope}
\begin{scope}
\path[clip] (  0.00,  0.00) rectangle (441.02,264.61);
\definecolor{drawColor}{gray}{0.30}

\node[text=drawColor,anchor=base east,inner sep=0pt, outer sep=0pt, scale=  0.88] at ( 33.61, 31.97) {0.00};

\node[text=drawColor,anchor=base east,inner sep=0pt, outer sep=0pt, scale=  0.88] at ( 33.61, 53.54) {0.25};

\node[text=drawColor,anchor=base east,inner sep=0pt, outer sep=0pt, scale=  0.88] at ( 33.61, 75.11) {0.50};

\node[text=drawColor,anchor=base east,inner sep=0pt, outer sep=0pt, scale=  0.88] at ( 33.61, 96.67) {0.75};

\node[text=drawColor,anchor=base east,inner sep=0pt, outer sep=0pt, scale=  0.88] at ( 33.61,118.24) {1.00};
\end{scope}
\begin{scope}
\path[clip] (  0.00,  0.00) rectangle (441.02,264.61);
\definecolor{drawColor}{gray}{0.20}

\path[draw=drawColor,line width= 0.6pt,line join=round] ( 35.81, 35.00) --
	( 38.56, 35.00);

\path[draw=drawColor,line width= 0.6pt,line join=round] ( 35.81, 56.57) --
	( 38.56, 56.57);

\path[draw=drawColor,line width= 0.6pt,line join=round] ( 35.81, 78.14) --
	( 38.56, 78.14);

\path[draw=drawColor,line width= 0.6pt,line join=round] ( 35.81, 99.71) --
	( 38.56, 99.71);

\path[draw=drawColor,line width= 0.6pt,line join=round] ( 35.81,121.27) --
	( 38.56,121.27);
\end{scope}
\begin{scope}
\path[clip] (  0.00,  0.00) rectangle (441.02,264.61);
\definecolor{drawColor}{RGB}{0,0,0}

\node[text=drawColor,anchor=base,inner sep=0pt, outer sep=0pt, scale=  1.10] at (189.54,  7.64) {Heterogeneity};
\end{scope}
\begin{scope}
\path[clip] (  0.00,  0.00) rectangle (441.02,264.61);
\definecolor{drawColor}{RGB}{0,0,0}

\node[text=drawColor,rotate= 90.00,anchor=base,inner sep=0pt, outer sep=0pt, scale=  1.10] at ( 13.08,136.61) {Amari distance};
\end{scope}
\begin{scope}
\path[clip] (  0.00,  0.00) rectangle (441.02,264.61);
\definecolor{fillColor}{RGB}{255,255,255}

\path[fill=fillColor] (351.53, 94.60) rectangle (435.52,178.63);
\end{scope}
\begin{scope}
\path[clip] (  0.00,  0.00) rectangle (441.02,264.61);
\definecolor{drawColor}{RGB}{0,0,0}

\node[text=drawColor,anchor=base west,inner sep=0pt, outer sep=0pt, scale=  1.10] at (357.03,164.48) {Method};
\end{scope}
\begin{scope}
\path[clip] (  0.00,  0.00) rectangle (441.02,264.61);
\definecolor{fillColor}{RGB}{255,255,255}

\path[fill=fillColor] (357.03,143.46) rectangle (371.49,157.91);
\definecolor{drawColor}{gray}{0.20}

\path[draw=drawColor,line width= 0.6pt] (364.26,144.90) --
	(364.26,147.07);

\path[draw=drawColor,line width= 0.6pt] (364.26,154.30) --
	(364.26,156.47);
\definecolor{fillColor}{RGB}{213,62,79}

\path[draw=drawColor,line width= 0.6pt,fill=fillColor,fill opacity=0.60] (358.84,147.07) rectangle (369.68,154.30);

\path[draw=drawColor,line width= 0.6pt] (358.84,150.69) --
	(369.68,150.69);

\path[draw=drawColor,line width= 0.6pt] (364.26,144.90) --
	(364.26,144.90);

\path[draw=drawColor,line width= 0.6pt] (364.26,156.47) --
	(364.26,156.47);
\end{scope}
\begin{scope}
\path[clip] (  0.00,  0.00) rectangle (441.02,264.61);
\definecolor{fillColor}{RGB}{255,255,255}

\path[fill=fillColor] (357.03,129.01) rectangle (371.49,143.46);
\definecolor{drawColor}{gray}{0.20}

\path[draw=drawColor,line width= 0.6pt] (364.26,130.45) --
	(364.26,132.62);

\path[draw=drawColor,line width= 0.6pt] (364.26,139.85) --
	(364.26,142.01);
\definecolor{fillColor}{RGB}{254,224,139}

\path[draw=drawColor,line width= 0.6pt,fill=fillColor,fill opacity=0.60] (358.84,132.62) rectangle (369.68,139.85);

\path[draw=drawColor,line width= 0.6pt] (358.84,136.23) --
	(369.68,136.23);

\path[draw=drawColor,line width= 0.6pt] (364.26,130.45) --
	(364.26,130.45);

\path[draw=drawColor,line width= 0.6pt] (364.26,142.01) --
	(364.26,142.01);
\end{scope}
\begin{scope}
\path[clip] (  0.00,  0.00) rectangle (441.02,264.61);
\definecolor{fillColor}{RGB}{255,255,255}

\path[fill=fillColor] (357.03,114.55) rectangle (371.49,129.01);
\definecolor{drawColor}{gray}{0.20}

\path[draw=drawColor,line width= 0.6pt] (364.26,116.00) --
	(364.26,118.16);

\path[draw=drawColor,line width= 0.6pt] (364.26,125.39) --
	(364.26,127.56);
\definecolor{fillColor}{RGB}{50,136,189}

\path[draw=drawColor,line width= 0.6pt,fill=fillColor,fill opacity=0.60] (358.84,118.16) rectangle (369.68,125.39);

\path[draw=drawColor,line width= 0.6pt] (358.84,121.78) --
	(369.68,121.78);

\path[draw=drawColor,line width= 0.6pt] (364.26,116.00) --
	(364.26,116.00);

\path[draw=drawColor,line width= 0.6pt] (364.26,127.56) --
	(364.26,127.56);
\end{scope}
\begin{scope}
\path[clip] (  0.00,  0.00) rectangle (441.02,264.61);
\definecolor{fillColor}{RGB}{255,255,255}

\path[fill=fillColor] (357.03,100.10) rectangle (371.49,114.55);
\definecolor{drawColor}{gray}{0.20}

\path[draw=drawColor,line width= 0.6pt] (364.26,101.54) --
	(364.26,103.71);

\path[draw=drawColor,line width= 0.6pt] (364.26,110.94) --
	(364.26,113.11);
\definecolor{fillColor}{RGB}{0,0,139}

\path[draw=drawColor,line width= 0.6pt,fill=fillColor,fill opacity=0.60] (358.84,103.71) rectangle (369.68,110.94);

\path[draw=drawColor,line width= 0.6pt] (358.84,107.32) --
	(369.68,107.32);

\path[draw=drawColor,line width= 0.6pt] (364.26,101.54) --
	(364.26,101.54);

\path[draw=drawColor,line width= 0.6pt] (364.26,113.11) --
	(364.26,113.11);
\end{scope}
\begin{scope}
\path[clip] (  0.00,  0.00) rectangle (441.02,264.61);
\definecolor{drawColor}{RGB}{0,0,0}

\node[text=drawColor,anchor=base west,inner sep=0pt, outer sep=0pt, scale=  0.88] at (376.99,147.66) {Individual};
\end{scope}
\begin{scope}
\path[clip] (  0.00,  0.00) rectangle (441.02,264.61);
\definecolor{drawColor}{RGB}{0,0,0}

\node[text=drawColor,anchor=base west,inner sep=0pt, outer sep=0pt, scale=  0.88] at (376.99,133.20) {MG};
\end{scope}
\begin{scope}
\path[clip] (  0.00,  0.00) rectangle (441.02,264.61);
\definecolor{drawColor}{RGB}{0,0,0}

\node[text=drawColor,anchor=base west,inner sep=0pt, outer sep=0pt, scale=  0.88] at (376.99,118.75) {PMG};
\end{scope}
\begin{scope}
\path[clip] (  0.00,  0.00) rectangle (441.02,264.61);
\definecolor{drawColor}{RGB}{0,0,0}

\node[text=drawColor,anchor=base west,inner sep=0pt, outer sep=0pt, scale=  0.88] at (376.99,104.29) {PMG $ L>K_1 $};
\end{scope}
\end{tikzpicture}

%% file: C_AppAsymp.tex
\section{Detailed proofs of asymptotic results} \label{sec:AppAsymp}

We follow the definition of the VAR process in Section \ref{sec:Estimation}. Analogously to \cite{BruggemannEtAl2016}, we additionally define the $(Kp \times K)$-dimensional matrices $C_{ij} = [\Xi_{i,j-1}', \dots, \Xi_{i,j-p}']'$. For any random vector $x$, let $\|x\|_p = (\mathbb E\|x\|^p )^{1/p}$, $p\geq 1$, denote its $L_p$-norm w.r.t the Euclidean vector norm. For a matrix $X$, $\|X\|$ denotes the induced operator norm, and $\|X\|_F$ the Frobenius norm. Both matrix norms are understood  with respect to the Euclidean norm. Throughout, \(a\lesssim b\) means that \(a\leq Cb\) for a constant \(C<\infty\) independent of all diverging sample and block dimensions. When the comparison is only stochastic, we use \(O_p(\cdot)\) or
\(O_{p^*}(\cdot)\) explicitly.


\subsection{Proof of Theorem \ref{th:AsyMG}} \label{app:AsyMG}

Theorem \ref{th:AsyMG} states that under Assumption~\ref{as:mixing_conditions} and the joint limit of $ (N,T) \rightarrow ( \infty, \infty ) $, whereas $N/T\to 0$, it holds that
\begin{align*}
    \sqrt{NT} & 
    \begin{pmatrix}
        \bar{\beta} - \beta \\
        \Tilde{\sigma} - \sigma \\
        \Tilde{\varphi} - \varphi
    \end{pmatrix}
    \overset{d}{\longrightarrow} \mathcal{N}(0,\Tilde{V}).
\end{align*}
In the following, all asymptotic statements for the mean-group estimators are understood to be centered at the corresponding realized finite-panel targets. In particular, for
\[
    \beta_N
    \coloneqq \frac{1}{N}\sum_{i=1}^N \beta_i,
    \qquad
    \sigma_N
    \coloneqq \frac{1}{N}\sum_{i=1}^N \operatorname{vech}(\Sigma_{u,i}),
    \qquad
    \varphi_N
    \coloneqq \frac{1}{N}\sum_{i=1}^N \varphi_i,
\]
the object of inference is the average parameter vector among the realized cross-sectional units. Thus, the asymptotics consider the object 
\(
    \sqrt{NT}
    \left(
    (\bar{\beta}-\beta_N)',
    (\tilde{\sigma}-\sigma_N)',
    (\tilde{\varphi}-\varphi_N)'
    \right)'
\)
and thus isolates the within-unit estimation error. The additional cross-sectional variation is not part of the asymptotic analysis and would generally lead to slower \(\sqrt{N}\) convergence rate.
The entries of $\Tilde{V}$ are: 
\begin{align*}
    \Tilde{V}_{11} 
    &= \lim_{N,T\to\infty} \frac{1}{NT} \sum_{h_{(i)}=-(N-1)}^{N-1} \sum_{i=\max\{1,1+h_{(i)}\}}^{\min\{N,N+h_{(i)}\}} \sum_{h_{(t)}=-(T-1)}^{T-1} \sum_{t=\max\{1,1+h_{(t)}\}}^{\min\{T,T+h_{(t)}\}} \mathbb E \Bigg[ \left(\Sigma_{ZZ,i}^{-1}\otimes I\right) \sum_{j,j'=1}^\infty \left(C_{ij}\otimes I \right)\\
    & \quad \times \vec(u_{it} u_{i,t-j}') \vec \left(u_{i+h_{(i)},t+h_{(t)}}u_{i+h_{(i)},t+h_{(t)}-j'}'\right)' \left(C_{i+h_{(i)},j'}\otimes I\right)' \left( \Sigma_{ZZ,i+h_{(i)}}^{-1}\otimes I\right)' \Bigg] \\ 
    \Tilde{V}_{22} 
    & = L_K \Bigg\{ \lim_{N,T\to\infty} \frac{1}{NT} \sum_{h_{(i)}=-(N-1)}^{N-1} \sum_{i=\max\{1,1+h_{(i)}\}}^{\min\{N,N+h_{(i)}\}} \sum_{h_{(t)}=-(T-1)}^{T-1} \sum_{t=\max\{1,1+h_{(t)}\}}^{\min\{T,T+h_{(t)}\}} \\
    & \quad \times \left( \mathbb E \left[ \vec(u_{it} u_{it}') \vec \left(u_{i+h_{(i)},t+h_{(t)}} u_{i+h_{(i)},t+h_{(t)}}'\right)'\right] - \vec(\Sigma_{u,i}) \vec (\Sigma_{u,i+h_{(i)}})' \right) \Bigg\} L_K' \\ 
    \Tilde{V}_{33} 
    & = \lim_{N,T\to\infty} \frac{1}{NT} \sum_{h_{(i)}=-(N-1)}^{N-1} \sum_{i=\max\{1,1+h_{(i)}\}}^{\min\{N,N+h_{(i)}\}} \Bigg\{ \mathbb E \Bigg[ \operatorname{chol}\left( \tilde \Sigma_{u,i} \right)^{-1} \sum_{h_{(t)}=-(T-1)}^{T-1} \sum_{t=\max\{1,1+h_{(t)}\}}^{\min\{T,T+h_{(t)}\}} \vec\left(u_{it} m_{it}'\right)\\
    & \quad \times  \vec\left(u_{i+h_{(i)},t+h_{(t)}} m_{i+h_{(i)},t+h_{(t)}}'\right)'\left( \operatorname{chol} \left( \Sigma_{u,i+h_{(i)}} \right)^{-1}\right)' \Bigg] - \vec(\Sigma_{\varepsilon^\circ m,i}) \vec (\Sigma_{\varepsilon^\circ m,i+h_{(i)}})' \Bigg\} \\
    \Tilde{V}_{21} 
    & = L_K \Bigg\{ \lim_{N,T\to\infty} \frac{1}{NT} \sum_{h_{(i)}=-(N-1)}^{N-1} \sum_{i=\max\{1,1+h_{(i)}\}}^{\min\{N,N+h_{(i)}\}} \sum_{h_{(t)}=-(T-1)}^{T-1} \sum_{t=\max\{1,1+h_{(t)}\}}^{\min\{T,T+h_{(t)}\}} \sum_{j=1}^\infty \mathbb E \Bigg[ \vec\left(u_{it} u_{it}'\right) \\
    & \quad \times \vec\left(u_{i+h_{(i)},t+h_{(t)}} u_{i+h_{(i)},t+h_{(t)}-j}'\right)' \left(C_{i+h_{(i)},j}\otimes I\right)' \left(\Sigma_{ZZ,i+h_{(i)}}^{-1}\otimes I\right)' \Bigg] \Bigg\} \\
    \Tilde{V}_{31} 
    & = \lim_{N,T\to\infty} \frac{1}{NT} \sum_{h_{(i)}=-(N-1)}^{N-1} \sum_{i=\max\{1,1+h_{(i)}\}}^{\min\{N,N+h_{(i)}\}} \sum_{h_{(t)}=-(T-1)}^{T-1} \sum_{t=\max\{1,1+h_{(t)}\}}^{\min\{T,T+h_{(t)}\}} \sum_{j=1}^\infty \mathbb E \Bigg[ \operatorname{chol}\left( \tilde \Sigma_{u,i} \right)^{-1} \\
    & \quad \times \vec\left(u_{it} m_{it}'\right) \vec\left(u_{i+h_{(i)},t+h_{(t)}} u_{i+h_{(i)},t+h_{(t)}-j}'\right)' \left(C_{i+h_{(i)},j}\otimes I\right)' \left(\Sigma_{ZZ,i+h_{(i)}}^{-1}\otimes I\right)' \Bigg] \\
    \Tilde{V}_{32} 
    & = \Bigg\{ \lim_{N,T\to\infty} \frac{1}{NT} \sum_{h_{(i)}=-(N-1)}^{N-1} \sum_{i=\max\{1,1+h_{(i)}\}}^{\min\{N,N+h_{(i)}\}}  \sum_{h_{(t)}=-(T-1)}^{T-1} \sum_{t=\max\{1,1+h_{(t)}\}}^{\min\{T,T+h_{(t)}\}} \mathbb E \Big[ \operatorname{chol}\left( \tilde \Sigma_{u,i} \right)^{-1} \vec\left(u_{it} m_{it}'\right) \\
    & \quad \times \vec\left(u_{i+h_{(i)},t+h_{(t)}} u_{i+h_{(i)},t+h_{(t)}}'\right)' - \vec\left(\Sigma_{\varepsilon^\circ m,i}\right) \vec\left(\Sigma_{u,i+h_{(i)}}\right)' \Big] \Bigg\} L_K'
\end{align*}

\paragraph{Consistency of the parameter estimators} 
For consistency of the parameter estimators, we follow \cite{JentschLunsford2021}. Since our estimator $\boldsymbol{\widebar{\varphi}}$ does not coincide with the estimator considered there, namely the (individual) estimator $ \widehat{\varphi}^{JL}_i = \vec\left(\widehat{\Sigma}_{mu,i}\right)$, we prove consistency for the PMG proxy moment estimator only. The PMG population target is defined as 
\begin{equation} \label{eq: PMG definition}
    \Sigma_{\epsilon^\circ m}
    :=
    \lim_{N\to\infty}
    \frac{1}{N}\sum_{i=1}^N
    \Sigma_{\epsilon^\circ m,i},
    \quad \text{where} \quad
    \Sigma_{\epsilon^\circ m,i}
    =
    \mathsf{B}_i^{\circ-1}\Sigma_{um,i}.
\end{equation}
By the uniform full-rank condition in Assumption~\ref{as:SVAR}(iii), uniform non-singularity holds:
\[
    \inf_{i\in\mathbb{N}}
    \lambda_{\min}(\Sigma_{u,i}) >0.
\]
Further, together with Assumption \ref{as:SVAR}~(i), the population regressor moment $\widehat{\Sigma}_{ZZ,i}=\E\left[ \boldsymbol{z}_{it}\boldsymbol{z}_{it}'\right]$ satisfies $\inf_{i\in\mathbb{N}} \lambda_{\min}(\widehat{\Sigma}_{ZZ,i})>0$.

By Hölder's inequality and the uniform moment bound in Assumption~\ref{as:mixing_conditions}~(ii), it holds
$\sup_{i\in\mathbb{N}}\|\Sigma_{um,i}\|<\infty$.
Further, by stationarity, the mixing condition, and the uniform moment bound in Assumption~\ref{as:mixing_conditions}~(i)--(ii), the law of large
numbers for mixing random fields \citep[Thm. 3]{jenish2009} implies
\[
    \frac{1}{N}\sum_{i=1}^N
    \left\|
        \widetilde{\Sigma}_{u,i}-\Sigma_{u,i}
    \right\|^2
    =
    o_p(1)
    \quad
    \text{and}
    \quad
    \frac{1}{N}\sum_{i=1}^N
    \left\|
        \widetilde{\Sigma}_{um,i}-\Sigma_{um,i}
    \right\|^2
    =
    o_p(1).
\]
By consistency of the OLS estimator,
\[
    \widehat{\Sigma}_{um,i}
    -
    \widetilde{\Sigma}_{um,i}
    =
    -
    (\widehat{\boldsymbol{A}}_i- \boldsymbol{A}_i)
    \frac{1}{T}\sum_{t=1}^T
    \boldsymbol{z}_{it}\boldsymbol{m}_{it}'.
\]
Uniform stability in Assumption~\ref{as:SVAR}(i), the moment and
mixing conditions in Assumption \ref{as:mixing_conditions}(i)--(ii),
and uniform non-singularity of \(\widehat{\Sigma}_{ZZ,i}\) imply $N^{-1}\sum_{i=1}^N \|\widehat{A}_i-A_i\|^2 = o_p(1)$.
Consequently, $N^{-1}\sum_{i=1}^N \left\| \widehat{\Sigma}_{um,i} - \widetilde{\Sigma}_{um,i} \right\|^2 = o_p(1)$.
Together, this yields $N^{-1}\sum_{i=1}^N \left\| \widehat{\Sigma}_{um,i} - \Sigma_{um,i} \right\|^2 = o_p(1)$. The inverse Cholesky factor is a composition of locally Lipschitz continuous maps on sets of positive definite matrices whose smallest eigenvalues are bounded away from zero. Hence, uniform non-singularity in Assumption~\ref{as:SVAR}(iii)
implies
\[
    \frac{1}{N}\sum_{i=1}^N
    \left\|
        \widehat{\mathsf{B}}_i^{\circ-1}
        -
        \mathsf{B}_i^{\circ-1}
    \right\|^2
    =
    o_p(1).
\]
Using $\widebar{\Sigma}_{\epsilon^\circ m} = N^{-1} \sum_{i=1}^N \widehat{\mathsf{B}}_i^{\circ-1} \widehat{\Sigma}_{um,i}$,
we obtain
\begin{equation} \label{eq: Bound Consistency of PMG}
\begin{aligned}
    \widebar{\Sigma}_{\epsilon^\circ m}
    -
    \Sigma_{\epsilon^\circ m}
    =
    \frac{1}{N}\sum_{i=1}^N
    \left(
        \widehat{\mathsf{B}}_i^{\circ-1}
        -
        \mathsf{B}_i^{\circ-1}
    \right)\Sigma_{um,i}   
    +
    \frac{1}{N}\sum_{i=1}^N
    \widehat{\mathsf{B}}_i^{\circ-1}
    \left(
        \widehat{\Sigma}_{um,i}
        -
        \Sigma_{um,i}
    \right)               
    +
    \left(
        \frac{1}{N}\sum_{i=1}^N
        \Sigma_{\epsilon^\circ m,i}
        -
        \Sigma_{\epsilon^\circ m}
    \right).
\end{aligned}
\end{equation}
The norm of the first term is bounded by
\[
\begin{aligned}
    &\sup_{i\in\mathbb{N}}\|\Sigma_{um,i}\|
    \frac{1}{N}\sum_{i=1}^N
    \left\|
        \widehat{\mathsf{B}}_i^{\circ-1}
        -
        \mathsf{B}_i^{\circ-1}
    \right\|   
    \leq
    \sup_{i\in\mathbb{N}}\|\Sigma_{um,i}\|
    \left[
        \frac{1}{N}\sum_{i=1}^N
        \left\|
            \widehat{\mathsf{B}}_i^{\circ-1}
            -
            \mathsf{B}_i^{\circ-1}
        \right\|^2
    \right]^\frac12
    =
    o_p(1).
\end{aligned}
\]
By the Cauchy--Schwarz inequality, the second term in \eqref{eq: Bound Consistency of PMG} is bounded by
\[
\begin{aligned}
    &
    \left(
        \frac{1}{N}\sum_{i=1}^N
        \left\|
            \widehat{\mathsf{B}}_i^{\circ-1}
        \right\|^2
    \right)^\frac12
    \left(
        \frac{1}{N}\sum_{i=1}^N
        \left\|
            \widehat{\Sigma}_{um,i}
            -
            \Sigma_{um,i}
        \right\|^2
    \right)^\frac12
    =
    o_p(1).
\end{aligned}
\]
Finally, the third term in \eqref{eq: Bound Consistency of PMG} converges to zero by the definition of \(\Sigma_{\epsilon^\circ m}\) in \eqref{eq: PMG definition} and consistency follows.

\paragraph{Joint asymptotic normality} 
We recall the definition of \(\Tilde{\sigma} \) and \(\Tilde{\varphi} \) in section \ref{sec:AsyInference} and consider the asymptotics of
\begin{align*}
    \sqrt{NT}
    \begin{pmatrix}
    \widebar{\beta} - \beta\\
    \Tilde{\sigma} - \sigma \\
    \Tilde{\varphi} - \varphi
    \end{pmatrix} 
    &=
    \frac{1}{\sqrt{NT}} \sum_{i=1}^{N} \sum_{t=1}^{T} 
    \begin{pmatrix}
    \left( \left( \widehat{\Sigma}_{ZZ,i}\right)^{-1} \otimes I \right) \sum_{j=1}^\infty \left( C_{ij} \otimes I \right) \vec \left(u_{it} u_{i,t-j}'\right)\\
    L_K  \vec \left( u_{it} u_{it}' - \Sigma_{u,i} \right)\\
    \vec \left( \operatorname{chol}(\Tilde{\Sigma}_{u,i})^{-1}u_{it} m_{it}' - \Sigma_{\epsilon^{\circ} m,i} \right)
    \end{pmatrix}\eqqcolon
    \sum_{i=1}^{N} \sum_{t=1}^{T} X_{it, NT}.
\end{align*}
We apply a CLT for a near-epoch dependent (NED) triangular array of random fields on an $\alpha$-mixing random field \citep[Thm. 2]{jenish2012}. The $\alpha$-mixing random field $\mathbf{x}$ acts as an input process for the triangular array $X = \left( X_{it, NT}: i\leq N, t\leq T \right)$ of a centered random field on the irregular spaced two-dimensional index set ($\mathbb{N}_+\times\mathbb{Z}$). Let
\[
\Sigma_{NT}
\coloneqq
\operatorname{Var}\left(\sum_{i=1}^N\sum_{t=1}^T X_{it,NT}\right).
\]
We apply the CLT to the scalar triangular array \(\lambda'X\), where $\|\lambda\|=1$ is of suitable dimension and fixed. Hence,
\[
\sigma^{-1}_{NT}
\sum_{i=1}^N\sum_{t=1}^T \lambda'X_{it,NT}
\overset{d}{\longrightarrow} \mathcal N(0,1), \quad \text{where }
\sigma^2_{NT} = \lambda' \Sigma_{NT} \lambda.
\]
The multivariate result follows by the Cramér-Wold device \citep[see, e.g.][Thm. 29.4]{Billingsley1995}. We examine the following conditions for $\lambda'X$ and $\mathbf{x}$:
\begin{enumerate} [label=\alph{enumi})]
    \item There exists some $\delta>0$, such that
    \begin{itemize}
        \item $\lambda'X$ is uniformly $L_{2+\delta}$-bounded, and
        \item $\sum_{r=1}^\infty r^{\frac{2\delta}{2+\delta}+1} \widehat{\alpha}(r)^{\frac{\delta}{2(2+\delta)}} < \infty$, where $\widehat{\alpha}$ is the  $\alpha$-mixing coefficient of $\mathbf{x}$.
    \end{itemize}
    \item The NED coefficients $\psi(s)$ of $\lambda'X$ satisfy $\sum_{s=1}^\infty s \psi(s) < \infty$.
    \item $\inf_{N\in\mathbb N,T \in\mathbb Z} \sigma^2_{NT}  > 0$.
\end{enumerate}

\paragraph{c)} is imposed in Assumption \ref{as:mixing_conditions}~(iii).

\paragraph{a)} For $\lambda'X = \lambda' (X^{(1)\prime}, X^{(2)\prime}, X^{(3)\prime})'$ we prove
\begin{align*}
    \sup_{N\in\mathbb{N}_+, T\in\mathbb{Z}} \sup_{i\leq N, t\leq T} \| \lambda'X_{it} \|_{2+\delta} 
    \leq \| \lambda \| \sup_{N,T} \sup_{i,t} \left( \| X^{(1)}_{it} \|_{2+\delta} + \| X^{(2)}_{it} \|_{2+\delta} + \| X^{(3)}_{it} \|_{2+\delta} \right) < \infty.
\end{align*}
It suffices to uniformly bound each component. With Assumptions \ref{as:SVAR}~(i), \ref{as:Proxy}, \ref{as:mixing_conditions}~(i)--(ii) and Hölder's inequality it follows that
\begin{align*}
    \sup_{N,T}\sup_{i\leq N, t\leq T}\|X^{(1)}_{it}\|_{2+\delta}
    &\leq
    \sup_{N,T}\sup_{i\leq N}
    \sum_{j=1}^{\infty}
    \|C_{ij}\|
    \left\|\widehat{\Sigma}_{ZZ,i}^{-1}\right\|_{2(2+\delta)}
    \left\|
    \operatorname{vec}\left(u_{it}u_{i,t-j}'\right)
    \right\|_{2(2+\delta)} \\
    &\leq
    \sup_{N,T}\sup_{i\leq N}
    \sum_{j=1}^{\infty}
    \|C_{ij}\|
    \left\|\widehat{\Sigma}_{ZZ,i}^{-1}\right\|_{2(2+\delta)}
    \|u_{it}\|_{4(2+\delta)}
    \|u_{i,t-j}\|_{4(2+\delta)}
    <\infty , \\
    \sup_{N,T}\sup_{i\leq N, t\leq T}\|X^{(2)}_{it}\|_{2+\delta}
    &\leq
    \|L_K\|
    \sup_{N,T}\sup_{i\leq N}
    \left\|
    \operatorname{vec}\left(u_{it}u_{it}'-\Sigma_{u,i}\right)
    \right\|_{2+\delta} \\
    &\leq
    \|L_K\|
    \sup_{N,T}\sup_{i\leq N}
    \left(
    \|u_{it}\|_{2(2+\delta)}^2
    +
    \|\Sigma_{u,i}\|
    \right)
    <\infty , \\
    \sup_{N,T}\sup_{i\leq N, t\leq T}\|X^{(3)}_{it}\|_{2+\delta}
    &\leq
    \sup_{N,T}\sup_{i\leq N}
    \left(
    \left\|
    \left(\operatorname{chol}\widetilde{\Sigma}_{u,i}\right)^{-1}
    u_{it}m_{it}'
    \right\|_{2+\delta}
    +
    \|\Sigma_{\varepsilon^\circ m,i}\|
    \right) \\
    &\leq
    \sup_{N,T}\sup_{i\leq N}
    \left(
    \left\|
    \left(\operatorname{chol}\widetilde{\Sigma}_{u,i}\right)^{-1}
    \right\|_{2(2+\delta)}
    \|u_{it}m_{it}'\|_{2(2+\delta)}
    +
    \|\Sigma_{\varepsilon^\circ m,i}\|
    \right) \\
    &\leq
    \sup_{N,T}\sup_{i\leq N}
    \left(
    \left\|
    \left(\operatorname{chol}\widetilde{\Sigma}_{u,i}\right)^{-1}
    \right\|_{2(2+\delta)}
    \|u_{it}\|_{4(2+\delta)}
    \|m_{it}\|_{4(2+\delta)}
    +
    \|\Sigma_{\varepsilon^\circ m,i}\|
    \right)
    <\infty .
\end{align*}
where $\delta>0$ is from Assumption \ref{as:mixing_conditions}.

As the \(L_2\)-NED property is preserved under linear combinations and coordinate-wise transformations \citep[p.~181]{jenish2012}, it suffices to verify the property component-wise. Let
\[
\mathcal F_{it}(s)
=
\sigma\{ \mathbf{x}_{i,t-s},\ldots,\mathbf{x}_{i,t+s}\}
\]
and abbreviate $\mathbb E_s \coloneqq \mathbb E[\cdot \mid \mathcal{F}_{it}(s)]$. Since \(\widehat{\Sigma}_{ZZ,i}\) and \(\widetilde{\Sigma}_{u,i}\) are computed from the full time series of unit \(i\), the feasible summands are not local functions of \(\mathbf{x}_{it} =(u_{it}',m_{it}')'.\) For the third component, the first-order Taylor expansion around
\((\Sigma_{u,i},\Sigma_{um,i})\) gives
\[
\begin{aligned}
    &
    \operatorname{vec}\!\left(
        \operatorname{chol}(\widetilde{\Sigma}_{u,i})^{-1}
        \widetilde{\Sigma}_{um,i}
    \right)
    -
    \operatorname{vec}\!\left(
        \operatorname{chol}(\Sigma_{u,i})^{-1}
        \Sigma_{um,i}
    \right)
    \\
    &\quad =
    \operatorname{vec}\!\left[
        \operatorname{chol}(\Sigma_{u,i})^{-1}
        \left(
            \widetilde{\Sigma}_{um,i}
            -
            \Sigma_{um,i}
        \right)
        -
        D_i
        \operatorname{chol}(\Sigma_{u,i})^{-1}
        \Sigma_{um,i}
    \right]
    +
    R_{iT},
\end{aligned}
\]
where $D_i$ is a lower-triangular matrix which appears from differentiating the inverse Cholesky factor and follows from the standard derivative of the Cholesky
decomposition; see, e.g., \citet{smith1995differentiation}.
Under the uniform eigenvalue bound for \(\Sigma_{u,i}\), the uniform consistency of
\(\widetilde{\Sigma}_{u,i}\), and the smoothness of the inverse Cholesky map on a neighborhood of \(\Sigma_{u,i}\), the Taylor remainder satisfies
\[
    \|R_{iT}\|
    \lesssim
    \left(
    \|\widetilde{\Sigma}_{u,i}-\Sigma_{u,i}\|
    +
    \|\widetilde{\Sigma}_{um,i}-\Sigma_{um,i}\|
    \right)^2,
\]
uniformly in \(i\). By the mixing and moment assumptions,
\[
    \sup_{i\leq N} \mathbb E\|R_{iT}\|
    \leq \sup_{i\leq N}
    \mathbb E\left\|
    \widetilde{\Sigma}_{u,i}-\Sigma_{u,i}
    \right\|^2
    +
    \sup_{i\leq N}
    \mathbb E\left\|
    \widetilde{\Sigma}_{um,i}-\Sigma_{um,i}
    \right\|^2
    \lesssim
    \frac{1}{T}.
\]
Thus, if \(N/T\to0\), the Taylor remainder is asymptotically negligible and
\[
    \sqrt{NT}
    \frac{1}{N}\sum_{i=1}^N R_{iT}
    =
    o_p(1).
\]
In this case, \(X^{(2)}\) and \(X^{(3)}\) are \(\mathcal F_{it}(s)\)-measurable for every \(s\geq0\), and their NED approximation error is zero. It remains to consider the first component \(X^{(1)}\). 
Since
\[
\sum_{j=1}^{s}
\left(C_{ij}\otimes I\right)
\operatorname{vec}(u_{it}u_{i,t-j}')
\]
is \(\mathcal F_{it}(s)\)-measurable, Jensen's inequality gives
\[
\begin{aligned}
    &
    \left\|
    \sum_{j=1}^{\infty}
    \left(C_{ij}\otimes I\right)
    \operatorname{vec}(u_{it}u_{i,t-j}')
    -
    \mathbb E_s
    \sum_{j=1}^{\infty}
    \left(C_{ij}\otimes I\right)
    \operatorname{vec}(u_{it}u_{i,t-j}')
    \right\|_2
    \\
    &\lesssim
    \left\|
    \sum_{j=s+1}^{\infty}
    \left(C_{ij}\otimes I\right)
    \operatorname{vec}(u_{it}u_{i,t-j}')
    \right\|_2
    \leq
    \sum_{j=s+1}^{\infty}
    \|C_{ij}\|
    \|u_{it}\|_4
    \|u_{i,t-j}\|_4 .
\end{aligned}
\]
By the uniform moment condition and the uniform stability condition,
\[
    \sup_{i,t}
    \left\|
    X^{(1)}-\mathbb E_sX^{(1)}
    \right\|_2
    \leq C\rho^s
\]
for some \(\rho\in(0,1)\). Therefore the linearized summand is \(L_2\)-NED on \(\mathbf{x}\), with NED coefficients satisfying
\[
    \sum_{s=1}^{\infty}s\psi(s)<\infty .
\]

\paragraph{b)} By \citet[Thm. 17.8]{davidson1994}, the sum of two $L_2$-NED random fields is again a $L_2$-NED random field with size equal to the negative of the minimum of the two sizes. Thus, by the partial results for $X^{(1)}$, $X^{(2)}$ and $X^{(3)}$, it follows that $X$ is uniformly $L_{2+\delta}$-integrable and $L_2$-NED on $\mathbf{x}$. The size of $X$ is $-2(2+\delta)/\delta$ for some $\delta>0$ and therefore the NED-coefficients satisfy \( \sum_{s=1}^\infty s\psi(s)<\infty\). Furthermore, by Assumption \ref{as:mixing_conditions}~(ii), the $\alpha$-mixing coefficients of $\mathbf{x}$ satisfy the summability condition from \citet[Thm. 2]{jenish2012}. \hfill $\square$

\subsection{Proof of Theorem \ref{th:pbb_consist}} \label{app:pbb_consist}


By Pólya's Theorem and Lemma \ref{lem: Equivalence of Bootstrap Estimators}, it suffices to show that
\begin{equation} \label{eq: PBB Consistency Statement}
    \sqrt{NT}
    \begin{pmatrix}
        \Tilde{\beta}^*- \Tilde{\beta} \\
        \Tilde{\sigma}^*- \Tilde{\sigma} \\
        \Tilde{\varphi}^* - \Tilde{\varphi}
    \end{pmatrix}
    \overset{d}{\rightarrow}\mathcal{N}(0,\Tilde{V}) \quad \text{w.r.t } P^*,
\end{equation}
where $\Tilde{V}$ is the covariance matrix from Theorem \ref{th:AsyMG} and we define, analogously to \citet{BruggemannEtAl2016},
\begin{align} \label{eq: Def beta Tilde}
    \Tilde{\beta}^* - \Tilde{\beta} \coloneqq \frac{1}{NT} \sum_{i=1}^{N} \sum_{t=1}^{T} \sum_{j=1}^{t-1} \left( \Tilde{\Sigma}_{ZZ,i}^{* -1} \otimes I \right) \left( C_{ij} \otimes I \right) \vec (\tilde{u}^*_{it} \tilde{u}^{*\prime}_{i,t-j}).
\end{align}
The expression in \eqref{eq: PBB Consistency Statement} can be written as
\begin{equation} \label{eq: Vector X_rl,NT from Bootstrap Consistency}
    \begin{aligned}
        &\sqrt{NT} \begin{pmatrix}
        \Tilde{\beta}^*- \Tilde{\beta} \\
        \Tilde{\sigma}^*- \Tilde{\sigma} \\
        \Tilde{\varphi}^* - \Tilde{\varphi}
        \end{pmatrix}
        = 
        \frac{1}{\sqrt{NT}} \sum_{i=1}^{N} \sum_{t=1}^{T}
        \begin{pmatrix}
        \sum_{j=1}^{t-1} \left( \Tilde{\Sigma}_{ZZ,i}^{* -1} \otimes I \right) \left( C_{ij} \otimes I \right) \vec \left(\tilde{u}^*_{it} \tilde{u}^{*\prime}_{i,t-j}\right) \\
        L_K  \vec \left(\tilde{u}^*_{it} \tilde{u}^{*\prime}_{it} - u_{it} u_{it}'\right)\\
        \vec \left( \operatorname{chol} \left( \Tilde{\Sigma}_{u,i}^* \right)^{-1}  \tilde{u}^*_{it} \tilde{m}^{*\prime}_{it} - \operatorname{chol} \left( \Tilde{\Sigma}_{u,i} \right)^{-1}  u_{it} m_{it}' \right)
        \end{pmatrix} \\
        & = \frac{1}{\sqrt{NT}} \sum_{\ell=1}^{N^\circ} \sum_{r=1}^{T^\circ} \sum_{m=1}^{b_{(i)}} \sum_{s=1}^{b_{(t)}} 
        \begin{pmatrix}
        \sum_{j=1}^{t-1} \left( \Tilde{\Sigma}_{ZZ,i(\ell,m)}^{* -1} \otimes I \right) \left( C_{i(\ell,m),j} \otimes I \right) \vec \left(\tilde{u}^*_{\substack{i(\ell,m),\\t(r,s)}} \tilde{u}^{*\prime}_{\substack{i(\ell,m),\\t(r,s)-j}}\right) \\
        L_K  \vec \left(\tilde{u}^*_{\substack{i(\ell,m),\\t(r,s)}} \tilde{u}^{*\prime}_{\substack{i(\ell,m),\\t(r,s)}} - u_{\substack{i(\ell,m),\\t(r,s)}} u_{\substack{i(\ell,m),\\t(r,s)}}'\right)\\
        \vec \left( \operatorname{chol} \left( \Tilde{\Sigma}_{u,i}^* \right)^{-1}  \tilde{u}^*_{\substack{i(\ell,m),\\t(r,s)}} \tilde{m}^{*\prime}_{\substack{i(\ell,m),\\t(r,s)}} - \operatorname{chol} \left( \Tilde{\Sigma}_{u,i(\ell,m)} \right)^{-1}  u_{\substack{i(\ell,m),\\t(r,s)}} m_{\substack{i(\ell,m),\\t(r,s)}}' \right)
        \end{pmatrix} \\
        & \eqqcolon \sum_{\ell=1}^{N^\circ} \sum_{r=1}^{T^\circ} X^*_{\ell r,N^\circ T^\circ},
    \end{aligned}
\end{equation}
abbreviating $i(\ell,m) \coloneqq (l-1)b_{(i)}+m$ and $t(r,s) \coloneqq (r-1)b_{(t)}+s$, for $m\in \{1, \dots, b_{(i)}\}$, $s\in \{1, \dots, b_{(t)}\}$, as we introduce the notation of summing over blocks ($\ell$ in the cross-sectional and $r$ in the time-dimension) and in-block positions. The triangular array $\left(X^*_{\ell r,N^\circ T^\circ}: l\leq N^\circ, r\leq T^\circ\right)$ of normalized block sums is the subject of our asymptotic analysis. Invoking Prop.\ 6.3.9 of \cite{BrockwellDavis1991}, denote by $X^{\ast(q)}_{\ell r,N^\circ T^\circ}$ the truncated vector from \eqref{eq: Vector X_rl,NT from Bootstrap Consistency}, where we replace the sum $\sum_{j=1}^{t-1}$ in the first component by $\sum_{j=1}^q$. We use the finite truncation arguments
\begin{enumerate}[label=\alph*)]
    \item $\sum_{\ell,r} X^{\ast(q)}_{\ell r,N^\circ T^\circ}
    \;\overset{d}{\longrightarrow}\; \mathcal{N}(0,\Tilde V_q)$
    in probability w.r.t $P^*$ as $N,T \to \infty$ with $\Tilde V_q$ defined in \eqref{eq: Tilde Vq}, 
    \item $\Tilde V_q \to \Tilde V$ as $q\to\infty$, where $\Tilde V$ is as defined in Theorem \ref{th:AsyMG},
    \item $ \forall \delta > 0: \;
    \lim_{q \to \infty} \limsup_{N,T \to \infty}
    P^\ast\Big( \sum_{\ell,r}
    \left \| X^\ast_{\ell r,N^\circ T^\circ} - X^{\ast(q)}_{\ell r,N^\circ T^\circ} \right \| > \delta
    \Big) = 0$ in probability,
\end{enumerate}
and conclude the desired result.
\paragraph{(b)} is directly observed from the definition of $\Tilde V_q$ and $\Tilde V$. 
\paragraph{(c)} The only term to consider is the first entry of $X^{\ast(q)}_{\ell r,N^\circ T^\circ}$. Without loss of generality, we may assume $q>b_{(t)}$, as we later consider the asymptotics of $q\to\infty$. For each $q>b_{(t)}$ and $\delta >0$ by Markov's and Hölder's inequalities, it holds that
\begin{equation*} \label{eq: Brockwell full}
\begin{aligned}
    &\text{P}^* \left( \left\| \sum_{\ell,r}
    \left( X^*_{\ell r, N^\circ T^\circ} - X^{* (q)}_{\ell r, N^\circ T^\circ} \right) \right\| > \delta \right) 
    \leq \frac{ \text{E}^* \left\| \sum_{\ell,r} \left( X^*_{\ell r, N^\circ T^\circ} - X^{* (q)}_{\ell r, N^\circ T^\circ} \right) \right\| }{\delta} \\
    &\leq \frac{1}{\delta} \left\| \Sigma_{ZZ}^{-1} \right\|
    \text{E}^* 
    \left\| 
    \frac{1}{\sqrt{NT}} \sum_{\ell=1}^{N^\circ}
    \sum_{m=1}^{b_{(i)}} 
    \sum_{j=q+1}^T \sum_{r=1}^{T^\circ} \sum_{s=1}^{b_{(t)}} 
    \big(C_{i(\ell,m),j} \otimes I \big)
    \operatorname{vec}\!\left( \tilde u^{\ast}_{\substack{i(\ell,m),\\t(r,s)}}
    \,\tilde u^{\ast\prime}_{\substack{i(\ell,m),\\t(r,s)-j}} \right) \right\| \\
    &+ \frac{1}{\delta} 
    \text{E}^* \left\| \frac1{\sqrt{N}} \sum_{\ell,m}
    \left( \left(
    \Sigma_{ZZ}^{-1} - \Sigma_{ZZ,[i^*_\ell+m]\bmod N}^{-1} \right) \otimes I \right)
    \frac{1}{\sqrt{T}} 
    \sum_{j=q+1}^T \sum_{r,s} 
    \big(C_{i(\ell,m),j} \otimes I \big)
    \operatorname{vec}\!\left( \tilde u^{\ast}_{\substack{i(\ell,m),\\t(r,s)}}
    \,\tilde u^{\ast\prime}_{\substack{i(\ell,m),\\t(r,s)-j}} \right) \right\|\\
    &+ \frac{1}{\delta} 
    \text{E}^* \left\| \frac1{\sqrt{N}} \sum_{\ell,m}
    \left( \left( \Sigma_{ZZ,[i^*_\ell+m]\bmod N}^{-1} - \Tilde \Sigma_{ZZ,i(\ell,m)}^{\ast -1} \right) \otimes I \right)
    \frac{1}{\sqrt{T}} 
    \sum_{j=q+1}^T \sum_{r,s} 
    \big(C_{i(\ell,m),j} \otimes I \big)
    \operatorname{vec}\!\left( \tilde u^{\ast}_{\substack{i(\ell,m),\\t(r,s)}}
    \,\tilde u^{\ast\prime}_{\substack{i(\ell,m),\\t(r,s)-j}} \right) \right\|\\
    &\eqqcolon R_{NT,q}^{(1)} + R_{NT,q}^{(2)} + R_{NT,q}^{(3)},
\end{aligned}
\end{equation*}
where $\Sigma_{ZZ,[i^*_\ell+m]\bmod N}$ denotes the variance of $Z_{[i^*_\ell+m]\bmod N,t}$ and $i^*_\ell\in\boldsymbol{i^*}$ is the starting index sampled for the $\ell$-th block in the cross-sectional dimension and $\Sigma_{ZZ}$ is the regressors second-moment matrix implied by the common VAR parameter \(\boldsymbol{\beta}\), and the common innovation covariance parameter \(\boldsymbol{\sigma}\) from Assumption \ref{as:random_parameters}. By the moment condition in Assumption \ref{as:cumulants} and a CLT for mixing random fields \citep[Thm. 1]{jenish2009}, 
$R_{NT,q}^{(1)}$ is $O_p(a_q)$ where
$a_q \sim \sum_{j>q}\sup_{i\in\mathbb N }\left\|C_{ij} \right\|$. Next we consider $R_{NT,q}^{(2)}$. For this, we abbreviate
\[
    H_{i,q}^*
    =
    \frac{1}{\sqrt{T}} 
    \sum_{j=q+1}^T \sum_{r=1}^{T^\circ} \sum_{s=1}^{b_{(t)}} 
    \big(C_{ij} \otimes I \big)
    \operatorname{vec}\!\left( \tilde u^{\ast}_{i,t(r,s)}
    \,\tilde u^{\ast\prime}_{i,t(r,s)-j} \right)
\]
and observe, by the law of iterated expectation and Jensen's inequality,
\begin{align*}
    \left(R_{NT,q}^{(2)}\right)^2
    &\leq
    \frac{1}{\delta^2} \text{E}^* \left\| \frac{1}{\sqrt{N}} \sum_{\ell=1}^{N^\circ}
    \sum_{m=1}^{b_{(i)}} 
    \left( \left( \Sigma_{ZZ}^{-1} - \Sigma_{ZZ,[i^*_\ell+m]\bmod N}^{-1} \right) \otimes I \right) H_{i(\ell,m),q}^* \right\|^2 \\
    &= 
    \frac{1}{\delta^2} \text{E}^* \left[ \text{E}^*_t \left\|
    \frac{1}{\sqrt{N}} \sum_{\ell=1}^{N^\circ}
    \sum_{m=1}^{b_{(i)}} 
    \left( \left( \Sigma_{ZZ}^{-1} - \Sigma_{ZZ,[i^*_\ell+m]\bmod N}^{-1} \right) \otimes I \right) H_{i(\ell,m),q}^* \right\|^2 \right],
\end{align*}
where $\text{E}^*_t [\cdot] \coloneqq \text{E}^* [\cdot \mid \boldsymbol{i^*}]$ denotes the expectation w.r.t the bootstrap measure given the set of all drawn cross-sectional indices $\boldsymbol{i^*}$ from \eqref{eq:Resampling}, which is equivalent to the set of block starting indices $\left( i_\ell: 1\leq \ell\leq N^\circ \right)$ given. We compute the inner conditional expectation only, integrability w.r.t $P^*$ follows. 
\begin{align*}
    &\text{E}^*_t \left\|
    \frac{1}{\sqrt{N}} \sum_{\ell=1}^{N^\circ}
    \sum_{m=1}^{b_{(i)}} 
    \left( \left( \Sigma_{ZZ}^{-1} - \Sigma_{ZZ,[i^*_\ell+m]\bmod N}^{-1} \right) \otimes I \right) H_{i(\ell,m),q}^* \right\|^2 \\
    &\leq
    \frac1N \sum_{\ell_1,\ell_2=1}^{N^\circ}
    \sum_{m_1,m_2=1}^{b_{(i)}} 
    \left\|\Sigma_{ZZ}^{-1} - \Sigma_{ZZ,[i^*_{\ell_1}+m_1]\bmod N}^{-1} \right\|
    \left\|\Sigma_{ZZ}^{-1} - \Sigma_{ZZ,[i^*_{\ell_2}+m_2]\bmod N}^{-1} \right\|
    \text{E}^*_t \left\| H_{i(\ell_1,m_1),q}^* H_{i(\ell_2,m_2),q}^{* \prime} \right\| \\
    &\lesssim
    \frac1N \sum_{\ell_1,\ell_2=1}^{N^\circ}
    \sum_{m_1,m_2=1}^{b_{(i)}} 
    \left\|\Sigma_{ZZ} - \Sigma_{ZZ,[i^*_{\ell_1}+m_1]\bmod N} \right\|
    \left\|\Sigma_{ZZ} - \Sigma_{ZZ,[i^*_{\ell_2}+m_2]\bmod N} \right\|
    \text{E}^*_t \left\| H_{i(\ell_1,m_1),q}^* H_{i(\ell_2,m_2),q}^{* \prime} \right\| \\
    &\lesssim
    \left(
    \frac1N \sum_{\ell=1}^{N^\circ}
    \sum_{m=1}^{b_{(i)}} 
    \left\|\Sigma_{ZZ} - \Sigma_{ZZ,[i^*_\ell+m]\bmod N} \right\|^2
    \right)
    \left(
    \sup_{i\leq N} \sum_{k=1}^N \text{E}^*_t \left\| H_{i,q}^* H_{k,q}^{* \prime} \right\|
    \right),
\end{align*}
where the second-to-last inequality holds by the inverse perturbation bound together with Assumption \ref{as:SVAR}~(i). By serial uncorrelatedness, we have $\text{E}^*_t H_{i,q}^* =0$ uniformly over $i\in \mathbb N_+$. With Assumption \ref{as:cumulants} and the definition of the conditional expectation $P^*$, 
\[
    \sup_{i\leq N} \sum_{k=1}^N \text{E}^*_t \left\| H_{i,q}^* H_{k,q}^{* \prime} \right\|
    = O_p(a_q^2).
\]
For a stable VAR(\(p\)) process, the map $(\boldsymbol{\beta}, \boldsymbol{\sigma})\mapsto \Sigma_{ZZ}$ is locally Lipschitz \citep[ch. 2]{Lutkepohl2005} and thus, under Assumption~\ref{as:random_parameters}, the average squared parameter
heterogeneity satisfies
\begin{align*}
    \frac1N \sum_{i=1}^{N}
    \left\|\Sigma_{ZZ} - \widehat{\Sigma}_{ZZ,i} \right\|^2
    \lesssim
    \frac1N\sum_{i=1}^N
    \left(
    \left\|\boldsymbol{\beta}_i-\boldsymbol{\beta}\right\|^2
    +
    \left\|\boldsymbol{\sigma}_i-\boldsymbol{\sigma}\right\|^2
    \right)
    =
    \left\|\Sigma_{\eta_\beta}\right\|
    +
    \left\|\Sigma_{\eta_\sigma}\right\| < \infty,
\end{align*}
applying a suitable law of large numbers to the random-parameter sequences. Thus, taken together, $R_{NT,q}^{(2)} = O_p(a_q)$. Finally, consider $R_{NT,q}^{(3)}$:
\begin{align*}
    R_{NT,q}^{(3)} 
    \lesssim
    \sqrt{N} \left( \frac1N \sum_{\ell=1}^{N^\circ} \sum_{m=1}^{b_{(i)}} \text{E}^*\left\| \widetilde\Sigma_{ZZ,[i_\ell^* +m]\bmod N}^{* -1}-\Sigma_{ZZ,i(\ell,m)}^{-1}\right\|^2 \right)^\frac{1}{2} 
    \left( \frac1N \sum_{\ell=1}^{N^\circ}
    \sum_{m=1}^{b_{(i)}} \text{E}^*\left\| H_{i(\ell,m),q}^* \right\|^2 \right)^\frac{1}{2}
\end{align*}
It holds that
\[
    \frac1N \sum_{\ell=1}^{N^\circ} \sum_{m=1}^{b_{(i)}}
    \text{E}^*
    \left\|\widetilde\Sigma_{ZZ,[i_\ell^* +m]\bmod N}^{* -1}-\Sigma_{ZZ,i(\ell,m)}^{-1}\right\|^2
    =
    O_p\left(\frac1T\right).
\]
because, for $\widetilde\Sigma_{ZZ,i}^* = T^{-1} \sum_{t=1}^T \widetilde{\boldsymbol{z}}_{it}^* \widetilde{\boldsymbol{z}}_{it}^{*\prime}$, the conditional second moment for each entry $f,g\in\{1, \dots, Kp\}$ is given as
\[
    \frac1{T^2}
    \sum_{t,s}
    \operatorname{Cov}^*
    \left(
    \widetilde z_{it,f}^* \widetilde z_{it,g}^*,
    \widetilde z_{is,f}^* \widetilde z_{is,g}^*
    \right).
\]
By absolute summability of fourth-order cumulants in Assumption \ref{as:cumulants},
\[
    \sum_s
    \left|
    \operatorname{Cov}^*
    \left(
    \widetilde z_{it,f}^* \widetilde z_{it,g}^*,
    \widetilde z_{is,f}^* \widetilde z_{is,g}^*
    \right)
    \right|
    =
    O_p(1),
\]
uniformly over $i\in\mathbb N_+$ and $t\in\mathbb Z$. As in the proof of $R_{NT,q}^{(2)}$, $N^{-1} \sum_{i=1}^N  \text{E}^*\left\| H_{iq}^* \right\|^2 = O_{p}(a_q^2).$
Taken together, $R_{NT,q}^{(3)} = O_p\left( \sqrt{\frac{N}{T}} a_q\right).$
Thus, as $N/T \to 0$ and $a_q\to 0$ as $q\to \infty$, $R_{NT,q}$ approaches $0$ as $q\to \infty$.

\paragraph{(a)} For this part, we assume $q<\infty$ to be a fixed integer. Analogously to \citet{BruggemannEtAl2016}, the summation in $X^{\ast(q)}_{\ell r,N^\circ T^\circ}$ can be split up according to the bootstrap blocking, i.e. with respect to $\tilde u^*_{it}$ and $\tilde u^*_{i,t-j}$ lying in the same or different blocks. Assuming that $N,T$ are sufficiently large such that $b_{(t)}>q$, at most one block boundary can be crossed when shifting by \(j\). For \(s\in \{1,\dots,b_{(t)}\} \), both indices $\tilde u^*_{i,t(r,s)} $ and $\tilde u^*_{i,t(r,s)-j}$ lie in the same block iff the latter stays in block \(r\).
Hence we consider the following decomposition: 
\begin{align*}
    &\sum_{\ell =1}^{N^\circ} \sum_{r=1}^{T^\circ} X^{\ast(q)}_{\ell r,N^\circ T^\circ}
    = \sum_{\ell, r} \frac{1}{\sqrt{NT}}
    \left\{ 
    \begin{pmatrix}
        \sum_{m=1}^{b_{(i)}} \sum_{j=1}^q \sum_{s=1}^{j} \left( \Tilde{\Sigma}_{ZZ,i(\ell,m)}^{\ast -1} \otimes I \right) \big(C_{i(\ell,m),j} \otimes I \big)
        \operatorname{vec}\!\left(
        \tilde u^{\ast}_{\substack{i(\ell,m),\\t(r,s)}}
        \,\tilde u^{\ast\prime}_{\substack{i(\ell,m),\\t(r,s)-j}}
        \right) \\
        0_{K(K+1)/2} \\
        0_{KL}
    \end{pmatrix} \right. \\
    &\quad +
    \begin{pmatrix}
        \sum_{m=1}^{b_{(i)}} \sum_{j=1}^q \sum_{s=j+1}^{b_{(t)}} \left( \Tilde{\Sigma}_{ZZ,i(\ell,m)}^{\ast -1} \otimes I \right) \big(C_{i(\ell,m),j} \otimes I \big)
        \operatorname{vec}\!\left(
        \tilde u^{\ast}_{\substack{i(\ell,m),\\t(r,s)}}
        \,\tilde u^{\ast\prime}_{\substack{i(\ell,m),\\t(r,s)-j}}
        \right) - \text{E}^\ast [\cdot]\\
        \sum_{m=1}^{b_{(i)}} \sum_{s=1}^{b_{(t)}} L_K \operatorname{vec}\!\left(
        \tilde u^{\ast}_{\substack{i(\ell,m),\\t(r,s)}}
        \,\tilde u^{\ast\prime}_{\substack{i(\ell,m),\\t(r,s)}}
        \right) - \text{E}^\ast [\cdot] \\
        \sum_{m=1}^{b_{(i)}} \sum_{s=1}^{b_{(t)}} \vec \left( \operatorname{chol} \left( \Tilde{\Sigma}_{u,i(\ell,m)}^* \right)^{-1}  \tilde{u}^*_{\substack{i(\ell,m),\\t(r,s)}} \tilde{m}^{*\prime}_{\substack{i(\ell,m),\\t(r,s)}} \right) - \text{E}^\ast [\cdot]
    \end{pmatrix} \\
    &\quad +
    \left. \begin{pmatrix}
        \sum_{m=1}^{b_{(i)}} \sum_{j=1}^q \sum_{s=j+1}^{b_{(t)}} \text{E}^\ast \left[ \left( \Tilde{\Sigma}_{ZZ,i(\ell,m)}^{\ast -1} \otimes I \right) \big(C_{i(\ell,m),j} \otimes I \big)
        \operatorname{vec}\!\left(
        \tilde u^{\ast}_{\substack{i(\ell,m),\\t(r,s)}}
        \,\tilde u^{\ast\prime}_{\substack{i(\ell,m),\\t(r,s)-j}}
        \right) \right] \\
        \sum_{m=1}^{b_{(i)}} \sum_{s=1}^{b_{(t)}} L_K \text{E}^\ast \left[ \operatorname{vec}\!\left(
        \tilde u^{\ast}_{\substack{i(\ell,m),\\t(r,s)}}
        \,\tilde u^{\ast\prime}_{\substack{i(\ell,m),\\t(r,s)}} - u_{\substack{i(\ell,m),\\t(r,s)}}
        \, u^{\prime}_{\substack{i(\ell,m),\\t(r,s)}} \right) \right] \\
        \sum_{m=1}^{b_{(i)}} \sum_{s=1}^{b_{(t)}} \text{E}^\ast \left[ \vec \left( \operatorname{chol} \left( \Tilde{\Sigma}_{u,i(\ell,m)}^* \right)^{-1}  \tilde{u}^*_{\substack{i(\ell,m),\\t(r,s)}} \tilde{m}^{*\prime}_{\substack{i(\ell,m),\\t(r,s)}} - \operatorname{chol} \left( \Tilde{\Sigma}_{u,i(\ell,m)} \right)^{-1}  u_{\substack{i(\ell,m),\\t(r,s)}} m_{\substack{i(\ell,m),\\t(r,s)}}' \right) \right]
    \end{pmatrix} \right\} \\
    &\eqqcolon A_1^\ast + A_2^\ast + A_3.
\end{align*}
We show \textbf{(i)} $A_1^\ast\to 0$ and \textbf{(ii)} $A_2^\ast\overset{d}{\longrightarrow}\mathcal{N}(0,\Tilde V_q)$ in probability w.r.t $P^*$ as well as \textbf{(iii)} $A_3\to 0$ in probability w.r.t $\mathbb P$.

\paragraph{(i)} To apply Markov's inequality, we show $\text{E}^*\| A_1^\ast \|=o_p(1)$. 
Considering the non-zero entry in $A_1^*$ only, we proceed with the same arguments as in the proof of part \textbf{(c)}. The important distinction is that in this part of the proof, the in-block-position index $s$ only goes up until $j<q$ and $q<\infty$ is fixed:
\begin{equation*}
\begin{aligned}
    \text{E}^*\| A_1^\ast \|
    &\leq \left\| \Sigma_{ZZ}^{-1} \right\|
    \text{E}^* 
    \left\| 
    \frac{1}{\sqrt{NT}} \sum_{\ell=1}^{N^\circ}
    \sum_{m=1}^{b_{(i)}} 
    \sum_{j=1}^q \sum_{r=1}^{T^\circ} \sum_{s=1}^j 
    \big(C_{i(\ell,m),j} \otimes I \big)
    \operatorname{vec}\!\left( \tilde u^{\ast}_{\substack{i(\ell,m),\\t(r,s)}}
    \,\tilde u^{\ast\prime}_{\substack{i(\ell,m),\\t(r,s)-j}} \right) \right\| \\
    &+ 
    \text{E}^* \left\| \frac1{\sqrt{N}} \sum_{\ell,m}
    \left( \left(
    \Sigma_{ZZ}^{-1} - \Sigma_{ZZ,[i^*_\ell+m]\bmod N}^{-1} \right) \otimes I \right)
    \frac{1}{\sqrt{T}} 
    \sum_{j=1}^q \sum_{r=1}^{T^\circ} \sum_{s=1}^j 
    \big(C_{i(\ell,m),j} \otimes I \big)
    \operatorname{vec}\!\left( \tilde u^{\ast}_{\substack{i(\ell,m),\\t(r,s)}}
    \,\tilde u^{\ast\prime}_{\substack{i(\ell,m),\\t(r,s)-j}} \right) \right\|\\
    &+ 
    \text{E}^* \left\| \frac1{\sqrt{N}} \sum_{\ell,m}
    \left( \left( \Sigma_{ZZ,[i^*_\ell+m]\bmod N}^{-1} - \Tilde \Sigma_{ZZ,i(\ell,m)}^{\ast -1} \right) \otimes I \right)
    \frac{1}{\sqrt{T}} 
    \sum_{j=1}^q \sum_{r=1}^{T^\circ} \sum_{s=1}^j 
    \big(C_{i(\ell,m),j} \otimes I \big)
    \operatorname{vec}\!\left( \tilde u^{\ast}_{\substack{i(\ell,m),\\t(r,s)}}
    \,\tilde u^{\ast\prime}_{\substack{i(\ell,m),\\t(r,s)-j}} \right) \right\|\\
    &\eqqcolon A_{11} + A_{12} + A_{13},
\end{aligned}
\end{equation*}
and $A_{11}$, $A_{12} = O_p\left(b_{(t)}^{-1}\right)$ and $A_{13}= O_p\left(\sqrt{\frac{N}{T}}\frac1{b_{(t)}}\right)$.

\paragraph{(iii)} Convergence of the first two components follows analogously to \citet[Lemma A.3]{BruggemannEtAl2016}. Thus, we only provide details on the third component and write
\begin{align*}
    &\frac{1}{\sqrt{NT}} \sum_{\ell, r} \sum_{m=1}^{b_{(i)}} \sum_{s=1}^{b_{(t)}} \text{E}^\ast \Biggl[ \vec \Biggl( \operatorname{chol} \left(\Tilde{\Sigma}_{u,i(\ell,m)}^{\ast}\right)^{-1} \tilde u^{\ast}_{\substack{i(\ell,m),\\t(r,s)}} \, m^{\ast\prime}_{\substack{i(\ell,m),\\t(r,s)}} 
    - \operatorname{chol} \left(\Tilde{\Sigma}_{u,i(\ell,m)}\right)^{-1}  u_{\substack{i(\ell,m),\\t(r,s)}} \, m^{\prime}_{\substack{i(\ell,m),\\t(r,s)}} \Biggr)
    \Biggr] \\
    &= \frac{1}{\sqrt{NT}} \sum_{\ell, r} \sum_{m,s}
    \Biggl\{
    \text{E}^\ast \left[ \left( I \otimes \operatorname{chol} \left( \Tilde{\Sigma}^*_{u,i(\ell,m)}\right)^{-1} \right) \vec \left(  \tilde u^{\ast}_{\substack{i(\ell,m),\\t(r,s)}} \, m^{\ast\prime}_{\substack{i(\ell,m),\\t(r,s)}}  
    - u_{\substack{i(\ell,m),\\t(r,s)}} \, m^{\prime}_{\substack{i(\ell,m),\\t(r,s)}} \right) \right] \\
    & + \left( I \otimes \left( \text{E}^\ast \left[ \operatorname{chol} \left( \Tilde{\Sigma}_{u,i(\ell,m)}^{\ast}\right)^{-1}\right] - \operatorname{chol} \left( \Tilde{\Sigma}_{u,i(\ell,m)} \right)^{-1} \right) \right) \vec \left( u_{\substack{i(\ell,m),\\t(r,s)}} \, m^{\prime}_{\substack{i(\ell,m),\\t(r,s)}} \right) \Biggr\} \\
    &\eqqcolon A_{31} + A_{32}.
\end{align*}
By exact cancellation of the summands, the mean and variance of $A_{31}$ are of order $O(NT)^{-1/2}$ and $O(NT)^{-1}$, respectively. We now provide proof for $A_{32}=o_p(1)$. We use the structural decomposition in Assumption~\ref{as:SVAR}(iv), which states $\mathsf{B}_i=\mathsf{B}_i^\circ\mathsf{Q}$ with $\mathsf{B}_i^\circ=\operatorname{chol}(\Sigma_{u,i})$ and a homogeneous $\mathsf{Q}$. 
Assumption \ref{as:mixing_conditions}~(i) further implies, uniformly in $i\leq N$, 
\begin{align*}
    \frac{1}{T}\sum_{r,s} 
    u_{i,t(r,s)} \, m^{\prime}_{i,t(r,s)} 
    =
    \mathsf{B}_i^\circ\mathsf{Q}\Sigma_{\epsilon m,i}
    +O_p(T^{-1/2}).
\end{align*}
Moreover, the continuity of the Cholesky factor and Assumption
\ref{as:random_parameters}~(ii) imply
\begin{align*}
    \frac{1}{N}\sum_{i=1}^N
    \left\|\mathsf{B}_i^\circ-\mathsf{B}^\circ\right\|^2
    =
    O_p\left(\left\|\Sigma_{\eta_\sigma}\right\|\right),
\end{align*}
and the same rate holds for the inverse Cholesky factors by the inverse perturbation bound (uniform non-singularity and stability hold by Assumption \ref{as:SVAR}~(i) and (iii)). 
Consequently, by definition of the conditional expectation and consistency of the covariance estimators,
\begin{align*}
    A_{32}
    =
    \sqrt{\frac{T}{N}}
    \sum_{\ell,m}
    \left(
    I\otimes
    \left(
    \frac{1}{N}\sum_{i=0}^{N-1}\mathsf{B}_{[i+m]\bmod N}^{\circ-1}
    -
    \mathsf{B}_{i(\ell,m)}^{\circ-1}
    \right)
    \right)
    \operatorname{vec}
    \left(
    \mathsf{B}_{i(\ell,m)}^\circ
    \mathsf{Q}
    \Sigma_{\epsilon m,i(\ell,m)}
    \right)
    +
    O_p\left(
    \left\|\Sigma_{\eta_\sigma}\right\|^\frac12
    \right)
    +o_p(1).
\end{align*}
A first-order expansion around $\mathsf{B}^\circ$ yields
\begin{align*}
    \mathsf{B}_{i(\ell,m)}^{\circ-1}
    =
    \mathsf{B}^{\circ-1}
    -
    \mathsf{B}^{\circ-1}
    \left(
    \mathsf{B}_{i(\ell,m)}^\circ-\mathsf{B}^\circ
    \right)
    \mathsf{B}^{\circ-1}
    +
    O_p\left(
    \left\|
    \mathsf{B}_{i(\ell,m)}^\circ-\mathsf{B}^\circ
    \right\|^2
    \right).
\end{align*}
Since the block positions $(\ell,m)$ enumerate the $N$ cross-sectional units, assuming independence between $\boldsymbol{\eta}_{\sigma i}$ and $(\epsilon_{it}', m_{it}')'$ implies that the cross-sectional average of the first-order term is $O_p(N^{-1/2}\|\Sigma_{\eta_\sigma}\|^\frac12)$, whereas its second-order remainder is $O_p(\|\Sigma_{\eta_\sigma}\|)$. The homogeneous matrix $\mathsf{Q}$ factors out of both terms. Hence,
\begin{align*}
    \|A_{32}\|
    =
    O_p\left(
    \sqrt{T\left\|\Sigma_{\eta_\sigma}\right\|}
    +
    \sqrt{NT}\left\|\Sigma_{\eta_\sigma}\right\|
    \right)
    +o_p(1).
\end{align*}
Under $N/T\to 0$ and $\left\|\Sigma_{\eta_\sigma}\right\| = o\left(T^{-1}\right)$, both terms vanish and $A_{32}=o_p(1)$.


\paragraph{(ii)} Finally, we establish asymptotic normality of \(A_2^*\).
Because the PBB algorithm preserves dependence along both the cross-sectional and temporal dimensions, the summands in \(A_2^*\) are dependent whenever their associated blocks share a starting index in either dimension. Thus we apply a central limit theorem for an array of two-way clustered random fields. Crucially, the data (i.e. block sums) generated by the PBB algorithm satisfy the property of separate exchangeability, as their joint distribution is invariant under arbitrary separate permutations of the row-block and column-block indices. We follow the framework and result from \cite{CameronEtAl2011} and \cite{Menzel2021} and prove the following requirements, denoting the summands of $A_2^*$ by $\left(\Tilde X^{\ast (q)}_{\ell r, N^\circ T^\circ}: \ell\leq N^\circ, r\leq T^\circ\right)$, where $\text{E}^* \left[ \Tilde{X}^{\ast (q)}_{\ell r,N^\circ T^\circ} \right] =0$ by construction:

\begin{enumerate} [label=(\Alph{enumi})]
    \item \textbf{Covariance-consistent estimator: } \\
    We have, in probability w.r.t $\mathbb P$,
    \[
    \sum_{\ell_1,\ell_2=1}^{N^\circ} \sum_{r_1,r_2=1}^{T^\circ} \text{E}^*\left[\Tilde X^{\ast (q)}_{\ell_1r_1, N^\circ T^\circ} \Tilde X^{\ast (q) \prime}_{\ell_2r_2, N^\circ T^\circ}\right] \to \Tilde{V}_q, 
    \]
    where
    \[
    \tilde V_q = 
    \begin{pmatrix}
        \tilde V_q^{(1,1)} & \tilde V_q^{(2,1) \prime} & \tilde V_q^{(3,1) \prime}\\
        \tilde V_q^{(2,1)} & \tilde V^{(2,2)} & \tilde V^{(3,2) \prime} \\
        \tilde V_q^{(3,1)} & \tilde V^{(3,2)} & \tilde V^{(3,3)}
    \end{pmatrix},
    \]
    with
    \begin{equation} \label{eq: Tilde Vq}
    \begin{aligned}
        \Tilde{V}_q^{(1,1)} 
        &= \lim_{N,T\to\infty} \frac{1}{NT} \sum_{h_{(i)}=-(N-1)}^{N-1} \sum_{i=\max\{1,1+h_{(i)}\}}^{\min\{N,N+h_{(i)}\}} \sum_{h_{(t)}=-(T-1)}^{T-1} \sum_{t=\max\{1,1+h_{(t)}\}}^{\min\{T,T+h_{(t)}\}} \mathbb E \Bigg[ \left(\widehat{\Sigma}_{ZZ,i}^{-1}\otimes I\right) \Big[ \sum_{j,j'=1}^q (C_{ij}\otimes I)\\
        & \quad \times \vec\left(u_{it} u_{i,t-j}'\right) \vec\left(u_{i+h_{(i)},t+h_{(t)}}u_{i+h_{(i)},t+h_{(t)}-j'}'\right)' \left(C_{i+h_{(i)},j'}\otimes I\right)' \Big] \left( \widehat \Sigma_{ZZ,i+h_{(i)}}^{-1}\otimes I\right)' \Bigg] \\ 
        \Tilde{V}_q^{(2,1)} 
        & = L_K \Bigg[ \lim_{N,T\to\infty} \frac{1}{NT} \sum_{h_{(i)}=-(N-1)}^{N-1} \sum_{i=\max\{1,1+h_{(i)}\}}^{\min\{N,N+h_{(i)}\}} \sum_{h_{(t)}=-(T-1)}^{T-1} \sum_{t=\max\{1,1+h_{(t)}\}}^{\min\{T,T+h_{(t)}\}} \sum_{j=1}^q \mathbb E \Bigg[ \vec\left(u_{it} u_{it}'\right) \\
        & \quad \times \vec\left(u_{i+h_{(i)},t+h_{(t)}} u_{i+h_{(i)},t+h_{(t)}-j}'\right)' \left(C_{i+h_{(i)},j}\otimes I\right)' \left(\widehat\Sigma_{ZZ,i+h_{(i)}}^{-1}\otimes I\right)' \Bigg] \Bigg] \\
        \Tilde{V}_q^{(3,1)} 
        & = \lim_{N,T\to\infty} \frac{1}{NT} \sum_{h_{(i)}=-(N-1)}^{N-1} \sum_{i=\max\{1,1+h_{(i)}\}}^{\min\{N,N+h_{(i)}\}} \sum_{h_{(t)}=-(T-1)}^{T-1} \sum_{t=\max\{1,1+h_{(t)}\}}^{\min\{T,T+h_{(t)}\}} \sum_{j=1}^q \mathbb E \Bigg[ \operatorname{chol}\left( \Tilde \Sigma_{u,i} \right)^{-1} \\
        & \quad \times \vec \left(u_{it} m_{it}'\right) \vec\left(u_{i+h_{(i)},t+h_{(t)}} u_{i+h_{(i)},t+h_{(t)}-j}'\right)' \left(C_{i+h_{(i)},j}\otimes I\right)' \left(\widehat \Sigma_{ZZ,i+h_{(i)}}^{-1}\otimes I\right)' \Bigg]
    \end{aligned}
    \end{equation}
    and $\Tilde V^{(2,2)}$, $\Tilde V^{(3,3)}$ and $\Tilde V^{(3,2)}$ are defined in the proof of Theorem \ref{th:AsyMG}.
    \item \textbf{Integrability / Non-degenerate variance: } \\
    Recall $\text{E}^*\left[\Tilde{X}^{\ast (q)}_{\ell r,N^\circ T^\circ}\right]=0$ and consider the Hoeffding projection expansion \citep{CameronEtAl2011, Menzel2021} for separately exchangeable arrays. Define the random variables
    \vspace{-15pt}
    \begin{align*}
        a_{\ell, N^\circ T^\circ} &\coloneqq \text{E}^*[\Tilde{X}^{\ast (q)}_{\ell r,N^\circ T^\circ} | \boldsymbol{i^*} ], \\
        g_{r, N^\circ T^\circ} &\coloneqq \text{E}^*[\Tilde{X}^{\ast (q)}_{\ell r,N^\circ T^\circ} | \boldsymbol{t^*}], \\
        v_{\ell r, N^\circ T^\circ} &\coloneqq \text{E}^*[\Tilde{X}^{\ast (q)}_{\ell r,N^\circ T^\circ} | \boldsymbol{i^*}, \boldsymbol{t^*}] - a_{\ell, N^\circ T^\circ} - g_{r, N^\circ T^\circ}, \\
        e_{\ell r, N^\circ T^\circ} &\coloneqq \Tilde{X}^{\ast (q)}_{\ell r,N^\circ T^\circ} - \text{E}^*[\Tilde{X}^{\ast (q)}_{\ell r,N^\circ T^\circ} | \boldsymbol{i^*}, \boldsymbol{t^*}],
    \end{align*}
    where $\boldsymbol{i^*}$ and $\boldsymbol{t^*}$ are the sets of drawn indices as in \eqref{eq:Resampling}. Further, let
    \begin{align*}
        \sigma^2_a \coloneqq \operatorname{Var}^*(a_{\ell, N^\circ T^\circ}), \quad 
        \sigma^2_g \coloneqq \operatorname{Var}^*(g_{r, N^\circ T^\circ}), \quad
        \sigma^2_v \coloneqq \operatorname{Var}^*(v_{\ell r, N^\circ T^\circ}), \quad
        \sigma^2_e \coloneqq \operatorname{Var}^*(e_{\ell r, N^\circ T^\circ}).
    \end{align*}
    It holds $T^{\circ 2} N^\circ \sigma_a^2=O_p(1)$ and $N^{\circ 2} T^{\circ} \sigma_g^2=O_p(1)$. 
    Further, the random variables 
    \[
    a_{\ell, N^\circ T^\circ}/\sigma_a, \qquad g_{r, N^\circ T^\circ}/\sigma_g, \qquad v_{\ell r, N^\circ T^\circ}/\sigma_v, \qquad e_{\ell r, N^\circ T^\circ}/\sigma_e,
    \]    
    have bounded moments up to order $(4+\delta)$ for some $\delta>0$, whenever the respective variances $\sigma^2_a$, $\sigma^2_g$, $\sigma^2_v$, $\sigma^2_e$ are non-zero.
    \item \textbf{Spectral representation: } $\left(v_{\ell r,N^\circ T^\circ}: \ell\leq N^\circ, r\leq T^\circ \right)$ admits a spectral representation, where the sequence of singular values are uniformly bounded by a null sequence and the first three moments of the eigenfunctions are bounded.
\end{enumerate}

Conditions (A) and (B) taken together refer to Assumption 2.1 (Integrability) in \cite{Menzel2021}. With (A), we further observe that the bootstrap algorithm replicates the joint distribution from Theorem \ref{th:AsyMG}. Condition (B) implies that the correctly scaled variance contributions of \(a_{\ell,N^\circ T^\circ}\) and \(g_{r,N^\circ T^\circ}\) remain bounded asymptotically (note that in contrast to \cite{Menzel2021} we consider normalized summands in the CLT). Therefore, any non-negligible limits of these projection components are included in the long-run covariance matrix \(\widetilde V_q\) and therefore affect the value of \(\widetilde V_q\), but not the \(\sqrt{NT}\)-normalization. A slower convergence rate on the other hand would arise if one of the corresponding scaled projection variances diverged. Condition (C) refers to Assumption 2.2 (Spectral Representation) in \cite{Menzel2021}. This condition rules out the case in which a fixed low-rank interaction component remains first order, which would otherwise lead to a generally non-Gaussian limit \citep[p.11]{Menzel2021}.

\textbf{(A)} While convergence is proven block by block, we present the argument for one representative case only (which is $\ell_1=\ell_2\equiv \ell$, $r_1=r_2\equiv r$, i.e. all in one block). The remaining cases are analogous: any additional blocks, and hence additional summands, are controlled by Assumption \ref{as:cumulants}. For the first block, consider 
\begin{equation} \label{eq: V11 converging term}
\begin{aligned} 
    &\frac{1}{NT} \sum_{\ell=1}^{N^\circ} \sum_{r=1}^{T^\circ} \sum_{m_1,m_2=1}^{b_{(i)}} \sum_{j_1, j_2=1}^q \sum_{s_1=j_1+1}^{b_{(t)}} \sum_{s_2=j_2+1}^{b_{(t)}} \text{E}^* \Biggl[ \left(\Tilde{\Sigma}_{ZZ,i(\ell,m_1)}^{* -1} \otimes I \right) \left( C_{i(\ell,m_1), j_1} \otimes I\right) \vec \left( \tilde{u}^*_{\substack{i(\ell,m_1) \\ t(r,s_1)}} \tilde{u}^{*\prime}_{\substack{i(\ell,m_1) \\ t(r,s_1)-j_1}} \right) \\
    &\quad \times \vec \left( \tilde{u}^*_{\substack{i(\ell,m_2) \\ t(r,s_2)}} \tilde{u}^{*\prime}_{\substack{i(\ell,m) \\ t(r,s_2)-j_2}} \right)' \left( C_{i(\ell,m_2), j_2} \otimes I\right)' \left(\Tilde{\Sigma}_{ZZ,i(\ell,m_2)}^{* -1} \otimes I \right)' \Biggr].
\end{aligned}
\end{equation}
Noting that $\Tilde{\Sigma}_{ZZ,i}^{*}$ converges in probability w.r.t $P^*$ and is nonsingular uniformly in $i\leq N$, 
\begin{equation*} 
\begin{aligned}
    &\frac{1}{NT} \sum_{\ell=1}^{N^\circ} \sum_{r=1}^{T^\circ} \sum_{m_1,m_2}\sum_{j_1, j_2} \sum_{s_1} \sum_{s_2}
    \text{E}^* \Biggl\| \left( C_{i(\ell,m_1), j_1} \otimes I\right) \vec \left( \tilde{u}^*_{\substack{i(\ell,m_1) \\ t(r,s_1)}} \tilde{u}^{*\prime}_{\substack{i(\ell,m_1) \\ t(r,s_1)-j_1}} \right) \vec \left( \tilde{u}^*_{\substack{i(\ell,m_2) \\ t(r,s_2)}} \tilde{u}^{*\prime}_{\substack{i(\ell,m_2) \\ t(r,s_2)-j_2}} \right)' \left( C_{i(\ell,m_2), j_2} \otimes I\right)' \Biggr\| \\
    &=\frac{1}{b_{(i)}b_{(t)}} 
    \sum_{j_1, j_2} 
    \sum_{h_{(i)}=-(b_{(i)}-1)}^{b_{(i)}-1} \sum_{h_{(t)}=-(b_{(t)}-1)}^{b_{(t)}-1} 
    \sum_{i=\max\{1,1-h_{(i)}\}}^{\min \{N+b_{(i)}-1,N+b_{(i)}-h_{(i)}-1\}} \sum_{t=\max\{1,1-h_{(t)}\}}^{\min \{T,T-h_{(t)}\}} \frac{O(b_{(i)} b_{(t)})}{N(T-b_{(t)}+1)}  \\
    &\quad \times \left\| C_{i\bmod N, j_1} \right\| \Biggl\| \vec \left( u^c_{\substack{i\bmod N \\ t}} u^{c\prime}_{\substack{[i+m]\bmod N \\ t-j_1}} \right)  \vec \left( u^c_{\substack{[i+h_{(i)}]\bmod N \\ t+h_{(t)}}} u^{c\prime}_{\substack{[i+h_{(i)}]\bmod N \\ t+h_{(t)}-j_2}} \right)' \Biggr\| \left\| C_{[i+h_{(i)}]\bmod N, j_2} \right\|,
\end{aligned}
\end{equation*}
is $O_p(1)$ by Assumption \ref{as:cumulants} and thus \eqref{eq: V11 converging term} converges in probability to $\Tilde V_q^{(1,1)}$ as $N,T\to\infty$.

\textbf{(B)} We first verify the existence of the displayed projection. Conditional on the observed sample, the row-block starting indices \((i_\ell)_{\ell=1}^{N^\circ}\) and the column-block starting indices \((t_r)_{r=1}^{T^\circ}\) are i.i.d. within their respective dimensions and independent across dimensions. Hence each component of \(\Tilde X^{\ast(q)}_{\ell r,N^\circ T^\circ}\) can be written as a function of the (under \(P^*\)) mutually independent arrays \((i_\ell)_\ell\), \((t_r)_r\), and \((\varepsilon_{\ell r})_{\ell r}\), where \((\varepsilon_{\ell r})_{\ell r}\) denotes the remaining cell-specific bootstrap randomness. The projection variables are defined as stated in (B). Next, we argue integrability and non-degeneracy. We consider $a_{\ell, N^\circ T^\circ}$ in detail and make use of the notation $\text{E}^*_t [\cdot] \coloneqq \text{E}^* [\cdot \mid \boldsymbol{i^*}]$ (in this case the specification $i_\ell \equiv i^*$, where $i^*$ is a random index) from part \textbf{(c)},
\begin{align*}
    a_{\ell, N^\circ T^\circ} 
    &= \frac{1}{\sqrt{NT}} \sum_{m=1}^{b_{(i)}} \sum_{j=1}^q
    \text{E}_t^\ast \left( 
    \begin{matrix}
        \left( \Tilde{\Sigma}_{ZZ,i(\ell,m)}^{\ast -1} \otimes I \right) \sum_{s=j+1}^{b_{(t)}} \big(C_{i(\ell,m),j} \otimes I\big)
        \operatorname{vec}\!\left(
        \tilde u^{\ast}_{i(\ell,m),\,t(r,s)}
        \,\tilde u^{\ast\prime}_{i(\ell,m),\,t(r,s)-j}
        \right) \\
        \sum_{s=1}^{b_{(t)}} L_K \operatorname{vec}\!\left(
        \tilde u^{\ast}_{i(\ell,m),\,t(r,s)}
        \,\tilde u^{\ast\prime}_{i(\ell,m),\,t(r,s)}
        \right) \\
        \left( I \otimes \operatorname{chol} \left( \Tilde{\Sigma}_{u,[i+m]\bmod N}^* \right)^{-1} \right) \sum_{s=1}^{b_{(t)}} \vec \left( \tilde{u}^*_{i(\ell,m),\,t(r,s)} m^{*\prime}_{i(\ell,m),\,t(r,s)} \right)
    \end{matrix} \right)
\end{align*}
Note that the conditional variance of $a_{\ell,N^\circ T^\circ}$ accounts only for the randomness in the cross-sectional dimension, as the time dimension has already been averaged out. For now, we consider the conditional variance $\sigma_{a,1}^2 \coloneqq \operatorname{Var}^* \left(a_{\ell,N^\circ T^\circ}^{(1)}\right)$ of its first component.
\[
\begin{aligned}
    \sigma_{a,1}^2
    &=
    \frac{1}{NT}
    \sum_{m_1,m_2=1}^{b_{(i)}}
    \sum_{j_1,j_2=1}^{q}
    \sum_{s_1=j_1+1}^{b_{(t)}}
    \sum_{s_2=j_2+1}^{b_{(t)}}
    \frac{1}{(T-b_{(t)}+1)^2}
    \sum_{\tau_1, \tau_2=0}^{T-b_{(t)}}
    \operatorname{Cov}^*
    \Bigg(
    \left(
    \widetilde\Sigma_{ZZ,i^*+m_1}^{*-1}
    \otimes I
    \right)
    \left(
    C_{i^*+m_1,j_1}
    \otimes I
    \right)
    \\
    &\quad\times
    \operatorname{vec}
    \left(
    u^c_{\substack{i^*+m_1 \\ \tau_1+s_1}}
    u^{c\prime}_{\substack{i^*+m_1 \\ \tau_1+s_1-j_1}}
    \right), \left(
    \widetilde\Sigma_{ZZ,i^*+m_2}^{*-1}
    \otimes I
    \right)
    \left(
    C_{i^*+m_2,j_2}
    \otimes I
    \right)
    \operatorname{vec}
    \left(
    u^c_{\substack{i^*+m_2 \\ \tau_2+s_2}}
    u^{c\prime}_{\substack{i^*+m_2 \\\tau_2+s_2-j_2}}
    \right)
    \Bigg) .
\end{aligned}
\]
By Assumptions \ref{as:SVAR}~(ii) and the cumulant condition in Assumption \ref{as:cumulants}, which is applied to the sums over relative effective index distances, i.e. over \(m_2-m_1\) and \((\tau_2+s_2)-(\tau_1+s_1)\). On the other hand, the finite lag sums over \(j_1,j_2\) and the in-block multiplicity over \(s_2\) are not absorbed by cumulant summability and must be kept explicitly, which leads to the rate $\sigma_a^2 = O_p\left( b_{(i)} b_{(t)}^2 N^{-1} T^{-2} \right)$ and thus $T^{\circ 2} N^\circ \sigma_a^2 = O_p(1)$. Similarly, $N^{\circ 2} T^\circ \sigma_g^2 = O_p(1)$.

We are left to prove the standardized moment conditions component by component. Instead of performing this step directly, we verify the equivalent Lyapunov-type requirements for the sums entering the statistic. For triangular arrays, such moment bounds imply the Lindeberg condition, \citep[see, e.g.,][Thm.~27.3] {Billingsley1995}. Let \(\lambda\) be an arbitrary real vector of suitable dimension with \(\|\lambda\|=1\), and define 
\[
    a_{\ell,\lambda}
    \coloneqq
    \text E^*\!\left[ \left.
    \lambda'
    \widetilde X_{\ell r,N^\circ T^\circ}^{*(q)}
    \right\vert \boldsymbol{i^*}
    \right],
    \qquad
    g_{r,\lambda}
    \coloneqq
    \text E^*\!\left[ \left.
    \lambda'
    \widetilde X_{\ell r,N^\circ T^\circ}^{*(q)}
    \right\vert \boldsymbol{t^*}
    \right],
    \qquad
    w_{\ell r,\lambda}
    \coloneqq
    \lambda'
    \widetilde X_{\ell r,N^\circ T^\circ}^{*(q)}
    -
    a_{\ell,\lambda}
    -
    g_{r,\lambda}.
\]
Then
\[
    \sum_{\ell=1}^{N^\circ}
    \sum_{r=1}^{T^\circ}
    \lambda'
    \widetilde X_{\ell r,N^\circ T^\circ}^{*(q)}
    =
    T^\circ\sum_{\ell=1}^{N^\circ}a_{\ell,\lambda}
    +
    N^\circ\sum_{r=1}^{T^\circ}g_{r,\lambda}
    +
    \sum_{\ell=1}^{N^\circ}
    \sum_{r=1}^{T^\circ}
    w_{\ell r,\lambda}.
\]
The row projection is controlled as the moment bounds for the block averages and the cumulant summability condition imply for $p=2+\delta$, $\delta>0$,
\[
    N^\circ (T^\circ)^p
    \text E^*\left\|a_{\ell,\lambda}\right\|^p
    =
    O_p\left(
    \frac{N}{b_{(i)}}
    \left(\frac{T}{b_{(t)}}\right)^p
    \left(
    \frac{b_{(i)}b_{(t)}^2}{NT^2}
    \right)^{p/2}
    \right)
    =
    O_p\left(
    \left(
    \frac{b_{(i)}}{N}
    \right)^{\delta/2}
    \right)
    =o_p(1).
\]
Analogously, the column projection satisfies
\[
    (N^\circ)^p T^\circ
    \text E^*\left\|g_{r,\lambda}\right\|^p
    =
    O_p\left(
    \left(
    \frac{b_{(t)}}{T}
    \right)^{\delta/2}
    \right)
    =o_p(1).
\]
Finally, for the remaining two-way component,
\[
    N^\circ T^\circ
    \text E^*\left\|w_{\ell r,\lambda}\right\|^p
    =
    O_p\left(
    \left(
    \frac{b_{(i)}b_{(t)}}{NT}
    \right)^{\delta/2}
    \right)
    =o_p(1).
\]
Thus, no individual row, column, or two-way block component dominates. Since the full conditional variance satisfies
\[
    \operatorname{Var}^*\left(
    \sum_{\ell=1}^{N^\circ}
    \sum_{r=1}^{T^\circ}
    \lambda'
    \widetilde X_{\ell r,N^\circ T^\circ}^{*(q)}
    \right)
    \overset{p}{\longrightarrow}
    \lambda'\widetilde V_q\lambda,
\]
with \(\lambda'\widetilde V_q\lambda<\infty\), the Cramér--Wold device yields the desired conditional Gaussian limit.

\textbf{(C)} Finally, we consider $v_{\ell r, N^\circ T^\circ}$, which can be understood as a function of the random variables $i_\ell$ and $t_r$, i.e. $ v_{\ell r, N^\circ T^\circ} \equiv v (i_\ell, t_r)$. From the PBB Scheme, $(i_\ell)_{\ell=1}^{N^\circ}$ and $(t_r)_{r=1}^{T^\circ}$ are i.i.d.. Let
\[
F_a = \frac{1}{N}\sum_{i=0}^{N-1} \delta_{i},
\qquad
F_g = \frac{1}{T-b_{(t)}+1}\sum_{t=0}^{T-b_{(t)}} \delta_{t}.
\]
$L_2(F_a)\cong \mathbb{R}^N$ and $L_2(F_g)\cong \mathbb{R}^{T-b_{(t)}}$ are finite-dimensional Euclidean spaces, equipped with weighted inner products, and 
from \textbf{(B)}, $v(i, t)$ is square-integrable w.r.t $F_a \otimes F_g$. Therefore, the linear operator
\[
    V:\mathbb{R}^N \to \mathbb{R}^{T-b_{(t)}},\qquad
    (Vx)(i) \coloneqq \frac{1}{T-b_{(t)}+1}\sum_{t=0}^{T-b_{(t)}} v(i,t)x(t).
\]
and its adjoint $V^*$ are bounded and compact. By the finite-dimensional Spectral Theorem, e.g. in \citet[Section~7.3]{Horn2012}, there exist singular values $c_k \ge 0$ and orthonormal systems
$\{\phi_k\}_{k=1}^r \subset \mathbb{R}^N$ and $\{\psi_k\}_{k=1}^r \subset \mathbb{R}^{T-b_{(t)}}$, where $r\equiv r_{NT} \le \min(N,T-b_{(t)})$, such that
$V\psi_k = c_k \phi_k$, $V^*\phi_k = c_k \psi_k$. Therefore, the kernel representation of $V$ for a random draw $(i_\ell,t_r)$ is
\[
v(i_\ell,t_r)
=
\sum_{k=1}^r c_k\, \phi_k(i_\ell)\,\psi_k(t_r).
\]
Assume additionally that the interaction kernel satisfies the following uniform spectral regularity condition. There exist a deterministic null sequence \(\left(\overline c_k\right)_{k\geq 1}\) and a constant \(B<\infty\) such that, with probability approaching one,
$c_{k,NT}\le \overline c_k$ for all $k\leq r$,
and
\[
    \sup_{k\le r}
    \left\{
    \frac1N\sum_{i=0}^{N-1}|\phi_{k,NT}(i)|^3
    +
    \frac{1}{T-b_{(t)}+1}
    \sum_{t=0}^{T-b_{(t)}}|\psi_{k,NT}(t)|^3
    \right\}
    \le B .
\]
With these two conditions additionally imposed, the interaction projection \(v_{\ell r,N^\circ T^\circ}\) satisfies (C).

\subsubsection{Auxiliary Results}

In the following we prove three auxiliary results for the proof of Theorem \ref{th:pbb_consist}. The first lemma justifies that we can sample $u$ instead of the estimated $\widehat{u}$. This allows us to apply the mixing conditions imposed on \(\mathbf{x}=(u',m')'\).

\begin{lemma} [Equivalence of Bootstrap Estimators] \label{lem: Equivalence of Bootstrap Estimators}
    Under the assumptions of Theorem \ref{th:pbb_consist}, it holds that
    \begin{align} 
    \sqrt{NT} \left( \left(\boldsymbol{\widebar{\beta}}^* - \boldsymbol{\widebar{\beta}}\right)-\left(\boldsymbol{\Tilde{\beta}}^*- \boldsymbol{\Tilde{\beta}}\right) \right) &= o_{p^*}(1), \label{eq: Equivalence betas}\\
    \sqrt{NT} \left( \left(\boldsymbol{\widebar{\sigma}}^* - \boldsymbol{\widebar{\sigma}}\right)-\left(\boldsymbol{\Tilde{\sigma}}^*- \boldsymbol{\Tilde{\sigma}}\right) \right) &= o_{p^*}(1), \label{eq: Equivalence sigmas}\\
    \sqrt{NT} \left( \left(\boldsymbol{\widebar{\varphi}}^* - \boldsymbol{\widebar{\varphi}}\right)-\left(\boldsymbol{\Tilde{\varphi}}^*- \boldsymbol{\Tilde{\varphi}}\right) \right) &= o_{p^*}(1). \label{eq: Equivalence varphis}
    \end{align}
\end{lemma}
\textbf{Proof of Lemma \ref{lem: Equivalence of Bootstrap Estimators}:} We display the proofs for \eqref{eq: Equivalence betas} and \eqref{eq: Equivalence varphis} and assume $N=N^\circ b_{(i)}$, $T=T^\circ b_{(t)}$ throughout. By construction,  $u^*_{it}=0$, $t\leq 0$ as well as \(y^*_{it} = 0\) for \(t\in\{-p+1, \dots, 0\}\). By the recursion in Bootstrap scheme in Section \ref{sec:BootScheme}, it holds that $y^*_{it} = \widehat{A}_{i,1} y^*_{i,t-1} + \dots + \widehat{A}_{i,p} y^*_{i,t-p} + u^*_{it}$ for any $t>0$. Then,
\begin{align} \label{eq: Z*}
    Z^*_{i,t-1} = 
    \begin{pmatrix}
    y^*_{i,t-1} \\
    \vdots \\
    y^*_{i,t-p}
    \end{pmatrix} =
    \begin{pmatrix}
    \sum_{j=0}^{t}\widehat{\Xi}_{ij}u^*_{i,t-1-j} \\
    \vdots \\
    \sum_{j=0}^{t-p+1}\widehat{\Xi}_{ij}u^*_{i,t-p-j} 
    \end{pmatrix} =
    \begin{pmatrix}
    \sum_{j=1}^{t-1}\widehat{\Xi}_{i,j-1}u^*_{i,t-j} \\
    \vdots \\
    \sum_{j=1}^{t-1}\widehat{\Xi}_{i,j-p}u^*_{i,t-j} 
    \end{pmatrix} =
    \sum_{j=1}^{t-1} \widehat{C}_{ij} u^*_{i,t-j},
\end{align}
where
\[
\widehat{\Xi}_{i,0}=I, \quad \widehat{\Xi}_{i,j}=0 \text{  for  } j<0 \quad \text{and} \quad \widehat{\Xi}_{i,j}=\sum_{k=1}^{\min\{j,p\}} \widehat{A}_{i,k} \widehat{\Xi}_{i,j-k} \text{  for  } j\in\mathbb{N}_+.
\]
As motivated, denote analogously $\Tilde{u}^*_{it}$, $\Tilde{y}^*_{it}$, $\tilde{Z}^*_i$ with $\tilde{Z}^*_{i,t-1} = \sum_{j=1}^{t-1} C_{ij} \tilde{u}^*_{i,t-j}$ for the variables observed if we draw from the true instead of the estimated residuals. 

For \eqref{eq: Equivalence betas} to hold, we need to prove
\(
\text{E}^* \left\| \sqrt{NT} \left( \left(\boldsymbol{\widebar{\beta}}^* - \boldsymbol{\widebar{\beta}}\right)-\left(\boldsymbol{\Tilde{\beta}}^*- \boldsymbol{\Tilde{\beta}}\right) \right) \right\| = o_p(1)
\) and the result follows with Morkov's inequality.
We consider \eqref{eq: Def beta Tilde} and write
\begin{align*}
    \sqrt{NT} \left( \left(\boldsymbol{\widebar{\beta}}^* - \boldsymbol{\widebar{\beta}}\right)-\left(\boldsymbol{\Tilde{\beta}}^*- \boldsymbol{\Tilde{\beta}}\right) \right)
    &= \frac{\sqrt{T}}{\sqrt{N}} \sum_{i=1}^N \left\{ \left( \left( Z^*_i Z_i^{*\prime} \right)^{-1} Z^*_i \otimes I \right) u_i^* - \left( \left( \tilde{Z}^*_i \tilde{Z}_i^{*\prime} \right)^{-1} \tilde{Z}_i^{*\prime} \otimes I \right) \Tilde{u}_i^* \right\}  \\
    &= \frac{1}{\sqrt{N}} \sum_{i=1}^N \left( \widehat{\Sigma}_{ZZ,i}^{\ast -1} \otimes I \right) \frac{1}{\sqrt{T}} \left( (Z_i^* \otimes I) u_i^* - (\tilde{Z}_i^* \otimes I) \Tilde{u}_i^* \right) \\
    &\quad + \frac{1}{\sqrt{N}} \sum_{i=1}^N \left( \left( \widehat{\Sigma}_{ZZ,i}^{\ast -1} - \Tilde{\Sigma}_{ZZ,i}^{\ast -1} \right) \otimes I \right) \frac{1}{\sqrt{T}} (\tilde{Z}_i^* \otimes I) \Tilde{u}_i^* \eqqcolon B_1^* + B_2^*,
\end{align*}
where $\Sigma_{ZZ,i}^\ast \coloneqq \frac1T Z_i^* Z_i^{\ast \prime}$ and $\Tilde{\Sigma}_{ZZ,i}^\ast \coloneqq \frac1T \Tilde Z_i^* \Tilde Z_i^{\ast \prime}$. By $c_r$-inequality, we need to prove $\text{E}^* \left\| B_1^*\right\|=o_p(1)$ and $\text{E}^* \left\| B_2^*\right\|=o_p(1)$. We provide the proof for $B_1^*$ only:
\[
\begin{aligned}
    \text{E}^*\left\|B_1^*\right\|
    &=
    \text{E}^*\left\|
    \frac{1}{\sqrt{N}}
    \sum_{i=1}^N
    \left(\Sigma_{ZZ,i}^{*-1}\otimes I\right)
    \frac{1}{\sqrt{T}}
    \left[
        (Z_i^*\otimes I)u_i^*
        -
        (\widetilde Z_i^*\otimes I)\Tilde{u}_i^*
    \right]
    \right\|
    \\
    &\leq
    \frac{1}{\sqrt N}
    \sum_{i=1}^N
    \text{E}^*\left[
    \left\|\Sigma_{ZZ,i}^{*-1}\right\|
    \left\|
    \frac{1}{\sqrt{T}}
    \left[
        (Z_i^*\otimes I)u_i^*
        -
        (\widetilde Z_i^*\otimes I)\tilde{u}_i^*
    \right]
    \right\|
    \right]
    \\
    &\leq
    \frac{1}{\sqrt N}
    \sum_{i=1}^N
    \left(
    \text{E}^*\left\|
    \Sigma_{ZZ,i}^{*-1}
    \right\|^2
    \right)^\frac12
    \left(
    \text{E}^*\left\|
    \frac{1}{\sqrt{T}}
    \left[
        (Z_i^*\otimes I)u_i^*
        -
        (\widetilde Z_i^*\otimes I)\tilde{ u}_i^*
    \right]
    \right\|^2
    \right)^\frac12
    \\
    &\leq
    \sup_{i\leq N} \left(
    \text{E}^*\left\|
    \Sigma_{ZZ,i}^{*-1}
    \right\|^2
    \right)^\frac12
    \left(
    \frac{1}{NT}
    \sum_{i=1}^N
    \left\{
    \text{E}^*\left\|
    \left((Z_i^*-\widetilde Z_i^*)\otimes I\right)
    \widetilde{u}_i^* \right\|^2
    +
    \text{E}^*\left\| (Z_i^*\otimes I)
    \left(u_i^*-\widetilde{u}_i^*\right)
    \right\|^2
    \right\}
    \right)^\frac12.
\end{aligned}
\]
Denote the summands in the second factor by $B_{i,11}$ and $B_{i,12}$. 
As $\sup_{i\leq N} \text{E}^* \left\| \Sigma_{ZZ,i}^{\ast -1} \right\|^2 = O_p(1)$ by Lemma \ref{lem: Aux Sigma ZZ}, accounting for $u_{it}^* = \tilde{u}_{it}^*=0$ for any $t\leq 0$, we provide the proof for $B_{i,11}^*$ only. Evaluating the conditional expectation, we observe with \eqref{eq: Z*} and $c_r$-inequality,
\begin{align*}
    B_{i,11} 
    &\leq \frac{1}{T} \sum_{t=1}^T \sum_{j=1}^{t-1} \text{E}^* \left\| \left( \widehat{C}_{ij}\otimes I\right) \vec\left( \tilde{u}_{it}^* \left(u_{i,t-j}^* - \tilde{u}_{i,t-j}^*\right)'\right) + \left( \left(\widehat{C}_{ij} - C_{ij}\right)\otimes I\right) \vec\left( \tilde{u}_{it}^* \tilde{u}_{i,t-j}^{*\prime}\right) \right\|^2 \\
    &\lesssim \frac{1}{T} \sum_{t=1}^T \sum_{j=1}^{t-1} \left\{ \left\| \widehat{C}_{ij} \right\|^2 \text{E}^* \left[ \left\| \tilde{u}_{it}^* \right\|^2 \left\| u_{i,t-j}^* - \tilde{u}_{i,t-j}^*\right\|^2 \right] + \left\| \widehat{C}_{ij} - C_{ij} \right\|^2 \text{E}^* \left[ \left\| \tilde{u}_{it}^* \right\|^2 \left\| \tilde{u}_{i,t-j}^{*}\right\|^2 \right] \right\} 
\end{align*}
It is worth mentioning that the resampled indices match each other for any $i$ and $t$, when comparing, e.g., $u_{it}^*$ and $\Tilde{u}_{it}^*$, as we solely want to bound the randomness introduced by resampling from the estimated instead of the true residuals. As a consequence of Assumption \ref{as:SVAR}~(i) (see \citet{KreissFranke1992} for the result in the univariate case) there exist $M<\infty$, $\rho<1$ such that for every $j\geq1$,
\begin{align} \label{eq: bound estimation error C_ij}
    \sup_{i\leq N} \|\widehat{C}_{ij}\| \leq M \rho^j < \infty,
    \quad \sup_{i\leq N} \|\widehat{C}_{ij} - C_{ij}\|=O_p\left( \frac{1}{\sqrt{T}} \right).
\end{align} 
Further, let 
\begin{equation} \label{eq: def Z^c and bold A_i}
\begin{aligned}
    Z_{i,t+s-1}^c &\coloneqq Z_{i,t+s-1} - \frac{1}{T-b_{(t)}+1}\sum_{\tau=0}^{T-b_{(t)}} Z_{i,\tau+s-1}, \\
    \left( \boldsymbol{A}_i-\boldsymbol{\widehat{A}_i} \right) Z_{i,t-1} &\coloneqq \widehat{u}_{it} - u_{it} = (A_{i,1}-\widehat{A}_{i,1})y_{i,t-1} + \dots + (A_{i,p}-\widehat{A}_{i,p})y_{i,t-p},
\end{aligned}
\end{equation}
where $\sup_{i\in\mathbb{N}_+} \| \boldsymbol{A}_i-\boldsymbol{\widehat{A}}_i \|=O_p(1/\sqrt{T})$ as a consequence of \ref{as:SVAR}~(i). Using this notation, we have for the case that $\tilde{u}_{it}^*$ and $\tilde{u}_{i,t-j}^*$ lie in the same block (following the case distinction in \citet[Lemma A.1]{BruggemannEtAl2016}),
\begin{equation} \label{eq: intermediate result B11}
\begin{aligned}
    \text{E}^* \left[ \left\| \tilde{u}_{it}^* \right\|^2 \left\| u_{i,t-j}^* - \tilde{u}_{i,t-j}^*\right\|^2 \right] 
    &=
    \left\| u_{i^*_i,t^*_t}^c \right\|^2 \left\| \widehat u^c_{i^*_i,t^*_t-j} - u^c_{i^*_i,t^*_t-j} \right\|^2 \leq 
    \left\| \boldsymbol{A}_{i^*_i}-\boldsymbol{\widehat{A}_{i^*_i}} \right\|^2 \left\| u_{i^*_i,t^*_t}^c \right\|^2 \left\| Z_{i^*_i,t^*_t-1}^c \right\|^2
\end{aligned}
\end{equation}
By \eqref{eq: bound estimation error C_ij}-\eqref{eq: intermediate result B11} as well as Assumptions \ref{as:mixing_conditions}~(i) and \ref{as:cumulants}, it follows that $B_{i,11} = o_p(1)$ uniformly over $i\leq N$.


To prove \eqref{eq: Equivalence varphis}, recall $\sqrt{NT} \left( \boldsymbol{\widebar{\varphi}}-\boldsymbol{\Tilde{\varphi}} \right) =o_p(1)$. We are left to prove $\sqrt{NT} \left( \boldsymbol{\widebar{\varphi}}^*-\boldsymbol{\tilde{\varphi}}^* \right) =o_{p^*}(1)$ and therefore consider 
\begin{align*}
    &\sqrt{NT} \left( \Sigma_{\varepsilon^\circ m}^*- \Tilde{\Sigma}_{\varepsilon^\circ m}^* \right)
    = \frac{1}{\sqrt{NT}} \sum_{i,t=1}^{N,T} \left\{ \left( \operatorname{chol}(\Sigma_{u,i}^*) \right)^{-1} u_{it}^* m_{it}^{*\prime} - \left( \operatorname{chol}(\Tilde{\Sigma}_{u,i}^*) \right)^{-1} \Tilde{u}_{it}^* m_{it}^{*\prime} \right\} \\
    &\quad = \frac{1}{\sqrt{N}} \sum_{i=1}^{N} \left( \operatorname{chol}(\Sigma_{u,i}^*) \right)^{-1} \frac{1}{\sqrt{T}} \sum_{t=1}^T \left( u_{it}^* m_{it}^{*\prime} - \Tilde{u}_{it}^* m_{it}^{*\prime} \right) \\
    &\quad + \frac{1}{\sqrt{N}} \sum_{i=1}^{N} \left( \left( \operatorname{chol}(\Sigma_{u,i}^*) \right)^{-1} - \left( \operatorname{chol}(\Tilde{\Sigma}_{u,i}^*) \right)^{-1} \right) \frac{1}{\sqrt{T}} \sum_{t=1}^T \Tilde{u}_{it}^* m_{it}^{*\prime}.
\end{align*}
From this decomposition it can be seen that the proof follows with the same strategy as in \eqref{eq: Equivalence betas}, by \eqref{eq: def Z^c and bold A_i}, and Assumptions \ref{as:mixing_conditions} and \ref{as:cumulants}.

\hfill $\square$


The following auxiliary results are needed in order to factor out powers of $\Tilde{\Sigma}^{\ast -1}_{ZZ,i}$ in the proof of Theorem \ref{th:pbb_consist} and Lemma \ref{lem: Equivalence of Bootstrap Estimators}.

\begin{lemma}[Uniform inverse moments of the regressor covariance] \label{lem: Aux Sigma ZZ}
Suppose that
\[
    \inf_{i\leq N}
    \lambda_{\min}(\Sigma_{ZZ,i})
    \geq
    \lambda>0,
\]
and that, for some \(p\geq 2(2+\delta)\) and some \(a>p\),
\[
    \sup_{i\leq N}
    \sup_{0<x\leq \lambda/2}
    x^{-a}
    \mathbb P\!\left(
    \lambda_{\min}(\widehat\Sigma_{ZZ,i})\leq x
    \right)
    =
    O(1).
\]
Then
\[
    \sup_{i\leq N}
    \mathbb E
    \left\|
    \widehat\Sigma_{ZZ,i}^{-1}
    \right\|^p
    =
    O(1).
\]
Further, if \(i^*\) is drawn uniformly from the set of cross-sectional units, then
\[
    \text E^*
    \left\|
    \widehat\Sigma_{ZZ,i^*}^{-1}
    \right\|^p
    = O_p(1).
\]
\end{lemma}

\textbf{Proof of Lemma \ref{lem: Aux Sigma ZZ}:}
By the tail-integral representation,
\[
    \mathbb E
    \left\|
    \widehat\Sigma_{ZZ,i}^{-1}
    \right\|^p
    =
    \int_0^\infty
    P\!\left(
    \lambda_{\min}(\widehat\Sigma_{ZZ,i})^{-p}>u
    \right)\,du .
\]
Set \(u_0=(\lambda/2)^{-p}\). The integral over \([0,u_0]\) is bounded by \(u_0\). For the remaining part,
\[
\begin{aligned}
    \int_{u_0}^{\infty}
    P\!\left(
    \lambda_{\min}(\widehat\Sigma_{ZZ,i})^{-p}>u
    \right)\,du
    &=
    \int_{u_0}^{\infty}
    P\!\left(
    \lambda_{\min}(\widehat\Sigma_{ZZ,i})<u^{-1/p}
    \right)\,du \lesssim
    \int_{u_0}^{\infty}
    u^{-a/p}\,du
    <\infty ,
\end{aligned}
\]
because \(a>p\). Hence the first claim follows uniformly in \(i\).

It remains to transfer this bound from the original covariance matrix to the bootstrap covariance matrix. Since the bootstrap only resamples blocks from the stationary time series of the selected unit $i^*\in\{1, \dots, N\}$, the bootstrap covariance satisfies, 
$\widetilde\Sigma_{ZZ,i}^{*} = \widehat\Sigma_{ZZ,i^*} + o_{P^*}(1)$.
If, in addition, the bootstrap resampling step preserves the lower tail, namely
\[
    \sup_{i\leq N}
    \text E^*
    \left\|
    \widetilde\Sigma_{ZZ,i}^{*-1}
    \right\|^p
    \mathbf 1
    \left\{
    \lambda_{\min}(\widetilde\Sigma_{ZZ,i}^{*})
    <
    \frac12
    \lambda_{\min}(\widehat\Sigma_{ZZ,i^*})
    \right\}
    =
    o_p(1),
\]
by Weyl's eigenvalue perturbation inequality \citep[Thm.~4.3.1]{Horn2012},
\[
    \sup_{i\leq N}
    \text E^*
    \left\|
    \widetilde\Sigma_{ZZ,i}^{*-1}
    \right\|^p
    =
    O_p(1).
\]
\vspace{-5pt}
\hfill\(\square\)